\documentclass[aps,a4paper,showkeys,nofootinbib,longbibliography,twocolumn]{revtex4-2} 
\usepackage{braket}
\usepackage[utf8]{inputenc}
\usepackage{filecontents}
\usepackage{natbib} 
\usepackage{amsmath,amssymb,bm,amsthm}
\usepackage{subfigure}
\usepackage{graphicx} 
\usepackage{dcolumn} 
\usepackage{bm} 
\usepackage[mathlines]{lineno} 
\usepackage{booktabs}
\usepackage{color}
\newcounter{one}
\usepackage{url}

\usepackage{textcase}

\usepackage{colortbl}
\usepackage{tabularx}
\usepackage{verbatim}
\usepackage{multirow}

\usepackage[T1]{fontenc} 
\usepackage{lmodern}
\usepackage{bbm}
\usepackage[utf8]{inputenc}
\usepackage{amsfonts}
\usepackage{array}

\usepackage[ruled,lined]{algorithm2e}

\usepackage{bbm}
\usepackage{enumitem}
\usepackage{umoline} 
\usepackage[usenames,svgnames]{xcolor}
\usepackage{natbib}
\usepackage[hyperindex,breaklinks]{hyperref} 
\hypersetup{
     colorlinks=true,   
     linkcolor=Navy,      
     citecolor=Navy,       
     filecolor=Navy, 
     urlcolor=Navy,  
    runcolor=cyan,
 } 
\usepackage{graphicx}
\usepackage{amsfonts}
\usepackage{amssymb}
\usepackage{amsmath}
\usepackage{bbm}
\usepackage{enumitem}
\usepackage{xcolor}
\usepackage{dsfont}

\newcommand{\half}[1]{{ \rm h}}
\newcommand{\Oorderof}{\mathcal{O}}
\newcommand{\orderof}[1]{\Oorderof(#1)}

\usepackage{yhmath}

\def\beq{\begin{equation}}
	\def\eeq{\end{equation}}
\def\nbeq{\begin{equation*}}
	\def\neeq{\end{equation*}} 
\def\<{\langle}
\def\>{\rangle}

\def\Tr{{\rm Tr}}

\newcommand{\mR}{{\mathcal{R}}}
\newcommand{\mE}{{\mathcal{E}}}

\def\Tr{{\rm Tr}}

\newtheorem{theorem}{Theorem}
\newtheorem{subtheorem}{Subtheorem}
\newtheorem{lemma}{Lemma}
\newtheorem{corollary}[lemma]{Corollary}

\newtheorem{prop}[subtheorem]{Proposition}  
\newtheorem{claim}[subtheorem]{Claim}

\newcommand{\br}[1]{\left( #1 \right)}
\newcommand{\brr}[1]{\left[ #1 \right]}

\newcommand{\norm}[1]{\left \|  #1 \right \|}

\newcommand{\abs}[1]{\left| #1 \right|}

 \newcommand{\hbin}{h_2}

\makeatletter

\newcommand{\supponecolumnfootnotes}{
	\onecolumn@grid@setup
	\let\set@footnotewidth\set@footnotewidth@one
	\let\compose@footnotes\compose@footnotes@one
}

\newcommand{\startsupplementarycontents}{
	\let\mainaddcontentsline\addcontentsline
	\renewcommand{\addcontentsline}[3]{
		\def\tempa{##1}
		\def\tempb{toc}
		\ifx\tempa\tempb
		\mainaddcontentsline{stoc}{##2}{##3}
		\else
		\mainaddcontentsline{##1}{##2}{##3}
		\fi
	}
}

\newcommand{\supplementarytableofcontents}{
	\begingroup
	\section*{\tocname}
	\@starttoc{stoc}
	\endgroup
}

\makeatother

\usepackage{mathtools}
\def\multiset#1#2{\ensuremath{\left(\kern-.3em\left(\genfrac{}{}{0pt}{}{#1}{#2}\right)\kern-.3em\right)}}

\renewcommand\thefootnote{*\arabic{footnote}} 

\begin{document}

\title{Collectivity limits quantum entanglement}

\author{Donghoon Kim$^1$}
\email{donghoon.kim@riken.jp}

\author{Tomotaka Kuwahara$^{1,2}$}
\email{tomotaka.kuwahara@riken.jp}

\affiliation{$^{1}$
	Analytical Quantum Complexity RIKEN Hakubi Research Team, RIKEN Center for Quantum Computing (RQC), Wako, Saitama 351-0198, Japan
}

\affiliation{$^{2}$
	RIKEN Pioneering Research Institute (PRI), Wako, Saitama 351-0198, Japan
}

\begin{abstract}
	
	Understanding what limits many-body quantum entanglement is a central problem in physics.
	Spatial locality has long provided a fundamental mechanism: correlations between a region and its complement must be mediated through a small spatial interface, thereby constraining their entanglement.
	Here we show that universal constraints on entanglement can persist even when interactions are strongly nonlocal, through collectivity: many weak interactions generate fluctuations controlled by their square-summed, rather than total, strength.
	We establish this mechanism rigorously for generic gapped Hamiltonians with Kac-normalized power-law interactions $r^{-\alpha}$ on a $D$-dimensional lattice.
	For arbitrary bipartitions, we prove that the ground-state entanglement scales at most logarithmically with system size for $\alpha<D/2$ and subextensively for $D/2<\alpha<D$, due to suppressed collective fluctuations around individual sites.
	For spatially regular bipartitions with codimension-one boundaries, we use a renormalization-group construction to extend this suppression to larger length scales, yielding polylogarithmic scaling for $D/2 < \alpha < (D+1) / 2$ and parametrically stronger subvolume bounds for $(D+1)/2<\alpha<D$.
	We also establish the corresponding optimality results: for arbitrary bipartitions, the logarithmic scaling for $\alpha<D/2$ is optimal and $\alpha=D/2$ marks the optimal threshold for universal logarithmic bounds, while for regular bipartitions the threshold at $\alpha=(D+1)/2$ is likewise optimal for universal polylogarithmic bounds when $D\ge2$.
	Together, these results reveal collectivity as a fundamental mechanism for constraining many-body entanglement alongside spatial locality.
	
\end{abstract}

\maketitle

Entanglement is a defining feature of quantum many-body systems and a central measure of their complexity~\cite{amico2008entanglement,horodecki2009quantum,laflorencie2016quantum}.
Its structure distinguishes quantum phases~\cite{osterloh2002scaling,vidal2003entanglement,kitaev2006topological,li2008entanglement,schuch2011classifying}, controls the efficiency of classical representations and numerical methods~\cite{vidal2003efficient,verstraete2008matrix,schollwock2011density,orus2014practical}, and determines the resources available for quantum information processing~\cite{horodecki2009quantum,bennett1996concentrating}.
For many-body ground states in particular, the scaling of entanglement with system size provides a fundamental link between microscopic interactions and the complexity of the resulting quantum state~\cite{verstraete2006matrix,hastings2007area,schuch2008entropy,eisert2010colloquium}.
Identifying the physical principles that constrain this scaling is therefore essential to understanding when strongly interacting quantum systems remain tractable.

Spatial locality has long served as a basic physical mechanism underlying such constraints~\cite{lieb1972finite,hastings2006spectral}.
When interactions are local, correlations between a region and its complement must be mediated through their spatial interface.
Together with a spectral gap, this geometric restriction gives rise to entanglement area laws in broad classes of quantum many-body systems~\cite{cramer2006entanglement,hastings2007area,eisert2010colloquium,arad2012improved,arad2013area,brandao2013area,brandao2015exponential,masanes2009area,gong2017entanglement,kuwahara2020area,AnshuAradGosset2022,abrahamsen2023entanglement,LiuYiZhouZou2025,KimKuwahara2026}, in which the ground-state entanglement of a region scales with its boundary rather than its volume.
These area laws underlie the efficient representation of low-energy states by tensor networks and much of the analytical and numerical understanding of gapped quantum systems~\cite{white1992density,white1993density,dukelsky1998equivalence,verstraete2008matrix,schollwock2011density,landau2015polynomial,arad2017rigorous,cirac2021matrix}.
Locality thus converts the geometry of interactions into a direct restriction on the complexity of many-body states.

The geometric mechanism above relies crucially on spatial locality, whereas long-range interactions progressively weaken this constraint by coupling degrees of freedom across increasing distances~\cite{defenu2023long}.
Once such interactions become sufficiently nonlocal, the boundary of a subsystem no longer directly controls the number or strength of interactions connecting it to its complement, and conventional area-law reasoning ceases to provide a general entanglement principle.
Indeed, entanglement constraints can fail on sufficiently nongeometric interaction structures~\cite{Aharonov2014}, while several nonlocal settings nevertheless exhibit strongly constrained ground-state entanglement~\cite{LatorreOrusRicoVidal2005,BarthelDusuelVidal2006,VidalDusuelBarthel2007,kim2024quantum}.
These contrasting observations leave open a broader question: \emph{what mechanism can impose rigorous constraints on ground-state entanglement in generic gapped systems as interactions become increasingly nonlocal?}

In this work, we identify collective suppression of quantum fluctuations as a mechanism that imposes rigorous constraints on ground-state entanglement beyond spatial locality.
We consider generic gapped two-local Hamiltonians on a $D$-dimensional lattice with Kac-normalized power-law interactions decaying with distance $r$ as $r^{-\alpha}$ for $0\le\alpha<D$, as illustrated in Fig.~\ref{fig:overview}(a), without assuming translation symmetry, permutation symmetry, or special operator structure.
Kac normalization keeps the total interaction strength per site finite while suppressing the fluctuation scale generated by the many individually weak interactions~\cite{kac1963van,defenu2023long}.
We first prove that the spectral gap converts this local fluctuation suppression into concentration of the Schmidt spectrum around individual sites.
For arbitrary bipartitions, this yields a logarithmic bound on the ground-state entanglement for $\alpha<D/2$, a squared-logarithmic bound at $\alpha=D/2$, and a subextensive bound for $D/2<\alpha<D$, without invoking any geometry of the bipartition.
We further show that the logarithmic scaling for $\alpha<D/2$ is optimal in its system-size dependence.
Moreover, $\alpha=D/2$ is the optimal threshold for universal logarithmic entanglement bounds across arbitrary bipartitions, since for every $\alpha>D/2$ there exist unique uniformly gapped examples with polynomially growing entanglement.

\begin{figure*}[t]
	\centering
	\includegraphics[width=\textwidth]{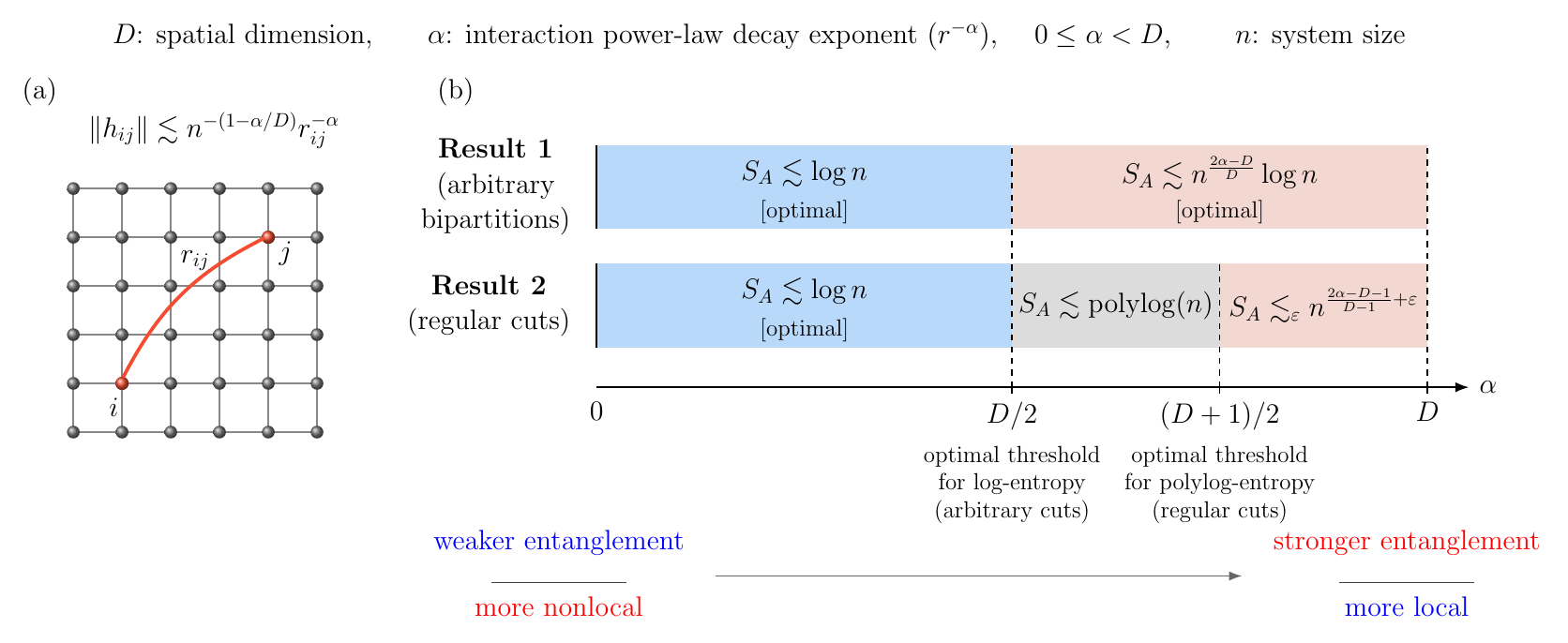}
	\caption{
		Model and overview of the entanglement bounds.
		(a) Kac-normalized power-law interactions on a $D$-dimensional lattice.
		(b) For arbitrary bipartitions (Result~1), the entanglement entropy $S_{A}(\ket{\Omega})$ is bounded logarithmically below the first threshold $\alpha=D/2$ and, above it, by the subvolume power $n^{(2\alpha-D)/D}$ up to a logarithmic factor.
		These bounds are optimal in their system-size dependence for $\alpha < D/2$ and $D / 2 < \alpha < D$, and $\alpha = D/2$ is the optimal threshold for universal logarithmic bounds across arbitrary bipartitions. 
		For regular codimension-one cuts (Result~2), such as half-system cuts or box-like regions, the entropy admits a polylogarithmic bound below the second threshold $\alpha=(D+1)/2$ and, above it, a subvolume power-law bound with exponent arbitrarily close to $(2\alpha-D-1)/(D-1)$.
		For $D \geq 2$, the second threshold is optimal for universal polylogarithmic bounds across regular codimension-one cuts.
	}
	\label{fig:overview}
\end{figure*}

Spatial geometry provides an additional source of suppression for regular bipartitions, such as half-system cuts or box-like regions with codimension-one boundaries.
For such bipartitions, we develop a renormalization-group construction that propagates the local collective suppression to successively larger length scales while keeping the effective interactions under control.
This multiscale mechanism improves the entanglement bound to polylogarithmic scaling for $D/2<\alpha<(D+1)/2$ and yields parametrically stronger subvolume bounds than for arbitrary bipartitions for $(D+1)/2<\alpha<D$, with subpolynomial scaling at $\alpha=(D+1)/2$.
For $D\ge2$, we also prove that the threshold $\alpha=(D+1)/2$ is optimal for universal polylogarithmic entanglement bounds across regular codimension-one cuts by constructing unique uniformly gapped examples with polynomially growing entanglement immediately above it.

The resulting entanglement bounds for arbitrary and regular bipartitions are summarized in Fig.~\ref{fig:overview}(b).
The first threshold, $\alpha=D/2$, is set by collective fluctuations around individual sites, whereas the second, $\alpha=(D+1)/2$, emerges from the corresponding suppression of fluctuations across codimension-one boundaries at larger length scales.
Our results thus establish collectivity as a distinct fundamental mechanism for constraining many-body entanglement, complementing the geometric constraints imposed by spatial locality.

\section*{Results}

\subsection*{Model and assumptions}

We consider a $D$-dimensional hypercubic lattice $\Lambda$ with $n=L^{D}$ sites, fixed local Hilbert-space dimension $d\ge2$, and either open or periodic boundary conditions.
The Hamiltonian is two-local,
\begin{equation}
	H = \sum_{i \in \Lambda} h_i + \sum_{i<j} h_{ij},
	\label{eq:main_model}
\end{equation}
with Kac-normalized power-law interactions~\cite{kac1963van,defenu2023long}
\begin{equation}
	\|h_{ij}\|
	\le
	\frac{J\mu_{ij}}{n^{1-\alpha/D}}r_{ij}^{-\alpha},
	\qquad
	0\le\alpha<D.
	\label{eq:main_interaction}
\end{equation}
Here $r_{ij}$ is the lattice distance between sites $i$ and $j$, and $\mu_{ij}=\mu_{0}$ when $i$ and $j$ are Moore neighbors, meaning that their coordinate displacement is at most one in every lattice direction, and $\mu_{ij}=1$ otherwise, with $\mu_{0}\ge1$.
We write $\beta:=1-\alpha/D>0$.
The factor $n^{-\beta}$ in Eq.~\eqref{eq:main_interaction} is the Kac normalization, which ensures that the interaction energy per site remains finite and the Hamiltonian is extensive despite the long-range connectivity.
We define the corresponding uniform local-energy bound by $\bar{g}:=\max_{i}\|h_{i}\|+J \max_{i}n^{-\beta} \sum_{j\neq i}\mu_{ij}r_{ij}^{-\alpha}$, so that $\|h_{i}\|+\sum_{j\neq i}\|h_{ij}\|\le\bar{g}$ for every $i$.
For bounded on-site terms, $\bar{g}=\mathcal{O}(1)$ uniformly in $n$.

We assume that $H$ has a unique normalized ground state $\ket{\Omega}$ with energy $E_{0}$ and a system-size-independent spectral gap $\Delta > 0$, so that
\begin{equation}
	H - E_{0}\mathds{1}
	\ge
	\Delta\bigl(\mathds{1}-\ket{\Omega}\bra{\Omega}\bigr).
	\label{eq:main_gap}
\end{equation}
All microscopic parameters are taken to be independent of the system size $n$.
No translation or permutation symmetry is assumed, and no special structure is imposed on the interaction operators.

Under these assumptions, our main question is how the ground-state entanglement scales with the system size.
For a subsystem $A \subseteq \Lambda$, let $A^{\mathrm{c}}:=\Lambda\setminus A$ and define the reduced state $\rho_A := \Tr_{A^{\mathrm{c}}}\br{\ket{\Omega}\bra{\Omega}}$.
The entanglement entropy across the bipartition $A|A^{\mathrm{c}}$ is $S_A(\ket{\Omega}):=-\Tr(\rho_A\log\rho_A)$.
We first derive a bound on $S_A(\ket{\Omega})$ for arbitrary bipartitions and then obtain substantially stronger bounds for geometrically regular cuts.

\subsection*{Result 1: Arbitrary-bipartition entanglement from collective suppression}

Our first result gives an entanglement bound for an arbitrary bipartition, without imposing any geometric condition on the cut.

\textbf{Theorem 1.}
For all sufficiently large $n$ and every bipartition $A|A^{\mathrm{c}}$,
\begin{equation}
	S_A(\ket{\Omega}) \le C_{\mathrm{arb}}
	\begin{cases}
		\log(en), & \alpha<D/2,\\ 
		[\log(en)]^{2}, & \alpha=D/2,\\ 
		n^{\frac{2\alpha-D}{D}}\log(en), & D/2<\alpha<D.
	\end{cases}
	\label{eq:main_thm_arbitrary}
\end{equation}
The constant is independent of $n$ and of the bipartition.
Its microscopic dependence may be chosen as $C_{\mathrm{arb}}=\mathcal{O}_{D,\alpha,d,\mu_0}\bigl[1+(J/\Delta)^{2}(\bar{g}/\Delta)\bigr]$.
The explicit prefactor is given in Methods~\ref{meth:microscopic}.

For an arbitrary bipartition, boundary counting alone provides no useful control, since a fragmented partition can separate an extensive number of interacting pairs.
What can be controlled uniformly instead is the local entanglement of each site with the rest of the system, encoded in its one-site Schmidt spectrum.
Accordingly, for each site $i\in\Lambda$, consider the bipartition $i\,|\,i^{\mathrm{c}}$ and write
\begin{equation}
	\ket{\Omega} = \sum_{a=0}^{d-1} \lambda_{a,i} \ket{a_i}\ket{\phi_{a,i}},
	\label{eq:main_schmidt}
\end{equation}
with Schmidt coefficients ordered as $\lambda_{0,i} \ge \lambda_{1,i} \ge \cdots \ge \lambda_{d-1,i} \ge 0$.
We define $q_{i} := 1 - \lambda_{0,i}^{2} = \sum_{a=1}^{d-1}\lambda_{a,i}^{2}$ and $q_{*}:=\max_i q_i$.
As illustrated in Fig.~\ref{fig:result1-schmidt}, the key structure is that $q_*$ becomes small with increasing system size, so that the one-site Schmidt spectrum is uniformly concentrated on its leading component: $\lambda_{0,i}^{2}\ge1-q_*$, while the total weight of all remaining components is at most $q_*$.
Such concentration directly controls the one-site entropies and, through subadditivity, the entanglement across an arbitrary bipartition.

The origin of this concentration is the suppression of collective fluctuations by the spectral gap.
For Hermitian one-site operators $f_j$, let $F:=\sum_j f_j$ and define $\operatorname{Var}_{\Omega}(F):=\bra{\Omega}F^{2}\ket{\Omega}-\bigl(\bra{\Omega}F\ket{\Omega}\bigr)^{2}$.
For the gapped two-local Hamiltonian considered here, one has~\cite{kuwahara2017local}
\begin{equation}
	\Delta \cdot \mathrm{Var}_{\Omega}(F) \le 4\bar{g} \sum_j \|f_j\|^{2}.
	\label{eq:main_variance_gap}
\end{equation}
The crucial feature is the $\ell_{2}$ square sum on the right-hand side. 
Thus, even when the total interaction strength remains of order one, the fluctuations in the sum of the interactions can decrease with system size.

To connect this fluctuation bound to the Schmidt spectrum, let $P_{i}:=\ket{0_{i}}\bra{0_{i}}$ project onto the leading Schmidt vector at site $i$, set $Q_{i}:=\mathds{1}-P_{i}$ and $u_{i}:=P_{i}-Q_{i}$, and define the reflection energy cost $\delta E_{i} :=  \bra{\Omega} u_{i}^{\dagger}Hu_{i} \ket{\Omega} - E_{0}$.
The reflection $u_{i}$ reverses the relative sign between the leading Schmidt component and its orthogonal complement.
For an upper bound on $\delta E_{i}$, the relevant energy change is controlled by the part of the interaction between site $i$ and the rest of the system that fluctuates around its ground-state expectation.
Equation~\eqref{eq:main_variance_gap} bounds precisely these fluctuations, while the amplitude in the orthogonal Schmidt sector is $\|Q_{i}\ket{\Omega}\|=\sqrt{q_{i}}$.
For the lower bound, $u_{i}$ is unitary and $\langle\Omega|u_{i}|\Omega\rangle=1-2q_{i}$, so applying the spectral-gap inequality~\eqref{eq:main_gap} to $u_{i}\ket{\Omega}$ gives $\delta E_{i}\ge\Delta\bigl[1-|\langle\Omega|u_{i}|\Omega\rangle|^{2}\bigr]=4\Delta q_{i}(1-q_{i})$.
Combining the fluctuation-based upper bound with the gap-based lower bound gives
\begin{align}
	4\Delta q_{i}(1-q_{i}) &\le \delta E_{i} \lesssim J\sqrt{\frac{\bar{g}}{\Delta}}\,\mathcal{L}_{\alpha}(n)\sqrt{q_{i}}, \label{eq:main_reflection_bootstrap} \\
	\mathcal{L}_{\alpha}(n) &:= \max_{i}\frac{1}{n^{\beta}}\left(\sum_{j\ne i}\mu_{ij}^{2}r_{ij}^{-2\alpha}\right)^{1/2}.
\end{align}
The factor $\sqrt{q_{i}}$ provides the self-consistency gain that drives the concentration of the one-site Schmidt spectrum; see Methods~\ref{meth:microscopic} for the detailed derivation.

\begin{figure}[t]
	\centering
	\includegraphics[width=\columnwidth]{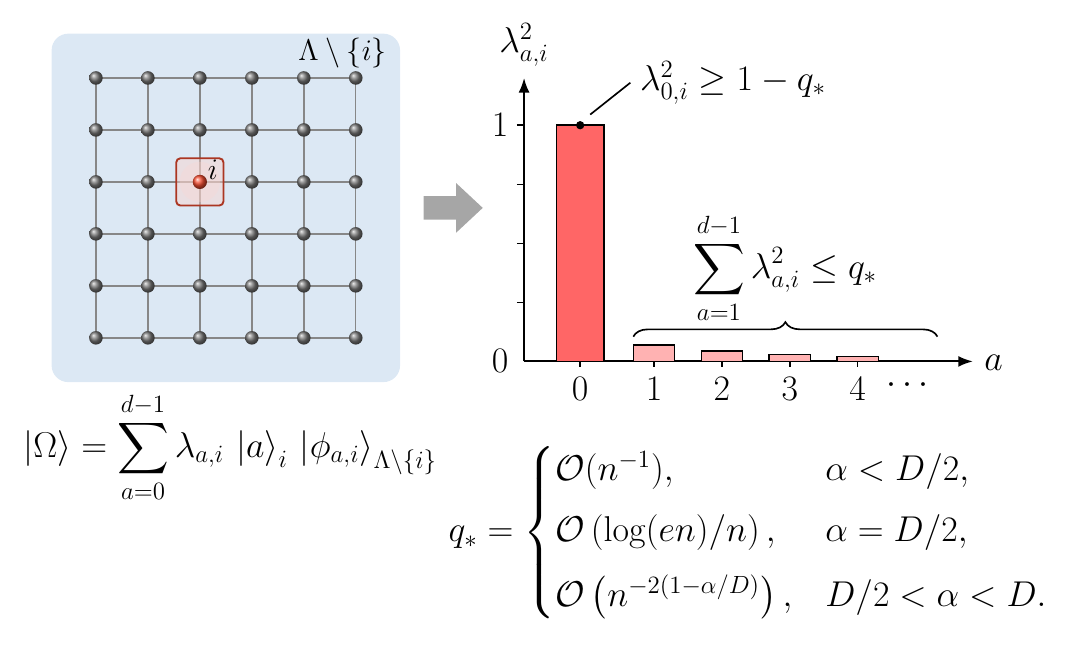}
	\caption{
		Collective suppression of the one-site Schmidt spectrum.
		The dominant Schmidt weight satisfies $\lambda_{0,i}^{2}\ge1-q_{*}$ at every site, while the total weight of all subleading Schmidt components is at most $q_{*}$.
		Throughout the Kac-normalized regime $0\le\alpha<D$, the maximal one-site Schmidt tail $q_{*}$ admits a system-size-dependent upper bound that vanishes as $n\to\infty$, so the local Schmidt spectrum becomes increasingly concentrated on its leading component; the corresponding regime-dependent scaling is given in Eq.~\eqref{eq:main_qstar_theorem}.
	}
	\label{fig:result1-schmidt}
\end{figure}

Because the one-site Hilbert space has fixed dimension $d$, the leading Schmidt weight satisfies $1-q_{i} = \lambda_{0,i}^{2} \ge1/d$.
Equation~\eqref{eq:main_reflection_bootstrap} therefore gives $q_{i} \lesssim \mathcal{L}_{\alpha}(n)^{2}$, so the system-size dependence of the Schmidt tail is determined by the square-summed interaction scale $\mathcal{L}_{\alpha}(n)$.
A direct lattice-counting estimate gives
\begin{equation}
	\mathcal{L}_{\alpha}(n)\lesssim
	\begin{cases}
		n^{-1/2}, & \alpha<D/2,\\ 
		\sqrt{\log(en)/n}, & \alpha=D/2,\\ 
		n^{-\beta}, & D/2<\alpha<D.
	\end{cases}
	\label{eq:main_Lalpha_scaling}
\end{equation}
The threshold $\alpha=D/2$ is precisely where the unnormalized square sum of the interaction strengths changes from power-law divergent to logarithmically divergent and then convergent.
After Kac normalization, $\mathcal{L}_{\alpha}(n)$ still vanishes for every fixed $\alpha<D$, but at different rates in the three regimes.
Consequently, for all sufficiently large $n$,
\begin{equation}
	q_{*}\le C_{\mathrm{mic}}
	\begin{cases}
		n^{-1}, & \alpha<D/2,\\ 
		\log(en)/n, & \alpha=D/2,\\ 
		n^{-2\beta}, & D/2<\alpha<D,
	\end{cases}
	\label{eq:main_qstar_theorem}
\end{equation}
where $C_{\mathrm{mic}} = \mathcal{O}_{D,\alpha,d,\mu_{0}}\brr{(J/\Delta)^{2}(\bar{g}/\Delta)}$.
Thus, throughout the full Kac-normalized regime $0\le\alpha<D$, the one-site Schmidt spectrum becomes increasingly concentrated on its leading component as the system size grows.

The one-site Schmidt concentration translates directly into the entropy bound in Theorem~1.
Each one-site reduced state satisfies $S(\rho_{i})\le h_{2}(q_{i})+q_{i}\log(d-1)$~\cite{nielsen2010quantum}, where $h_{2}(q):=-q\log q-(1-q)\log(1-q)$ is the binary entropy.
If $X$ denotes the smaller side of the bipartition, purity and subadditivity give $S_{A}(\ket{\Omega})=S_{X}(\ket{\Omega})\le\sum_{i\in X}S(\rho_{i})$~\cite{araki1970entropy}.
For sufficiently large $n$, the upper bound on $q_{*}$ in Eq.~\eqref{eq:main_qstar_theorem} is at most $1/2$, and since $q\mapsto h_{2}(q)+q\log(d-1)$ is increasing on $[0,1/2]$, we have $S_{A}(\ket{\Omega}) \le n\left[h_{2}(q_{*})+q_{*}\log(d-1)\right]$.
Substituting the corresponding upper bound on $q_{*}$ from Eq.~\eqref{eq:main_qstar_theorem} into this increasing function, and then using $h_{2}(q)\le q\log(e/q)$, gives precisely the entropy scalings stated in Theorem~1.

For arbitrary bipartitions, the bounds in Theorem~1 are optimal in their system-size dependence for $\alpha<D/2$ and $D/2<\alpha<D$, with $\alpha=D/2$ marking the optimal threshold for universal logarithmic entanglement bounds.
For $D/2<\alpha<D$, this optimality is realized by the dimer construction of Supplementary Proposition~1,
\begin{equation}
	H_{\mathrm{dimer}} = \sum_{i \in \Lambda} \frac{1-\sigma_{i}^{z}}{2} - n^{-\beta}\sum_{(i,j)\in\mathcal{M}}\sigma_{i}^{x}\sigma_{j}^{x},
	\label{eq:main_dimer_counterexample}
\end{equation}
where $\mathcal{M}$ is a nearest-neighbor perfect matching and $\sigma_{i}^{x}$ and $\sigma_{i}^{z}$ are the Pauli operators acting on site $i$.
Its ground state $\ket{\Omega_{\mathrm{dimer}}}$ is unique and uniformly gapped, while a bipartition containing exactly one endpoint of every dimer satisfies $S_{A}(\ket{\Omega_{\mathrm{dimer}}})=\Theta\br{n^{(2\alpha-D)/D}\log n}$.
This matches the third line of Theorem~1 throughout $D/2<\alpha<D$, showing that the arbitrary-bipartition scaling is optimal in this regime and that a universal logarithmic bound cannot extend beyond $\alpha=D/2$.
For $0\le\alpha<D/2$, Supplementary Proposition~13 constructs an explicit Kac-normalized power-law spin family with a unique uniformly gapped ground state satisfying $S_{A}(\ket{\Omega_{F}})\ge c_{D,\alpha}\log n$ across a regular codimension-one cut.
Together with the first line of Theorem~1, this shows that the logarithmic scaling is optimal throughout $0\le\alpha<D/2$.

\begin{figure*}[t]
	\centering
	\includegraphics[width=\textwidth]{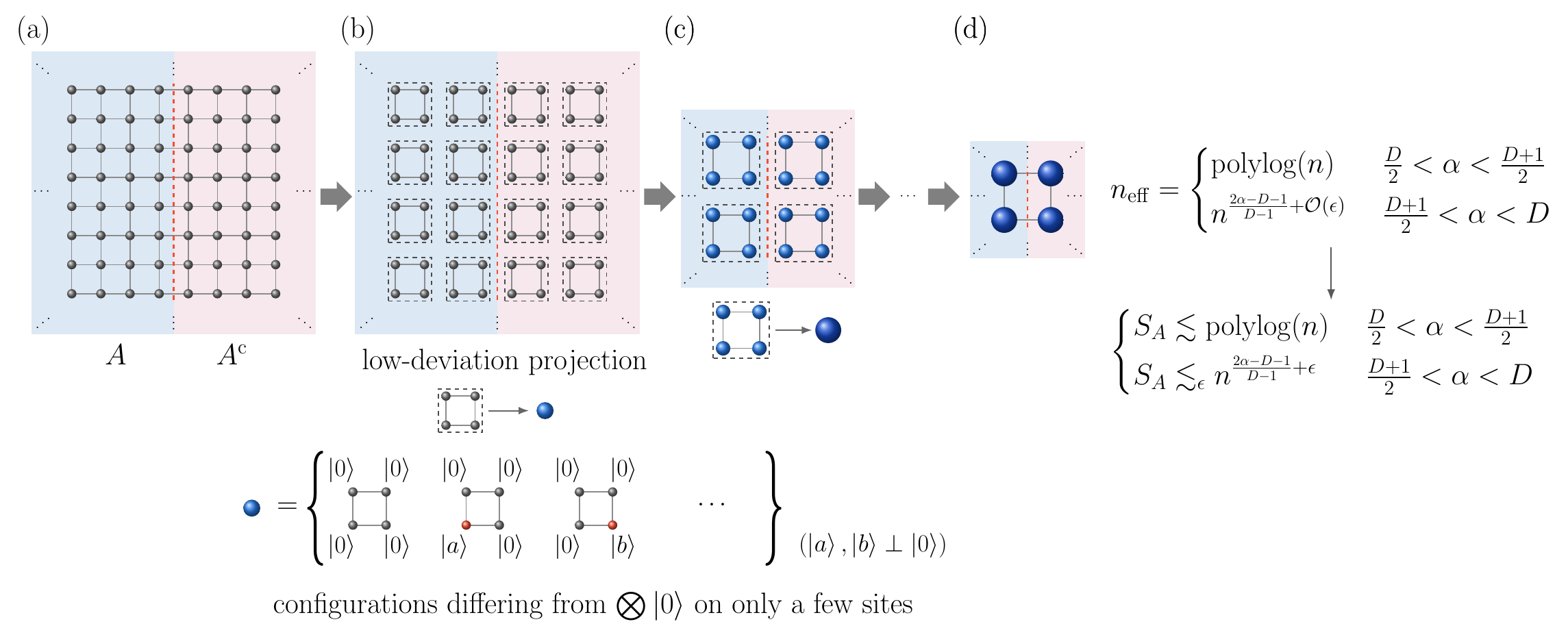}
	\caption{
		Multiscale collective suppression across a regular cut.
		(a) A microscopic lattice separated by a regular codimension-one cut.
		(b) The lattice is partitioned into blocks, and each block is restricted to a low-deviation sector.
		(c) The retained block space becomes one effective site on a coarse lattice, allowing the same construction to be iterated.
		(d) The coarse-graining procedure is continued up to a stopping scale at which the effective interaction and local-deviation bounds remain controlled.
		For $D/2<\alpha<(D+1)/2$, this reduces the effective system to
		$n_{\mathrm{eff}}=\operatorname{polylog}(n)$ sites.
		For $(D+1)/2<\alpha<D$, the controlled coarse-graining procedure reaches
		$n_{\mathrm{eff}}=n^{\frac{2\alpha-D-1}{D-1}+\mathcal{O}(\varepsilon)}$ effective sites.
		Together with the contributions accumulated from cut-crossing blocks at preceding scales, these effective sizes yield the regular-cut entropy bounds of Theorem~2.
	}
	\label{fig:result2}
\end{figure*}

\subsection*{Result 2: Entanglement across regular cuts from multiscale collective suppression}

We now consider geometrically regular bipartitions, such as half-system cuts or box-like regions, with codimension-one boundaries rather than highly fragmented cuts.
More precisely, for every $1 \le r \le L$, we require that the boundary of $A$ can be covered by at most $C_{A}(L/r)^{D-1}$ blocks of linear size $r$, where $C_{A}$ is independent of $n$.

\textbf{Theorem 2.}
Let $A|A^{\mathrm{c}}$ be a regular codimension-one bipartition satisfying the uniform covering condition above.
Then, for all sufficiently large $n$,
\begin{equation}
	S_{A}(\ket{\Omega})\le
	\begin{cases}
		C_{\mathrm{ent}}[\log(en)]^{\mathfrak{a}}, & D/2<\alpha<(D+1)/2,\\
		C_{\varepsilon}n^{\varepsilon}, & \alpha=(D+1)/2<D,\\
		C_{\varepsilon}n^{\frac{2\alpha-D-1}{D-1}+\varepsilon}, & (D+1)/2<\alpha<D.
	\end{cases}
	\label{eq:main_thm_regular}
\end{equation}
In the first line, $C_{\mathrm{ent}}$ and $\mathfrak{a}$ depend only on the fixed microscopic parameters and the regularity constant $C_{A}$, and are independent of $n$.
More explicitly, with $\mathcal{G}:=1+\bar{g}/\Delta+\mu_{0}J/\Delta+(J/\Delta)^{2}(\bar{g}/\Delta)+(\mu_{0}J/\Delta)^{2}\max\{1,\bar{g}/\Delta,\mu_{0}J/\Delta\}$, one may choose $C_{\mathrm{ent}} = \mathcal{O}_{D,\alpha,d,\mu_{0}} \bigl((1+C_{A})\mathcal{G}\bigr)$ and $\mathfrak{a} = \mathcal{O}_{D,\alpha,d,\mu_{0}} \bigl(1+\log\mathcal{G}\bigr)$.
At the marginal point, for every fixed $\varepsilon>0$, there are constants $C_{\varepsilon}$ and $n_{0}(\varepsilon)$, depending only on $\varepsilon$, the fixed microscopic parameters, and $C_{A}$, such that the second line holds for all $n \ge n_{0}(\varepsilon)$.
For $(D+1)/2 < \alpha < D$, the same dependence applies to the third line.
For $D\ge2$, the threshold $\alpha=(D+1)/2$ is optimal for universal polylogarithmic entanglement bounds across regular codimension-one cuts, as discussed below.
However, no optimality of the exponent $(2\alpha-D-1)/(D-1)$ in the upper bound above the threshold is asserted.

The improvement over Theorem~1 comes from combining the same local Schmidt suppression with the geometry of a regular bipartition.
For an arbitrary bipartition, the one-site contributions may accumulate over order-$n$ sites, producing the factor $nq_{*}$ in Theorem~1.
For a regular cut, grouping nearby sites becomes advantageous because only a boundary-sized number of blocks intersect both sides of the bipartition, whereas for an arbitrary bipartition, this number can remain extensive.
We therefore introduce an iterative coarse-graining procedure, which we refer to as the mean-field renormalization group (MFRG), closely related in spirit to real-space renormalization~\cite{kadanoff1966scaling,wilson1974renormalization,wilson1975renormalization, arad2017rigorous}.

To make the MFRG construction explicit, we introduce the effective degrees of freedom and the quantities propagated from one scale to the next.
At scale $s$, let $\Lambda^{(s)}$ denote the effective lattice and $n^{(s)}:=|\Lambda^{(s)}|$ its number of sites.
We write $H^{(s)}$ for the effective Hamiltonian and $\ket{\Omega^{(s)}}$ and $\Delta^{(s)}$ for its ground state and spectral gap, with $\Lambda^{(0)}=\Lambda$, $n^{(0)}=n$, and $H^{(0)}=H$.
For each effective site $i\in\Lambda^{(s)}$, let $\rho_{i}^{(s)}$ be its reduced state and let $P_{i}^{(s)}:=\ket{0_{i}^{(s)}}\bra{0_{i}^{(s)}}$ project onto its leading eigenvector.
We define $Q_{i}^{(s)}:=\mathds{1}-P_{i}^{(s)}$, $q_{s,i}:=\bra{\Omega^{(s)}}Q_{i}^{(s)}\ket{\Omega^{(s)}}$, and $q_{s,*}:=\max_{i}q_{s,i}$.

To pass from scale $s$ to $s+1$, neighboring effective sites are grouped into blocks.
For a block $L\subseteq\Lambda^{(s)}$ containing $b_{s+1}$ effective sites, define
\begin{equation}
	M_{L}^{(s)}:=\sum_{i\in L}Q_{i}^{(s)}, \qquad \Pi_{L,\le z}^{(s)}:=\mathds{1}_{[0,z]} \br{M_{L}^{(s)}}.
\end{equation}
Thus, $M_{L}^{(s)}$ counts the number of sites in $L$ outside their leading local sectors, while $\Pi_{L,\le z}^{(s)}$ retains the subspace containing at most $z$ such deviations.
Each retained block space is identified with a single effective site of $\Lambda^{(s+1)}$, and the same construction is repeated at the next scale; see Fig.~\ref{fig:result2}.

The small deviation probability $q_{s,*}$ implies a small mean deviation number in each block, but an accurate truncation requires control of the entire high-deviation tail.
The spectral gap and two-locality provide this stronger control by suppressing transitions to sectors with increasingly many deviations.
Let $\epsilon_{s}$ denote the one-step ground-state error associated with the low-deviation truncation at scale $s$.
Crucially, the coarse-graining parameters can be chosen so that the cutoff $z$ is independent of $n$ while keeping $\epsilon_{s}$ small and preserving a nonzero effective spectral gap.
The corresponding truncation and gap estimates are given in Methods~\ref{meth:truncation}.

A further key ingredient is that the collective suppression remains stable under coarse graining.
At scale $s$, the effective pair interactions retain the power-law form
\begin{equation}
	\|h_{ij}^{(s)}\|\le\frac{J^{(s)}}{[n^{(s)}]^{\beta}}\mu_{ij}^{(s)}r_{ij}^{-\alpha},
\end{equation}
where $J^{(s)}$ is the effective interaction scale, $r_{ij}$ is the distance on $\Lambda^{(s)}$, and $\mu_{ij}^{(s)}$ accounts for the short-distance enhancement generated under coarse graining.
To isolate the fluctuating part, the interaction is centered with respect to the local reference vectors $\{\ket{0_{i}^{(s)}}\}_{i\in\Lambda^{(s)}}$, with the resulting one-sided terms absorbed into the on-site terms.
This reorganization changes the full Hamiltonian only by an additive scalar, leaving its ground state and spectral gap unchanged.
For two blocks $B$ and $B'$, let $\acute{H}_{B,B'}^{(s)}$ denote the resulting centered interaction between them.
If the retained sectors of $B$ and $B'$ contain at most $z$ and $z'$ deviations, respectively, we show that
\begin{align}
	&\left\|\Pi_{B,\le z}^{(s)}\Pi_{B',\le z'}^{(s)}\acute{H}_{B,B'}^{(s)}\Pi_{B,\le z}^{(s)}\Pi_{B',\le z'}^{(s)}\right\| \nonumber\\
	&\qquad \qquad \qquad \qquad \qquad \quad  \le \frac{16J^{(s)}\sqrt{zz'}}{[n^{(s)}]^{\beta}}\gamma_{B,B'}^{(2\alpha,s)}.
	\label{eq:main_block_l2}
\end{align}
Here $\gamma_{B,B'}^{(2\alpha,s)}:=\left(\sum_{i\in B,\,j\in B'}[\mu_{ij}^{(s)}]^{2}r_{ij}^{-2\alpha}\right)^{1/2}$ is the square-sum measure of the coupling between the two blocks.
Thus, the effective interaction between retained blocks is governed by the square sum of the underlying couplings rather than by the sum of all bond strengths.
This is the block-level counterpart of the collective suppression in Result~1, and the bound remains independent of the growing effective local Hilbert-space dimension.

For neighboring blocks of linear size $\ell$, the square sum has the scaling
\begin{equation}
	\gamma_{B,B'}^{(2\alpha,s)}\lesssim
	\begin{cases}
		\ell^{D-\alpha}, & 2\alpha<D+1,\\
		\ell^{(D-1)/2}\sqrt{\log(e\ell)}, & 2\alpha=D+1,\\
		\mu_{s}\ell^{(D-1)/2}, & 2\alpha>D+1,
	\end{cases}
	\label{eq:main_block_geometry}
\end{equation}
where $\mu_{s}$ denotes the short-distance enhancement accumulated under coarse graining.
For $2\alpha<D+1$, the factor $\ell^{D-\alpha}$ is exactly absorbed by the change in the Kac normalization, since $\ell^{D-\alpha}/[n^{(s)}]^{\beta}=1/[n^{(s+1)}]^{\beta}$ using $D\beta=D-\alpha$ and $n^{(s+1)}=n^{(s)}/\ell^{D}$.
At $\alpha=(D+1)/2$, an additional factor $\sqrt{\log(e\ell)}$ remains.
For $\alpha>(D+1)/2$, the Kac rescaling instead leaves the positive factor $\ell^{\alpha-(D+1)/2}$ in the short-distance interaction, which grows under coarse graining.
This change in the block-interface scaling identifies the second threshold at $\alpha=(D+1)/2$.
The corresponding block square-sum estimates and scale-dependent interaction bounds are summarized in Methods~\ref{meth:interactions}.

The regularity of the bipartition enters through a separate counting argument.
Each scale-$s$ effective site corresponds to a block of sites in the original lattice, and we let $\mathcal{C}_{s}(A)$ denote the effective sites whose underlying blocks intersect both $A$ and $A^{\mathrm{c}}$.
We write
\begin{equation}
	|\mathcal{C}_{s}(A)|\le C_{A}[n^{(s)}]^{\theta}, \qquad 0\le\theta<1.
	\label{eq:main_crossing_general}
\end{equation}
For the regular codimension-one cuts in Theorem~2, the uniform covering condition gives $\theta=(D-1)/D$.
More generally, the geometric condition needed for the polylogarithmic coarse-graining construction is $\theta<2\beta$.

The small Schmidt tail $q_{s,*}$ must remain small under each coarse-graining step.
If a block contains $b_{s+1}$ scale-$s$ sites, the total weight outside the product of their leading local sectors is at most $b_{s+1}q_{s,*}$, while the one-step coarse-graining error contributes an additional $\epsilon_{s}$.
Consequently, the dominant weight of the corresponding scale-$(s+1)$ effective site remains close to one whenever $b_{s+1}q_{s,*}+\epsilon_{s}$ is small.
This supplies the input needed to apply the same gap--variance bootstrap as in Result~1 to the scale-$(s+1)$ effective Hamiltonian.
For $D/2<\alpha<(D+1)/2$, the resulting induction gives, for a fixed constant $C$ independent of $n$,
\begin{equation}
	q_{s,*}\lesssim C^{s}[n^{(s)}]^{-2\beta}, \quad \Delta^{(s)}\ge\frac{\Delta}{2}, \quad \epsilon_{s}\le[n^{(s)}]^{-p}.
	\label{eq:main_RG_summary}
\end{equation}
Here $p\ge3$ is fixed independently of $n$.
The effective-site bootstrap and its iteration across successive coarse-graining steps are described in Methods~\ref{meth:bootstrap}--\ref{meth:closure}.
Thus, at every effective scale, the Schmidt tail retains the same decay exponent $2\beta$ with respect to the current effective system size $n^{(s)}$, up to the scale-dependent prefactor $C^{s}$.

For $D/2<\alpha<(D+1)/2$, the contribution to the entanglement entropy generated at each coarse-graining scale is controlled by the effective regions that intersect both sides of the bipartition.
Their number satisfies $|\mathcal{C}_{s}(A)|\le C_{A}[n^{(s)}]^{\theta}$ by Eq.~\eqref{eq:main_crossing_general}, while the local weights entering the entropy estimate are controlled by the small Schmidt tail $q_{s,*}\lesssim C^{s}[n^{(s)}]^{-2\beta}$ from Eq.~\eqref{eq:main_RG_summary}.
When $\theta<2\beta$, this suppression is strong enough to keep the total contribution from these regions under control throughout repeated coarse graining.
For a codimension-one bipartition, $\theta=(D-1)/D$, so this condition is equivalent to $\alpha<(D+1)/2$.
As the coarse graining proceeds, additional scale-dependent factors accumulate, and the induction can be closed as long as these factors remain dominated by the residual powers of $n^{(s)}$.
This control extends down to a polylogarithmic effective size, which is reached after $\mathcal{O}(\log\log n)$ coarse-graining steps and leaves $n_{\mathrm{eff}}=\mathrm{polylog}(n)$ effective sites.

The reduction to $n_{\mathrm{eff}}=\operatorname{polylog}(n)$ effective sites provides the terminal contribution to the entropy bound, while the contributions generated at the preceding coarse-graining scales are controlled by the scale-by-scale estimate above.
To recover the entropy of the exact ground state, we use the fact that the retained spaces are nested, so that $\ket{\Omega}$ can be decomposed into the component retained to the final scale and components discarded at the preceding steps.
The weights of the discarded components are controlled by the truncation errors $\epsilon_{s}$.
Combining the polylogarithmic terminal contribution with the controlled contributions from the preceding scales gives $S_{A}(\ket{\Omega})\le C_{\mathrm{ent}}[\log(en)]^{\mathfrak{a}}$, which is the first line of Theorem~2.
The corresponding entropy reconstruction is summarized in Methods~\ref{meth:entropy}.

At $\alpha=(D+1)/2$, the coarse-graining construction becomes marginal.
For a codimension-one bipartition, one has $\theta=2\beta$, while the block interaction acquires the logarithmic enhancement in Eq.~\eqref{eq:main_block_geometry}.
These two effects prevent the same growing-depth construction used below the threshold, although a fixed number of coarse-graining steps can still be controlled.
Choosing a finite number of steps depending on a fixed $\varepsilon>0$ gives $S_{A}(\ket{\Omega})\le C_{\varepsilon}n^{\varepsilon}$, which is the second line of Theorem~2.

For $(D+1)/2<\alpha<D$, Eq.~\eqref{eq:main_block_geometry} shows that the short-distance part of the effective interaction grows under coarse graining, so the present construction can no longer be continued to arbitrarily small effective sizes.
If $N_{s}:=n/n^{(s)}$ denotes the number of original lattice sites combined into each scale-$s$ effective site, this growth is measured by the factor $N_{s}^{(D-1)/(2D)}/n^{\beta}$.
The coarse graining is therefore stopped while this factor remains sufficiently small for the effective interaction to remain controlled, leaving $n_{\mathrm{eff}} = n^{\frac{2\alpha-D-1}{D-1}+\mathcal{O}(\varepsilon)}$ effective sites.
Combining this terminal contribution with the controlled contributions from the preceding scales gives $S_{A}(\ket{\Omega}) \lesssim_{\varepsilon} n^{\frac{2\alpha-D-1}{D-1}+\varepsilon}$, recovering the third line of Theorem~2.

The optimality of the regular-cut threshold at $\alpha=(D+1)/2$ for $D\ge2$ follows from Supplementary Proposition~13, which constructs an explicit Kac-normalized power-law spin family that, for every $(D+1)/2 < \alpha < (3D+1)/4$, has a unique uniformly gapped ground state $\ket{\Omega_{F}}$ satisfying $S_{A}(\ket{\Omega_{F}})\ge c_{D,\alpha}n^{(2\alpha-D-1)/D}\log n$ across a flat codimension-one cut.
Since $(2\alpha-D-1)/D>0$ in this regime, polynomially growing regular-cut entanglement is therefore possible for exponents arbitrarily close to $\alpha=(D+1)/2$ from above.
Thus, universal polylogarithmic entanglement bounds across regular codimension-one cuts cannot extend into any right neighborhood of the threshold.

\section*{Discussion}

Our results establish collective suppression as a distinct mechanism for constraining ground-state entanglement beyond spatial locality.
In the Kac-normalized regime, each site can interact with a macroscopic fraction of the system while the total interaction strength incident on that site remains of order one.
The relevant small quantity is instead the collective fluctuation of these interactions, whose square-summed strength decreases with system size.
The spectral gap converts this collective suppression into concentration of the local Schmidt spectrum, without assuming permutation symmetry or an a priori reduction to a small number of collective modes.

Our results reveal two distinct entanglement thresholds arising from the interplay between collective suppression and geometry.
The threshold $\alpha=D/2$ is set by the square sum of the couplings incident on a single site and marks the point below which collective suppression yields a logarithmic entropy bound for arbitrary bipartitions.
This threshold is optimal for universal logarithmic bounds, and the arbitrary-bipartition entropy bounds are themselves optimal in their system-size dependence for $\alpha < D/2$ and $D/2 < \alpha < D$.
For regular cuts, multiscale coarse-graining shifts the relevant square-sum geometry from a single site to codimension-one interfaces, leading to the second threshold $\alpha=(D+1)/2$.
Below the second threshold, collective suppression remains stable across a growing number of RG scales, leading to a polylogarithmic entropy bound for regular cuts.
For $D \geq 2$, this threshold is also optimal for universal polylogarithmic bounds, since polynomially growing regular-cut entanglement is possible arbitrarily close to it from above.
Beyond this threshold, the attainable coarse-graining depth is reduced, resulting in weaker but still subvolume entanglement bounds.

Several questions remain open.
Most immediately, the exponent $(2\alpha-D-1)/(D-1)$ obtained above the second threshold is not claimed to be optimal, and determining the sharp regular-cut scaling in this regime remains an open problem.
More broadly, the central role of square-sum suppression in the present proof motivates asking whether the mechanism extends beyond Kac-normalized power-law interactions to broader classes of interaction profiles for which this suppression remains sufficiently strong and stable under coarse graining.
It is also natural to ask whether collective suppression can constrain entanglement beyond gapped ground states, including low-energy, thermal, and dynamically generated states, and under what conditions analogous entanglement bounds can be established.
These questions point toward sharpening the quantitative entanglement bounds established here and clarifying the broader range of systems in which collective suppression can operate.

\section*{Methods}
\renewcommand{\thesubsection}{M\arabic{subsection}}
\setcounter{subsection}{0}

\subsection{Gap--variance bootstrap for Schmidt concentration}
\label{meth:microscopic}

We first derive the microscopic Schmidt-tail bound used in Result~1 and later as the starting point for the coarse-graining construction in Result~2.
The idea is to reflect a single site about its dominant Schmidt vector and estimate the resulting energy cost in two complementary ways.
The spectral gap gives a lower bound proportional to the Schmidt tail $q_{i}$, whereas the fluctuation of the interaction connecting that site to the rest of the system gives an upper bound proportional to $\sqrt{q_{i}}$.
Comparing the two bounds yields a self-consistent estimate for $q_{i}$ whose system-size dependence is determined by the square sum of the interaction strengths.
We begin with the gap--variance estimate that controls the fluctuation entering the upper bound.

Let $K=\sum_{i}k_{i}+\sum_{i<j}k_{ij}$ have a unique normalized ground state $\ket{\omega}$, spectral gap $\delta>0$, and local energy scale $g:=\max_{i}\left(\|k_{i}\|+\sum_{j\neq i}\|k_{ij}\|\right)$.
For Hermitian one-site operators $a_{i}$ and $A:=\sum_{i}a_{i}$, the gap inequality gives~\cite{kuwahara2017local}
\begin{equation}
	\delta \cdot \mathrm{Var}_{\omega}(A) \le \frac{1}{2}\bra{\omega}[A,[K,A]]\ket{\omega}.
\end{equation}
Because $K$ is two-local, the double commutator reduces to terms involving only the one-site operators acting on the support of each Hamiltonian term, yielding
\begin{equation}
	\delta \cdot \mathrm{Var}_{\omega} \left(\sum_{i}a_{i}\right) \le 4g\sum_{i}\|a_{i}\|^{2}.
	\label{eq:methods_variance_gap}
\end{equation}
The square sum on the right-hand side is the key feature: when the $a_{i}$ are chosen from the interactions incident on a fixed site, it converts the many interaction terms into a collective fluctuation scale that controls the upper bound on the reflection energy.
A detailed derivation of Eq.~\eqref{eq:methods_variance_gap} is given in Supplementary Lemma~8.

To apply Eq.~\eqref{eq:methods_variance_gap} to the microscopic ground state $\ket{\Omega}$, we first introduce the local reflection whose energy cost will be estimated.
For a fixed site $i$, let $P_{i}:=\ket{0_{i}}\bra{0_{i}}$ project onto the leading eigenvector of the one-site reduced state $\rho_{i}$, and define $Q_{i}:=\mathds{1}-P_{i}$ and $u_{i}:=P_{i}-Q_{i}$.
The corresponding Schmidt tail is $q_{i}:=\bra{\Omega}Q_{i}\ket{\Omega}$.
We define the reflection energy cost by
\begin{equation}
	\delta E_{i}:=\bra{\Omega}u_{i}^{\dagger}Hu_{i}\ket{\Omega}-E_{0}.
	\label{eq:methods_reflection_energy}
\end{equation}

To identify the terms that contribute to $\delta E_{i}$, write
\begin{equation}
	H=h_{i}+H_{i,i^{\mathrm{c}}}+H_{i^{\mathrm{c}}},
	\label{eq:methods_site_decomposition}
\end{equation}
where
\begin{equation}
	H_{i,i^{\mathrm{c}}}:=\sum_{j\neq i}h_{ij}
\end{equation}
denotes the interaction between site $i$ and its complement $i^{\mathrm{c}}$, with $h_{ji}:=h_{ij}$ when $j<i$, and $H_{i^{\mathrm{c}}}$ contains all on-site and interaction terms supported entirely on $i^{\mathrm{c}}$.
Since $u_{i}$ acts only on site $i$, it commutes with $H_{i^{\mathrm{c}}}$.
Moreover, because $P_{i}$ is a spectral projector of $\rho_{i}$, one has $[u_{i},\rho_{i}]=0$, and hence $\bra{\Omega}u_{i}^{\dagger}h_{i}u_{i}\ket{\Omega}=\bra{\Omega}h_{i}\ket{\Omega}$.
Therefore the reflection energy cost is determined entirely by the interaction $H_{i,i^{\mathrm{c}}}$.

Choose a fixed Hermitian operator basis $\{\mathcal{X}_{i,v}\}_{v=1}^{d^{2}}$ on site $i$ and expand
\begin{equation}
	H_{i,i^{\mathrm{c}}}=\sum_{v=1}^{d^{2}}\mathcal{X}_{i,v}\otimes B_{i,v}, \qquad B_{i,v}:=\sum_{j\neq i}b_{i,j,v},
	\label{eq:methods_connected_expansion}
\end{equation}
where each $b_{i,j,v}$ is Hermitian and supported on site $j$, with $\|\mathcal{X}_{i,v}\|\le1$ and $\|b_{i,j,v}\|\le\|h_{ij}\|$.
Since $P_{i}$ is a spectral projector of $\rho_{i}$, one has $[P_{i},\rho_{i}]=0$ and hence
\begin{equation}
	\bra{\Omega} P_{i}\mathcal{X}_{i,v}Q_{i} \ket{\Omega} = \Tr_{i}\left(\rho_{i}P_{i}\mathcal{X}_{i,v}Q_{i}\right) = 0.
\end{equation}
Therefore, replacing $B_{i,v}$ by $B_{i,v} - c_{v} \mathds{1}$ for any $c_{v} \in \mathbb{R}$ leaves the corresponding off-diagonal matrix element unchanged.

Using $u_{i}=P_{i}-Q_{i}$ and $P_{i}+Q_{i}=\mathds{1}$, while the on-site term $h_{i}$ and $H_{i^{\mathrm{c}}}$ give no change under the reflection, we obtain
\begin{align}
	\delta E_{i}
	&=-4\operatorname{Re}\bra{\Omega}P_{i}H_{i,i^{\mathrm{c}}}Q_{i}\ket{\Omega} \notag\\
	&=4\left|\operatorname{Re}\bra{\Omega}P_{i}H_{i,i^{\mathrm{c}}}Q_{i}\ket{\Omega}\right| \notag\\
	&\le4\sum_{v=1}^{d^{2}}\left|\bra{\Omega}P_{i}\mathcal{X}_{i,v}Q_{i}(B_{i,v}-c_{v}\mathds{1})\ket{\Omega}\right|,
\end{align}
where the second equality uses $\delta E_{i}\ge0$, and each $c_{v}\in\mathbb{R}$ is arbitrary by the scalar subtraction above.
Since $B_{i,v}-c_{v}\mathds{1}$ acts only on $i^{\mathrm{c}}$, it commutes with $P_{i}\mathcal{X}_{i,v}Q_{i}$, and the Cauchy--Schwarz inequality together with $\|\mathcal{X}_{i,v}\|\le1$ and $\|Q_{i}\ket{\Omega}\|=\sqrt{q_{i}}$ gives
\begin{equation}
	\delta E_{i}\le4\sqrt{q_{i}}\sum_{v=1}^{d^{2}}\inf_{c\in\mathbb{R}}\|(B_{i,v}-c\mathds{1})\ket{\Omega}\|.
	\label{eq:methods_reflection_fluctuation}
\end{equation}
Since $B_{i,v}$ is Hermitian, $\inf_{c\in\mathbb{R}}\|(B_{i,v}-c\mathds{1})\ket{\Omega}\|=\sqrt{\operatorname{Var}_{\Omega}(B_{i,v})}$, and hence
\begin{equation}
	\delta E_{i}\le4\sqrt{q_{i}}\sum_{v=1}^{d^{2}}\sqrt{\operatorname{Var}_{\Omega}(B_{i,v})}.
\end{equation}
Applying Eq.~\eqref{eq:methods_variance_gap} to $B_{i,v}=\sum_{j\neq i}b_{i,j,v}$ and using $\|b_{i,j,v}\|\le\|h_{ij}\|\le J\mu_{ij}n^{-\beta}r_{ij}^{-\alpha}$ gives
\begin{equation}
	\sqrt{\operatorname{Var}_{\Omega}(B_{i,v})}\le2J\sqrt{\frac{\bar{g}}{\Delta}}\frac{1}{n^{\beta}}\left(\sum_{j\neq i}\mu_{ij}^{2}r_{ij}^{-2\alpha}\right)^{1/2}.
	\label{eq:methods_connected_variance}
\end{equation}
Writing $\mathcal{L}_{\alpha}(n):=\max_{i}n^{-\beta}\left(\sum_{j\neq i}\mu_{ij}^{2}r_{ij}^{-2\alpha}\right)^{1/2}$, we obtain
\begin{equation}
	\delta E_{i}\le8d^{2}J\sqrt{\frac{\bar{g}}{\Delta}}\,\mathcal{L}_{\alpha}(n)\sqrt{q_{i}}.
	\label{eq:methods_reflection_upper}
\end{equation}

The complementary lower bound follows directly from the spectral gap.
Since $u_{i}$ is unitary and $\bra{\Omega}u_{i}\ket{\Omega}=1-2q_{i}$, Eq.~\eqref{eq:main_gap} gives
\begin{equation}
	\delta E_{i}\ge\Delta\left(1-\left|\bra{\Omega}u_{i}\ket{\Omega}\right|^{2}\right)=4\Delta q_{i}(1-q_{i}).
	\label{eq:methods_reflection_lower}
\end{equation}
The leading one-site eigenvalue satisfies $1-q_{i}\ge1/d$, so combining Eqs.~\eqref{eq:methods_reflection_upper} and \eqref{eq:methods_reflection_lower} first gives $q_{i}=\mathcal{O}(\mathcal{L}_{\alpha}(n)^{2})$ and hence $q_{i}<1/2$ for all sufficiently large $n$.
Using $1-q_{i} \ge 1/2$ in Eq.~\eqref{eq:methods_reflection_lower} and applying the same upper bound then yields
\begin{equation}
	q_{i} \le 16d^{4}\left(\frac{J}{\Delta}\right)^{2} \frac{\bar{g}}{\Delta}\,\mathcal{L}_{\alpha}(n)^{2}.
	\label{eq:methods_q_bound}
\end{equation}

The remaining system-size dependence follows from the geometric square sum in $\mathcal{L}_{\alpha}(n)$.
A direct lattice-counting estimate gives
\begin{equation}
	\mathcal{L}_{\alpha}(n)^{2} \le
	\begin{cases}
		\Xi_{<}\,n^{-1}, & \alpha<D/2,\\
		\Xi_{=}\,\log(en)/n, & \alpha=D/2,\\ 
		\Xi_{>}\,n^{-2\beta}, & D/2<\alpha<D,
	\end{cases}
	\label{eq:methods_L_explicit}
\end{equation}
where
\begin{equation}
	\Xi_{\alpha}:=
	\begin{cases}
		a_{D}D^{D-2\alpha} \dfrac{D-2\alpha+1}{D-2\alpha}+\xi_{\mu_{0}}, & \alpha<D/2,\\[6pt]
		a_{D}(1+\log D)+\xi_{\mu_{0}}, & \alpha=D/2,\\[6pt]
		a_{D} \dfrac{2\alpha-D+1}{2\alpha-D}+\xi_{\mu_{0}}, & D/2<\alpha<D,
	\end{cases}
	\label{eq:methods_Xi}
\end{equation}
with $a_{D} := 2^{D} D^{D-1}$ and $\xi_{\mu_{0}}:=(3^{D}-1)(\mu_{0}^{2}-1)$, and with $\Xi_{<}$, $\Xi_{=}$, and $\Xi_{>}$ denoting the corresponding three branches of $\Xi_{\alpha}$.
Substituting Eq.~\eqref{eq:methods_L_explicit} into Eq.~\eqref{eq:methods_q_bound} gives
\begin{equation}
	q_{*}\le C_{\mathrm{mic}}
	\begin{cases}
		n^{-1}, & \alpha<D/2,\\ 
		\log(en)/n, & \alpha=D/2,\\ 
		n^{-2\beta}, & D/2<\alpha<D,
	\end{cases}
	\label{eq:methods_qstar}
\end{equation}
where
\begin{equation}
	C_{\mathrm{mic}} := 16d^{4} \left(\frac{J}{\Delta}\right)^{2} \frac{\bar{g}}{\Delta}\,\Xi_{\alpha}.
	\label{eq:methods_Cmic}
\end{equation}

For sufficiently large $n$, the upper bound on $q_{*}$ in Eq.~\eqref{eq:methods_qstar} is at most $1/2$, so substituting this bound into the increasing function $q\mapsto h_{2}(q)+q\log(d-1)$ and then using $h_{2}(q)\le q\log(e/q)$ gives the three entropy scalings in Theorem~1.
The corresponding prefactor may be chosen as
\begin{equation}
	C_{\mathrm{arb}} := [1+\log(d-1)]\max\{1,C_{\mathrm{mic}}\}.
	\label{eq:methods_Carb}
\end{equation}
Consequently, $C_{\mathrm{arb}} = \mathcal{O}_{D,\alpha,d,\mu_{0}}\bigl[1+(J/\Delta)^{2}(\bar{g}/\Delta)\bigr]$ by Eq.~\eqref{eq:methods_Cmic}.

\subsection{Low-deviation truncation and gap preservation}
\label{meth:truncation}

We next show that each coarse-graining block can be restricted to a low-deviation subspace while keeping the ground-state error small and preserving a nonzero spectral gap.
At scale $s$, let $\Lambda^{(s)}$ denote the effective lattice, let $n^{(s)}:=|\Lambda^{(s)}|$ be its number of sites, let $d^{(s)}$ denote the largest local Hilbert-space dimension, and let $H^{(s)}$ denote the effective Hamiltonian obtained after $s$ coarse-graining steps, with $H^{(0)}=H$.
Write $H^{(s)}=\sum_{i}h_{i}^{(s)}+\sum_{i<j}h_{ij}^{(s)}$, with a unique normalized ground state $\ket{\Omega^{(s)}}$ and spectral gap $\Delta^{(s)}>0$, and let $\bar{g}^{(s)}$ denote a uniform local-energy bound satisfying $\|h_{i}^{(s)}\|+\sum_{j\neq i}\|h_{ij}^{(s)}\|\le\bar{g}^{(s)}$ for every $i$.
For each effective site $i\in\Lambda^{(s)}$, let $P_{i}^{(s)} := \ket{0_{i}^{(s)}}\bra{0_{i}^{(s)}}$ project onto the leading eigenvector of its reduced state, define $Q_{i}^{(s)}:=\mathds{1}-P_{i}^{(s)}$, and write $q_{s,*}:=\max_{i}\bra{\Omega^{(s)}}Q_{i}^{(s)}\ket{\Omega^{(s)}}$.
For a set of effective sites $C$, define the deviation-number operator
\begin{equation}
	M_{C}^{(s)} := \sum_{i\in C}Q_{i}^{(s)}.
\end{equation}
Its spectral projections are
\begin{equation}
	\begin{aligned}
		\Pi_{C,m}^{(s)} &:= \mathds{1}_{\{m\}}\bigl(M_{C}^{(s)}\bigr),\\
		\Pi_{C,\le m}^{(s)} &:= \mathds{1}_{[0,m]}\bigl(M_{C}^{(s)}\bigr),\\
		\Pi_{C,>m}^{(s)} &:= \mathds{1}_{(m,\infty)}\bigl(M_{C}^{(s)}\bigr).
	\end{aligned}
	\label{eq:methods_deviation_projections}
\end{equation}
For a coarse-graining block $L$, we retain the low-deviation sector $\operatorname{Ran}\Pi_{L,\le z}^{(s)}$.
The key question is whether the high-deviation weight $\|\Pi_{L,>z}^{(s)}\ket{\Omega^{(s)}}\|$ can be made sufficiently small without allowing the cutoff $z$ to grow with the microscopic system size.

To expose the fluctuating part of the interaction, define $\mathcal{E}_{i}^{(s)}(X) := \bra{0_{i}^{(s)}} X \ket{0_{i}^{(s)}} \otimes \mathds{1}_{i}$, with the tensor factors placed in their original order, and set $\mathcal{R}_{i}^{(s)} := \mathrm{Id} - \mathcal{E}_{i}^{(s)}$.
The centered pair interaction
\begin{equation}
	\acute{h}_{ij}^{(s)} := \mathcal{R}_{i}^{(s)}\mathcal{R}_{j}^{(s)}\bigl(h_{ij}^{(s)}\bigr)
	\label{eq:methods_centered_pair}
\end{equation}
satisfies $P_{i}^{(s)} \acute{h}_{ij}^{(s)} P_{i}^{(s)} = P_{j}^{(s)}\acute{h}_{ij}^{(s)}P_{j}^{(s)} = 0$.
Absorbing the resulting one-sided pieces into the on-site terms gives a centered Hamiltonian $\acute{H}^{(s)}$ that differs from $H^{(s)}$ only by a scalar and therefore has the same ground state and spectral gap.

The mean deviation number in a block is at most $|L|q_{s,*}$, but this alone does not provide the accuracy needed to truncate the high-deviation sector.
We therefore first control the deviation tail in the complementary region $L^{\mathrm{c}}$.
Using the spectral gap and two-locality, one can choose a cutoff $m_{\mathrm{c},s}$ such that
\begin{equation}
	\left\|\Pi_{L^{\mathrm{c}},>m_{\mathrm{c},s}}^{(s)}\ket{\Omega^{(s)}}\right\| \le \epsilon_{\mathrm{c},s},
\end{equation}
with
\begin{equation}
	m_{\mathrm{c},s} \le 2n^{(s)}q_{s,*} + \sqrt{\frac{2\bar{g}^{(s)}n^{(s)}}{\Delta^{(s)}}} \log\frac{2}{\epsilon_{\mathrm{c},s}} +2.
	\label{eq:methods_complement_cutoff}
\end{equation}
The detailed derivation of this complement-tail bound is given in Supplementary Note 2.
Let $\ket{\widehat{\overline{\Omega}}^{(s)}}$ denote the ground state after this complement restriction, embedded back into the scale-$s$ Hilbert space, and define $\eta_{\mathrm{c},s} := \|\ket{\Omega^{(s)}}-\ket{\widehat{\overline{\Omega}}^{(s)}}\|$.
For a sufficiently small choice of $\epsilon_{\mathrm{c},s}$, the complement restriction produces a small state error $\eta_{\mathrm{c},s}$ while preserving a gap of at least $\Delta^{(s)}/2$.

We now bound directly the ground-state weight in sectors of $L$ containing more than $z$ deviations.
Two-locality provides the key structure: the Hamiltonian can change the deviation number in $L$ by at most two at a time.
Let $m_{*}$ be the smallest integer such that $\|\Pi_{L,\le m_{*}}^{(s)}\ket{\widehat{\overline{\Omega}}^{(s)}}\|^{2}\ge1/2$.
Since $\|\Pi_{L,>0}^{(s)}\ket{\Omega^{(s)}}\|\le\sqrt{|L|q_{s,*}}$, we have $m_{*}=0$ whenever $\sqrt{|L|q_{s,*}}+\eta_{\mathrm{c},s}\le1/\sqrt{2}$, so that at least half of the auxiliary ground-state weight lies in the zero-deviation sector.
We then control how strongly the Hamiltonian connects a sector with $m$ deviations to the sectors with $m+1$ and $m+2$ deviations.
Let $\overline{J}_{L}[0,z]$ be a uniform bound on these couplings for $0 \le m \le z$, and let $\overline{\Delta}_{\mathrm{c},s}$ denote the gap after the complement restriction.
Defining $\rho_{z} := \overline{J}_{L}[0,z] / \overline{\Delta}_{\mathrm{c},s}$, the small-hopping estimate applies when $\sqrt{|L|q_{s,*}} + \eta_{\mathrm{c},s}\le1/\sqrt{2}$ and $\rho_{z}<1$.
Under these conditions,  
\begin{equation}
	\left\| \Pi_{L,>z}^{(s)}\ket{\Omega^{(s)}} \right\| \le \eta_{\mathrm{c},s} + \rho_{z}^{z/8}.
	\label{eq:methods_shell_tail}
\end{equation}
The first term is the error introduced by the complement restriction, while the second is small when the hopping between successive deviation sectors is sufficiently weak compared with the restricted gap.
Methods~\ref{meth:closure} verifies the required small-hopping condition uniformly over the active coarse-graining scales in each regime.
Consequently, the high-deviation weight can be made uniformly small with a cutoff $z$ independent of $n$ over the corresponding active scales.

For the step $s\to s+1$, let $\mathcal{B}_{s+1}$ be the block partition and define the global retained projector
\begin{equation}
	\mathsf{P}_{s\rightarrow s+1}:=\prod_{L\in\mathcal{B}_{s+1}}\Pi_{L,\le z_{s+1}}^{(s)}.
	\label{eq:methods_global_retained_projection}
\end{equation}
For each block $L\in\mathcal{B}_{s+1}$, let $V_{s+1,L}$ be a coisometry from the scale-$s$ Hilbert space of $L$ onto the corresponding scale-$(s+1)$ effective-site Hilbert space, with $V_{s+1,L}^{\dagger}V_{s+1,L}=\Pi_{L,\le z_{s+1}}^{(s)}$.
Define the global blockwise coisometry by $V_{s+1}:=\bigotimes_{L\in\mathcal{B}_{s+1}}V_{s+1,L}$, so that $V_{s+1}^{\dagger}V_{s+1}=\mathsf{P}_{s\rightarrow s+1}$ and $V_{s+1}V_{s+1}^{\dagger}=\mathds{1}^{(s+1)}$, and define
\begin{equation}
	H^{(s+1)}:=V_{s+1}\acute{H}^{(s)}V_{s+1}^{\dagger}.
	\label{eq:methods_effective_hamiltonian}
\end{equation}
Let $E_{0}^{(s)}$ denote the ground-state energy of $H^{(s)}$, define the discarded component $\ket{\xi_{s}}:=(\mathds{1}-\mathsf{P}_{s\rightarrow s+1})\ket{\Omega^{(s)}}$, and write $w_{s}:=\|\ket{\xi_{s}}\|^{2}$.
The corresponding energy defect is $e_{s}:=\bra{\xi_{s}}(H^{(s)}-E_{0}^{(s)}\mathds{1})\ket{\xi_{s}}/(1-w_{s})$.
Writing $\ket{\Omega^{(s+1)}}$ for the normalized ground state of $H^{(s+1)}$ and choosing its phase appropriately, define
\begin{equation}
	\epsilon_{s}:=\left\|\ket{\Omega^{(s)}}-V_{s+1}^{\dagger}\ket{\Omega^{(s+1)}}\right\|.
	\label{eq:methods_step_error}
\end{equation}
Because $\acute{H}^{(s)}$ differs from $H^{(s)}$ only by a scalar, it has the same ground state and spectral gap.
Applying the compression estimate of Supplementary Lemma~6 to the retained space gives
\begin{equation}
	\Delta^{(s+1)}\ge\Delta^{(s)}-e_{s}, \qquad \epsilon_{s}\le\sqrt{\frac{2e_{s}}{\Delta^{(s)}}}.
	\label{eq:methods_step_gap}
\end{equation}
Methods~\ref{meth:closure} shows that the discarded weight $w_{s}$, and hence the energy defect $e_{s}$, remains sufficiently small at every active scale, so that the ground-state error is controlled and the effective spectral gap remains nonzero.

\subsection{Effective interactions under coarse graining}
\label{meth:interactions}

We now ask how the interaction between two blocks changes after restricting each block to its low-deviation sector.
The low-deviation constraint allows the resulting block interaction to be controlled by the square-summed strength of the couplings connecting the two blocks.
At scale $s$, let $\acute{h}_{ij}^{(s)}$ denote the centered pair term obtained from $h_{ij}^{(s)}$ as in Eq.~\eqref{eq:methods_centered_pair}.
For two disjoint blocks $B,B'\subseteq\Lambda^{(s)}$, define their centered interaction by
\begin{equation}
	\acute{H}_{B,B'}^{(s)}:=\sum_{i\in B}\sum_{j\in B'}\acute{h}_{ij}^{(s)}.
	\label{eq:methods_centered_block_interaction}
\end{equation}
Assume that the scale-$s$ pair interactions satisfy
\begin{equation}
	\|h_{ij}^{(s)}\|\le\frac{J^{(s)}}{[n^{(s)}]^{\beta}}\mu_{ij}^{(s)}r_{ij}^{-\alpha}.
	\label{eq:methods_stage_kac_bound}
\end{equation}
Here $\mu_{ij}^{(s)}=\mu_{s}$ when $j$ lies in the coarse Moore neighborhood of $i$, that is, when the blocks represented by $i$ and $j$ share a face, edge, or corner, and $\mu_{ij}^{(s)}=1$ otherwise, with $\mu_{s}\ge1$.
Accordingly, we may choose
\begin{equation}
	\bar{g}^{(s)}:=\max_{i}\|h_{i}^{(s)}\|+J^{(s)}\max_{i}[n^{(s)}]^{-\beta}\sum_{j\neq i}\mu_{ij}^{(s)}r_{ij}^{-\alpha},
\end{equation}
which satisfies $\|h_{i}^{(s)}\|+\sum_{j\neq i}\|h_{ij}^{(s)}\|\le\bar{g}^{(s)}$ for every $i$.
Define the block square sum by
\begin{equation}
	\gamma_{B,B'}^{(2\alpha,s)}:=\left(\sum_{i\in B}\sum_{j\in B'}[\mu_{ij}^{(s)}]^{2}r_{ij}^{-2\alpha}\right)^{1/2}.
	\label{eq:methods_block_gamma}
\end{equation}
If the retained sectors of $B$ and $B'$ contain at most $z$ and $z'$ deviations, respectively, centering removes the purely reference-sector contribution from each pair interaction, so that its matrix elements are controlled by the deviation amplitudes on the two blocks.
Using the low-deviation constraints and applying the Cauchy--Schwarz inequality to the resulting sum over pairs gives
\begin{align}
	&\left\|\Pi_{B,\le z}^{(s)} \Pi_{B',\le z'}^{(s)} \acute{H}_{B,B'}^{(s)} \Pi_{B,\le z}^{(s)} \Pi_{B',\le z'}^{(s)}\right\| \nonumber\\
	&\qquad \qquad \qquad \qquad \qquad \le \frac{16J^{(s)}\sqrt{zz'}}{[n^{(s)}]^{\beta}} \gamma_{B,B'}^{(2\alpha,s)}.
	\label{eq:methods_block_square_sum}
\end{align}
Thus, the interaction between the retained blocks is controlled by the square-summed coupling strength $\gamma_{B,B'}^{(2\alpha,s)}$, rather than by the sum of all pair-interaction strengths, with no dependence on the effective local Hilbert-space dimension.
The detailed derivation is given in Supplementary Proposition~2.

The geometric behavior of $\gamma_{B,B'}^{(2\alpha,s)}$ determines how the interaction changes after blocking.
For two neighboring hypercubic blocks of side length $\ell$ on the scale-$s$ lattice, counting the pairs at each distance from the block interface gives
\begin{equation}
	\gamma_{B,B'}^{(2\alpha,s)}\lesssim
	\begin{cases}
		\ell^{D-\alpha}, & 2\alpha<D+1,\\[2pt]
		\ell^{(D-1)/2}\sqrt{\log(e\ell)}, & 2\alpha=D+1,\\[2pt]
		\mu_{s}\ell^{(D-1)/2}, & 2\alpha>D+1,
	\end{cases}
	\label{eq:methods_block_geometry}
\end{equation}
where the last line includes the short-distance enhancement of coarse neighboring sites.
For blocks that are not coarse Moore neighbors, the short-distance enhancement is absent, and if $R$ denotes their coarse-lattice distance,
\begin{equation}
	\gamma_{B,B'}^{(2\alpha,s)}\lesssim\ell^{D-\alpha}R^{-\alpha}.
	\label{eq:methods_separated_blocks}
\end{equation}
Thus non-neighboring effective sites retain the original power-law tail, while coarse Moore neighbors can acquire an additional short-distance factor.
The detailed geometric counting is given in Supplementary Lemma~9.

Let $\ell_{s+1}$ be the block side length and $z_{s+1}$ the deviation cutoff used in the step $s\to s+1$, so that
\begin{equation}
	b_{s+1}:=\ell_{s+1}^{D}, \qquad n^{(s+1)}=\frac{n^{(s)}}{b_{s+1}}.
	\label{eq:methods_size_recursion}
\end{equation}
For $D/2<\alpha\le(D+1)/2$, we absorb the microscopic short-distance factor into $J_{\mathrm{K}}:=\mu_{0}J$ and use the reparametrized stage-$0$ bound with $J^{(0)}=J_{\mathrm{K}}$ and $\mu_{ij}^{(0)}=1$.
For $(D+1)/2<\alpha<D$, we instead retain the original stage-$0$ parameters $J^{(0)}=J$ and $\mu_{0}$.
The square-sum estimate in Eq.~\eqref{eq:methods_block_geometry} determines how the interaction changes under one coarse-graining step.
For $D/2< \alpha <(D+1)/2$, the neighboring-block square sum contributes a factor $\ell_{s+1}^{D-\alpha}$, and since $n^{(s+1)} = n^{(s)}/\ell_{s+1}^{D}$ and $D\beta = D-\alpha$, one has
\begin{equation}
	\frac{\ell_{s+1}^{D-\alpha}}{[n^{(s)}]^{\beta}}
	=
	\frac{1}{[n^{(s+1)}]^{\beta}}.
	\label{eq:methods_kac_scale_covariance}
\end{equation}
Thus, the block-size dependence is absorbed exactly into the Kac normalization and no additional power of $\ell_{s+1}$ is generated.
At $\alpha=(D+1)/2$, the same cancellation removes the factor $\ell_{s+1}^{D-\alpha}$, while the additional factor $\sqrt{\log(e\ell_{s+1})}$ in Eq.~\eqref{eq:methods_block_geometry} remains and accumulates in the effective interaction strength.
For $(D+1)/2<\alpha<D$, Eq.~\eqref{eq:methods_block_geometry} instead gives a neighboring-block contribution proportional to $\mu_{s}\ell_{s+1}^{(D-1)/2}$, and
\begin{equation}
	\frac{\ell_{s+1}^{(D-1)/2}}{[n^{(s)}]^{\beta}}
	=
	\frac{\ell_{s+1}^{\alpha-(D+1)/2}}{[n^{(s+1)}]^{\beta}}.
	\label{eq:methods_mu_scale_covariance}
\end{equation}
Hence each coarse-graining step generates the additional short-distance factor $\ell_{s+1}^{\alpha-(D+1)/2}$.

For successive steps, define
\begin{equation}
	\begin{aligned}
		\ell_{1:s}:=\prod_{t=1}^{s}\ell_{t}, \quad Z_{s}:=\prod_{t=1}^{s}z_{t}, \quad 
		\chi_{s}:=\prod_{t=1}^{s}\sqrt{\log(e\ell_{t})}.
	\end{aligned}
	\label{eq:methods_cumulative_factors}
\end{equation}
Iterating these one-step behaviors, with a constant $c_{\mathrm{int}}\ge1$ depending only on $D$ and $\alpha$, gives
\begin{equation}
	\mu_{s}
	=
	\begin{cases}
		1, & D/2<\alpha\le(D+1)/2,\\ 
		\mu_{0}\ell_{1:s}^{\alpha-(D+1)/2}, & (D+1)/2<\alpha<D,
	\end{cases}
	\label{eq:methods_mu_recursion}
\end{equation}
and
\begin{equation}
	J^{(s)}
	=
	\begin{cases}
		c_{\mathrm{int}}^{s}J_{\mathrm{K}}Z_{s}, & D/2<\alpha<(D+1)/2,\\ 
		c_{\mathrm{int}}^{s}J_{\mathrm{K}}Z_{s}\chi_{s}, & \alpha=(D+1)/2<D,\\ 
		c_{\mathrm{int}}^{s}JZ_{s}, & (D+1)/2<\alpha<D.
	\end{cases}
	\label{eq:methods_J_recursion}
\end{equation}
The detailed derivation of these interaction recursions is given in Supplementary Proposition~7.

For $D/2<\alpha<(D+1)/2$, the same block estimates also control the local energy scale.
The new on-site term contains corrections of order $J^{(s)}\sqrt{b_{s+1}}/[n^{(s)}]^{\beta}$ and $\bar{g}^{(s)}\sqrt{b_{s+1}q_{s,*}}$, while the sum of pair interactions from an effective site to the rest of the system is controlled by the row sum of the effective interaction bound and is of order $J^{(s+1)}$.
When $\sqrt{b_{s+1}}/[n^{(s)}]^{\beta}$ is bounded and $b_{s+1}q_{s,*}$ is sufficiently small, these contributions are bounded by constant multiples of $J^{(s)}$, $\bar{g}^{(s)}$, and $J^{(s+1)}$, respectively.
For a fixed deviation cutoff, Eq.~\eqref{eq:methods_J_recursion} gives $J^{(s+1)}=\mathcal{O}(J^{(s)})$, and hence
\begin{equation}
	\bar{g}^{(s+1)}\le C_{g}\left(\bar{g}^{(s)}+J^{(s)}\right).
	\label{eq:methods_local_energy_recursion}
\end{equation}
Here $C_{g}$ is independent of $n$ and $s$.

\subsection{Effective-scale Schmidt bootstrap}
\label{meth:bootstrap}

We next show that the same Schmidt concentration mechanism used at the microscopic scale remains available after successive coarse-graining steps.
For $0\le k<s$ and $v\in\Lambda^{(s)}$, let $D_{k\leftarrow s}(v)\subseteq\Lambda^{(k)}$ denote the set of scale-$k$ sites that are coarse-grained into the scale-$s$ site $v$, and let $V_{k\to s}^{v}$ denote the cumulative local coisometry that maps this descendant block onto the effective site $v$ at scale $s$.
The corresponding retained projection on the scale-$k$ descendant block is
\begin{equation}
	\mathsf{P}_{k\to s}^{v}:=(V_{k\to s}^{v})^{\dagger}V_{k\to s}^{v}.
	\label{eq:methods_local_retained_projection}
\end{equation}
For $I\subseteq D_{k\leftarrow s}(v)$ and Hermitian one-site operators $a_{i}^{(k)}$, write $A^{(k)}:=\sum_{i\in I}a_{i}^{(k)}$.
For each intermediate scale $k<t\le s$, consider the coarse-graining step from scale $t-1$ to scale $t$, with deviation cutoff $z_{t}$.
After subtracting the expectation of each relevant operator in the local reference state $\ket{0_{w}^{(t-1)}}$ of the corresponding scale-$(t-1)$ site $w$, the centered operator has no component acting entirely within the reference sector.
Since the retained block contains at most $z_{t}$ deviations from these reference states, the Cauchy--Schwarz inequality controls the compressed sum by the square sum of the operator norms, with an additional factor $2\sqrt{z_{t}}+\sqrt{z_{t}+1}$.
Iterating this estimate from scale $k$ to scale $s$ and writing $\eta_{k+1:s}:=\prod_{t=k+1}^{s}\left(2\sqrt{z_{t}}+\sqrt{z_{t}+1}\right)$ gives
\begin{equation}
	\inf_{c\in\mathbb{R}} \left\|(A^{(k)}-c\mathds{1})\mathsf{P}_{k\to s}^{v}\right\| \le  \eta_{k+1:s} \left(\sum_{i\in I}\|a_{i}^{(k)}\|^{2}\right)^{1/2}.
	\label{eq:methods_retained_square_sum}
\end{equation}
Thus, compression to an effective site preserves the square-sum structure without introducing any dependence on the effective local Hilbert-space dimension.
Writing $A^{[k\to s]}:=V_{k\to s}^{v}A^{(k)}(V_{k\to s}^{v})^{\dagger}$, Eq.~\eqref{eq:methods_retained_square_sum} controls the norm of the centered compressed operator by the same square sum.
Since subtracting a scalar does not change the variance, the variance--gap estimate in Eq.~\eqref{eq:methods_variance_gap} applied to $H^{(s)}$ then gives
\begin{equation}
	\Delta^{(s)} \cdot \mathrm{Var}_{\Omega^{(s)}}\left(A^{[k\to s]}\right)  \le 4\eta_{k+1:s}^{2}\bar{g}^{(s)} \sum_{i\in I}\|a_{i}^{(k)}\|^{2}.
	\label{eq:methods_effective_variance_gap}
\end{equation}
For fixed deviation cutoffs, $\eta_{k+1:s}$ grows only exponentially in the number of coarse-graining steps.
The detailed derivations of Eqs.~\eqref{eq:methods_retained_square_sum} and \eqref{eq:methods_effective_variance_gap} are given in Supplementary Propositions~8 and~9, respectively.

For an effective site $v\in\Lambda^{(s)}$, write $q_{s,v}:=\bra{\Omega^{(s)}}Q_{v}^{(s)}\ket{\Omega^{(s)}}$, so that $q_{s,*}=\max_{v}q_{s,v}$, and define $u_{v}^{(s)} := P_{v}^{(s)} - Q_{v}^{(s)}$.  
We now repeat at scale $s$ the reflection-energy argument of Methods~\ref{meth:microscopic}, using the effective-scale fluctuation bound in Eq.~\eqref{eq:methods_effective_variance_gap} in place of its microscopic counterpart.
Let $E_{0}^{(s)}$ denote the ground-state energy of $H^{(s)}$ and define
\begin{equation}
	\delta E_{s,v}:=\bra{\Omega^{(s)}} (u_{v}^{(s)})^{\dagger} H^{(s)} u_{v}^{(s)} \ket{\Omega^{(s)}} - E_{0}^{(s)}.
	\label{eq:methods_effective_reflection_energy}
\end{equation}
As in the microscopic argument, the upper bound on the reflection energy is controlled by the fluctuations of the interactions between $v$ and its complement.
The square-sum control in Eqs.~\eqref{eq:methods_retained_square_sum} and \eqref{eq:methods_effective_variance_gap}, together with the effective interaction bounds of Methods~\ref{meth:interactions}, gives a dimension-independent upper bound on $\delta E_{s,v}$.
For fixed deviation cutoffs, the cutoff-dependent factors accumulated over successive coarse-graining steps grow at most exponentially in $s$, and we collect them into a factor $C^{s}$.
The resulting upper bound takes the three regime-dependent forms
\begin{align}
	\delta E_{s,v}
	&\lesssim
	C^{s}\sqrt{\frac{\bar{g}^{(s)}}{\Delta^{(s)}}}
	\sqrt{q_{s,v}}\,[n^{(s)}]^{-\beta}
	\nonumber\\
	&\quad\times
	\begin{cases}
		J_{\mathrm{K}}, & D/2<\alpha<(D+1)/2,\\
		J_{\mathrm{K}}\chi_{s}, & \alpha=(D+1)/2<D,\\
		J\mu_{0}\ell_{1:s}^{\alpha-(D+1)/2}, & (D+1)/2<\alpha<D.
	\end{cases}
	\label{eq:methods_effective_reflection}
\end{align}
The lower bound follows from the spectral gap exactly as at the microscopic scale,
\begin{equation}
	\delta E_{s,v}\ge4\Delta^{(s)}q_{s,v}(1-q_{s,v}).
	\label{eq:methods_effective_reflection_lower}
\end{equation}
Once the leading eigenvalue $1-q_{s,v}$ remains larger than $1/2$, Eqs.~\eqref{eq:methods_effective_reflection} and \eqref{eq:methods_effective_reflection_lower} therefore reproduce the same self-consistent suppression of the effective Schmidt tail as in Methods~\ref{meth:microscopic}.
The corresponding effective-scale reflection estimate is established in Supplementary Proposition~11.

We finally show how the condition that the leading Schmidt weight exceeds $1/2$ is inherited from one scale to the next.
For the block $L_{j}^{(s)}$ that is coarse-grained into the scale-$(s+1)$ site $j$, one has $|L_{j}^{(s)}| = b_{s+1}$.
The product state $\bigotimes_{i\in L_{j}^{(s)}} \ket{0_{i}^{(s)}}$ lies in the low-deviation subspace $\operatorname{Ran} \Pi_{L_{j}^{(s)},\le z_{s+1}}^{(s)}$, and the ground state $\ket{\Omega^{(s)}}$ has weight at least $1 - b_{s+1}q_{s,*}$ in the corresponding product reference sector.
If $p_{0,j}^{(s+1)}:= 1-q_{s+1,j}$ denotes the largest eigenvalue of the reduced state of the scale-$(s+1)$ site $j$, then, using the one-step state error $\epsilon_{s}$ defined in Eq.~\eqref{eq:methods_step_error},
\begin{equation}
	p_{0,j}^{(s+1)} \ge 1 - b_{s+1}q_{s,*} - \epsilon_{s}.
	\label{eq:methods_inherited_seed}
\end{equation}
Thus, whenever $b_{s+1}q_{s,*}+\epsilon_{s} < 1/2$, one has $p_{0,j}^{(s+1)} > 1/2$, so the leading effective eigenvalue is nondegenerate and the reflection bootstrap can be applied at the next scale.
For every child $i\in L_{j}^{(s)}$, the local coisometry $V_{s+1,j}$ also satisfies
\begin{equation}
	\bra{0_{j}^{(s+1)}} V_{s+1,j} Q_{i}^{(s)} V_{s+1,j}^{\dagger} \ket{0_{j}^{(s+1)}} \le  2(q_{s,*}+\epsilon_{s}).
	\label{eq:methods_child_occupation}
\end{equation}
This estimate is the link between the effective Schmidt suppression and the entropy contributions generated at the preceding coarse-graining scales.
The detailed derivation of Eqs.~\eqref{eq:methods_inherited_seed} and \eqref{eq:methods_child_occupation} is given in Supplementary Lemma~14.

\subsection{Closing the coarse-graining induction}
\label{meth:closure}

We now close the coarse-graining induction and determine the stopping scale in the three ranges of $\alpha$ considered above.
For a bipartition satisfying Eq.~\eqref{eq:main_crossing_general}, let $\theta$ denote its crossing exponent, so that $|\mathcal{C}_{s}(A)|\le C_{A}[n^{(s)}]^{\theta}$.
For the regular codimension-one bipartitions considered in Theorem~2, one has $\theta=(D-1)/D$.

For $D/2<\alpha<(D+1)/2$, take block volumes $b_{s+1}\asymp[n^{(s)}]^{\nu}$ with
\begin{equation}
	0<\nu<\min\left\{\beta,4\beta^{2},\frac{2\beta-\theta}{1-\theta}\right\}.
	\label{eq:methods_subcritical_nu}
\end{equation}
Such a $\nu$ exists precisely when $\theta<2\beta$.
We use a fixed $p\ge3$ and a sufficiently large deviation cutoff $z$, both independent of $n$.
To control the accumulation of scale-dependent factors, introduce the stopping threshold $T_{n}:=[\log(en)]^{K}$ with a sufficiently large fixed $K$, and define $s_{\mathrm{fin}}$ as the first scale satisfying $n^{(s_{\mathrm{fin}})}\le T_{n}$.
The successive coarse-graining steps then give $T_{n}^{1-\nu}<n^{(s_{\mathrm{fin}})}\le T_{n}$ and $s_{\mathrm{fin}}=\mathcal{O}(\log\log n)$.
Consequently, any factor of the form $C^{s}$ accumulated over the coarse-graining steps grows only as a fixed power of $\log n$ up to the terminal scale.

We show by induction that, for a constant $C\ge1$ depending only on the fixed microscopic parameters and the chosen cutoff, the following bounds hold throughout the active scales:
\begin{equation}
	\begin{aligned}
		q_{s,*}\le C^{s+1}[n^{(s)}]^{-2\beta}&, \qquad \bar{g}^{(s)} \le C^{s+1} \Delta, \\
		\Delta^{(s)}\ge\frac{\Delta}{2}&, \qquad \epsilon_{s}\le[n^{(s)}]^{-p}.
	\end{aligned}
	\label{eq:methods_closed_induction_a}
\end{equation}
At scale $s=0$, the bounds on $q_{0,*}$, $\bar{g}^{(0)}$, and $\Delta^{(0)}$ follow from the microscopic Schmidt estimate together with $H^{(0)}=H$, $\Delta^{(0)}=\Delta$, and $\bar{g}^{(0)}=\mathcal{O}(1)$.
Suppose now that the bounds on $q_{s,*}$, $\bar{g}^{(s)}$, and $\Delta^{(s)}$ hold at scale $s$, and that $\epsilon_{t}\le[n^{(t)}]^{-p}$ for all $t<s$.
The chosen block size and fixed deviation cutoff keep the mean deviation number $b_{s+1}q_{s,*}$ and the hopping between deviation sectors sufficiently small, so Eq.~\eqref{eq:methods_shell_tail} gives an inverse-polynomially small discarded norm for each block.
Since there are $n^{(s)}/b_{s+1}$ blocks, a union bound keeps the total discarded weight $w_{s} = \|(\mathds{1}-\mathsf{P}_{s\rightarrow s+1}) \ket{\Omega^{(s)}}\|^{2}$ small.
The corresponding energy defect $e_{s}$ is then bounded by $\bar{g}^{(s)}$ times a higher inverse power of $n^{(s)}$.
Since $\bar{g}^{(s)}$ grows at most exponentially in $s$, $s_{\mathrm{fin}}=\mathcal{O}(\log\log n)$, and $n^{(s)}>T_{n}$ at every active step, a sufficiently large fixed $K$ keeps the accumulated energy defects below $\Delta/2$.
Equation~\eqref{eq:methods_step_gap} then gives $\Delta^{(s+1)}\ge\Delta/2$ and the new one-step bound $\epsilon_{s}\le[n^{(s)}]^{-p}$.
Together with the bound on $q_{s,*}$, this gives $b_{s+1}q_{s,*}+\epsilon_{s}<1/2$, so Eq.~\eqref{eq:methods_inherited_seed} supplies the dominant Schmidt branch at scale $s+1$.
Eqs.~\eqref{eq:methods_J_recursion} and \eqref{eq:methods_local_energy_recursion} then show that the bound on $\bar{g}^{(s+1)}$ closes at scale $s+1$, with at most exponential growth in $s$.
With this local-energy bound and $\Delta^{(s+1)}\ge\Delta/2$, Eqs.~\eqref{eq:methods_effective_reflection} and \eqref{eq:methods_effective_reflection_lower} recover the $[n^{(s+1)}]^{-2\beta}$ suppression of $q_{s+1,*}$ up to a factor exponential in $s$.
This closes the induction.
The complete argument is given in Supplementary Proposition~12.

At $\alpha=(D+1)/2<D$, a codimension-one bipartition has $\theta=(D-1)/D=2\beta$.
For any fixed $\nu>0$, the scale-decay exponent $2\beta-\theta-\nu(1-\theta)$ used above is then negative, while the interaction scale in Eq.~\eqref{eq:methods_J_recursion} acquires the accumulated factor $\chi_{s}$.
These two effects prevent the preceding coarse-graining procedure from being continued for $\mathcal{O}(\log\log n)$ steps with the same control.
We therefore use only a fixed number of coarse-graining steps.
For a fixed $0<\varepsilon<1$, choose a sufficiently small fixed block exponent $\nu>0$ and let $s_{\mathrm{fin}}$ be the first scale for which $n^{(s_{\mathrm{fin}})}\le n^{\varepsilon/2}$.
The block sizes can be chosen so that
\begin{equation}
	n^{\varepsilon/4} < n^{(s_{\mathrm{fin}})} \le n^{\varepsilon/2}, \quad s_{\mathrm{fin}} \le s_{\mathrm{fin},*}(\varepsilon) = \mathcal{O}_{\varepsilon}(1).
	\label{eq:methods_critical_stopping}
\end{equation}
Because the number of steps is bounded independently of $n$, one has $\chi_{s}\le[\log(en)]^{s_{\mathrm{fin},*}(\varepsilon)/2}$ throughout the active scales.
Moreover, $n^{(s)}\ge n^{\varepsilon/4}$ at every active scale, so these logarithmic factors do not overcome the power-law smallness required by the low-deviation truncation.
The local-energy scale can correspondingly be bounded as $\bar{g}^{(s)}\le G_{s}[\log(en)]^{s/2}$, where $G_{s}$ is independent of $n$.
The truncation and gap estimate in Eq.~\eqref{eq:methods_step_gap}, the inherited Schmidt branch in Eq.~\eqref{eq:methods_inherited_seed}, and the effective reflection bounds in Eqs.~\eqref{eq:methods_effective_reflection} and \eqref{eq:methods_effective_reflection_lower} can then be iterated over this finite number of scales.
This gives
\begin{equation}
	q_{s,*}\le\frac{Q_{s}[\log(en)]^{3s/2}}{[n^{(s)}]^{(D-1)/D}}, \quad \Delta^{(s)}\ge\frac{\Delta}{2}, \quad \epsilon_{s}\le n^{-p}.
	\label{eq:methods_critical_envelope}
\end{equation}
Here $Q_{s}$ is independent of $n$ at every active scale.
Thus, the low-deviation truncation remains accurate, the effective gap stays open, and the effective Schmidt tail remains controlled through the terminal scale $n^{(s_{\mathrm{fin}})}\le n^{\varepsilon/2}$.
The corresponding finite-depth closure is established in the proof of Supplementary Theorem~3.

For $(D+1)/2<\alpha<D$, define $\eta_{\alpha} := (2\alpha-D-1)/(D-1)$ and let $N_{s} := n/n^{(s)}$ denote the number of microscopic sites combined into each scale-$s$ effective site.
The short-distance factor enters the effective interaction through the combination $\mu_{s}[n^{(s)}]^{-\beta}$, and Eq.~\eqref{eq:methods_mu_recursion} gives $\mu_{s}[n^{(s)}]^{-\beta}=\mu_{0}\mathfrak{r}_{s}$, where
\begin{equation}
	\mathfrak{r}_{s}:=\frac{N_{s}^{(D-1)/(2D)}}{n^{\beta}} = \left(\frac{n^{\eta_{\alpha}}}{n^{(s)}}\right)^{\frac{D-1}{2D}}.
	\label{eq:methods_supercritical_smallness}
\end{equation}
The same factor $\mathfrak{r}_{s}$ also controls the third branch of the effective reflection bound in Eq.~\eqref{eq:methods_effective_reflection}.
Thus $\mathfrak{r}_{s}$ increases as more microscopic sites are combined into each effective site and determines how far the coarse graining can be continued.

For a fixed $0<\varepsilon<1-\eta_{\alpha}$, choose a sufficiently small fixed block exponent $\nu>0$ and let $s_{\mathrm{fin}}$ be the first scale for which $n^{(s_{\mathrm{fin}})}\le n^{\eta_{\alpha}+\varepsilon/2}$.
The block sizes can be chosen so that
\begin{equation}
	n^{\eta_{\alpha}+\varepsilon/4}<n^{(s_{\mathrm{fin}})}\le n^{\eta_{\alpha}+\varepsilon/2}, \qquad s_{\mathrm{fin}}=\mathcal{O}_{\varepsilon}(1).
	\label{eq:methods_supercritical_stopping}
\end{equation}
Throughout the active scales, $\mathfrak{r}_{s}^{2}$ remains bounded by a negative power of the current effective size.
This keeps the deviation-sector hopping small, so the low-deviation truncation and the gap estimate in Eq.~\eqref{eq:methods_step_gap} can be iterated through the terminal scale.
Together with these bounds, Eq.~\eqref{eq:methods_inherited_seed} preserves the dominant Schmidt branch, while the local-energy scale remains bounded as $\bar{g}^{(s)}\le G_{s}$ with $G_{s}$ independent of $n$.
The effective reflection bounds in Eqs.~\eqref{eq:methods_effective_reflection} and \eqref{eq:methods_effective_reflection_lower} can therefore be applied at each scale, giving
\begin{equation}
	q_{s,*}\le Q_{s}\mathfrak{r}_{s}^{2}, \qquad \Delta^{(s)}\ge\frac{\Delta}{2}, \qquad \epsilon_{s}\le n^{-p}.
	\label{eq:methods_supercritical_envelope}
\end{equation}
The coefficients $Q_{s}$ remain independent of $n$ throughout the active scales.
The terminal effective size therefore satisfies
\begin{equation}
	n_{\mathrm{eff}}=n^{\eta_{\alpha}+\mathcal{O}(\varepsilon)}=n^{\frac{2\alpha-D-1}{D-1}+\mathcal{O}(\varepsilon)}.
	\label{eq:methods_supercritical_effective_size}
\end{equation}
The positive $\varepsilon$ margin keeps $\mathfrak{r}_{s}$ polynomially small throughout the finite sequence of coarse-graining steps and is retained in the final entropy bound.
The finite-depth induction and stopping-scale analysis in this range are given in the proof of Supplementary Theorem~4.

\subsection{Entropy reconstruction across coarse-graining scales}
\label{meth:entropy}

We finally convert the effective-scale suppression into an entropy bound for the original ground state $\ket{\Omega}$.
Recall from Eq.~\eqref{eq:main_crossing_general} that $\mathcal{C}_{s}(A)$ denotes the scale-$s$ effective regions that intersect both sides of the bipartition, with $|\mathcal{C}_{s}(A)|\le C_{A}[n^{(s)}]^{\theta}$.
Let $V_{0\to t}$ denote the cumulative coisometry through scale $t$, and let $\mathsf{P}_{0\to t}:=V_{0\to t}^{\dagger}V_{0\to t}$ be the corresponding microscopic retained projector.

The entropy contribution associated with the step from scale $s$ to scale $s+1$ is localized near effective regions that intersect both sides of the bipartition.
Each region in $\mathcal{C}_{s+1}(A)$ contains $b_{s+1}$ scale-$s$ children, and therefore
\begin{equation}
	b_{s+1}|\mathcal{C}_{s+1}(A)| \leq C_{A}[n^{(s)}]^{\theta}b_{s+1}^{1-\theta}.
	\label{eq:methods_cut_count}
\end{equation}
To convert this counting estimate into an entropy bound, decompose $A$ into maximal coarse-grained regions that lie entirely inside $A$.
Every nonterminal region at scale $s$ is a child of a region in $\mathcal{C}_{s+1}(A)$, so there are at most $b_{s+1}|\mathcal{C}_{s+1}(A)|$ such regions.

For the reduced state $\rho_{v}$ of each such scale-$s$ region, let $q_{v}$ denote its occupation outside the reference sector.
When $q_{v}$ is small, most of the weight is concentrated on a single reference direction, so the entropy can only come from the small orthogonal component, with a logarithmic dependence on the available local dimension.
This gives $S(\rho_{v})\le 2y_{s}+\omega_{s}q_{v}$, where $y_{s}:=\frac{1}{2}[n^{(s)}]^{-2\beta}$ and $\omega_{s}:=\log\left(2d^{(s)}[n^{(s)}]^{2\beta}\right)$.
The first term therefore contributes at most $2y_{s}b_{s+1}| \mathcal{C}_{s+1}(A)|$ at scale $s$.
The term $\omega_{s}q_{v}$ cannot be bounded directly by $q_{s,*}+\epsilon_{s}$, because Eq.~\eqref{eq:methods_child_occupation} controls the child occupation only in the reference state of its parent, not in an arbitrary retained state.
We therefore rewrite the child occupation in terms of a small reference contribution and the occupation of its parent, and repeat this step up to scale $t$.
Because each retained block contains at most $z$ deviations, the coefficient multiplying an occupation can increase by at most a factor $2z$ at each coarse-graining step.
Defining $B_{t}:=(2z)^{t}\max_{0\le r\le t}\omega_{r}$ therefore bounds the largest coefficient accumulated along this process, and Eq.~\eqref{eq:methods_child_occupation} gives the scale-$s$ occupation contribution $4B_{t}(q_{s,*}+\epsilon_{s})b_{s+1}| \mathcal{C}_{s+1}(A)|$.
Thus, the total contribution generated at scale $s$ is bounded by $2\left(2B_{t}(q_{s,*}+\epsilon_{s})+y_{s}\right)b_{s+1}|\mathcal{C}_{s+1}(A)|$.
Moreover, because the retained block contains at most the fixed number $z$ of deviations, its effective dimension satisfies $d^{(s+1)}\le\sum_{m=0}^{z}\binom{b_{s+1}}{m}[d^{(s)}-1]^{m}$. 
Since $z$ is fixed, iterating this dimension bound gives $\log d^{(s)}\le C^{s}\log(en)$ for some constant $C$ independent of $n$.

Collecting these contributions over all scales gives the entropy estimate formalized in Supplementary Lemma~16.
Every normalized state $\ket{\psi}\in\operatorname{Ran}\mathsf{P}_{0\to t}$ satisfies
\begin{equation}
	\begin{aligned}
		S_{A}(\ket{\psi}) &\le (B_{t}+2y_{t})n^{(t)} \\
		&\quad + 2\sum_{s=0}^{t-1} \left(2B_{t}(q_{s,*}+\epsilon_{s})+y_{s}\right)b_{s+1} |\mathcal{C}_{s+1}(A)|.
	\end{aligned}
	\label{eq:methods_retained_entropy_bound}
\end{equation}
The first term is the terminal contribution from the scale-$t$ regions, of which there are at most $n^{(t)}$, while the sum collects the contributions generated at the preceding coarse-graining scales near the bipartition boundary.

For $D/2<\alpha<(D+1)/2$, Eqs.~\eqref{eq:methods_closed_induction_a} and \eqref{eq:methods_cut_count}, together with $b_{s+1}\asymp[n^{(s)}]^{\nu}$, show that each scale-$s$ term in the sum of Eq.~\eqref{eq:methods_retained_entropy_bound} decays, up to the scale-dependent prefactors, as $[n^{(s)}]^{-\lambda_{\mathrm{cut}}}$, where $\lambda_{\mathrm{cut}}:=2\beta-\theta-\nu(1-\theta)>0$ by Eq.~\eqref{eq:methods_subcritical_nu}.
The power-law part of each preceding-scale contribution therefore scales as a negative power of the current effective size, while $B_{t}$ and the remaining step-dependent factors contribute only fixed powers of $\log n$ because $s_{\mathrm{fin}}=\mathcal{O}(\log\log n)$.
Since the terminal effective size is also $\operatorname{polylog}(n)$, every normalized $\ket{\psi}\in\operatorname{Ran}\mathsf{P}_{0\to s_{\mathrm{fin}}}$ satisfies
\begin{equation}
	S_{A}(\ket{\psi}) \le C_{\mathrm{ent}}[\log(en)]^{\mathfrak{a}},
	\label{eq:methods_subcritical_retained_entropy}
\end{equation}
where $C_{\mathrm{ent}}$ and $\mathfrak{a}>0$ are independent of $n$.

At $\alpha=(D+1)/2$, one has $2\beta=\theta$, so the positive scale-decay exponent is no longer available.
The finite-depth construction in Methods~\ref{meth:closure} instead has $s_{\mathrm{fin}}=\mathcal{O}_{\varepsilon}(1)$.
Equation~\eqref{eq:methods_critical_envelope}, together with Eq.~\eqref{eq:methods_retained_entropy_bound}, bounds the accumulated scale contributions by $n^{\varepsilon/2}$ times a fixed power of $\log(en)$.
Using $n^{(s_{\mathrm{fin}})}\le n^{\varepsilon/2}$ from Eq.~\eqref{eq:methods_critical_stopping}, one obtains
\begin{equation}
	S_{A}(\ket{\psi})\le C_{\varepsilon}n^{\varepsilon/2}[\log(en)]^{c_{\varepsilon}}
	\label{eq:methods_critical_retained_entropy}
\end{equation}
for every normalized $\ket{\psi}\in\operatorname{Ran}\mathsf{P}_{0\to s_{\mathrm{fin}}}$, where $C_{\varepsilon}$ and $c_{\varepsilon}$ are independent of $n$.

For $(D+1)/2<\alpha<D$, the relevant occupation scale is $\mathfrak{r}_{s}^{2}$ from Eq.~\eqref{eq:methods_supercritical_smallness}.
With $\theta=(D-1)/D$, Eqs.~\eqref{eq:methods_supercritical_envelope} and \eqref{eq:methods_cut_count} give
\begin{equation}
	\mathfrak{r}_{s}^{2} b_{s+1} |\mathcal{C}_{s+1}(A)| \lesssim n^{\theta\eta_{\alpha}} b_{s+1}^{1/D} \le n^{\eta_{\alpha}+\varepsilon/4}.
	\label{eq:methods_supercritical_cut_count}
\end{equation}
Since the number of coarse-graining steps is $\mathcal{O}_{\varepsilon}(1)$ and the terminal scale satisfies $n^{(s_{\mathrm{fin}})}\le n^{\eta_{\alpha}+\varepsilon/2}$ by Eq.~\eqref{eq:methods_supercritical_stopping}, Eq.~\eqref{eq:methods_retained_entropy_bound} gives
\begin{equation}
	S_{A}(\ket{\psi})\le C_{\varepsilon}n^{\eta_{\alpha}+\varepsilon/2}\log(en)
	\label{eq:methods_supercritical_retained_entropy}
\end{equation}
for every normalized $\ket{\psi}\in\operatorname{Ran}\mathsf{P}_{0\to s_{\mathrm{fin}}}$, where $C_{\varepsilon}$ is independent of $n$.

We next transfer these retained-space bounds to the exact ground state $\ket{\Omega}$.
Because each successive retained space is contained in the preceding one, $\mathsf{P}_{0\to t+1}\le\mathsf{P}_{0\to t}$.
For $D/2<\alpha<(D+1)/2$, define $\mathfrak{p}_{t}:=\|(\mathsf{P}_{0\to t}-\mathsf{P}_{0\to t+1})\ket{\Omega}\|^{2}$ for $0\le t<s_{\mathrm{fin}}$, and let $\ket{\omega_{t}}$ be the corresponding normalized component whenever $\mathfrak{p}_{t}>0$.
At the terminal scale, define $\mathfrak{p}_{s_{\mathrm{fin}}}:=\|\mathsf{P}_{0\to s_{\mathrm{fin}}}\ket{\Omega}\|^{2}$ and let $\ket{\omega_{s_{\mathrm{fin}}}}$ be the corresponding normalized component.
The nested retained spaces give the exact orthogonal decomposition
\begin{equation}
	\ket{\Omega}=\sum_{t=0}^{s_{\mathrm{fin}}}\sqrt{\mathfrak{p}_{t}}\ket{\omega_{t}}.
	\label{eq:methods_nested_decomposition}
\end{equation}
The nonzero components $\ket{\omega_{t}}$ are mutually orthogonal and $\sum_{t=0}^{s_{\mathrm{fin}}}\mathfrak{p}_{t}=1$.
For $t<s_{\mathrm{fin}}$, telescoping the one-step errors gives $\mathfrak{p}_{t}\le(t+1)^{2}[n^{(t)}]^{-2p}$.
Moreover, each $\ket{\omega_{t}}$ belongs to $\operatorname{Ran}\mathsf{P}_{0\to t}$ and therefore obeys Eq.~\eqref{eq:methods_retained_entropy_bound} at scale $t$.

Supplementary Lemma~17 applied to Eq.~\eqref{eq:methods_nested_decomposition} gives
\begin{equation}
	S_{A}(\ket{\Omega})\le (s_{\mathrm{fin}}+1)\left[H(\{\mathfrak{p}_{t}\})+\sum_{t=0}^{s_{\mathrm{fin}}}\mathfrak{p}_{t}S_{A}(\ket{\omega_{t}})\right].
	\label{eq:methods_orthogonal_entropy_transfer}
\end{equation}
For $D/2<\alpha<(D+1)/2$, one has $s_{\mathrm{fin}}=\mathcal{O}(\log\log n)$ and hence $H(\{\mathfrak{p}_{t}\})\le\log(s_{\mathrm{fin}}+1)$.
For every nonterminal scale, the decay $\mathfrak{p}_{t}\le(t+1)^{2}[n^{(t)}]^{-2p}$ suppresses the corresponding retained-space entropy contribution.
The terminal component is controlled by Eq.~\eqref{eq:methods_subcritical_retained_entropy}.
Thus, after absorbing the remaining polylogarithmic factors into the constants and exponent, the exact ground state $\ket{\Omega}$ satisfies the same polylogarithmic form of entropy bound.
The complete derivation of this bound is given in Supplementary Theorem~2.

For the two finite-depth constructions, corresponding to $\alpha=(D+1)/2$ and $(D+1)/2<\alpha<D$, the cumulative approximation error remains polynomially small in $n$.
We therefore decompose the exact ground state $\ket{\Omega}$ into its component inside the final retained space and the small component orthogonal to that space.
Writing $\delta_{\mathrm{fin}}:=\|(\mathds{1}-\mathsf{P}_{0\to s_{\mathrm{fin}}})\ket{\Omega}\|^{2}$, the cumulative state error gives $\delta_{\mathrm{fin}}\lesssim_{\varepsilon}n^{-2p}$.
If $\ket{\psi}$ and $\ket{\xi}$ denote the normalized retained and discarded components, respectively, Supplementary Lemma~17 with two components gives
\begin{equation}
	S_{A}(\ket{\Omega}) \le 2 h_{2}(\delta_{\mathrm{fin}})+2S_{A}(\ket{\psi})+2\delta_{\mathrm{fin}}n\log d.
	\label{eq:methods_finite_depth_entropy_transfer}
\end{equation}
The first and last terms are negligible for the fixed choice $p\ge3$.

At $\alpha=(D+1)/2$, Eq.~\eqref{eq:methods_critical_retained_entropy} then gives $S_{A}(\ket{\Omega})\le C_{\varepsilon}n^{\varepsilon}$ after absorbing the fixed logarithmic factor into the remaining power-law margin.
The complete finite-depth argument in this case is given in Supplementary Theorem~3.
For $(D+1)/2<\alpha<D$, Eq.~\eqref{eq:methods_supercritical_retained_entropy}, together with $\log(en)\le n^{\varepsilon/2}$ for sufficiently large $n$, gives $S_{A}(\ket{\Omega})\le C_{\varepsilon}n^{\eta_{\alpha}+\varepsilon}$.
The corresponding derivation is given in Supplementary Theorem~4.
Together with the polylogarithmic bound established above for $D/2<\alpha<(D+1)/2$, these are the three entropy bounds stated in Theorem~2.

The exact divisibility required by the intermediate block partitions is only technical.
For a general cubic side length, the system can be embedded into a compatible cubic box of size larger by at most a bounded factor by adding sites in a product state.
This leaves the entropy of the original bipartition unchanged and modifies the interaction and geometric constants only by fixed factors.
The same bounded-ratio embedding applies for both open and periodic boundary conditions in all three regimes; see Supplementary Corollary~18 and the proofs of Supplementary Theorems~3 and~4.

\begin{acknowledgments}
	
	Both authors acknowledge support from RIKEN Hakubi Project.
	D.K. was supported by RIKEN Special Postdoctoral Researcher Program and JSPS KAKENHI Grant Number JP26K17060.
	T.K. was supported by Japan Science and Technology Agency through its Exploratory Research for Advanced Technology program (JST ERATO Grant Number JPM-JER2302) and by JSPS KAKENHI Grant Numbers JP23K25796, JP25K24674, JP26K00019, JP26K00020, and JP26H02015.
	The authors used generative AI tools, including OpenAI's GPT-5.6 and GPT-6 models, to assist in developing and refining proof ideas and improving the clarity and presentation of the manuscript. All results, proofs, calculations, and references were independently verified by the authors, who take full responsibility for the content of the manuscript.
\end{acknowledgments}

\bibliography{Collectivity}

\clearpage
\newpage
\setcounter{page}{1}

\onecolumngrid
\supponecolumnfootnotes

\renewcommand\thefootnote{*\arabic{footnote}}

\setcounter{equation}{0}
\setcounter{section}{0}
\setcounter{figure}{0}

\renewcommand{\theequation}{S.\arabic{equation}}
\renewcommand{\figurename}{Supplementary Figure}

\renewcommand{\tablename}{Supplementary Table}
\renewcommand{\thetable}{\arabic{table}}

\renewcommand{\thesection}{Supplementary Note~\arabic{section}}

\renewcommand{\thesubsection}{\Alph{subsection}}

\begin{center}
	{\large \bf Supplementary Information for ``Collectivity limits quantum entanglement''}\\
	\vspace*{0.3cm}
	Donghoon Kim$^{1}$ and Tomotaka Kuwahara$^{1,2}$ \\
	\vspace*{0.1cm}
	$^{1}${\small \it Analytical Quantum Complexity RIKEN Hakubi Research Team, \\ RIKEN Center for Quantum Computing (RQC), Wako, Saitama 351-0198, Japan} \\ 
	$^{2}${\small \it RIKEN Pioneering Research Institute (PRI), Wako, Saitama 351-0198, Japan} 
\end{center}

\vspace*{0.3cm}

In this Supplementary Information, we provide detailed proofs and technical analyses supporting the results presented in the main text. In particular, we establish the microscopic estimates underlying the entanglement bounds and develop the mean-field renormalization-group framework used in the analysis.

\supplementarytableofcontents
\startsupplementarycontents

\section{Setup}
\label{sec:microscopic-setup}

\subsection{Geometry, normalization, and the spectral gap}

Let $\Lambda$ be a cubic lattice box with $n = \ell_{\Lambda}^{D} \ge 2$ sites, with either open or periodic boundary conditions.
The microscopic estimates also hold for rectangular boxes of uniformly bounded aspect ratio, with modified geometric constants.
Each microscopic site carries a $d$-dimensional Hilbert space, where $d\ge2$ is fixed independently of $n$.
All logarithms are natural.

For $X \subseteq \Lambda$, write $X^{\mathrm{c}} := \Lambda \setminus X$.
The distance $r_{i,j}$ is the nearest-neighbor graph distance, and $r_{X,Y} := \min_{i\in X, j\in Y} r_{i,j}$ for nonempty disjoint sets.
With the convention $r_{i,\varnothing} := +\infty$, we use the internal boundary
\begin{align}
	\partial X := \{i\in X: r_{i,X^{\mathrm{c}}}=1\}.
\end{align}
For a hypercubic block $X$ of side length $\ell \ge 2$,
\begin{align}
	|\partial X| \le \ell^{D}-(\ell-2)^{D} \le 2D\ell^{D-1} = 2D|X|^{(D-1)/D}. \label{partial_L_size}
\end{align}
The first relation is an inequality because a face of $X$ may coincide with an open boundary of the full system.

Let $\mathcal{S}_{i}$ be the Moore neighborhood of $i$: the sites other than $i$ whose coordinate displacement from $i$ is at most one in every direction, with the periodic convention when appropriate.
Thus $|\mathcal{S}_{i}|\le3^{D}-1$.
We will repeatedly use the elementary counting bounds
\begin{align}
	|\{j:r_{i,j}=r\}|&\le a_{D} r^{D-1}, \qquad |\{j:1\le r_{i,j}\le R\}|\le a'_{D} R^{D}, \label{eq:preMFRG-lattice-counting}
\end{align}
with the explicit choices $a_{D}:=2^{D}D^{D-1}$ and $a'_{D}:=3^{D}$, uniform in the boundary condition.
These follow by counting integer displacement vectors of $\ell_{1}$ length $r$; restricting to a box or identifying periodic copies cannot increase that count.
Indeed, $|\{j: r_{i,j} = r\}| \le 2^{D} \binom{r+D-1}{D-1} \le a_{D}r^{D-1}$, while $|\{j:1 \le r_{i,j} \le R\}| \le (2R+1)^{D}\le a'_{D}R^{D}$ for $R\ge1$.
We fix $C_{D}:= 2D(a_{D}+3^{D})$ whenever a dimension-only upper bound is used below.

For reference, an operator $O$ is $k$-local and $g$-extensive if it has a decomposition
\begin{align}
	O=\sum_{Z:\,1\le |Z|\le k} o_{Z}, \qquad \max_{i\in\Lambda} \sum_{Z\ni i} \norm{o_{Z}} \le g, \label{eq:def-k-local}
\end{align}
where $o_{Z}$ is supported on $Z$.
In particular,
\begin{align}
	\norm{O} \le \sum_{Z}\norm{o_{Z}} \le \sum_{i\in\Lambda} \sum_{Z\ni i} \norm{o_{Z}} \le gn.
\end{align}
Locality here refers to the number of sites in a term, not to its spatial diameter.

We consider a Hermitian two-local Hamiltonian
\begin{align}
	H=\sum_{i \in \Lambda} h_{i} + \sum_{i<j}h_{i,j}, \label{eq:hamiltonian}
\end{align}
with Hermitian terms $h_{i}$ and $h_{i,j}$, the convention $h_{j,i}:=h_{i,j}$, and
\begin{align}
	\norm{h_{i,j}} \le \frac{J \mu_{i,j}}{n^{1-\alpha/D}} r_{i,j}^{-\alpha}, \qquad 0\le\alpha<D, \label{eq:long-range-interaction}
\end{align}
where
\begin{align}
	\mu_{i,j}:= \begin{cases} \mu_{0}, &j\in\mathcal{S}_{i}, \\
		1, &j\notin\mathcal{S}_{i}\cup\{i\}, \end{cases} \qquad \mu_{0}\ge1. \label{mu_i_i'_choice}
\end{align}
The parameters $J$ and $\mu_{0}$ are independent of $n$.
It is convenient to write
\begin{align}
	\beta := 1 - \frac{\alpha}{D} > 0. \label{eq:preMFRG-beta}
\end{align}
The factor $1/n^{1-\alpha/D} = n^{-\beta}$ is the Kac normalization~\cite{kac1963van,defenu2023long}, which ensures that the Hamiltonian is extensive, namely $\norm{H}=\orderof{n}$ provided the on-site terms are uniformly bounded.
Weakening or removing the Kac normalization changes the scaling problem and is outside the scope of the present note.
The normalization in Eq.~\eqref{eq:long-range-interaction} is kept fixed throughout the arguments below.

Define
\begin{align}
	g_{0}&:=\max_{i}\norm{h_{i}}, \qquad \gamma_{0}:=\max_{i}\sum_{j\ne i} \frac{\mu_{i,j}r_{i,j}^{-\alpha}}{n^{\beta}}, \qquad \bar{g} :=g_{0}+J\gamma_{0}. \label{eq:local-energy-scale}
\end{align}
Indeed, Eq.~\eqref{eq:preMFRG-lattice-counting} and $r_{i,j}\le D\ell_{\Lambda}$ imply
\begin{align}
	\sum_{j\ne i}\mu_{i,j}r_{i,j}^{-\alpha} &\le a_{D}\sum_{r=1}^{D\ell_{\Lambda}}r^{D-1-\alpha} +(3^{D}-1)\mu_{0} \le C_{D,\alpha,\mu_{0}}n^{\beta}.
\end{align}
Consequently $\gamma_{0}=\mathcal{O}(1)$, and
\begin{align}
	\norm{h_{i}}+\sum_{j\ne i}\norm{h_{i,j}}\le\bar{g}, \qquad \norm{H}\le\bar{g} n. \label{eq:preMFRG-local-and-global-norm}
\end{align}
Constants in geometric estimates may depend on $D$ and $\alpha$; any additional dependence on $\mu_{0}$ will be stated or displayed.
They never depend on a renormalized local Hilbert-space dimension.

Assume that $H$ has a unique normalized ground state $\ket{\Omega}$, ground energy $E_{0}$, and gap $\Delta>0$:
\begin{align}
	H - E_{0}\mathds{1} \ge \Delta\bigl(\mathds{1}-\ket{\Omega}\bra{\Omega}\bigr). \label{eq:unique-ground-state-gap}
\end{align}
For asymptotic entropy statements, $D,\alpha,d,J,\mu_{0},\bar{g}$, and the gap lower bound $\Delta$ are fixed independently of $n$.
The quantitative estimates below retain their dependence on $\bar{g}$ and $\Delta$.

An energy shift leaves eigenvectors, gaps, and off-diagonal shell matrix elements unchanged.
When an operator-norm bound is needed, however, the shift is included explicitly; in particular,
\begin{align}
	0 \le H - E_{0}\mathds{1} \le 2\bar{g} n\,\mathds{1}. \label{eq:preMFRG-shifted-norm}
\end{align}

\subsection{Entropy notation}

For a density matrix $\rho$, we write
\begin{align}
	S(\rho) := -\operatorname{Tr}(\rho\log\rho)
\end{align}
for its von Neumann entropy.
For a normalized pure state $\ket{\psi}$ and a bipartition $A|A^{\mathrm{c}}$, let
\begin{align}
	\rho_{A}^{\psi}:= \operatorname{Tr}_{A^{\mathrm{c}}}\ket{\psi}\bra{\psi}, \qquad S_{A}(\ket{\psi}):=S(\rho_{A}^{\psi}). \label{eq:entanglement-entropy-definition}
\end{align}
Thus $S_{A}(\ket{\psi})=S_{A^{\mathrm{c}}}(\ket{\psi})$ for every pure state $\ket{\psi}$.

\subsection{Reference vectors and deviation-number projections}

For a site $i$, let $\rho_{i} := \operatorname{Tr}_{i^{\mathrm{c}}}(\ket{\Omega}\bra{\Omega})$ and choose a Schmidt decomposition
\begin{align}
	\ket{\Omega} = \sum_{a=0}^{d-1} \lambda_{a,i} \ket{a_{i}} \otimes \ket{\phi_{a}}_{i^{\mathrm{c}}}, \qquad \lambda_{0,i}\ge\lambda_{1,i}\ge\cdots\ge0. \label{eq:schmidt-decomposition}
\end{align}
The vector $\ket{0}_{i}$ associated with the largest coefficient is our local reference, or mean-field, vector.
It is selected from the exact reduced state; it is not assumed to minimize a separate product-state energy.
If the largest eigenvalue is degenerate, choose any normalized vector in that eigenspace and keep that choice fixed.

Set
\begin{align}
	P_{i} &:=\ket{0}_{i}\bra{0}_{i}, \qquad Q_{i}:=\mathds{1}_{i}-P_{i},
\end{align}
and
\begin{align}
	q_{i} &:=\bra{\Omega} Q_{i}\ket{\Omega}=1-\lambda_{0,i}^{2}, \qquad q_{*}:=\max_{i}q_{i}, \qquad \lambda_{*}:=\sqrt{1-q_{*}}=\min_{i}\lambda_{0,i}>0. \label{eq:microscopic-flip-density}
\end{align}
For $X\subseteq\Lambda$, write $\ket{\mathbf{0}}_{X}:=\bigotimes_{i\in X}\ket{0}_{i}$ and define
\begin{align}
	M_{X}:=\sum_{i\in X}Q_{i}.
\end{align}
For a Borel set $I\subseteq\mathbb{R}$, let $\mathds{1}_{I}(M_{X})$ denote the spectral projection of $M_{X}$ onto the part of its spectrum contained in $I$.
In particular,
\begin{align}
	\Pi_{X,z}&:=\mathds{1}_{\{z\}}(M_{X}), \qquad \Pi_{X,\le z}:=\mathds{1}_{[0,z]}(M_{X}). \label{eq:preMFRG-flip-projections}
\end{align}
We also use $\Pi_{X,>z}:=\mathds{1}-\Pi_{X,\le z}$ and set $\Pi_{X,z}=0$ when $z\notin\{0,\ldots,|X|\}$.
A ``flip'' means a deviation from $\ket{0}_{i}$; for $d>2$ it need not be a physical spin flip.

These projections are useful for two different reasons.
First,
\begin{align}
	\bra{\Omega} M_{X}\ket{\Omega}=\sum_{i\in X}q_{i}\le|X|q_{*}. \label{eq:preMFRG-mean-flips}
\end{align}
Second, the retained subspace is small when the cutoff is fixed:
\begin{align}
	\operatorname{rank}\Pi_{X,\le z} =\sum_{r=0}^{\min\{z,|X|\}}\binom{|X|}{r}(d-1)^{r} \le(d|X|)^{z} \qquad(z\in\mathbb{Z}_{\ge0},\ X\ne\varnothing). \label{eq:preMFRG-low-flip-dimension}
\end{align}
For $z\ge1$, the last inequality follows from $\binom{|X|}{r}\le|X|^{r}$, $1+(d-1)|X|\le d|X|$, and $\sum_{r=0}^{z}a^{r}\le(1+a)^{z}$ with $a=(d-1)|X|$; the case $z=0$ is immediate.

We do not yet assume that $q_{*}$ is small.
Even if it is small, the mean bound alone gives only
\begin{align}
	\norm{\Pi_{X,>z} \ket{\Omega}}^{2} \le \frac{|X|q_{*}}{z+1}, \label{eq:preMFRG-Markov-tail}
\end{align}
because $M_{X}\ge(z+1)\Pi_{X,>z}$.
This is generally insufficient for a high-accuracy truncation of many blocks.
The purpose of the estimates below is to obtain a much stronger tail from the Hamiltonian and its gap.

\subsection{Scope of the entropy statements and a dimer obstruction}
\label{sec:scope-counterexample}

Before developing the truncation estimates, it is important to separate what they can establish for arbitrary cuts from what requires geometric information.
The following example shows that the distinction is physical, not merely a feature of the proof.

\begin{prop}
	\label{prop:dimer-obstruction}
	Assume that $n$ is even and that $\Lambda$ admits a nearest-neighbor perfect matching $\mathcal{M}$.
	For qubits, let
	\begin{align}
		\varepsilon_{n}:=n^{-\beta}, \qquad H_{\mathrm{dimer}}:=\sum_{i\in\Lambda}\frac{1-\sigma_{i}^{z}}{2} - \varepsilon_{n}\sum_{(i,j)\in\mathcal{M}}\sigma_{i}^{x}\sigma_{j}^{x}. \label{eq:dimer-counterexample}
	\end{align}
	Here $\sigma_{i}^{x}$ and $\sigma_{i}^{z}$ denote the Pauli operators acting on site $i$.
	This Hamiltonian satisfies Eq.~\eqref{eq:long-range-interaction} with $J=\mu_{0}=1$, has a unique ground state $\ket{\Omega_{\mathrm{dimer}}}$, and has a gap bounded below uniformly in $n$.
	Nevertheless, a cut containing exactly one endpoint of every dimer has
	\begin{align}
		S_{A}(\ket{\Omega_{\mathrm{dimer}}})=\Theta\left(n^{2\alpha/D-1}\log n\right). \label{eq:dimer-arbitrary-cut-entropy}
	\end{align}
	Thus the generic assumptions do not imply a polylogarithmic entropy bound for every bipartition when $\alpha>D/2$.
\end{prop}

\begin{proof}
	Since $\mathcal{M}$ is a perfect matching, the Hamiltonian decomposes as
	\begin{align}
		H_{\mathrm{dimer}} = \sum_{(i,j)\in\mathcal{M}}H_{i,j}, \qquad H_{i,j}:=\frac{1-\sigma_{i}^{z}}{2}+\frac{1-\sigma_{j}^{z}}{2}-\varepsilon_{n}\sigma_{i}^{x}\sigma_{j}^{x},
	\end{align}
	and the terms $H_{i,j}$ corresponding to different dimers act on disjoint tensor factors.
	
	For a single dimer, in the even- and odd-parity bases $\{\ket{00},\ket{11}\}$ and $\{\ket{01},\ket{10}\}$, respectively,
	\begin{align}
		H_{i,j}\big|_{\mathrm{even}}&=\begin{pmatrix}0&-\varepsilon_{n}\\-\varepsilon_{n}&2\end{pmatrix}, \qquad H_{i,j}\big|_{\mathrm{odd}}=\begin{pmatrix}1&-\varepsilon_{n}\\-\varepsilon_{n}&1\end{pmatrix}.
	\end{align}
	Hence the four eigenvalues are $1-\sqrt{1+\varepsilon_{n}^{2}}$, $1-\varepsilon_{n}$, $1+\varepsilon_{n}$, and $1+\sqrt{1+\varepsilon_{n}^{2}}$.
	Since $\sqrt{1+\varepsilon_{n}^{2}}>\varepsilon_{n}$, the lowest eigenvalue is strictly $e_{0}(\varepsilon_{n})=1-\sqrt{1+\varepsilon_{n}^{2}}$, so the dimer ground state is unique.
	The corresponding normalized ground state lies in the even-parity sector and can be written as
	\begin{align}
		\ket{\omega_{\varepsilon_{n}}}=\sqrt{1-p_{n}}\ket{00}+\sqrt{p_{n}}\ket{11}, \qquad p_{n}=\frac{1}{2}\left(1-\frac{1}{\sqrt{1+\varepsilon_{n}^{2}}}\right).
	\end{align}
	The first excited energy is $1-\varepsilon_{n}$, and therefore the single-dimer gap is
	\begin{align}
		\Delta_{\mathrm{dimer}}=(1-\varepsilon_{n})-\left(1-\sqrt{1+\varepsilon_{n}^{2}}\right)=\sqrt{1+\varepsilon_{n}^{2}}-\varepsilon_{n}.
	\end{align}
	Because $0<\varepsilon_{n}\le1$, $\Delta_{\mathrm{dimer}}\ge\sqrt{2}-1$.
	Since the full Hamiltonian is a sum of independent dimer Hamiltonians, its ground state is the tensor product
	\begin{align}
		\ket{\Omega_{\mathrm{dimer}}}=\bigotimes_{(i,j)\in\mathcal{M}}\ket{\omega_{\varepsilon_{n}}}_{i,j},
	\end{align}
	which is unique, and the global gap is exactly $\Delta_{\mathrm{dimer}}$.
	Thus the gap is bounded below uniformly in $n$.
	
	The interaction bound is also satisfied.
	For every matching edge, $r_{i,j}=1$ and $\norm{-\varepsilon_{n}\sigma_{i}^{x}\sigma_{j}^{x}}=\varepsilon_{n}=n^{-\beta}=\frac{1}{n^{1-\alpha/D}}r_{i,j}^{-\alpha}$, where $\beta=1-\alpha/D$.
	All non-matching pair terms vanish.
	Thus Eq.~\eqref{eq:long-range-interaction} holds with $J=\mu_{0}=1$.
	Moreover, each site belongs to exactly one dimer, so
	\begin{align}
		\norm{\frac{1-\sigma_{i}^{z}}{2}}+\sum_{j\ne i}\norm{h_{i,j}}=1+\varepsilon_{n}\le2.
	\end{align}
	
	Now choose a bipartition containing exactly one endpoint of every dimer.
	Across this cut, each dimer has Schmidt coefficients $\sqrt{1-p_{n}}$ and $\sqrt{p_{n}}$, and hence contributes $\hbin(p_{n}):=-p_{n}\log p_{n}-(1-p_{n})\log(1-p_{n})$ to the entanglement entropy.
	Since there are $n/2$ dimers,
	\begin{align}
		S_{A}(\ket{\Omega_{\mathrm{dimer}}})=\frac{n}{2}\hbin(p_{n}).
	\end{align}
	
	As $\varepsilon_{n}\to0$,
	\begin{align}
		p_{n}=\frac{1}{2}\left(1-\frac{1}{\sqrt{1+\varepsilon_{n}^{2}}}\right)=\frac{\varepsilon_{n}^{2}}{4}+\mathcal{O}(\varepsilon_{n}^{4}).
	\end{align}
	Therefore
	\begin{align}
		\hbin(p_{n})=\Theta\left(\varepsilon_{n}^{2}\log\frac{1}{\varepsilon_{n}}\right).
	\end{align}
	Using $\varepsilon_{n}=n^{-\beta}$ and $\beta=1-\alpha/D$,
	\begin{align}
		S_{A}(\ket{\Omega_{\mathrm{dimer}}}) = \Theta\left(n\varepsilon_{n}^{2}\log\frac{1}{\varepsilon_{n}}\right) = \Theta\left(n^{1-2\beta}\log n\right)=\Theta\left(n^{2\alpha/D-1}\log n\right).
	\end{align}
	This proves Eq.~\eqref{eq:dimer-arbitrary-cut-entropy}.
\end{proof}

For a cut with $|\partial A|=\mathcal{O}(n^{(D-1)/D})$, at most $|\partial A|$ of the matching edges cross the cut.
The same construction thus gives
\begin{align}
	S_{A}(\ket{\Omega_{\mathrm{dimer}}}) = \mathcal{O} \left(n^{(2\alpha-D-1)/D}\log n\right). \label{eq:dimer-regular-cut-entropy}
\end{align}
This illustrates why $D/2$ is the relevant threshold for arbitrary cuts, whereas $(D+1)/2$ appears naturally in boundary-sensitive estimates.
The regular-cut theorem below uses an explicit bound on the number of hierarchy blocks straddling the cut.
This is a geometric restriction, not a restriction on the quantum state.
In what follows, no assumption of low operator Schmidt rank is imposed on a microscopic pair interaction.
Removing dependence on the growing \emph{renormalized} local dimension is a separate issue from imposing a collective-structure condition on the microscopic interaction.

\section{\texorpdfstring{$\mR\mE$}{RE} decomposition}
\label{sec:preMFRG-RE}

\subsection{Centering the Hamiltonian}

A direct triangle-inequality estimate of an interaction between two large blocks counts all microscopic pairs.
Such an estimate does not benefit from a low-flip constraint.
We therefore reorganize the Hamiltonian so that every site in the support of an interaction must participate in a deviation from its reference state, on at least one side of a matrix element.
This is the role of the $\mathcal{R}\mathcal{E}$ decomposition.

For an operator $\Phi$, define
\begin{align}
	\mathcal{E}_{i}(\Phi):=\bra{0}_{i}\Phi\ket{0}_{i}\otimes\mathds{1}_{i}, \qquad \mathcal{R}_{i}(\Phi):=\Phi-\mathcal{E}_{i}(\Phi). \label{def_superop_R_0_i}
\end{align}
The tensor factors in the first expression are placed in their original order.
Thus both maps act on operators on the full Hilbert space.
Maps associated with distinct sites commute, and
\begin{align}
	\mathcal{E}_{i}^{2}=\mathcal{E}_{i}, \qquad \mathcal{E}_{i}\mathcal{R}_{i}=0, \qquad \norm{\mathcal{E}_{i}(\Phi)}\le\norm{\Phi}, \qquad \norm{\mathcal{R}_{i}(\Phi)}\le2\norm{\Phi}. \label{norm_change_RE}
\end{align}
For example, the contraction estimate follows by evaluating the matrix elements of $\Phi$ between vectors of the form $\ket{0}_{i}\otimes\ket{\psi}$.

For each pair,
\begin{align}
	h_{i,j} = \br{\mathcal{R}_{i} + \mathcal{E}_{i}} \br{\mathcal{R}_{j} + \mathcal{E}_{j}} (h_{i,j}) = \mathcal{R}_{i} \mathcal{R}_{j}(h_{i,j}) + \mathcal{R}_{i} \mathcal{E}_{j}(h_{i,j}) + \mathcal{E}_{i} \mathcal{R}_{j}(h_{i,j}) + \mathcal{E}_{i} \mathcal{E}_{j}(h_{i,j}). \label{h_Z_Z_p_decomp}
\end{align}
The four terms are respectively a centered two-site term, two one-site terms, and a scalar.
Absorbing the one-site contributions into the on-site terms, define
\begin{align}
	\acute{h}_{i,j}&:=\mathcal{R}_{i}\mathcal{R}_{j}(h_{i,j}), \label{eq:microscopic-centered-pair-definition} \\
	\acute{h}_{i}&:=\mathcal{R}_{i}(h_{i}) +\sum_{j\ne i}\mathcal{R}_{i}\mathcal{E}_{j}(h_{i,j}), \label{acute_h_i_i'}\\
	\acute{H}&:=\sum_{i}\acute{h}_{i}+\sum_{i<j}\acute{h}_{i,j}. \label{acute_H_def}
\end{align}
Then
\begin{align}
	\acute{H}=H-c_{\mathrm{ref}}\mathds{1}, \qquad c_{\mathrm{ref}}:=\bra{\mathbf{0}}_{\Lambda} H\ket{\mathbf{0}}_{\Lambda}. \label{eq:preMFRG-RE-scalar-shift}
\end{align}
Indeed, the terms removed from $H$ are precisely $\sum_{i}\mathcal{E}_{i}(h_{i})+\sum_{i<j}\mathcal{E}_{i}\mathcal{E}_{j}(h_{i,j})$.
The decomposition therefore preserves the ground state and the gap.

By construction,
\begin{align}
	\mathcal{E}_{i}(\acute{h}_{i})&=0, \qquad \mathcal{E}_{i}(\acute{h}_{i,j}) =\mathcal{E}_{j}(\acute{h}_{i,j})=0, \label{decompostion_H_E_R}\\
	P_{i}\acute{h}_{i}P_{i}&=0, \qquad P_{i}\acute{h}_{i,j}P_{i}=P_{j}\acute{h}_{i,j}P_{j}=0. \label{acute_h_i_i'_proj}
\end{align}
Moreover,
\begin{align}
	\norm{\acute{h}_{i,j}} \le \frac{4J\mu_{i,j}}{n^{\beta}}r_{i,j}^{-\alpha}, \qquad \sum_{j\ne i}\norm{\acute{h}_{i,j}} \le4J\gamma_{0}, \qquad \norm{\acute{h}_{i}} \le2\norm{h_{i}}+2\sum_{j\ne i}\norm{h_{i,j}} \le 2\bar{g}. \label{acute_h_norm_bounds}
\end{align}
Thus $\acute{H}$ is still two-local and extensive, with local scale at most $2g_{0}+6J\gamma_{0}$.

For disjoint sets $X,Y$, let
\begin{align}
	H_{X,Y}:=\sum_{i\in X,j\in Y}h_{i,j}, \qquad \acute{H}_{X,Y}:=\sum_{i\in X,j\in Y}\acute{h}_{i,j}. \label{acute_H_L_L'}
\end{align}
In particular,
\begin{align}
	H_{i,Y}:=\sum_{j\in Y}h_{i,j}, \qquad \acute{H}_{i,Y}:=\sum_{j\in Y}\acute{h}_{i,j}, \qquad Y\subseteq i^{\mathrm{c}}. \label{eq:microscopic-site-set-interaction}
\end{align}
For $X\subseteq\Lambda$, define the internal Hamiltonians
\begin{align}
	H_{X} := \sum_{i\in X}h_{i} +\sum_{\substack{i<j\\i,j\in X}}h_{i,j}, \qquad \acute{H}_{X} := \sum_{i\in X}\acute{h}_{i} +\sum_{\substack{i<j\\i,j\in X}}\acute{h}_{i,j}. \label{eq:microscopic-acute-subsystem-definitions}
\end{align}
Notice that $\acute{h}_{i}$ includes the mean-field contributions of all bonds touching $i$, including bonds leaving $X$.

\subsubsection{The relation between original and centered interactions}
\label{sec:prop:RE_decomp}

Although the full Hamiltonians differ only by a scalar, a subsystem interaction changes by one-site terms as well.
Summing Eq.~\eqref{h_Z_Z_p_decomp} gives
\begin{align}
	H_{X,Y} = \acute{H}_{X,Y} +\bra{\mathbf{0}}_{Y}H_{X,Y}\ket{\mathbf{0}}_{Y}\otimes\mathds{1}_{Y} +\mathds{1}_{X}\otimes \bra{\mathbf{0}}_{X}H_{X,Y}\ket{\mathbf{0}}_{X} - \bra{\mathbf{0}}_{X\cup Y}H_{X,Y}\ket{\mathbf{0}}_{X\cup Y} \mathds{1}_{X\cup Y}. \label{Decompose_RE_vs_original}
\end{align}
Partial expectations are again embedded with identities on the averaged factors.
Keeping these one-site terms is essential when estimating a block Hamiltonian.

\subsection{A basic estimate for centered pair interactions}

The following observation is the algebraic reason that centering is useful.
It involves only projections and operator norms and is independent of local Hilbert-space dimension.

\begin{lemma}
	\label{lem:preMFRG-centered-pair}
	Let $T=T^{\dagger}$ be supported on two distinct sites $i,j$ and satisfy $P_{i}TP_{i}=P_{j}TP_{j}=0$.
	For a normalized vector $\ket{\psi}$, set $p_{k}:=\bra{\psi} Q_{k}\ket{\psi}$ for $k\in\{i,j\}$ and $p_{i,j}:=\bra{\psi} Q_{i}Q_{j}\ket{\psi}$.
	Then
	\begin{align}
		|\bra{\psi} T\ket{\psi}| \le2\norm{T}\left(\sqrt{p_{i,j}}+\sqrt{p_{i}p_{j}}\right). \label{eq:preMFRG-centered-pair-bound}
	\end{align}
\end{lemma}

\begin{proof}
	Insert $\mathds{1}=P_{i}+Q_{i}$ and $\mathds{1}=P_{j}+Q_{j}$ on both sides of $T$.
	Among the resulting sixteen terms, any term containing $P_{i}TP_{i}$ or $P_{j}TP_{j}$ vanishes.
	Hence the nine surviving terms are
	\begin{align}
		T={}&P_{i}P_{j}TQ_{i}Q_{j} +P_{i}Q_{j}TQ_{i}P_{j} +P_{i}Q_{j}TQ_{i}Q_{j} \notag\\
		&+Q_{i}P_{j}TP_{i}Q_{j} +Q_{i}P_{j}TQ_{i}Q_{j} \notag\\
		&+Q_{i}Q_{j}TP_{i}P_{j} +Q_{i}Q_{j}TP_{i}Q_{j} +Q_{i}Q_{j}TQ_{i}P_{j} +Q_{i}Q_{j}TQ_{i}Q_{j}.
	\end{align}
	Equivalently, grouping these terms gives
	\begin{align}
		T={}&P_{i}P_{j}TQ_{i}Q_{j} +P_{i}Q_{j}TQ_{i} +Q_{i}P_{j}TQ_{j} +Q_{i}Q_{j}T.
	\end{align}
	For the first term,
	\begin{align}
		\left|\bra{\psi} P_{i}P_{j}TQ_{i}Q_{j}\ket{\psi}\right| &\le \norm{T}\, \norm{P_{i}P_{j}\ket{\psi}}\, \norm{Q_{i}Q_{j}\ket{\psi}} \le \norm{T}\sqrt{p_{i,j}}.
	\end{align}
	For the second term,
	\begin{align}
		\left|\bra{\psi} P_{i}Q_{j}TQ_{i}\ket{\psi}\right| &\le \norm{T}\, \norm{P_{i}Q_{j}\ket{\psi}}\, \norm{Q_{i}\ket{\psi}} \le \norm{T}\sqrt{p_{j}p_{i}},
	\end{align}
	where we used $P_{i}Q_{j}\le Q_{j}$.
	Similarly,
	\begin{align}
		\left|\bra{\psi} Q_{i}P_{j}TQ_{j}\ket{\psi}\right| &\le \norm{T}\, \norm{Q_{i}P_{j}\ket{\psi}}\, \norm{Q_{j}\ket{\psi}} \le \norm{T}\sqrt{p_{i}p_{j}},
	\end{align}
	since $Q_{i}P_{j}\le Q_{i}$.
	Finally,
	\begin{align}
		\left|\bra{\psi} Q_{i}Q_{j}T\ket{\psi}\right| &\le \norm{T}\, \norm{Q_{i}Q_{j}\ket{\psi}}\, \norm{\ket{\psi}} \le \norm{T}\sqrt{p_{i,j}}.
	\end{align}
	Adding the four bounds yields
	\begin{align}
		|\bra{\psi} T\ket{\psi}| \le2\norm{T} \left(\sqrt{p_{i,j}}+\sqrt{p_{i}p_{j}}\right),
	\end{align}
	which proves the claim.
\end{proof}

\subsection{Interaction strength on low-flip subspaces}

\subsubsection{Interaction between two blocks}

Let $L,L'$ be disjoint sets and define
\begin{align}
	\gamma_{L,L'}^{(2\alpha)} :=\left(\sum_{i\in L,j\in L'}\mu_{i,j}^{2}r_{i,j}^{-2\alpha}\right)^{1/2}. \label{Prop/main_ineq:renormalized_gamma_def}
\end{align}
We write $\gamma_{i,Y}^{(2\alpha)}$ for $\gamma_{\{i\},Y}^{(2\alpha)}$.
The square sum is the natural geometric quantity because the low-flip constraint controls sums of excitation probabilities, to which Cauchy--Schwarz can be applied.

\begin{prop}
	\label{Prop:renormalized_interaction}
	For $z,z'\in\mathbb{Z}_{\ge0}$,
	\begin{align}
		\norm{\Pi_{L,\le z} \Pi_{L',\le z'} \acute{H}_{L,L'} \Pi_{L,\le z} \Pi_{L',\le z'}} \le\frac{16J\sqrt{zz'}}{n^{\beta}}\gamma_{L,L'}^{(2\alpha)}. \label{Prop/main_ineq:renormalized_interaction}
	\end{align}
\end{prop}

\begin{proof}
	Take a normalized $\ket{\psi}$ in the range of $\Pi_{L,\le z}\Pi_{L',\le z'}$.
	By Eq.~\eqref{eq:preMFRG-centered-pair-bound},
	\begin{align}
		|\bra{\psi}\acute{H}_{L,L'}\ket{\psi}| &\le \sum_{i\in L,j\in L'} 2\norm{\acute{h}_{i,j}} \left(\sqrt{p_{i,j}}+\sqrt{p_{i}p_{j}}\right)\notag\\
		&\le \frac{8J}{n^{\beta}} \sum_{i\in L,j\in L'} \mu_{i,j}r_{i,j}^{-\alpha} \left(\sqrt{p_{i,j}}+\sqrt{p_{i}p_{j}}\right), \label{Proof_inequality_reduce_1}
	\end{align}
	where the second inequality follows from Eq.~\eqref{acute_h_norm_bounds}.
	
	For the correlated term, Cauchy--Schwarz and Eq.~\eqref{Prop/main_ineq:renormalized_gamma_def} give
	\begin{align}
		\sum_{i\in L,j\in L'} \mu_{i,j}r_{i,j}^{-\alpha}\sqrt{p_{i,j}} &\le \gamma_{L,L'}^{(2\alpha)} \left( \sum_{i\in L,j\in L'}p_{i,j} \right)^{1/2}\notag\\
		&= \gamma_{L,L'}^{(2\alpha)} \left( \bra{\psi} M_{L}M_{L'}\ket{\psi} \right)^{1/2}\notag\\
		&\le \gamma_{L,L'}^{(2\alpha)}\sqrt{zz'}. \label{Q_i_Q_i'_sum}
	\end{align}
	Here we used Eq.~\eqref{eq:preMFRG-flip-projections} and the fact that $\ket{\psi}$ lies in the range of $\Pi_{L,\le z}\Pi_{L',\le z'}$.
	
	For the product term, another application of Cauchy--Schwarz gives
	\begin{align}
		\sum_{i\in L,j\in L'} \mu_{i,j}r_{i,j}^{-\alpha}\sqrt{p_{i}p_{j}} &\le \gamma_{L,L'}^{(2\alpha)} \left( \sum_{i\in L,j\in L'}p_{i}p_{j} \right)^{1/2}\notag\\
		&= \gamma_{L,L'}^{(2\alpha)} \left(\sum_{i\in L}p_{i}\right)^{1/2} \left(\sum_{j\in L'}p_{j}\right)^{1/2}\notag\\
		&= \gamma_{L,L'}^{(2\alpha)} \left(\bra{\psi} M_{L}\ket{\psi}\right)^{1/2} \left(\bra{\psi} M_{L'}\ket{\psi}\right)^{1/2}\notag\\
		&\le \gamma_{L,L'}^{(2\alpha)}\sqrt{zz'}.
	\end{align}
	
	Substituting the last two bounds into Eq.~\eqref{Proof_inequality_reduce_1} yields
	\begin{align}
		|\bra{\psi}\acute{H}_{L,L'}\ket{\psi}| \le \frac{16J\sqrt{zz'}}{n^{\beta}} \gamma_{L,L'}^{(2\alpha)}.
	\end{align}
	Since $\Pi_{L,\le z}\Pi_{L',\le z'} \acute{H}_{L,L'} \Pi_{L,\le z}\Pi_{L',\le z'}$ is Hermitian, the variational characterization of the operator norm proves Eq.~\eqref{Prop/main_ineq:renormalized_interaction}.
\end{proof}

\subsubsection{The internal two-body part of one block}

For a single block, define the internal square-sum strength
\begin{align}
	\bar{\gamma}_{L}:=\max_{i\in L} \left(\sum_{j\in L\setminus\{i\}} \mu_{i,j}^{2}r_{i,j}^{-2\alpha}\right)^{1/2}. \label{eq:preMFRG-internal-gamma}
\end{align}
If $|L|\le1$, we set this quantity to zero.
For every integer $z\ge0$, we claim that
\begin{align}
	\left\|\Pi_{L,\le z} \sum_{\substack{i<j\\i,j\in L}}\acute{h}_{i,j} \Pi_{L,\le z}\right\| \le\frac{8Jz\bar{\gamma}_{L}\sqrt{|L|}}{n^{\beta}}. \label{Starting_inequality_amplitude_second_term}
\end{align}
Indeed, take a normalized $\ket{\psi}\in\operatorname{Ran}\Pi_{L,\le z}$.
By Eqs.~\eqref{eq:preMFRG-centered-pair-bound} and \eqref{acute_h_norm_bounds},
\begin{align}
	\left| \bra{\psi} \sum_{\substack{i<j\\i,j\in L}}\acute{h}_{i,j} \ket{\psi} \right| \le \frac{8J}{n^{\beta}} \sum_{\substack{i<j\\i,j\in L}} \mu_{i,j}r_{i,j}^{-\alpha} \left(\sqrt{p_{i,j}}+\sqrt{p_{i}p_{j}}\right). \label{eq:internal-pair-start}
\end{align}
Since the $Q_{i}$ commute and satisfy $Q_{i}^{2}=Q_{i}$,
\begin{align}
	M_{L}^{2}-M_{L} =2\sum_{\substack{i<j\\i,j\in L}}Q_{i}Q_{j}.
\end{align}
Hence, using $\ket{\psi}\in\operatorname{Ran}\Pi_{L,\le z}$,
\begin{align}
	\sum_{\substack{i<j\\i,j\in L}}p_{i,j} &=\frac{1}{2}\bra{\psi}(M_{L}^{2}-M_{L})\ket{\psi} \le\frac{1}{2} z(z-1) \le\frac{1}{2} z^{2}, \label{eq:internal-correlated-probability-sum}\\
	\sum_{\substack{i<j\\i,j\in L}}p_{i}p_{j} &\le\frac{1}{2} \left(\sum_{i\in L}p_{i}\right)^{2} =\frac{1}{2}\bra{\psi} M_{L}\ket{\psi}^{2} \le\frac{1}{2} z^{2}. \label{eq:internal-product-probability-sum}
\end{align}
Moreover, by Eq.~\eqref{eq:preMFRG-internal-gamma},
\begin{align}
	\sum_{\substack{i<j\\i,j\in L}} \mu_{i,j}^{2}r_{i,j}^{-2\alpha} &= \frac{1}{2}\sum_{i\in L} \sum_{j\in L\setminus\{i\}} \mu_{i,j}^{2}r_{i,j}^{-2\alpha} \le\frac{|L|}{2}\bar{\gamma}_{L}^{2}. \label{eq:internal-coupling-square-sum}
\end{align}
Applying Cauchy--Schwarz to the two sums in Eq.~\eqref{eq:internal-pair-start}, and using
Eqs.~\eqref{eq:internal-correlated-probability-sum}--\eqref{eq:internal-coupling-square-sum}, we obtain
\begin{align}
	\sum_{\substack{i<j\\i,j\in L}} \mu_{i,j}r_{i,j}^{-\alpha} \left(\sqrt{p_{i,j}}+\sqrt{p_{i}p_{j}}\right) &\le \left( \sum_{\substack{i<j\\i,j\in L}} \mu_{i,j}^{2}r_{i,j}^{-2\alpha} \right)^{1/2} \left[ \left(\sum_{\substack{i<j\\i,j\in L}}p_{i,j}\right)^{1/2} + \left(\sum_{\substack{i<j\\i,j\in L}}p_{i}p_{j}\right)^{1/2} \right] \notag\\
	&\le z\bar{\gamma}_{L}\sqrt{|L|}.
\end{align}
Substituting this into Eq.~\eqref{eq:internal-pair-start} gives
\begin{align}
	\left| \bra{\psi} \sum_{\substack{i<j\\i,j\in L}}\acute{h}_{i,j} \ket{\psi} \right| \le \frac{8Jz\bar{\gamma}_{L}\sqrt{|L|}}{n^{\beta}}.
\end{align}
Taking the supremum over normalized vectors in the retained subspace proves Eq.~\eqref{Starting_inequality_amplitude_second_term}, since the compressed operator is Hermitian.

\subsection{The on-site term: why the ground-state equation is needed}

Centering eliminates the purely reference-sector component $P_{i}\acute{h}_{i}P_{i}$, but the off-diagonal matrix element $Q_{i}\acute{h}_{i}P_{i}$ generally remains.
If it were controlled only by the trivial bound $\norm{\acute{h}_{i}}\le2\bar{g}$, then summing over a block with at most $z$ deviations would produce a contribution of order $\bar{g}\sqrt{z|L|}$.
Since this grows with the block size, we need a stronger estimate on the off-diagonal part of $\acute{h}_{i}$.

The additional structure comes from the way the reference vectors were chosen.
The vector $\ket{0}_{i}$ is the leading Schmidt vector of the exact ground state, so its coupling to the orthogonal sector is constrained by the ground-state equation.
Projecting $(H-E_{0}\mathds{1}) \ket{\Omega}=0$ onto the corresponding complement Schmidt vector $\ket{\phi_{0}}_{i^{\mathrm{c}}}$ relates the off-diagonal part of the on-site term to the interaction between $i$ and $i^{\mathrm{c}}$, conditioned on $\ket{\phi_{0}}_{i^{\mathrm{c}}}$.
The next lemma first bounds this conditional interaction; the resulting estimate will then be used to control $Q_{i}\acute{h}_{i}P_{i}$ by a quantity proportional to $\sqrt{q_{*}}$ rather than by the crude scale $\bar{g}$.

We use the normalized coupling row sum $\gamma_{0}$ defined in Eq.~\eqref{eq:local-energy-scale}.

\begin{lemma}
	\label{lem:norm_conditional_phi0}
	For the complement Schmidt vector $\ket{\phi_{0}}_{i^{\mathrm{c}}}$,
	\begin{align}
		\left\| \bra{\phi_{0}}_{i^{\mathrm{c}}}\acute{H}_{i,i^{\mathrm{c}}} \ket{\phi_{0}}_{i^{\mathrm{c}}} \right\| \le \frac{16J\gamma_{0}}{\lambda_{*}}\sqrt{q_{*}}. \label{eq:conditional-l1-improvement}
	\end{align}
\end{lemma}

\begin{proof}
	The conditional operator is Hermitian, so by the variational characterization of the norm it suffices to evaluate it on an arbitrary unit vector $\ket{\chi}_{i}$.
	Set $\ket{\psi} = \ket{\chi_{i}} \otimes \ket{\phi_{0}}_{i^{\mathrm{c}}}$.
	For every $j\ne i$, this product structure gives $p_{i,j}=p_{i}p_{j}$.
	Therefore, using $p_{i}\le1$ together with Eqs.~\eqref{eq:preMFRG-centered-pair-bound} and \eqref{acute_h_norm_bounds},
	\begin{align}
		\left\| \bra{\phi_{0}}_{i^{\mathrm{c}}}\acute{H}_{i,i^{\mathrm{c}}} \ket{\phi_{0}}_{i^{\mathrm{c}}} \right\| \le \frac{16J}{n^{\beta}} \sum_{j\ne i} \mu_{i,j}r_{i,j}^{-\alpha} \left( \bra{\phi_{0}}_{i^{\mathrm{c}}}Q_{j} \ket{\phi_{0}}_{i^{\mathrm{c}}} \right)^{1/2}. \label{eq:conditional-interaction-weighted-sum}
	\end{align}
	Since $P_{i}\ket{\Omega} = \lambda_{0,i}\ket{0}_{i}\ket{\phi_{0}}_{i^{\mathrm{c}}}$, and $P_{i}$ commutes with $Q_{j}$ for $j\ne i$,
	\begin{align}
		\bra{\phi_{0}}_{i^{\mathrm{c}}}Q_{j}\ket{\phi_{0}}_{i^{\mathrm{c}}} = \frac{ \bra{\Omega} P_{i}Q_{j}P_{i}\ket{\Omega} }{\lambda_{0,i}^{2}} \le \frac{ \bra{\Omega} Q_{j}\ket{\Omega} }{\lambda_{0,i}^{2}} = \frac{q_{j}}{\lambda_{0,i}^{2}} \le \frac{q_{*}}{\lambda_{0,i}^{2}}.
	\end{align}
	Substituting this pointwise bound into Eq.~\eqref{eq:conditional-interaction-weighted-sum} and using $n^{-\beta} \sum_{j\ne i}\mu_{i,j}r_{i,j}^{-\alpha} \le \gamma_{0}$ from Eq.~\eqref{eq:local-energy-scale}, we obtain
	\begin{align}
		\left\| \bra{\phi_{0}}_{i^{\mathrm{c}}}\acute{H}_{i,i^{\mathrm{c}}} \ket{\phi_{0}}_{i^{\mathrm{c}}} \right\| &\le \frac{16J\gamma_{0}}{\lambda_{0,i}}\sqrt{q_{*}} \le \frac{16J\gamma_{0}}{\lambda_{*}}\sqrt{q_{*}}.
	\end{align}
	Here we also used $\lambda_{0,i}\ge\lambda_{*}$.
	This proves Eq.~\eqref{eq:conditional-l1-improvement}.
\end{proof}

\begin{claim}
	\label{prop:mean_field_local_Ham_re}
	The centered on-site term satisfies
	\begin{align}
		\norm{Q_{i}\acute{h}_{i}P_{i}} \le\left\| \bra{\phi_{0}}_{i^{\mathrm{c}}} \acute{H}_{i,i^{\mathrm{c}}} \ket{\phi_{0}}_{i^{\mathrm{c}}} \right\| +\frac{\bar{g}}{\lambda_{0,i}} \left(1+\frac{1}{\lambda_{0,i}}\right)\sqrt{q_{i}}. \label{eq:onsite-offdiagonal-local-bound}
	\end{align}
\end{claim}

\begin{proof}
	Let
	\begin{align}
		\widehat{H}_{i}:=h_{i}+\sum_{j\ne i}h_{i,j}, \qquad H-E_{0}\mathds{1}=\widehat{H}_{i}+H_{i^{\mathrm{c}}}-E_{0}\mathds{1},
	\end{align}
	where $E_{0}$ is the ground-state energy of $H$.
	Then $H_{i^{\mathrm{c}}}-E_{0}\mathds{1}$ acts only on $i^{\mathrm{c}}$ and $\norm{\widehat{H}_{i}}\le\bar{g}$.
	Write $R_{0}:=\mathds{1}_{i}\otimes \ket{\phi_{0}}_{i^{\mathrm{c}}}\bra{\phi_{0}}_{i^{\mathrm{c}}}$.
	
	Introduce the conditional local operator
	\begin{align}
		\check h_{i}:=\acute{h}_{i}+ \bra{\phi_{0}}_{i^{\mathrm{c}}} \acute{H}_{i,i^{\mathrm{c}}} \ket{\phi_{0}}_{i^{\mathrm{c}}}.
	\end{align}
	Decompose the centered Hamiltonian as $\acute{H}=\acute{h}_{i}+\acute{H}_{i,i^{\mathrm{c}}}+\acute{H}_{i^{\mathrm{c}}}$, according to whether the terms are supported on $i$, connect $i$ to $i^{\mathrm{c}}$, or are supported entirely on $i^{\mathrm{c}}$.
	Since $H$ and $\acute{H}$ differ only by a scalar, and $\acute{H}_{i^{\mathrm{c}}}$ acts entirely on $i^{\mathrm{c}}$, taking the expectation in $\ket{\phi_{0}}_{i^{\mathrm{c}}}$ shows that $\bra{\phi_{0}}_{i^{\mathrm{c}}}(H-E_{0}\mathds{1}) \ket{\phi_{0}}_{i^{\mathrm{c}}}=\check h_{i}+c_{i}\mathds{1}_{i}$ for some real scalar $c_{i}$.
	Hence, since $Q_{i}P_{i}=0$, the two operators have the same $Q_{i}(\cdot)P_{i}$ block.
	Since $P_{i}=\ket{0}_{i}\bra{0}_{i}$, the operator $Q_{i}\check h_{i}P_{i}$ has norm $\norm{Q_{i}\check h_{i}\ket{0}_{i}}$.
	Moreover, $P_{i}\ket{\Omega}=\lambda_{0,i}\ket{0}_{i}\ket{\phi_{0}}_{i^{\mathrm{c}}}$, and $R_{0}$ projects the complement onto $\ket{\phi_{0}}_{i^{\mathrm{c}}}$.
	Therefore,
	\begin{align}
		\norm{Q_{i}\check h_{i}P_{i}} &=\frac{1}{\lambda_{0,i}} \norm{Q_{i}R_{0}(H-E_{0}\mathds{1})P_{i}\ket{\Omega}} \notag\\
		&= \frac{1}{\lambda_{0,i}} \norm{Q_{i}R_{0}(H-E_{0}\mathds{1})Q_{i}\ket{\Omega}} \notag\\
		&\le \frac{1}{\lambda_{0,i}} \norm{R_{0}(H-E_{0}\mathds{1})Q_{i}\ket{\Omega}}. \label{eq:preMFRG-onsite-ground-equation}
	\end{align}
	Here the second equality follows from $(H-E_{0}\mathds{1})\ket{\Omega}=0$ and $P_{i}+Q_{i}=\mathds{1}$, which imply $(H-E_{0}\mathds{1})P_{i}\ket{\Omega} =-(H-E_{0}\mathds{1})Q_{i}\ket{\Omega}$.
	The last inequality uses that $Q_{i}$ is an orthogonal projection.
	
	The complement contribution is controlled without taking the extensive norm of $H_{i^{\mathrm{c}}}-E_{0}\mathds{1}$.
	Since $H_{i^{\mathrm{c}}}-E_{0}\mathds{1}$ commutes with $P_{i}$,
	\begin{align}
		\lambda_{0,i}\norm{(H_{i^{\mathrm{c}}}-E_{0}\mathds{1}) \ket{\phi_{0}}_{i^{\mathrm{c}}}} =\norm{(H_{i^{\mathrm{c}}}-E_{0}\mathds{1})P_{i}\ket{\Omega}} = \norm{P_{i}(H_{i^{\mathrm{c}}}-E_{0}\mathds{1})\ket{\Omega}} = \norm{P_{i}\widehat{H}_{i}\ket{\Omega}} \le\bar{g}. \label{eq:preMFRG-conditional-complement-energy}
	\end{align}
	Here the third equality follows from $(H-E_{0}\mathds{1})\ket{\Omega}=0$ together with $H-E_{0}\mathds{1} =\widehat{H}_{i}+H_{i^{\mathrm{c}}}-E_{0}\mathds{1}$, while the last inequality follows from $\norm{\widehat{H}_{i}}\le \norm{h_{i}}+\sum_{j\ne i}\norm{h_{i,j}}\le\bar{g}$, by Eq.~\eqref{eq:preMFRG-local-and-global-norm}.
	Furthermore, since $q_{i}=\sum_{a>0}\lambda_{a,i}^{2}$, the Schmidt expansion gives
	\begin{align}
		\norm{R_{0}(H_{i^{\mathrm{c}}}-E_{0}\mathds{1})Q_{i}\ket{\Omega}}^{2} &=\sum_{a>0}\lambda_{a,i}^{2} \left| \bra{\phi_{0}}_{i^{\mathrm{c}}} (H_{i^{\mathrm{c}}}-E_{0}\mathds{1}) \ket{\phi_{a}}_{i^{\mathrm{c}}} \right|^{2}\notag\\
		&\le q_{i}\sum_{a>0} \left| \bra{\phi_{0}}_{i^{\mathrm{c}}} (H_{i^{\mathrm{c}}}-E_{0}\mathds{1}) \ket{\phi_{a}}_{i^{\mathrm{c}}} \right|^{2}\notag\\
		&\le q_{i} \norm{(H_{i^{\mathrm{c}}}-E_{0}\mathds{1}) \ket{\phi_{0}}_{i^{\mathrm{c}}}}^{2},
	\end{align}
	where Hermiticity of $H_{i^{\mathrm{c}}}-E_{0}\mathds{1}$ and the orthonormality of the complement Schmidt vectors are used in the last step.
	Also,
	\begin{align}
		\norm{R_{0}\widehat{H}_{i}Q_{i}\ket{\Omega}} \le \norm{\widehat{H}_{i}}\norm{Q_{i}\ket{\Omega}} \le \bar{g}\sqrt{q_{i}},
	\end{align}
	where we used $\norm{R_{0}}\le1$, $\norm{\widehat{H}_{i}}\le\bar{g}$, and $\norm{Q_{i}\ket{\Omega}}=\sqrt{q_{i}}$.
	Therefore, using $H-E_{0}\mathds{1}=\widehat{H}_{i}+H_{i^{\mathrm{c}}}-E_{0}\mathds{1}$ in Eq.~\eqref{eq:preMFRG-onsite-ground-equation}, we obtain
	\begin{align}
		\norm{Q_{i}\check h_{i}P_{i}} &\le\frac{1}{\lambda_{0,i}} \left[ \norm{R_{0}\widehat{H}_{i}Q_{i}\ket{\Omega}} +\norm{R_{0}(H_{i^{\mathrm{c}}}-E_{0}\mathds{1})Q_{i}\ket{\Omega}} \right]\notag\\
		&\le\frac{1}{\lambda_{0,i}} \left[ \bar{g}\sqrt{q_{i}} +\frac{\bar{g}}{\lambda_{0,i}}\sqrt{q_{i}} \right] =\frac{\bar{g}}{\lambda_{0,i}} \left(1+\frac{1}{\lambda_{0,i}}\right)\sqrt{q_{i}}.
	\end{align}
	
	Finally, by the definition of $\check h_{i}$, $Q_{i}\acute{h}_{i}P_{i} =Q_{i}\check h_{i}P_{i} -Q_{i}\bra{\phi_{0}}_{i^{\mathrm{c}}} \acute{H}_{i,i^{\mathrm{c}}} \ket{\phi_{0}}_{i^{\mathrm{c}}}P_{i}$.
	Thus the triangle inequality yields
	\begin{align}
		\norm{Q_{i}\acute{h}_{i}P_{i}} \le \left\| \bra{\phi_{0}}_{i^{\mathrm{c}}} \acute{H}_{i,i^{\mathrm{c}}} \ket{\phi_{0}}_{i^{\mathrm{c}}} \right\| +\frac{\bar{g}}{\lambda_{0,i}} \left(1+\frac{1}{\lambda_{0,i}}\right)\sqrt{q_{i}}.
	\end{align}
	This proves Eq.~\eqref{eq:onsite-offdiagonal-local-bound}.
\end{proof}

Combining Lemma~\ref{lem:norm_conditional_phi0} with Claim~\ref{prop:mean_field_local_Ham_re}, and using $\lambda_{0,i}\ge\lambda_{*}$ and $q_{i}\le q_{*}$, we obtain
\begin{align}
	\norm{Q_{i}\acute{h}_{i}P_{i}} \le \frac{1}{\lambda_{*}} \left[16J\gamma_{0} +\bar{g}\left(1+\frac{1}{\lambda_{*}}\right)\right]\sqrt{q_{*}}.
\end{align}
We therefore define
\begin{align}
	\delta_{*}:= \frac{1}{\lambda_{*}} \left[16J\gamma_{0} +\bar{g}\left(1+\frac{1}{\lambda_{*}}\right)\right]\sqrt{q_{*}}. \label{eq:block-strength-parameters}
\end{align}
Then, uniformly in $i$,
\begin{align}
	\max_{i}\norm{Q_{i}\acute{h}_{i}P_{i}} =\max_{i}\norm{P_{i}\acute{h}_{i}Q_{i}} \le\delta_{*}. \label{eq:onsite-offdiagonal-delta-star}
\end{align}
Thus $\delta_{*}$ is merely a shorthand for the proved off-diagonal matrix-element bound, rather than an additional assumption on the ground state.

If $\lambda_{*}^{2}\ge\frac{1}{2}$, then $\lambda_{*}^{-1}\le\sqrt{2}$, $\lambda_{*}^{-2}\le2$, and $J\gamma_{0}\le\bar{g}$.
Therefore
\begin{align}
	\delta_{*} \le (2+17\sqrt{2})\bar{g}\sqrt{q_{*}}. \label{eq:dimension-free-onsite-simplification}
\end{align}
This form is useful in the RG induction because it controls the on-site off-diagonal block in terms of the local energy scale $\bar{g}$.

\begin{prop}
	\label{Prop:renormalized_interaction/on-site}
	For any nonempty $L\subseteq\Lambda$ and integer $z\ge0$, let $\bar{\gamma}_{L}$ denote the internal coupling square-sum strength defined in Eq.~\eqref{eq:preMFRG-internal-gamma}.
	Then
	\begin{align}
		\norm{\Pi_{L,\le z}\acute{H}_{L}\Pi_{L,\le z}} \le\frac{8Jz\bar{\gamma}_{L}\sqrt{|L|}}{n^{\beta}} +2\bar{g} z+2\delta_{*}\sqrt{z|L|}. \label{Prop/main_ineq:renormalized_interaction_onsite}
	\end{align}
\end{prop}

\begin{proof}
	Recall from Eq.~\eqref{eq:microscopic-acute-subsystem-definitions} that
	\begin{align}
		\acute{H}_{L} = \sum_{i\in L}\acute{h}_{i} + \sum_{\substack{i<j\\ i,j\in L}}\acute{h}_{i,j}.
	\end{align}
	We bound the two contributions separately.
	
	For the internal two-body part, Eq.~\eqref{Starting_inequality_amplitude_second_term} gives
	\begin{align}
		\left\| \Pi_{L,\le z} \left( \sum_{\substack{i<j\\ i,j\in L}}\acute{h}_{i,j} \right) \Pi_{L,\le z} \right\| \le \frac{8Jz\bar{\gamma}_{L}\sqrt{|L|}}{n^{\beta}}.
	\end{align}
	
	For the on-site part, Eq.~\eqref{acute_h_i_i'_proj} gives $P_{i}\acute{h}_{i}P_{i}=0$, and hence $\acute{h}_{i} =Q_{i}\acute{h}_{i}Q_{i}+P_{i}\acute{h}_{i}Q_{i}+Q_{i}\acute{h}_{i}P_{i}$.
	Take a normalized $\ket{\psi}\in\operatorname{Ran}\Pi_{L,\le z}$.
	Using $\norm{\acute{h}_{i}}\le2\bar{g}$ from Eq.~\eqref{acute_h_norm_bounds} and $\norm{P_{i}\acute{h}_{i}Q_{i}} =\norm{Q_{i}\acute{h}_{i}P_{i}}\le\delta_{*}$ from Eq.~\eqref{eq:onsite-offdiagonal-delta-star}, we obtain
	\begin{align}
		\left|\bra{\psi}\sum_{i\in L}\acute{h}_{i}\ket{\psi}\right| &\le \sum_{i\in L}\norm{\acute{h}_{i}}\norm{Q_{i}\ket{\psi}}^{2} +2\sum_{i\in L} \norm{P_{i}\acute{h}_{i}Q_{i}} \norm{P_{i}\ket{\psi}}\norm{Q_{i}\ket{\psi}} \notag\\
		&\le 2\bar{g}\sum_{i\in L}\bra{\psi} Q_{i}\ket{\psi} +2\delta_{*}\sum_{i\in L}\norm{Q_{i}\ket{\psi}} \notag\\
		&\le 2\bar{g} z +2\delta_{*}\sqrt{|L|} \left( \sum_{i\in L}\bra{\psi} Q_{i}\ket{\psi} \right)^{1/2} \notag\\
		&\le 2\bar{g} z+2\delta_{*}\sqrt{z|L|}. \label{eq:onsite-split-bound}
	\end{align}
	Here we used Cauchy--Schwarz, together with $M_{L}=\sum_{i\in L}Q_{i}$ and $\bra{\psi} M_{L}\ket{\psi}\le z$.
	
	Since $\Pi_{L,\le z}\left(\sum_{i\in L}\acute{h}_{i}\right)\Pi_{L,\le z}$ is Hermitian, the preceding expectation bound also bounds its operator norm.
	Combining it with the internal two-body estimate above and using the triangle inequality proves Eq.~\eqref{Prop/main_ineq:renormalized_interaction_onsite}.
\end{proof}

\subsection{Transitions between deviation-number sectors}
\label{sec:preMFRG-shell-hopping}

The preceding estimates control the Hamiltonian after restriction to a low-flip subspace of a block $L$.
To prove that this restriction is an accurate approximation to the ground state, we must instead control the matrix elements that transfer amplitude from a low-flip sector of $L$ to higher-flip sectors.
We therefore estimate $\Pi_{L,z'}H\Pi_{L,z}$ for $z'\ne z$.

The crossing interaction between $L$ and $L^{\mathrm{c}}$ also depends on the occupation of the complement.
Without any restriction on $L^{\mathrm{c}}$, an operator-norm estimate must allow configurations with an extensive number of complement excitations, which is too crude for the desired shell-hopping bound.
We therefore introduce an auxiliary low-flip cutoff on $L^{\mathrm{c}}$ and later prove that the ground state has only a small component outside this cutoff.

Fix a nonempty proper block $L$ and, for now, an arbitrary complement integer cutoff $0\le m_{0}\le|L^{\mathrm{c}}|$.
All projections involving $M_{L}$ commute with $\Pi_{L^{\mathrm{c}},\le m_{0}}$.
Because each Hamiltonian term acts on at most two sites,
\begin{align}
	\Pi_{L,z'}H\Pi_{L,z}=0 \qquad\text{if }|z'-z|>2. \label{eq:two-local-bandedness}
\end{align}
Indeed, a term supported on $Z$ leaves the flip number on $L\setminus Z$ unchanged, and the flip number on $L\cap Z$ can change by at most $|L\cap Z|\le2$.
Also, for $z'\ne z$,
\begin{align}
	\Pi_{L^{\mathrm{c}},\le m_{0}}\Pi_{L,z'}H\Pi_{L,z}\Pi_{L^{\mathrm{c}},\le m_{0}} =\Pi_{L^{\mathrm{c}},\le m_{0}}\Pi_{L,z'}\acute{H}\Pi_{L,z}\Pi_{L^{\mathrm{c}},\le m_{0}}, \label{eq:preMFRG-shell-shift-invariance}
\end{align}
since the scalar in Eq.~\eqref{eq:preMFRG-RE-scalar-shift} has no matrix element between distinct sectors.

\subsubsection{An off-diagonal sector estimate}

\begin{lemma}
	\label{lemm_flip_prob}
	For a Hermitian operator $\Phi$ and distinct integers $z,z'\in\{0,\ldots,|L|\}$,
	\begin{align}
		\norm{\Pi_{L,z'}\Phi\Pi_{L,z}} \le(\sqrt{z}+\sqrt{z'})\sqrt{|L|} \max_{i\in L}\norm{P_{i}\Phi Q_{i}}. \label{eq:flip-sector-lemma}
	\end{align}
\end{lemma}

\begin{proof}
	Order the sites of $L$ as $1,\ldots,|L|$, and define
	\begin{align}
		\mathcal{L}_{i}(\Phi) := P_{i}\Phi P_{i}+Q_{i}\Phi Q_{i}, \qquad \mathcal{K}_{i}(\Phi) :=P_{i}\Phi Q_{i}+Q_{i}\Phi P_{i}.
	\end{align}
	Since $\mathcal{K}_{i}=\mathrm{Id}-\mathcal{L}_{i}$, the identity $\mathcal{L}_{1}\cdots\mathcal{L}_{i-1} -\mathcal{L}_{1}\cdots\mathcal{L}_{i} =\mathcal{L}_{1}\cdots\mathcal{L}_{i-1}\mathcal{K}_{i}$ telescopes to
	\begin{align}
		\Phi=\mathcal{L}_{1}\cdots\mathcal{L}_{|L|}(\Phi) +\sum_{i=1}^{|L|}\mathcal{L}_{1}\cdots\mathcal{L}_{i-1}\mathcal{K}_{i}(\Phi). \label{eq:pinching-telescope}
	\end{align}
	The first term is block diagonal with respect to every $P_{i}\oplus Q_{i}$, and therefore commutes with $M_{L}=\sum_{i\in L}Q_{i}$.
	Since $z\ne z'$, it has no matrix element between the two sectors.
	Set $\Phi_{<i}:=\mathcal{L}_{1}\cdots\mathcal{L}_{i-1}(\Phi)$.
	Since $\mathcal{K}_{i}$ commutes with $\mathcal{L}_{j}$ for every $j<i$, the $i$th summand above satisfies $\mathcal{L}_{1}\cdots\mathcal{L}_{i-1}\mathcal{K}_{i}(\Phi) =\mathcal{K}_{i}(\Phi_{<i}) =P_{i}\Phi_{<i}Q_{i}+Q_{i}\Phi_{<i}P_{i}$.
	Thus, for normalized $\ket{\psi}\in\operatorname{Ran}\Pi_{L,z}$ and $\ket{\psi'}\in\operatorname{Ran}\Pi_{L,z'}$,
	\begin{align}
		\abs{\bra{\psi'}\Phi\ket{\psi}} &\le\sum_{i\in L}\bigl[ \norm{P_{i}\Phi_{<i}Q_{i}}\norm{Q_{i}\ket{\psi}} +\norm{Q_{i}\ket{\psi'}}\norm{Q_{i}\Phi_{<i}P_{i}}\bigr].
	\end{align}
	Each map $\mathcal{L}_{j}$ is contractive, since $\mathcal{L}_{j}(X)=\tfrac{1}{2}[X+(P_{j}-Q_{j})X(P_{j}-Q_{j})]$ and $P_{j}-Q_{j}$ is unitary.
	Hence $\norm{\mathcal{L}_{j}(X)}\le\norm{X}$.
	Moreover, for $j<i$, $\mathcal{L}_{j}$ commutes with left multiplication by $P_{i}$ and right multiplication by $Q_{i}$.
	Hence $P_{i}\Phi_{<i}Q_{i} =\mathcal{L}_{1}\cdots\mathcal{L}_{i-1}(P_{i}\Phi Q_{i})$, and therefore
	\begin{align}
		\norm{P_{i}\Phi_{<i}Q_{i}}\le\norm{P_{i}\Phi Q_{i}}.
	\end{align}
	Since $\Phi$ and $\Phi_{<i}$ are Hermitian, the same bound holds for the adjoint block.
	Finally, by Cauchy--Schwarz,
	\begin{align}
		\sum_{i\in L}\norm{Q_{i}\ket{\psi}} &\le\sqrt{|L|\sum_{i\in L}\norm{Q_{i}\ket{\psi}}^{2}} =\sqrt{|L|z}, \label{eq:flip-sector-CS-unprimed}\\
		\sum_{i\in L}\norm{Q_{i}\ket{\psi'}} &\le\sqrt{|L|\sum_{i\in L}\norm{Q_{i}\ket{\psi'}}^{2}} =\sqrt{|L|z'}. \label{eq:flip-sector-CS-primed}
	\end{align}
	Substituting these bounds into the preceding estimate and taking the supremum over the two normalized vectors proves the lemma.
\end{proof}

\subsubsection{Weighted complement occupation}

To exploit the complement cutoff, we need a bound that remembers that at most $m_{0}$ complement sites are excited.
Define
\begin{align}
	F_{\alpha}(n,m):= \begin{cases} n^{1/2-\alpha/D}\sqrt{m},&0\le\alpha<D/2,\\
		m^{1-\alpha/D},&D/2\le\alpha<D. \end{cases} \label{eq:F-alpha-definition}
\end{align}
The two branches reflect whether the weight $r^{D-2\alpha}$ increases or decreases with distance.

\begin{lemma}
	\label{lem:preMFRG-weighted-occupation}
	With the explicit constants
	\begin{align}
		C_{\mathrm{occ},1} &:= \begin{cases} D^{D/2-\alpha},&0\le\alpha<D/2,\\
			(a'_{D})^{(2\alpha-D)/(2D)}(2-2\alpha/D)^{-1/2},&D/2\le\alpha<D, \end{cases} \qquad C_{\mathrm{occ},2}:=\sqrt{a_{D}(1+\log D)}, \label{eq:weighted-occupation-constants}
	\end{align}
	the following bounds hold:
	\begin{align}
		\left\|\Pi_{L^{\mathrm{c}},\le m_{0}}\sum_{j\in L^{\mathrm{c}}} r_{i,j}^{D-2\alpha}Q_{j}\Pi_{L^{\mathrm{c}},\le m_{0}}\right\| &\le C_{\mathrm{occ},1}^{2}F_{\alpha}(n,m_{0})^{2} \qquad \forall i \in L, \label{eq:preMFRG-weighted-occupation-bound}\\
		\sum_{i\in L,j\in L^{\mathrm{c}}}r_{i,j}^{-D} &\le C_{\mathrm{occ},2}^{2}|L|\log(en). \label{eq:preMFRG-crossing-harmonic-sum}
	\end{align}
\end{lemma}

\begin{proof}
	Since all $Q_{j}$ commute, choose their common eigenbasis.
	On $\operatorname{Ran}\Pi_{L^{\mathrm{c}},\le m_{0}}$, each basis vector is labeled by a set $S\subseteq L^{\mathrm{c}}$ with $|S|\le m_{0}$, where $Q_{j}=1$ precisely for $j\in S$.
	Therefore
	\begin{align}
		\left\| \Pi_{L^{\mathrm{c}},\le m_{0}} \sum_{j\in L^{\mathrm{c}}}r_{i,j}^{D-2\alpha}Q_{j} \Pi_{L^{\mathrm{c}},\le m_{0}} \right\| = \max_{\substack{S\subseteq L^{\mathrm{c}}\\ |S|\le m_{0}}} \sum_{j\in S}r_{i,j}^{D-2\alpha}. \label{eq:weighted-occupation-max-subset}
	\end{align}
	If $\alpha<D/2$, then $D-2\alpha>0$.
	Since $r_{i,j}\le D\ell_{\Lambda}$ and $\ell_{\Lambda}^{D}=n$,
	\begin{align}
		\sum_{j\in S}r_{i,j}^{D-2\alpha} \le |S|(D\ell_{\Lambda})^{D-2\alpha} \le C_{\mathrm{occ},1}^{2}m_{0} n^{(D-2\alpha)/D} = C_{\mathrm{occ},1}^{2}m_{0}n^{1-2\alpha/D} = C_{\mathrm{occ},1}^{2}F_{\alpha}(n,m_{0})^{2}. \label{eq:weighted-occupation-alpha-below-half}
	\end{align}
	If $\alpha\ge D/2$, set $p:=2\alpha-D\in[0,D)$.
	Then $r_{i,j}^{D-2\alpha}=r_{i,j}^{-p}$ decreases with distance.
	Let $r_{1}\le r_{2}\le\cdots$ be the distances from $i$ to the sites of $L^{\mathrm{c}}$ in increasing order.
	By Eq.~\eqref{eq:preMFRG-lattice-counting},
	\begin{align}
		a \le \bigl|\{j:1\le r_{i,j}\le r_{a}\}\bigr| \le a'_{D} r_{a}^{D}, \qquad\text{so}\qquad r_{a}\ge\left(\frac{a}{a'_{D}}\right)^{1/D}.
	\end{align}
	Hence, for every $S\subseteq L^{\mathrm{c}}$ with $|S|\le m_{0}$,
	\begin{align}
		\sum_{j\in S}r_{i,j}^{D-2\alpha} &\le\sum_{a=1}^{m_{0}}r_{a}^{-p} \le(a'_{D})^{p/D} \sum_{a=1}^{m_{0}}a^{-p/D} \le C_{\mathrm{occ},1}^{2}m_{0}^{1-p/D} = C_{\mathrm{occ},1}^{2}m_{0}^{2-2\alpha/D} = C_{\mathrm{occ},1}^{2}F_{\alpha}(n,m_{0})^{2}. \label{eq:weighted-occupation-alpha-above-half}
	\end{align}
	The case $m_{0}=0$ is immediate.
	Using the displayed value of $C_{\mathrm{occ},1}$ in
	Eqs.~\eqref{eq:weighted-occupation-max-subset}--\eqref{eq:weighted-occupation-alpha-above-half} proves Eq.~\eqref{eq:preMFRG-weighted-occupation-bound}.
	
	For the second estimate, fix $i\in L$.
	Since $\sum_{r=1}^{D\ell_{\Lambda}}r^{-1}\le1+\log D+D^{-1}\log n\le(1+\log D)\log(en)$, the stated value of $C_{\mathrm{occ},2}$ applies.
	Using the shell-counting bound in Eq.~\eqref{eq:preMFRG-lattice-counting},
	\begin{align}
		\sum_{j\in L^{\mathrm{c}}}r_{i,j}^{-D} &\le \sum_{j\ne i}r_{i,j}^{-D} = \sum_{r=1}^{D\ell_{\Lambda}} |\{j:r_{i,j}=r\}|\,r^{-D} \le a_{D}\sum_{r=1}^{D\ell_{\Lambda}}\frac{1}{r} \le C_{\mathrm{occ},2}^{2}\log(en).
	\end{align}
	Summing over $i\in L$ gives
	\begin{align}
		\sum_{i\in L,j\in L^{\mathrm{c}}}r_{i,j}^{-D} \le C_{\mathrm{occ},2}^{2}|L|\log(en),
	\end{align}
	which proves Eq.~\eqref{eq:preMFRG-crossing-harmonic-sum}.
\end{proof}

\subsubsection{The shell-hopping bound}

\begin{prop}
	\label{hopping_amplitude}
	Let $L$ be a hypercubic block, and let $z'>z\ge0$ be integers.
	Then
	\begin{align}
		\norm{\Pi_{L^{\mathrm{c}},\le m_{0}}\Pi_{L,z'}H\Pi_{L,z}\Pi_{L^{\mathrm{c}},\le m_{0}}} &\le \frac{8Jz'\bar{\gamma}_{L}\sqrt{|L|}}{n^{\beta}} +\delta_{*}(\sqrt{z}+\sqrt{z'})\sqrt{|L|}\notag\\
		&\quad\quad+ \frac{16J}{n^{\beta}} \left[ C_{\mathrm{occ},1}C_{\mathrm{occ},2}F_{\alpha}(n,m_{0})\sqrt{z'|L|\log(en)} +(3^{D}-1)(\mu_{0}-1)\sqrt{z'|\partial L|} \right]. \label{Prop/main_ineq:hopping-amplitude}
	\end{align}
	For a general subset $L$, the same estimate holds with $\partial L$ replaced by its Moore boundary $\partial_{\mathrm{M}}L:=\{i\in L:\mathcal{S}_{i}\cap L^{\mathrm{c}}\ne\varnothing\}$.
\end{prop}

\begin{proof}
	Use Eq.~\eqref{eq:preMFRG-shell-shift-invariance} and decompose
	\begin{align}
		\acute{H}=\acute{H}_{L^{\mathrm{c}}} +\sum_{\substack{i<j\\i,j\in L}}\acute{h}_{i,j} +\sum_{i\in L}\acute{h}_{i}+\acute{H}_{L,L^{\mathrm{c}}}.
	\end{align}
	The term $\acute{H}_{L^{\mathrm{c}}}$ acts only on $L^{\mathrm{c}}$, so it commutes with $M_{L}$ and therefore has no matrix element between the distinct sectors $\Pi_{L,z}$ and $\Pi_{L,z'}$.
	Hence $\Pi_{L,z'}\acute{H}_{L^{\mathrm{c}}}\Pi_{L,z}=0$ for $z'\ne z$.
	
	Since $z<z'$, both $\operatorname{Ran}\Pi_{L,z}$ and $\operatorname{Ran}\Pi_{L,z'}$ are contained in $\operatorname{Ran}\Pi_{L,\le z'}$.
	Hence, using that multiplication by $\Pi_{L^{\mathrm{c}},\le m_{0}}$ is contractive,
	\begin{align}
		\left\| \Pi_{L^{\mathrm{c}},\le m_{0}}\Pi_{L,z'} \left(\sum_{\substack{i<j\\ i,j\in L}}\acute{h}_{i,j}\right) \Pi_{L,z}\Pi_{L^{\mathrm{c}},\le m_{0}} \right\| \le \left\| \Pi_{L,\le z'} \left(\sum_{\substack{i<j\\ i,j\in L}}\acute{h}_{i,j}\right) \Pi_{L,\le z'} \right\| \le \frac{8Jz'\bar{\gamma}_{L}\sqrt{|L|}}{n^{\beta}},
	\end{align}
	where the last inequality is Eq.~\eqref{Starting_inequality_amplitude_second_term}.
	
	For the on-site part, set $\Phi:=\sum_{i\in L}\acute{h}_{i}$.
	Since $\acute{h}_{j}$ acts trivially on site $i$ for $j\ne i$, we have $P_{i}\Phi Q_{i}=P_{i}\acute{h}_{i}Q_{i}$.
	Therefore, Eq.~\eqref{eq:flip-sector-lemma} together with Eq.~\eqref{eq:onsite-offdiagonal-delta-star} gives
	\begin{align}
		\left\| \Pi_{L^{\mathrm{c}},\le m_{0}}\Pi_{L,z'} \left(\sum_{i\in L}\acute{h}_{i}\right) \Pi_{L,z}\Pi_{L^{\mathrm{c}},\le m_{0}} \right\| &\le \left\| \Pi_{L,z'} \left(\sum_{i\in L}\acute{h}_{i}\right) \Pi_{L,z} \right\|\notag\\
		&\le (\sqrt{z}+\sqrt{z'})\sqrt{|L|} \max_{i\in L}\norm{P_{i}\acute{h}_{i}Q_{i}}\notag\\
		&\le \delta_{*}(\sqrt{z}+\sqrt{z'})\sqrt{|L|}.
	\end{align}
	
	It remains to control the crossing term $\acute{H}_{L,L^{\mathrm{c}}}$.
	Take a normalized $\ket{\psi}\in \operatorname{Ran}(\Pi_{L,\le z'}\Pi_{L^{\mathrm{c}},\le m_{0}})$ and write $p_{i}:=\bra{\psi} Q_{i}\ket{\psi}$ and $p_{i,j}:=\bra{\psi} Q_{i}Q_{j}\ket{\psi}$.
	By Lemma~\ref{lem:preMFRG-centered-pair} and Eq.~\eqref{acute_h_norm_bounds},
	\begin{align}
		\left|\bra{\psi}\acute{H}_{L,L^{\mathrm{c}}}\ket{\psi}\right| \le \frac{8J}{n^{\beta}} \sum_{i\in L,j\in L^{\mathrm{c}}} \mu_{i,j}r_{i,j}^{-\alpha} \left(\sqrt{p_{i,j}}+\sqrt{p_{i}p_{j}}\right). \label{eq:crossing-pair-start}
	\end{align}
	Thus it remains to bound the two weighted sums on the right-hand side.
	By Eq.~\eqref{mu_i_i'_choice}, $\mu_{i,j}r_{i,j}^{-\alpha} =r_{i,j}^{-\alpha} +(\mu_{0}-1)\mathds{1}_{\{j\in\mathcal{S}_{i}\}}r_{i,j}^{-\alpha}$.
	Since $r_{i,j}\ge1$ and $\alpha\ge0$, $r_{i,j}^{-\alpha}\le1$, and hence
	\begin{align}
		\mu_{i,j}r_{i,j}^{-\alpha} \le r_{i,j}^{-\alpha} +(\mu_{0}-1)\mathds{1}_{\{j\in\mathcal{S}_{i}\}}.
	\end{align}
	We therefore treat the power-law contribution and the additional Moore-neighborhood contribution separately.
	
	For the correlated part of the power-law contribution, Cauchy--Schwarz inequality gives
	\begin{align}
		\sum_{i\in L,j\in L^{\mathrm{c}}} r_{i,j}^{-\alpha}\sqrt{p_{i,j}} &\le \left( \sum_{i\in L,j\in L^{\mathrm{c}}}r_{i,j}^{-D} \right)^{1/2} \left( \sum_{i\in L,j\in L^{\mathrm{c}}} r_{i,j}^{D-2\alpha}p_{i,j} \right)^{1/2}. \label{Q_i_Q_i'_sum_refine}
	\end{align}
	The first factor is bounded by Eq.~\eqref{eq:preMFRG-crossing-harmonic-sum}.
	For the second factor, using that $\ket{\psi}$ lies in $\operatorname{Ran}\Pi_{L^{\mathrm{c}},\le m_{0}}$ and applying Eq.~\eqref{eq:preMFRG-weighted-occupation-bound}, we obtain
	\begin{align}
		\sum_{i\in L,j\in L^{\mathrm{c}}} r_{i,j}^{D-2\alpha}p_{i,j} &= \sum_{i\in L} \bra{\psi} Q_{i} \left( \sum_{j\in L^{\mathrm{c}}} r_{i,j}^{D-2\alpha}Q_{j} \right) \ket{\psi}\notag\\
		&\le C_{\mathrm{occ},1}^{2}F_{\alpha}(n,m_{0})^{2} \sum_{i\in L}\bra{\psi} Q_{i}\ket{\psi}\notag\\
		&\le C_{\mathrm{occ},1}^{2}F_{\alpha}(n,m_{0})^{2}z'.
	\end{align}
	Therefore
	\begin{align}
		\sum_{i\in L,j\in L^{\mathrm{c}}} r_{i,j}^{-\alpha}\sqrt{p_{i,j}} \le C_{\mathrm{occ},1}C_{\mathrm{occ},2}F_{\alpha}(n,m_{0}) \sqrt{z'|L|\log(en)}. \label{eq:preMFRG-correlated-crossing-bound}
	\end{align}
	
	For the product contribution, the same Cauchy--Schwarz decomposition gives
	\begin{align}
		\sum_{i\in L,j\in L^{\mathrm{c}}} r_{i,j}^{-\alpha}\sqrt{p_{i}p_{j}} &\le \left( \sum_{i\in L,j\in L^{\mathrm{c}}}r_{i,j}^{-D} \right)^{1/2} \left( \sum_{i\in L,j\in L^{\mathrm{c}}} r_{i,j}^{D-2\alpha}p_{i}p_{j} \right)^{1/2}.
	\end{align}
	Moreover,
	\begin{align}
		\sum_{i\in L,j\in L^{\mathrm{c}}} r_{i,j}^{D-2\alpha}p_{i}p_{j} &= \sum_{i\in L}p_{i} \sum_{j\in L^{\mathrm{c}}}r_{i,j}^{D-2\alpha}p_{j}\notag\\
		&\le C_{\mathrm{occ},1}^{2}F_{\alpha}(n,m_{0})^{2} \sum_{i\in L}p_{i}\notag\\
		&\le C_{\mathrm{occ},1}^{2}F_{\alpha}(n,m_{0})^{2}z',
	\end{align}
	where we again used Eq.~\eqref{eq:preMFRG-weighted-occupation-bound}.
	Hence
	\begin{align}
		\sum_{i\in L,j\in L^{\mathrm{c}}} r_{i,j}^{-\alpha}\sqrt{p_{i}p_{j}} \le C_{\mathrm{occ},1}C_{\mathrm{occ},2}F_{\alpha}(n,m_{0}) \sqrt{z'|L|\log(en)}. \label{eq:preMFRG-product-crossing-bound}
	\end{align}
	
	It remains to control the additional Moore-neighborhood contribution.
	Let $\mathcal{B}_{\mathrm{M}} :=\{(i,j):i\in L,\ j\in L^{\mathrm{c}}\cap\mathcal{S}_{i}\}$.
	For a hypercubic block, every $i$ participating in such a bond lies on $\partial L$, and each $i$ has at most $3^{D}-1$ Moore neighbors.
	Therefore
	\begin{align}
		|\mathcal{B}_{\mathrm{M}}| \le(3^{D}-1)|\partial L|.
	\end{align}
	For a general subset $L$, the same counting holds with $\partial L$ replaced by $\partial_{\mathrm{M}}L$.
	Using $Q_{i}Q_{j}\le Q_{i}$ and $\sum_{i\in L}p_{i}\le z'$, we also have
	\begin{align}
		\sum_{(i,j)\in\mathcal{B}_{\mathrm{M}}}p_{i,j} \le (3^{D}-1)\sum_{i\in L}p_{i} \le(3^{D}-1)z', \qquad 
		\sum_{(i,j)\in\mathcal{B}_{\mathrm{M}}}p_{i}p_{j} \le (3^{D}-1)\sum_{i\in L}p_{i} \le(3^{D}-1)z',
	\end{align}
	where the second line uses $p_{j}\le1$.
	Thus Cauchy--Schwarz inequality yields
	\begin{align}
		\sum_{(i,j)\in\mathcal{B}_{\mathrm{M}}}\sqrt{p_{i,j}} \le (3^{D}-1)\sqrt{z'|\partial L|}, \qquad 
		\sum_{(i,j)\in\mathcal{B}_{\mathrm{M}}}\sqrt{p_{i}p_{j}} \le (3^{D}-1)\sqrt{z'|\partial L|}. \label{eq:Moore-crossing-bounds}
	\end{align}
	
	Combining Eqs.~\eqref{eq:preMFRG-correlated-crossing-bound}, \eqref{eq:preMFRG-product-crossing-bound}, and \eqref{eq:Moore-crossing-bounds} with Eq.~\eqref{eq:crossing-pair-start}, we obtain
	\begin{align}
		\left|\bra{\psi}\acute{H}_{L,L^{\mathrm{c}}}\ket{\psi}\right| \le \frac{16J}{n^{\beta}} \left[ C_{\mathrm{occ},1}C_{\mathrm{occ},2}F_{\alpha}(n,m_{0})\sqrt{z'|L|\log(en)} +(3^{D}-1)(\mu_{0}-1)\sqrt{z'|\partial L|} \right].
	\end{align}
	Since the restriction of $\acute{H}_{L,L^{\mathrm{c}}}$ to $\operatorname{Ran}(\Pi_{L,\le z'}\Pi_{L^{\mathrm{c}},\le m_{0}})$ is Hermitian, taking the supremum over normalized $\ket{\psi}$ bounds its operator norm, and hence also the desired $z$-to-$z'$ off-diagonal block.
	Combining this with the internal-pair and on-site bounds proves Eq.~\eqref{Prop/main_ineq:hopping-amplitude}.
\end{proof}

By Eq.~\eqref{eq:two-local-bandedness}, a state in the $x$-flip sector can reach only the $x+1$ and $x+2$ sectors in the forward direction.
We therefore collect these two hopping amplitudes into a single bound.
For $x\in\mathbb{Z}_{\ge0}$, define
\begin{align}
	\bar{J}_{m_{0},L}(x) &:= \sum_{a=1}^{2} \Bigg[ \frac{8J(x+a)\bar{\gamma}_{L}\sqrt{|L|}}{n^{\beta}} +2\delta_{*}\sqrt{(x+a)|L|}\notag\\
	&\qquad \qquad +\frac{16J}{n^{\beta}} \left( C_{\mathrm{occ},1}C_{\mathrm{occ},2}F_{\alpha}(n,m_{0}) \sqrt{(x+a)|L|\log(en)} +(3^{D}-1)(\mu_{0}-1) \sqrt{(x+a)|\partial L|} \right) \Bigg]. \label{Parameter_bar_J_m_0/L}
\end{align}
For a general subset $L$, replace $\partial L$ by $\partial_{\mathrm{M}}L$.

Indeed, applying Proposition~\ref{hopping_amplitude} with $(z,z')=(x,x+a)$ and using $\sqrt{x}+\sqrt{x+a}\le2\sqrt{x+a}$ for $a=1,2$, we obtain
\begin{align}
	\sum_{a=1}^{2} \norm{ \Pi_{L^{\mathrm{c}},\le m_{0}} \Pi_{L,x+a}H\Pi_{L,x} \Pi_{L^{\mathrm{c}},\le m_{0}} } \le \bar{J}_{m_{0},L}(x). \label{eq:preMFRG-adjacent-shell-envelope}
\end{align}
Thus $\bar{J}_{m_{0},L}(x)$ bounds the total forward hopping out of the $x$-flip sector under the complement cutoff.

\subsection{From shell hopping to a controlled block truncation}
\label{sec:preMFRG-block-truncation}

The preceding proposition holds for any chosen $m_{0}$.
To use it for the true ground state, we must now justify the complement cutoff, verify that the compressed Hamiltonian remains gapped, and then convert its small hopping into a local tail.
We prove these three steps explicitly.

\subsubsection{Choosing the complement cutoff}

A general exponential-tail estimate is sufficient for this first, coarse cutoff.
The following elementary polynomial argument also fixes the dependence on the requested error.

\begin{lemma}
	\label{lem:preMFRG-coarse-cutoff}
	Set $G:=2\bar{g} n$.
	For $C\subseteq\Lambda$, let
	\begin{align}
		\mu_{C}:=\bra{\Omega} M_{C}\ket{\Omega}, \qquad a_{C}:=\min\{|C|,\lfloor2\mu_{C}\rfloor\}.
	\end{align}
	For every integer $r\ge0$,
	\begin{align}
		\norm{\Pi_{C,>a_{C}+2r}\ket{\Omega}} \le2\exp\left(-2r\sqrt{\frac{\Delta}{G}}\right). \label{eq:preMFRG-coarse-cutoff-tail}
	\end{align}
\end{lemma}

\begin{proof}
	If $a_{C}<|C|$, then $a_{C}=\lfloor2\mu_{C}\rfloor$, so $a_{C}+1>2\mu_{C}$.
	Since $M_{C}\ge(a_{C}+1)\Pi_{C,>a_{C}}$, Markov's inequality gives
	\begin{align}
		\norm{\Pi_{C,>a_{C}}\ket{\Omega}}^{2} =\bra{\Omega}\Pi_{C,>a_{C}}\ket{\Omega} \le\frac{\bra{\Omega} M_{C}\ket{\Omega}}{a_{C}+1} =\frac{\mu_{C}}{a_{C}+1} <\frac{1}{2}.
	\end{align}
	If $a_{C}=|C|$, then $\Pi_{C,>a_{C}}=0$, so the same conclusion is trivial.
	Hence, in all cases, $p:=\norm{\Pi_{C,\le a_{C}}\ket{\Omega}}^{2}\ge1/2$.
	Define the normalized low-deviation component $\ket{\varphi}:=\Pi_{C,\le a_{C}}\ket{\Omega}/\sqrt{p}$.
	
	Now set $K:=H-E_{0}\mathds{1}$.
	By the gap assumption, $K\ket{\Omega}=0$ and every excited eigenvalue of $K$ is at least $\Delta$.
	Moreover, Eq.~\eqref{eq:preMFRG-shifted-norm} gives $0\le K\le G\mathds{1}$ with $G=2\bar{g} n$.
	Hence the ground-state eigenvalue of $K$ is $0$, while all other eigenvalues lie in $[\Delta,G]$.
	
	First suppose $\Delta<G$.
	For the Chebyshev polynomial $T_{r}$ of the first kind, as commonly used in polynomial ground-space filtering~\cite{arad2013area,arad2017rigorous}, define
	\begin{align}
		p_{r}(x):= \frac{T_{r}((G+\Delta-2x)/(G-\Delta))}{T_{r}((G+\Delta)/(G-\Delta))}.
	\end{align}
	Then $p_{r}(0)=1$.
	On $[\Delta,G]$ the numerator has absolute value at most one, while $T_{r}(y)=\cosh(r\operatorname{arcosh}y)$ for $y>1$.
	Consequently,
	\begin{align}
		\sup_{x\in[\Delta,G]}|p_{r}(x)| &\le2\exp\left[-r\operatorname{arcosh} \left(\frac{G+\Delta}{G-\Delta}\right)\right] \le2e^{-2r\sqrt{\Delta/G}}.
	\end{align}
	Here we used $\cosh u\ge e^{u}/2$ and $\operatorname{arcosh}((1+t)/(1-t)) =2\operatorname{artanh}\sqrt{t}\ge2\sqrt{t}$ for $0<t<1$.
	
	Write $\ket{\varphi}=\sqrt{p}\ket{\Omega}+\ket{\eta}$, where $\bra{\Omega}\eta\rangle = 0$ and $\norm{\ket{\eta}} = \sqrt{1-p}$.
	The spectral bound therefore implies
	\begin{align}
		\norm{p_{r}(K)\ket{\varphi}-\sqrt{p}\ket{\Omega}} \le2\sqrt{1-p}\,e^{-2r\sqrt{\Delta/G}}.
	\end{align}
	By the same two-locality argument as in Eq.~\eqref{eq:two-local-bandedness}, each application of $K$ changes $M_{C}$ by at most two.
	Hence $p_{r}(K)\ket{\varphi}$, being obtained from a degree-$r$ polynomial in $K$, is supported on $M_{C}\le a_{C}+2r$, so $\Pi_{C,>a_{C}+2r}p_{r}(K)\ket{\varphi}=0$.
	Applying $\Pi_{C,>a_{C}+2r}$ to the preceding estimate therefore gives
	\begin{align}
		\sqrt{p}\,\norm{\Pi_{C,>a_{C}+2r}\ket{\Omega}} &= \norm{ \Pi_{C,>a_{C}+2r} \left( p_{r}(K)\ket{\varphi}-\sqrt{p}\,\ket{\Omega} \right) }\notag\\
		&\le \norm{ p_{r}(K)\ket{\varphi}-\sqrt{p}\,\ket{\Omega} }\notag\\
		&\le 2\sqrt{1-p}\,e^{-2r\sqrt{\Delta/G}}.
	\end{align}
	Since $p\ge1/2$, it follows that
	\begin{align}
		\norm{\Pi_{C,>a_{C}+2r}\ket{\Omega}} \le 2\sqrt{\frac{1-p}{p}}\,e^{-2r\sqrt{\Delta/G}} \le 2e^{-2r\sqrt{\Delta/G}},
	\end{align}
	which proves the stated estimate.
	
	It remains to consider the case $\Delta=G$, for which the preceding Chebyshev polynomial is not defined because $G-\Delta=0$.
	In this case, every excited eigenvalue of $K$ is exactly $G$.
	Hence the degree-one polynomial $q(x):=1-x/G$ satisfies $q(K)\ket{\varphi}=\sqrt{p}\,\ket{\Omega}$.
	Since one application of $K$ changes $M_{C}$ by at most two, $q(K)\ket{\varphi}$ is supported on $M_{C}\le a_{C}+2$.
	Hence, for every $r\ge1$,
	\begin{align}
		0 &= \Pi_{C,>a_{C}+2r}q(K)\ket{\varphi} = \sqrt{p}\,\Pi_{C,>a_{C}+2r}\ket{\Omega}.
	\end{align}
	Since $p\ge1/2$, it follows that $\Pi_{C,>a_{C}+2r}\ket{\Omega}=0$.
	For $r=0$, the claimed bound is trivial, since $\norm{\Pi_{C,>a_{C}}\ket{\Omega}}\le1$ while $2e^{-2r\sqrt{\Delta/G}}=2$.
\end{proof}

Given a norm-tail tolerance $0<\epsilon\le1/2$, apply the lemma with $C=L^{\mathrm{c}}$ and choose
\begin{align}
	m_{0}:=\min\left\{|L^{\mathrm{c}}|, a_{L^{\mathrm{c}}}+ 2\left\lceil\frac{1}{2}\sqrt{\frac{G}{\Delta}} \log\frac{2}{\epsilon}\right\rceil\right\}. \label{choice_of_m_0}
\end{align}
Then
\begin{align}
	\norm{\Pi_{L^{\mathrm{c}},>m_{0}}\ket{\Omega}}\le\epsilon.
\end{align}
Moreover, $\mu_{L^{\mathrm{c}}} =\sum_{i\in L^{\mathrm{c}}}q_{i} \le |L^{\mathrm{c}}|q_{*}\le nq_{*}$, and hence $a_{L^{\mathrm{c}}}\le2\mu_{L^{\mathrm{c}}}\le2nq_{*}$.
Using also $G=2\bar{g} n$ and $2\lceil x/2\rceil\le x+2$, we obtain
\begin{align}
	m_{0} \le 2nq_{*} +\sqrt{\frac{2\bar{g} n}{\Delta}} \log\frac{2}{\epsilon} +2 = \mathcal{O}(nq_{*}) +\mathcal{O}\left( \sqrt{\frac{\bar{g} n}{\Delta}} \log\frac{1}{\epsilon} \right), \label{eq:preMFRG-cutoff-scaling}
\end{align}
with the cutoff capped by $|L^{\mathrm{c}}|$.
The first term reflects the mean occupation of the complement, while the second is the additional cutoff required to suppress the tail to norm error $\epsilon$.

\subsubsection{Compression, accuracy, and preservation of the gap}

A small norm tail must next be converted into a spectral statement.
We work on the retained Hilbert space itself, rather than interpreting the discarded sector of a parent-space compression as physical zero energy states.

\begin{lemma}
	\label{lem:preMFRG-compression}
	Let $P$ be an orthogonal projection of rank at least two, and let $V$ be a coisometry onto an abstract retained space:
	\begin{align}
		V^{\dagger} V=P, \qquad VV^{\dagger}=\mathds{1}.
	\end{align}
	Define
	\begin{align}
		\varepsilon_{P}:=\norm{(\mathds{1}-P)\ket{\Omega}}^{2}<1, \qquad \varepsilon_{E}:=\bra{\Omega}(\mathds{1}-P)(H-E_{0}\mathds{1})(\mathds{1}-P)\ket{\Omega}, \qquad e_{P}:=\frac{\varepsilon_{E}}{1-\varepsilon_{P}}.
	\end{align}
	For $H_{P}:=VHV^{\dagger}$,
	\begin{align}
		E_{0}(H_{P}) \le E_{0}+e_{P}, \qquad E_{1}(H_{P}) \ge E_{0}+\Delta. \label{eq:preMFRG-compression-eigenvalues}
	\end{align}
	In particular, if $e_{P}<\Delta$, then $H_{P}$ has a unique ground state, its gap is at least $\Delta-e_{P}$, and its normalized ground state $\ket{\Omega_{P}}$ can be phased so that
	\begin{align}
		\norm{V^{\dagger}\ket{\Omega_{P}}-\ket{\Omega}} \le\sqrt{\frac{2e_{P}}{\Delta}}. \label{eq:preMFRG-compression-state-error}
	\end{align}
\end{lemma}
\begin{proof}
	Since $V^{\dagger} V=P$, we have $\norm{V\ket{\Omega}}^{2}=\bra{\Omega} P\ket{\Omega}=1-\varepsilon_{P}$.
	Thus $\ket{\phi_{P}}:=V\ket{\Omega}/\sqrt{1-\varepsilon_{P}}$ is a normalized trial state for $H_{P}$.
	Using $(H-E_{0}\mathds{1})\ket{\Omega}=0$, we have $(H-E_{0}\mathds{1})P\ket{\Omega} =-(H-E_{0}\mathds{1})(\mathds{1}-P)\ket{\Omega}$, and hence
	\begin{align}
		\bra{\phi_{P}}(H_{P}-E_{0}\mathds{1})\ket{\phi_{P}} = \frac{ \bra{\Omega} P(H-E_{0}\mathds{1})P\ket{\Omega} }{1-\varepsilon_{P}} = \frac{ \bra{\Omega}(\mathds{1}-P)(H-E_{0}\mathds{1}) (\mathds{1}-P)\ket{\Omega} }{1-\varepsilon_{P}} =e_{P}.
	\end{align}
	The variational principle therefore gives $E_{0}(H_{P})\le E_{0}+e_{P}$.
	
	Next, $H_{P}$ is unitarily equivalent to the restriction of $H$ to $\operatorname{Ran}P$.
	Let $\mathcal{S}\subseteq\operatorname{Ran}P$ be any two-dimensional subspace.
	There exists a normalized $\ket{\psi}\in\mathcal{S}$ with $\braket{\Omega|\psi}=0$, and the gap of $H$ implies $\bra{\psi} H\ket{\psi}\ge E_{0}+\Delta$.
	Hence
	\begin{align}
		E_{1}(H_{P}) &= \min_{\substack{\mathcal{S}\subseteq\operatorname{Ran}P\\ \dim\mathcal{S}=2}} \max_{\substack{\psi\in\mathcal{S}\\ \norm{\ket{\psi}}=1}} \bra{\psi} H\ket{\psi} \ge E_{0}+\Delta.
	\end{align}
	Combining the two eigenvalue bounds yields
	\begin{align}
		E_{1}(H_{P})-E_{0}(H_{P}) \ge\Delta-e_{P}.
	\end{align}
	In particular, if $e_{P}<\Delta$, the ground state of $H_{P}$ is unique.
	
	Let $\ket{\Omega_{P}}$ be this normalized ground state and set $\ket{\widehat{\Omega}_{P}}:=V^{\dagger}\ket{\Omega_{P}}$.
	Since $V^{\dagger}$ is an isometry, $\norm{\ket{\widehat{\Omega}_{P}}}=1$, and
	\begin{align}
		\bra{\widehat{\Omega}_{P}}(H-E_{0}\mathds{1}) \ket{\widehat{\Omega}_{P}} &= E_{0}(H_{P})-E_{0} \le e_{P}.
	\end{align}
	Applying the original gap inequality $H-E_{0}\mathds{1}\ge \Delta(\mathds{1}-\ket{\Omega}\bra{\Omega})$ to the normalized state $\ket{\widehat{\Omega}_{P}}$ gives
	\begin{align}
		\Delta\left( 1-|\braket{\Omega|\widehat{\Omega}_{P}}|^{2} \right) &\le \bra{\widehat{\Omega}_{P}}(H-E_{0}\mathds{1}) \ket{\widehat{\Omega}_{P}} \le e_{P}.
	\end{align}
	Choosing the phase of $\ket{\Omega_{P}}$ so that $\braket{\Omega|\widehat{\Omega}_{P}}\ge0$, we obtain
	\begin{align}
		\norm{\ket{\widehat{\Omega}_{P}}-\ket{\Omega}}^{2} = 2\left(1-\braket{\Omega|\widehat{\Omega}_{P}}\right) \le 2\left(1-|\braket{\Omega|\widehat{\Omega}_{P}}|^{2}\right) \le\frac{2e_{P}}{\Delta}.
	\end{align}
	This proves Eq.~\eqref{eq:preMFRG-compression-state-error}.
\end{proof}

We now apply Lemma~\ref{lem:preMFRG-compression} to the complement cutoff.
On the full Hilbert space, set $P:=\mathds{1}_{L}\otimes\Pi_{L^{\mathrm{c}},\le m_{0}}$.
Choose $v_{\mathrm{c}}$ so that its restriction to $\operatorname{Ran}\Pi_{L^{\mathrm{c}},\le m_{0}}$ is a unitary identification with a Hilbert space containing only the complement states retained by the cutoff, and extend it by zero on the orthogonal complement.
Thus $v_{\mathrm{c}}^{\dagger} v_{\mathrm{c}} =\Pi_{L^{\mathrm{c}},\le m_{0}}$ and $v_{\mathrm{c}}v_{\mathrm{c}}^{\dagger}=\mathds{1}$.
Then $V_{\mathrm{c}}:=\mathds{1}_{L}\otimes v_{\mathrm{c}}$ satisfies $V_{\mathrm{c}}^{\dagger} V_{\mathrm{c}}=P$ and $V_{\mathrm{c}}V_{\mathrm{c}}^{\dagger}=\mathds{1}$, as required in Lemma~\ref{lem:preMFRG-compression}.
Define the Hamiltonian after this complement cutoff by
\begin{align}
	\bar{H}_{\le m_{0}}:=V_{\mathrm{c}}HV_{\mathrm{c}}^{\dagger}. \label{eq:abstract-complement-cutoff-H}
\end{align}
Thus $\bar{H}_{\le m_{0}}$ acts only on the Hilbert space in which the sector $\operatorname{Ran}\Pi_{L^{\mathrm{c}},\le m_{0}}$ of the complement has been retained; states in the discarded sector $\operatorname{Ran}\Pi_{L^{\mathrm{c}},>m_{0}}$ are simply absent rather than appearing as artificial zero-energy states.

By Eq.~\eqref{eq:preMFRG-cutoff-scaling}, $\varepsilon_{P}=\norm{(\mathds{1}-P)\ket{\Omega}}^{2} =\norm{\Pi_{L^{\mathrm{c}},>m_{0}}\ket{\Omega}}^{2}\le\epsilon^{2}$.
Moreover, since $0\le H-E_{0}\mathds{1}\le G\mathds{1}$, the corresponding energy-weighted error satisfies $\varepsilon_{E}\le G\epsilon^{2}$.
Hence, denoting the quantity $e_{P}$ of Lemma~\ref{lem:preMFRG-compression} in the present application by $e_{\mathrm{c}}$,
\begin{align}
	e_{\mathrm{c}} \le\frac{G\epsilon^{2}}{1-\epsilon^{2}}.
\end{align}
Therefore, whenever $G\epsilon^{2}/(1-\epsilon^{2})\le\Delta/2$, Lemma~\ref{lem:preMFRG-compression} gives a compressed gap $\bar{\Delta}\ge\Delta/2$.

Let $\ket{\bar{\Omega}}$ be the normalized ground state of $\bar{H}_{\le m_{0}}$, and let $\ket{\widehat{\bar{\Omega}}} :=V_{\mathrm{c}}^{\dagger}\ket{\bar{\Omega}}$ be its embedding back into the original Hilbert space.
Then
\begin{align}
	\norm{\ket{\Omega}-\ket{\widehat{\bar{\Omega}}}} \le\eta_{\mathrm{c}}:= \sqrt{\frac{2G\epsilon^{2}}{\Delta(1-\epsilon^{2})}}. \label{eq:preMFRG-complement-error}
\end{align}
Thus the complement cutoff preserves both the spectral gap and the ground state, provided that its discarded energy-weighted error is sufficiently small.

\subsubsection{A direct finite-interval small-hopping estimate}

The next lemma gives the stronger local tail.
Its proof uses only the gap, bandwidth two, and hopping norms on the finite interval that must be crossed.
It is valid for shells of arbitrary dimension.

\begin{lemma}
	\label{lem:preMFRG-small-hopping}
	Let $K=K^{\dagger}$ be a finite-dimensional operator with a unique normalized ground state $\ket{\omega}$ and gap $\delta>0$.
	Let $\{R_{x}\}_{x=0}^{N}$ be orthogonal projections summing to the identity such that $R_{y}KR_{x}=0$ for $|y-x|>2$.
	Put $R_{\le x}:=\sum_{y\le x}R_{y}$ and $R_{>x}:=\sum_{y>x}R_{y}$, and let $m_{*}$ be the smallest integer with $\norm{R_{\le m_{*}}\ket{\omega}}^{2}\ge1/2$.
	For an integer $z>m_{*}$, define
	\begin{align}
		j(x)&:=\sum_{a=1}^{2}\norm{R_{x+a}KR_{x}}, \qquad J_{[m_{*},z]}:=\max_{m_{*}\le x\le z}j(x), \qquad \rho:=\frac{J_{[m_{*},z]}}{\delta},
	\end{align}
	where out-of-range shell projections are zero.
	If $0<\rho<1$, then
	\begin{align}
		\norm{R_{>z}\ket{\omega}} \le\left(\frac{8\rho^{2}}{1+8\rho^{2}}\right)^{\frac{1}{2}\lceil(z-m_{*})/2\rceil} \le\rho^{(z-m_{*})/8}. \label{eq:preMFRG-finite-interval-tail}
	\end{align}
	If $J_{[m_{*},z]}=0$, the left-hand side vanishes.
\end{lemma}

\begin{proof}
	Let $E_{0}(K)$ denote the ground energy of $K$, and set $\widetilde{K}:=K-E_{0}(K)\mathds{1}$.
	Then $\widetilde{K}\ket{\omega}=0$ and $\widetilde{K}\ge\delta(\mathds{1}-\ket{\omega}\bra{\omega})$.
	Moreover, $R_{y}\widetilde{K} R_{x}=R_{y}KR_{x}$ for $x\ne y$, so the bandwidth condition and the quantities $j(x)$, $J_{[m_{*},z]}$, and $\rho$ are unchanged by the scalar shift.
	
	Write $t_{m}:=\norm{R_{>m}\ket{\omega}}$ and set $t_{-1}:=1$.
	For $m\ge m_{*}$, the definition of $m_{*}$ gives $t_{m}^{2}\le1/2$.
	Restricting the gap inequality $\widetilde{K}\ge\delta(\mathds{1}-\ket{\omega}\bra{\omega})$ to $\operatorname{Ran}R_{>m}$ gives
	\begin{align}
		R_{>m}\widetilde{K} R_{>m} &\ge \delta\left( R_{>m} - R_{>m}\ket{\omega}\bra{\omega} R_{>m} \right) \ge \delta(1-t_{m}^{2})R_{>m} \ge \frac{\delta}{2}R_{>m},
	\end{align}
	where we used $R_{>m}\ket{\omega}\bra{\omega} R_{>m} \le\norm{R_{>m}\ket{\omega}}^{2}R_{>m} =t_{m}^{2}R_{>m}$.
	If $R_{>m}=0$, the tail already vanishes.
	Otherwise, $R_{>m}\widetilde{K} R_{>m}$ is invertible on $\operatorname{Ran}R_{>m}$ with $\norm{(R_{>m}\widetilde{K} R_{>m})^{-1}}\le2/\delta$.
	Applying $R_{>m}$ to $\widetilde{K}\ket{\omega}=0$ and using $\mathds{1}=R_{>m}+(\mathds{1}-R_{>m})$, we obtain
	\begin{align}
		0 &= R_{>m}\widetilde{K} R_{>m}\ket{\omega} + R_{>m}\widetilde{K}(\mathds{1}-R_{>m})\ket{\omega}.
	\end{align}
	Since $R_{>m}\widetilde{K} R_{>m}$ is invertible on $\operatorname{Ran}R_{>m}$ and $R_{>m}\ket{\omega}\in\operatorname{Ran}R_{>m}$, it follows that
	\begin{align}
		R_{>m}\ket{\omega} = -(R_{>m}\widetilde{K} R_{>m})^{-1} R_{>m}\widetilde{K}(\mathds{1}-R_{>m})\ket{\omega}. \label{eq:preMFRG-tail-resolvent}
	\end{align}
	
	By the bandwidth-two condition, $R_{>m}\widetilde{K}(\mathds{1}-R_{>m}) =R_{>m}\widetilde{K}(R_{m-1}+R_{m})$.
	For $\ket{u}\in\operatorname{Ran}R_{m-1}$ and $\ket{v}\in\operatorname{Ran}R_{m}$, the scalar shift does not affect the relevant off-diagonal shell blocks, and hence
	\begin{align}
		\norm{R_{>m}\widetilde{K}(\ket{u}+\ket{v})} \le j(m-1)\norm{\ket{u}} +j(m)\norm{\ket{v}} \le \sqrt{j(m-1)^{2}+j(m)^{2}} \sqrt{\norm{\ket{u}}^{2}+\norm{\ket{v}}^{2}}.
	\end{align}
	Since $\operatorname{Ran}R_{m-1}$ and $\operatorname{Ran}R_{m}$ are orthogonal, taking the supremum over $\norm{\ket{u}}^{2}+\norm{\ket{v}}^{2}=1$ gives
	\begin{align}
		\norm{R_{>m}\widetilde{K}(R_{m-1}+R_{m})} \le\sqrt{j(m-1)^{2}+j(m)^{2}}.
	\end{align}
	For $m_{*}+1\le m\le z$, both $j(m-1)$ and $j(m)$ are bounded by $J_{[m_{*},z]}$, and hence $\norm{R_{>m}\widetilde{K}(R_{m-1}+R_{m})} \le\sqrt{2}J_{[m_{*},z]}$.
	Using Eq.~\eqref{eq:preMFRG-tail-resolvent}, this bound, $\norm{(R_{>m}\widetilde{K} R_{>m})^{-1}}\le2/\delta$, and $\rho=J_{[m_{*},z]}/\delta$, we obtain
	\begin{align}
		t_{m} = \norm{R_{>m}\ket{\omega}} \le 2\sqrt{2}\rho\, \norm{(R_{m-1}+R_{m})\ket{\omega}} = 2\sqrt{2}\rho \sqrt{t_{m-2}^{2}-t_{m}^{2}}.
	\end{align}
	Indeed, $\norm{(R_{m-1}+R_{m})\ket{\omega}}^{2} =t_{m-2}^{2}-t_{m}^{2}$.
	Squaring the preceding inequality gives $(1+8\rho^{2})t_{m}^{2}\le8\rho^{2}t_{m-2}^{2}$, and hence
	\begin{align}
		t_{m} \le \vartheta(\rho)t_{m-2}, \qquad \vartheta(\rho):= \frac{2\sqrt{2}\rho}{\sqrt{1+8\rho^{2}}}. \label{eq:preMFRG-tail-recursion}
	\end{align}
	Iterating Eq.~\eqref{eq:preMFRG-tail-recursion} from $z$ downward in steps of two terminates at $m_{*}$ or $m_{*}-1$.
	Since $t_{m_{*}}\le1$ and $t_{m_{*}-1}\le1$, this yields $t_{z}\le\vartheta(\rho)^{\lceil(z-m_{*})/2\rceil}$.
	If $J_{[m_{*},z]}=0$, then $\rho=0$, and Eq.~\eqref{eq:preMFRG-tail-recursion} immediately gives $t_{z}=0$.
	
	Now suppose $0<\rho<1$.
	We claim that $\vartheta(\rho)\le\rho^{1/4}$.
	Indeed, after squaring and dividing by $\rho^{1/2}>0$, this inequality is equivalent to $8\rho^{3/2}(1-\sqrt{\rho})\le1$.
	Setting $u:=\sqrt{\rho}\in(0,1)$, the left-hand side becomes $8u^{3}(1-u)$, whose maximum on $[0,1]$ is $27/32<1$.
	Therefore $t_{z}\le\rho^{\lceil(z-m_{*})/2\rceil/4}$, and since $\lceil(z-m_{*})/2\rceil\ge(z-m_{*})/2$,
	\begin{align}
		t_{z} \le \rho^{(z-m_{*})/8}.
	\end{align}
	This proves Eq.~\eqref{eq:preMFRG-finite-interval-tail}.
\end{proof}

\subsubsection{Application to the block and the resulting criterion}

After retaining only the complement sector $\operatorname{Ran}\Pi_{L^{\mathrm{c}},\le m_{0}}$, we consider the corresponding $x$-flip sectors of the block $L$.
Their projections on the retained Hilbert space are
\begin{align}
	\bar{\Pi}_{L,x}:=V_{\mathrm{c}}\Pi_{L,x}V_{\mathrm{c}}^{\dagger}.
\end{align}
Since $\Pi_{L,x}$ acts on $L$ whereas $\Pi_{L^{\mathrm{c}},\le m_{0}}$ acts on $L^{\mathrm{c}}$, they commute.
Hence each $\Pi_{L,x}$ preserves the retained complement sector, and the operators $\bar{\Pi}_{L,x}$ form mutually orthogonal shell projections there.
Moreover, $V_{\mathrm{c}}$ is unitary between the retained subspace and its image, so the corresponding shell-hopping norms are unchanged:
\begin{align}
	\norm{\bar{\Pi}_{L,x+a}\bar{H}_{\le m_{0}}\bar{\Pi}_{L,x}} = \norm{\Pi_{L^{\mathrm{c}},\le m_{0}}\Pi_{L,x+a} H\Pi_{L,x}\Pi_{L^{\mathrm{c}},\le m_{0}}}.
\end{align}
Therefore Eq.~\eqref{eq:preMFRG-adjacent-shell-envelope} gives exactly the shell-hopping bounds required in Lemma~\ref{lem:preMFRG-small-hopping}.

Let $m_{*}$ be the median shell of $\ket{\bar{\Omega}}$ and define
\begin{align}
	\bar{J}_{m_{0},L}[m_{*},z] :=\max_{\substack{x\in\mathbb{Z}\\m_{*}\le x\le z}} \bar{J}_{m_{0},L}(x), \qquad \rho_{z}:=\frac{\bar{J}_{m_{0},L}[m_{*},z]}{\bar{\Delta}}. \label{eq:finite-interval-hopping-bound}
\end{align}
If $z>m_{*}$ and $0<\rho_{z}<1$, then
\begin{align}
	\norm{\Pi_{L,>z}\ket{\widehat{\bar{\Omega}}}} \le\rho_{z}^{(z-m_{*})/8}. \label{eq:scale-covariant-small-hopping-tail}
\end{align}
For $\rho_{z}=0$ the tail is zero.
Returning to the microscopic ground state and using the triangle inequality together with $\norm{\Pi_{L,>z}}\le1$, we obtain
\begin{align}
	\norm{\Pi_{L,>z}\ket{\Omega}} \le \norm{\Pi_{L,>z} (\ket{\Omega}-\ket{\widehat{\bar{\Omega}}})} +\norm{\Pi_{L,>z}\ket{\widehat{\bar{\Omega}}}} \le \norm{\ket{\Omega}-\ket{\widehat{\bar{\Omega}}}} +\norm{\Pi_{L,>z}\ket{\widehat{\bar{\Omega}}}} \le \eta_{\mathrm{c}}+\rho_{z}^{(z-m_{*})/8}. \label{eq:tail-after-global-truncation}
\end{align}
Here $\eta_{\mathrm{c}}$ is the complement-compression error defined in Eq.~\eqref{eq:preMFRG-complement-error}.

To make the preceding tail bound useful, we next control its starting shell $m_{*}$.
In particular, we give a simple condition ensuring $m_{*}=0$.
Using $\norm{\ket{\widehat{\bar{\Omega}}}-\ket{\Omega}}\le\eta_{\mathrm{c}}$ from Eq.~\eqref{eq:preMFRG-complement-error} and the contractivity of the projection $\Pi_{L,>0}$, we have
\begin{align}
	\norm{\Pi_{L,>0}\ket{\widehat{\bar{\Omega}}}} \le \norm{\Pi_{L,>0}\ket{\Omega}} +\norm{\Pi_{L,>0} (\ket{\widehat{\bar{\Omega}}}-\ket{\Omega})} \le \norm{\Pi_{L,>0}\ket{\Omega}}+\eta_{\mathrm{c}}.
\end{align}
Moreover, since $\Pi_{L,>0}\le M_{L}$ and $\bra{\Omega} M_{L}\ket{\Omega}=\sum_{i\in L}q_{i}\le|L|q_{*}$, we obtain $\norm{\Pi_{L,>0}\ket{\Omega}}\le\sqrt{|L|q_{*}}$.
Therefore
\begin{align}
	\norm{\Pi_{L,>0}\ket{\widehat{\bar{\Omega}}}} \le \sqrt{|L|q_{*}}+\eta_{\mathrm{c}}.
\end{align}
Consequently,
\begin{align}
	\sqrt{|L|q_{*}}+\eta_{\mathrm{c}}\le\frac{1}{\sqrt{2}} \quad \Rightarrow \quad m_{*}=0. \label{eq:preMFRG-zero-median-criterion}
\end{align}
Indeed, under this condition, $\norm{\Pi_{L,>0}\ket{\widehat{\bar{\Omega}}}}^{2}\le1/2$, and hence $\norm{\Pi_{L,0}\ket{\widehat{\bar{\Omega}}}}^{2}\ge1/2$.
By the definition of $m_{*}$ as the smallest shell index for which the cumulative weight is at least $1/2$, this implies $m_{*}=0$.
Thus a small mean deviation number controls the starting shell of the tail estimate, while the decay beyond that shell still requires the separate small-hopping condition $\rho_{z}<1$.

For a desired block norm error $\tau$, we want $\norm{\Pi_{L,>z}\ket{\Omega}}\le\tau$.
By Eq.~\eqref{eq:tail-after-global-truncation}, it is sufficient to bound the complement-compression error and the small-hopping tail separately by $\tau/2$.
We therefore choose the complement tolerance $\epsilon$ so that $\eta_{\mathrm{c}}\le\tau/2$ and, at the same time, $\bar{\Delta}\ge\Delta/2$.
Once $\epsilon$ has been fixed in this way, we choose $z>m_{*}$ so that the small-hopping estimate applies and contributes at most $\tau/2$, namely
\begin{align}
	\rho_{z}<1, \qquad \rho_{z}^{(z-m_{*})/8}\le\tau/2. \label{eq:preMFRG-block-truncation-criterion}
\end{align}
Under these conditions, Eq.~\eqref{eq:tail-after-global-truncation} gives $\norm{\Pi_{L,>z}\ket{\Omega}}\le\tau$.
The two requirements on the complement compression can be imposed simultaneously: the bound $\bar{\Delta}\ge\Delta/2$ follows if $G\epsilon^{2}/(1-\epsilon^{2})\le\Delta/2$, while Eq.~\eqref{eq:preMFRG-complement-error} shows that $\eta_{\mathrm{c}}\le\tau/2$ follows if $G\epsilon^{2}/(1-\epsilon^{2})\le\Delta\tau^{2}/8$.
Hence it is sufficient to require
\begin{align}
	\frac{G\epsilon^{2}}{1-\epsilon^{2}} \le\Delta\min\left\{\frac{1}{2},\frac{\tau^{2}}{8}\right\}.
\end{align}
The remaining condition concerns the shell hopping.
Since $\rho_{z}=\bar{J}_{m_{0},L}[m_{*},z]/\bar{\Delta}$ and $\bar{J}_{m_{0},L}[m_{*},z]$ is the maximum hopping strength over $m_{*}\le x\le z$, increasing $z$ may also increase $\rho_{z}$.
Therefore $\rho_{z}<1$ must be verified for the chosen cutoff $z$; it does not follow from the Kac normalization alone.

Finally, suppose that the system is partitioned into $B$ disjoint blocks, and let $P_{j}$ denote the low-flip projection retained on block $j$.
Since the blocks are disjoint, the projections $P_{j}$ commute, and
\begin{align}
	\mathds{1}-\prod_{j=1}^{B}P_{j} =\sum_{j=1}^{B}\left(\prod_{k<j}P_{k}\right)(\mathds{1}-P_{j}) \le\sum_{j=1}^{B}(\mathds{1}-P_{j}).
\end{align}
Therefore, if each block truncation satisfies $\norm{(\mathds{1}-P_{j})\ket{\Omega}}\le\tau_{j}$, then, since $\mathds{1}-\prod_{j=1}^{B}P_{j}$ is itself a projection,
\begin{align}
	\left\|\left(\mathds{1}-\prod_{j=1}^{B}P_{j}\right)\ket{\Omega}\right\|^{2} \le\sum_{j=1}^{B} \bra{\Omega}(\mathds{1}-P_{j})\ket{\Omega} \le\sum_{j=1}^{B}\tau_{j}^{2}. \label{eq:preMFRG-block-union}
\end{align}
Thus the total truncation error is controlled by the sum of the squared block errors, so the number of blocks must be included when choosing the local error tolerances.

At this point, we have established the ingredients needed for a controlled block truncation: a low-flip retained space on each block, a bound on the resulting ground-state error, preservation of a nonzero spectral gap, and bounds on the Hamiltonian terms after truncation.
In Section~\ref{sec:closed-induction}, we use these estimates to define and analyze the renormalization procedure across successive length scales, together with an explicit stopping criterion.

\section{Microscopic robustness and the arbitrary-cut theorem}
\label{sec:microscopic-robustness}

Before coarse-graining, we prove the small-occupation estimate that initializes the induction.
At this scale the local dimension is the fixed number $d$, so a microscopic operator expansion is harmless.
We will not repeat such an expansion on a renormalized site.

\subsection{A variance--gap estimate for a one-site sum}

We use a variance--gap relation for sums of local observables, closely related to the fluctuation bounds underlying local reversibility of gapped ground states~\cite{kuwahara2017local}.

\begin{lemma}
	\label{lem:basic-variance-gap}
	Let $K=K^{\dagger}=\sum_{i} k_{i}+\sum_{i<j}k_{ij}$, where $k_{i}$ is supported on $i$ and $k_{ij}$ on $i,j$, with the convention $k_{ji}:=k_{ij}$.
	Assume that $K$ has a unique ground state $\ket{\omega}$, gap $\delta$, and $\max_{i}(\norm{k_{i}}+\sum_{j\ne i}\norm{k_{ij}})\le g$.
	For a Hermitian operator $A$, write $\operatorname{Var}_{\omega}(A):= \bra{\omega} A^{2}\ket{\omega}- \bigl(\bra{\omega} A\ket{\omega}\bigr)^{2}$.
	Then, for Hermitian one-site operators $a_{i}$,
	\begin{align}
		\delta\,\operatorname{Var}_{\omega}\left(\sum_{i} a_{i}\right) \le4g\sum_{i}\norm{a_{i}}^{2}. \label{eq:basic-variance-gap}
	\end{align}
	There is no assumption on the individual site dimensions.
\end{lemma}

\begin{proof}
	Set $A:=\sum_{i} a_{i}$ and let $E_{0}(K)$ be the ground energy of $K$.
	Since
	\begin{align}
		K-E_{0}(K)\mathds{1} \ge\delta\bigl(\mathds{1}-\ket{\omega}\bra{\omega}\bigr),
	\end{align}
	applying this inequality to $A\ket{\omega}$ gives
	\begin{align}
		\delta\operatorname{Var}_{\omega}(A) &= \delta\bra{\omega} A\bigl(\mathds{1}-\ket{\omega}\bra{\omega}\bigr)A \ket{\omega} \le \bra{\omega} A(K-E_{0}(K)\mathds{1})A\ket{\omega}.
	\end{align}
	Using $(K-E_{0}(K)\mathds{1})\ket{\omega}=0$ and its adjoint, the right-hand side can be written as
	\begin{align}
		\bra{\omega} A(K-E_{0}(K)\mathds{1})A\ket{\omega} = \frac{1}{2}\bra{\omega}[A,[K,A]]\ket{\omega}. \label{eq:abstract-gap-double-commutator}
	\end{align}
	
	We now bound the double commutator using the locality of $K$.
	Since $k_{i}$ is supported only on site $i$, all $a_{j}$ with $j\ne i$ commute with $k_{i}$, and hence
	\begin{align}
		[A,[k_{i},A]] = [a_{i},[k_{i},a_{i}]].
	\end{align}
	Likewise, since $k_{ij}$ is supported only on $i,j$,
	\begin{align}
		[A,[k_{ij},A]] = [a_{i}+a_{j},[k_{ij},a_{i}+a_{j}]].
	\end{align}
	For arbitrary operators $X$ and $k$, $\norm{[X,[k,X]]}\le4\norm{X}^{2}\norm{k}$.
	Therefore
	\begin{align}
		\frac{1}{2}\norm{[A,[K,A]]} &\le 2\sum_{i}\norm{a_{i}}^{2}\norm{k_{i}} +2\sum_{i<j}\norm{a_{i}+a_{j}}^{2}\norm{k_{ij}}\notag\\
		&\le 2\sum_{i}\norm{a_{i}}^{2}\norm{k_{i}} +4\sum_{i<j} \bigl(\norm{a_{i}}^{2}+\norm{a_{j}}^{2}\bigr)\norm{k_{ij}},
	\end{align}
	where we used $\norm{a_{i}+a_{j}}^{2}\le 2(\norm{a_{i}}^{2}+\norm{a_{j}}^{2})$.
	
	Collecting all terms containing a fixed $\norm{a_{i}}^{2}$ gives
	\begin{align}
		\frac{1}{2}\norm{[A,[K,A]]} &\le \sum_{i}\norm{a_{i}}^{2} \left( 2\norm{k_{i}} +4\sum_{j\ne i}\norm{k_{ij}} \right) \le 4g\sum_{i}\norm{a_{i}}^{2}.
	\end{align}
	Combining this with Eq.~\eqref{eq:abstract-gap-double-commutator} proves Eq.~\eqref{eq:basic-variance-gap}.
\end{proof}

\subsection{An explicit microscopic operator basis}
\label{sec:fixed-microscopic-basis}

Fix an orthonormal basis $\{\ket{a}\}_{a=1}^{d}$ on a microscopic site, and write $E_{ab}:=\ket{a}\bra{b}$.
The $d^{2}$ Hermitian operators $E_{aa}$, $E_{ab}+E_{ba}$, and $i(E_{ab}-E_{ba})$ for $a<b$ form a basis of the one-site Hermitian operators, and each has operator norm one.

Let $X_{ij}=X_{ij}^{\dagger}$ be a two-site operator.
Expanding with respect to this basis on site $i$, define the operator on site $j$ by $X_{ab}^{(j)}:=(\bra{a}_{i}\otimes\mathds{1}_{j})X_{ij} (\ket{b}_{i}\otimes\mathds{1}_{j})$.
Then $\norm{X_{ab}^{(j)}}\le\norm{X_{ij}}$, and Hermiticity of $X_{ij}$ implies $X_{ba}^{(j)}=(X_{ab}^{(j)})^{\dagger}$.
Starting from $X_{ij}=\sum_{a,b}E_{ab}^{(i)}\otimes X_{ab}^{(j)}$ and grouping the terms with indices $(a,b)$ and $(b,a)$, we obtain
\begin{align}
	X_{ij} &= \sum_{a=1}^{d}E_{aa}^{(i)}\otimes X_{aa}^{(j)} +\sum_{a<b} (E_{ab}^{(i)}+E_{ba}^{(i)}) \otimes \frac{X_{ab}^{(j)}+X_{ba}^{(j)}}{2} + \sum_{a<b} i(E_{ab}^{(i)}-E_{ba}^{(i)}) \otimes \frac{X_{ab}^{(j)}-X_{ba}^{(j)}}{2i}.
\end{align}
The operators in the second tensor factors are Hermitian, and each has norm at most $\norm{X_{ij}}$.
Since the first tensor factors are exactly the $d^{2}$ one-site basis operators above and all have norm one, we may relabel the terms as
\begin{align}
	X_{ij}=\sum_{v=1}^{d^{2}}\mathcal{X}_{i,v}\otimes x_{j,v}, \qquad \norm{\mathcal{X}_{i,v}}\le1, \qquad \norm{x_{j,v}}\le \norm{X_{ij}}. \label{eq:basis-reconstruction-constant}
\end{align}

\subsection{Coarse and refined local-unitary estimates}

For $Y\subseteq i^{\mathrm{c}}$, expand
\begin{align}
	H_{i,Y}=\sum_{v=1}^{d^{2}}\mathcal{X}_{i,v}\otimes B_{i,v}, \qquad B_{i,v}=\sum_{j\in Y}b_{i,j,v}, \label{eq:microscopic-d2-expansion}
\end{align}
where each $b_{i,j,v}$ is Hermitian and supported on $j$.
Define
\begin{align}
	\Theta_{i,v}:=\inf_{c\in\mathbb{R}}\norm{(B_{i,v}-c)\ket{\Omega}} =\sqrt{\operatorname{Var}_{\Omega}(B_{i,v})}.
\end{align}
The equality follows by minimizing the quadratic function of $c$.
Applying Lemma~\ref{lem:basic-variance-gap} to the one-site sum $B_{i,v}=\sum_{j\in Y}b_{i,j,v}$, with $K=H$, $\delta=\Delta$, and $g=\bar{g}$, gives
\begin{align}
	\Theta_{i,v} \le 2\sqrt{\frac{\bar{g}}{\Delta}} \left(\sum_{j\in Y}\norm{b_{i,j,v}}^{2}\right)^{1/2} \le 2J\sqrt{\frac{\bar{g}}{\Delta}}\, \frac{1}{n^{\beta}} \left(\sum_{j\in Y}\mu_{i,j}^{2}r_{i,j}^{-2\alpha}\right)^{1/2}, \label{eq:s0-theta-variance}
\end{align}
where in the second inequality we used Eqs.~\eqref{eq:basis-reconstruction-constant} and \eqref{eq:long-range-interaction}, namely $\norm{b_{i,j,v}}\le \norm{h_{i,j}} \le J\mu_{i,j}r_{i,j}^{-\alpha}/n^{\beta}$.
The advantage of using the variance estimate is that it produces a square sum of the interaction strengths, rather than the $\ell_{1}$ sum that would follow from a direct bound on $\norm{B_{i,v}}$.
This improvement is essential for the microscopic scaling derived below.

Let $u_{i}:=P_{i}-Q_{i}=2P_{i}-\mathds{1}_{i}$.
Since $\ket{0_{i}}$ is an eigenvector of $\rho_{i}$, we have $[P_{i},\rho_{i}]=0$, and hence $[u_{i},\rho_{i}]=0$.
Moreover, $u_{i}=u_{i}^{\dagger}$ and $u_{i}^{2}=\mathds{1}_{i}$.
Define
\begin{align}
	\delta_{0,i}(Y):= \left| \bra{\Omega}u_{i}^{\dagger} H_{i,Y}u_{i}\ket{\Omega} -\bra{\Omega} H_{i,Y}\ket{\Omega} \right|.
\end{align}
Using $u_{i}=P_{i}-Q_{i}$ and $P_{i}+Q_{i}=\mathds{1}_{i}$, we have
\begin{align}
	u_{i}^{\dagger} H_{i,Y}u_{i}-H_{i,Y} &= -2\left(P_{i}H_{i,Y}Q_{i}+Q_{i}H_{i,Y}P_{i}\right).
\end{align}
Since $H_{i,Y}$ is Hermitian, it follows that
\begin{align}
	\delta_{0,i}(Y) &= 4\left| \operatorname{Re}\bra{\Omega} P_{i}H_{i,Y}Q_{i}\ket{\Omega} \right| \le 4\left|\bra{\Omega} P_{i}H_{i,Y}Q_{i}\ket{\Omega}\right|.
\end{align}

Using Eq.~\eqref{eq:microscopic-d2-expansion},
\begin{align}
	\left|\bra{\Omega} P_{i}H_{i,Y}Q_{i}\ket{\Omega}\right| \le \sum_{v=1}^{d^{2}} \left| \bra{\Omega} P_{i}\mathcal{X}_{i,v}Q_{i} B_{i,v} \ket{\Omega} \right|.
\end{align}
For each $v$, we may subtract an arbitrary scalar $c\in\mathbb{R}$ from $B_{i,v}$.
Indeed, $\bra{\Omega} P_{i}\mathcal{X}_{i,v}Q_{i}\ket{\Omega}=0$ because $[P_{i},\rho_{i}]=0$.
Since $B_{i,v}-c$ acts outside site $i$, it commutes with $P_{i}\mathcal{X}_{i,v}Q_{i}$, and Cauchy--Schwarz gives
\begin{align}
	\left| \bra{\Omega} P_{i}\mathcal{X}_{i,v}Q_{i}(B_{i,v}-c) \ket{\Omega} \right| &= \left| \bra{\Omega} (B_{i,v}-c)P_{i}\mathcal{X}_{i,v}Q_{i} \ket{\Omega} \right|\notag\\
	&\le \norm{(B_{i,v}-c)\ket{\Omega}}\, \norm{P_{i}\mathcal{X}_{i,v}Q_{i}\ket{\Omega}}\notag\\
	&\le \norm{(B_{i,v}-c)\ket{\Omega}}\sqrt{q_{i}},
\end{align}
where we used $\norm{P_{i}}=1$, $\norm{\mathcal{X}_{i,v}}\le1$, and $\norm{Q_{i}\ket{\Omega}}=\sqrt{q_{i}}$.
Minimizing over $c$ gives $\Theta_{i,v}$, and summing over $v$ therefore yields
\begin{align}
	\delta_{0,i}(Y) \le 4\left|\bra{\Omega} P_{i}H_{i,Y}Q_{i}\ket{\Omega}\right| \le 4\sqrt{q_{i}}\sum_{v=1}^{d^{2}}\Theta_{i,v}. \label{eq:s0-refined-direct}
\end{align}

We next derive a bound that applies to an arbitrary one-site unitary.
Take $Y=i^{\mathrm{c}}$, so that $H_{i,i^{\mathrm{c}}}$ contains every interaction term touching site $i$, and let $u$ be any one-site unitary satisfying $[u,\rho_{i}]=0$.
Terms supported entirely on $i^{\mathrm{c}}$ are unchanged under $u$, while $[u,\rho_{i}]=0$ implies $\bra{\Omega} u^{\dagger} h_{i} u\ket{\Omega}=\bra{\Omega} h_{i}\ket{\Omega}$.
Thus the total energy change comes entirely from $H_{i,i^{\mathrm{c}}}$.
For each term in Eq.~\eqref{eq:microscopic-d2-expansion}, the scalar part of $B_{i,v}$ again makes no contribution, since $[u,\rho_{i}]=0$ implies $\bra{\Omega}(u^{\dagger}\mathcal{X}_{i,v}u-\mathcal{X}_{i,v})\ket{\Omega}=0$.
Hence, for any $c\in\mathbb{R}$, Cauchy--Schwarz and $\norm{\mathcal{X}_{i,v}}\le1$ give
\begin{align}
	\left| \bra{\Omega} (u^{\dagger}\mathcal{X}_{i,v}u-\mathcal{X}_{i,v}) (B_{i,v}-c) \ket{\Omega} \right| &\le 2\norm{(B_{i,v}-c)\ket{\Omega}}.
\end{align}
Minimizing over $c$ and summing over $v$, and using the variational inequality $\bra{\Omega} u^{\dagger} H u\ket{\Omega}\ge E_{0}$, we obtain
\begin{align}
	0 \le \bra{\Omega} u^{\dagger} H u\ket{\Omega}-E_{0} \le 2\sum_{v=1}^{d^{2}}\Theta_{i,v}. \label{eq:s0-coarse-bound-proved}
\end{align}
In particular, for the reflection $u=u_{i}$, the quantity $\delta_{0,i}(i^{\mathrm{c}})$ is exactly this total energy excess, whereas for a proper subset $Y\subsetneq i^{\mathrm{c}}$ it measures only the contribution from the interactions between $i$ and $Y$.

\subsection{The bootstrap and its two ways of initialization}

\begin{prop}
	\label{prop:robustness-fixed-point}
	Let $\ket{\omega}$ be the unique ground state of a Hermitian Hamiltonian $K$, with ground energy $E_{0}(K)$ and gap $\delta>0$.
	At a fixed site, let $\rho$ be the reduced state, let $p_{0}=1-q$ be its largest eigenvalue, and choose a corresponding normalized eigenvector $\ket{0}$.
	Define $P:=\ket{0}\bra{0}$, $Q:=\mathds{1}-P$, and the reflection $u_{0}:=P-Q$.
	Its energy excess is defined as
	\begin{align}
		\delta_{\rm refl}:= \bra{\omega} u_{0}^{\dagger} K u_{0}\ket{\omega}-E_{0}(K).
	\end{align}
	Suppose that
	\begin{align}
		\delta_{\rm refl}\le\mathcal{A}\sqrt{q}.
	\end{align}
	If $p_{0}\ge1/2$, then
	\begin{align}
		q\le\frac{\mathcal{A}^{2}}{4\delta^{2}}, \qquad \delta_{\rm refl}\le\frac{\mathcal{A}^{2}}{2\delta}. \label{eq:robustness-fixed-point-result}
	\end{align}
	One sufficient condition for $p_{0}>1/2$ is that every one-site unitary $u$ satisfying $[u,\rho]=0$ obeys
	\begin{align}
		0\le \bra{\omega} u^{\dagger} K u\ket{\omega}-E_{0}(K) \le\mathcal{A}<\frac{\delta}{2}. \label{eq:coarse-local-unitary-bound}
	\end{align}
\end{prop}

\begin{proof}
	Since $u_{0}=P-Q$ and $\bra{\omega} P\ket{\omega}=p_{0}=1-q$, $\bra{\omega} Q\ket{\omega}=q$, we have
	\begin{align}
		\bra{\omega} u_{0}\ket{\omega} = \bra{\omega}(P-Q)\ket{\omega} = p_{0}-q = 1-2q.
	\end{align}
	Moreover, $u_{0}$ is unitary, so $u_{0}\ket{\omega}$ is normalized.
	Applying the gap inequality $K-E_{0}(K)\mathds{1} \ge\delta(\mathds{1}-\ket{\omega}\bra{\omega})$ to $u_{0}\ket{\omega}$ gives
	\begin{align}
		\delta_{\rm refl} = \bra{\omega} u_{0}^{\dagger}(K-E_{0}(K)\mathds{1})u_{0} \ket{\omega} \ge \delta\left( 1- |\bra{\omega} u_{0}\ket{\omega}|^{2} \right) = \delta\left[1-(1-2q)^{2}\right] = 4\delta q(1-q). \label{Q_i_proj_norm}
	\end{align}
	If $p_{0}=1-q\ge1/2$, then $q\le1/2$, and hence $4q(1-q)\ge2q$.
	Therefore Eq.~\eqref{Q_i_proj_norm} gives $\delta_{\rm refl}\ge2\delta q$.
	Combining this with the assumed upper bound $\delta_{\rm refl}\le\mathcal{A}\sqrt{q}$, we obtain, for $q>0$,
	\begin{align}
		2\delta q \le \mathcal{A}\sqrt{q}, \qquad\text{hence}\qquad \sqrt{q}\le\frac{\mathcal{A}}{2\delta}.
	\end{align}
	It follows that $q\le\mathcal{A}^{2}/(4\delta^{2})$, and substituting this bound into $\delta_{\rm refl}\le\mathcal{A}\sqrt{q}$ gives $\delta_{\rm refl}\le\mathcal{A}^{2}/(2\delta)$.
	The case $q=0$ is immediate.
	
	It remains to show that the coarse condition Eq.~\eqref{eq:coarse-local-unitary-bound} is sufficient to ensure $p_{0}>1/2$.
	Diagonalize the reduced state as $\rho=\sum_{a} p_{a}\ket{a}\bra{a}$, where $p_{0}=\max_{a} p_{a}$, and consider the phase unitaries $u_{\boldsymbol\theta}:=\sum_{a} e^{i\theta_{a}}\ket{a}\bra{a}$.
	By construction, $[u_{\boldsymbol\theta},\rho]=0$, so Eq.~\eqref{eq:coarse-local-unitary-bound} applies to every choice of the phases $\boldsymbol\theta$.
	Moreover,
	\begin{align}
		\bra{\omega} u_{\boldsymbol\theta}\ket{\omega} = \operatorname{Tr}(\rho u_{\boldsymbol\theta}) = \sum_{a} p_{a} e^{i\theta_{a}}.
	\end{align}
	Since $u_{\boldsymbol\theta}\ket{\omega}$ is normalized, the gap inequality $K-E_{0}(K)\mathds{1} \ge\delta(\mathds{1}-\ket{\omega}\bra{\omega})$ gives
	\begin{align}
		\delta\left( 1- \left|\sum_{a} p_{a} e^{i\theta_{a}}\right|^{2} \right) &\le \bra{\omega} u_{\boldsymbol\theta}^{\dagger} (K-E_{0}(K)\mathds{1}) u_{\boldsymbol\theta} \ket{\omega} \le \mathcal{A}.
	\end{align}
	Hence, for every choice of the phases, $|\sum_{a} p_{a} e^{i\theta_{a}}|^{2} \ge1-\mathcal{A}/\delta$.
	Now average this inequality over independent uniform phases $\theta_{a}\in[0,2\pi)$.
	Expanding the squared modulus gives
	\begin{align}
		\left|\sum_{a} p_{a} e^{i\theta_{a}}\right|^{2} = \sum_{a} p_{a}^{2} + \sum_{a\ne b}p_{a}p_{b} e^{i(\theta_{a}-\theta_{b})}.
	\end{align}
	The phase average of every term with $a\ne b$ vanishes, while the diagonal terms are unchanged.
	Therefore
	\begin{align}
		\sum_{a} p_{a}^{2} \ge 1-\frac{\mathcal{A}}{\delta} > \frac{1}{2},
	\end{align}
	where the last inequality uses $\mathcal{A}<\delta/2$.
	Finally, since $p_{a}\le p_{0}$ for every $a$ and $\sum_{a} p_{a}=1$,
	\begin{align}
		\sum_{a} p_{a}^{2} \le p_{0}\sum_{a} p_{a} = p_{0}.
	\end{align}
	Thus $p_{0}>1/2$.
	In particular, the largest eigenvalue is nondegenerate, since two eigenvalues equal to $p_{0}>1/2$ would have sum greater than one.
\end{proof}

There is a second way to verify the hypothesis $p_{0}>1/2$ at the microscopic scale.
Here we specialize $K=H$, $\ket{\omega}=\ket{\Omega}$, and $\delta=\Delta$ in Proposition~\ref{prop:robustness-fixed-point}.
Since the reduced state acts on a $d$-dimensional space, its largest eigenvalue satisfies $p_{0}\ge1/d$.
Hence Eq.~\eqref{Q_i_proj_norm} gives $\delta_{\rm refl}\ge4\Delta q p_{0}\ge4\Delta q/d$.
Combining this with the refined hypothesis $\delta_{\rm refl}\le\mathcal{A}\sqrt{q}$ gives, for $q>0$,
\begin{align}
	q\le\frac{d^{2}\mathcal{A}^{2}}{16\Delta^{2}}.
\end{align}
Thus, if $\mathcal{A}$ is sufficiently small, then $q<1/2$ and hence $p_{0}>1/2$.
At the microscopic scale, the bootstrap can therefore be initialized either by the coarse all-unitary condition Eq.~\eqref{eq:coarse-local-unitary-bound} or directly from the fixed local dimension together with the refined reflection bound.

We now verify the refined hypothesis of Proposition~\ref{prop:robustness-fixed-point} uniformly over all microscopic sites.
For a fixed site $i$, the remaining correspondence with the notation of that proposition is $\rho=\rho_{i}$, $p_{0}=\lambda_{0,i}^{2}=1-q_{i}$, $P=P_{i}$, $Q=Q_{i}$, and $u_{0}=u_{i}=P_{i}-Q_{i}$.
Taking $Y=i^{\mathrm{c}}$ in Eq.~\eqref{eq:s0-refined-direct}, we now identify its left-hand side with the reflection energy excess in Proposition~\ref{prop:robustness-fixed-point}.
Write $H=h_{i}+H_{i,i^{\mathrm{c}}}+H_{i^{\mathrm{c}}}$, where $H_{i^{\mathrm{c}}}$ contains all terms supported entirely on $i^{\mathrm{c}}$.
Since $u_{i}$ acts only on site $i$, it commutes with $H_{i^{\mathrm{c}}}$.
Moreover, $[u_{i},\rho_{i}]=0$ implies
\begin{align}
	\bra{\Omega} u_{i}^{\dagger} h_{i} u_{i}\ket{\Omega} = \operatorname{Tr}\left(\rho_{i} u_{i}^{\dagger} h_{i} u_{i}\right) = \operatorname{Tr}\left(u_{i}\rho_{i} u_{i}^{\dagger} h_{i}\right) = \bra{\Omega} h_{i}\ket{\Omega}.
\end{align}
Therefore, using $\bra{\Omega} H\ket{\Omega}=E_{0}$,
\begin{align}
	\delta_{\rm refl} = \bra{\Omega} u_{i}^{\dagger} H u_{i}\ket{\Omega}-E_{0} = \bra{\Omega} u_{i}^{\dagger} H_{i,i^{\mathrm{c}}}u_{i}\ket{\Omega} -\bra{\Omega} H_{i,i^{\mathrm{c}}}\ket{\Omega}.
\end{align}
The left-hand side is nonnegative by the variational principle, so taking the absolute value does not change it.
Hence $\delta_{\rm refl}=\delta_{0,i}(i^{\mathrm{c}})$.
Combining Eqs.~\eqref{eq:s0-refined-direct} and \eqref{eq:s0-theta-variance}, and using the $d^{2}$ terms in the microscopic operator expansion, gives
\begin{align}
	\delta_{0,i}(i^{\mathrm{c}}) \le 4\sqrt{q_{i}}\sum_{v=1}^{d^{2}}\Theta_{i,v} \le 8d^{2}J \sqrt{\frac{\bar{g}}{\Delta}}\, \frac{\gamma_{i,i^{\mathrm{c}}}^{(2\alpha)}}{n^{\beta}} \sqrt{q_{i}}.
\end{align}
Here $\gamma_{i,i^{\mathrm{c}}}^{(2\alpha)}$ is the coupling square-sum defined in Eq.~\eqref{Prop/main_ineq:renormalized_gamma_def}, namely $\gamma_{i,i^{\mathrm{c}}}^{(2\alpha)} =\bigl(\sum_{j\ne i}\mu_{i,j}^{2}r_{i,j}^{-2\alpha}\bigr)^{1/2}$.
To make this estimate uniform in $i$, define
\begin{align}
	\mathcal{L}_{\alpha}(n) := \max_{i} \frac{\gamma_{i,i^{\mathrm{c}}}^{(2\alpha)}}{n^{\beta}}, \qquad \mathcal{A}_{0} :=8d^{2}J \sqrt{\frac{\bar{g}}{\Delta}}\,\mathcal{L}_{\alpha}(n). \label{eq:A0-robustness-coefficient}
\end{align}
Then, for every site $i$, $\delta_{0,i}(i^{\mathrm{c}})\le\mathcal{A}_{0}\sqrt{q_{i}}$.
Thus $\mathcal{A}=\mathcal{A}_{0}$ in the refined hypothesis of Proposition~\ref{prop:robustness-fixed-point}.

It remains to estimate the size of $\mathcal{L}_{\alpha}(n)$, and hence of $\mathcal{A}_{0}$.
Using the expression for $\gamma_{i,i^{\mathrm{c}}}^{(2\alpha)}$ recalled above and $\mu_{i,j}^{2}\le 1+(\mu_{0}^{2}-1)\mathds{1}_{\{j\in\mathcal{S}_{i}\}}$, we have
\begin{align}
	\bigl(\gamma_{i,i^{\mathrm{c}}}^{(2\alpha)}\bigr)^{2} = \sum_{j\ne i}\mu_{i,j}^{2}r_{i,j}^{-2\alpha} \le \sum_{j\ne i}r_{i,j}^{-2\alpha} +(\mu_{0}^{2}-1) \sum_{j\in\mathcal{S}_{i}}r_{i,j}^{-2\alpha} \le a_{D}\sum_{r\le D\ell_{\Lambda}}r^{D-1-2\alpha} +(3^{D}-1)(\mu_{0}^{2}-1).
\end{align}
Here the last inequality follows from the lattice shell-counting bound in Eq.~\eqref{eq:preMFRG-lattice-counting}; the first term counts $\mathcal{O}(r^{D-1})$ sites at distance $r$, while the Moore neighborhood contains only $\mathcal{O}_{D}(1)$ sites.
Define the explicit geometric constant
\begin{align}
	\Xi_{D,\alpha,\mu_{0}} := \begin{cases} a_{D}D^{D-2\alpha}\left(1+\frac{1}{D-2\alpha}\right)+(3^{D}-1)(\mu_{0}^{2}-1),&\alpha<D/2,\\
		a_{D}(1+\log D)+(3^{D}-1)(\mu_{0}^{2}-1),&\alpha=D/2,\\
		a_{D}\left(1+\frac{1}{2\alpha-D}\right)+(3^{D}-1)(\mu_{0}^{2}-1),&D/2<\alpha<D. \end{cases} \label{eq:microscopic-square-sum-constant}
\end{align}
Here we use $\sum_{r=1}^{R}r^{u-1}\le(1+u^{-1})R^{u}$ for $u>0$, $\sum_{r=1}^{R}r^{-1}\le1+\log R$, and $\sum_{r=1}^{\infty}r^{-1-u}\le1+u^{-1}$, with $R=D\ell_{\Lambda}$, together with $\log(en)\ge1$.
It follows that
\begin{align}
	\mathcal{L}_{\alpha}(n)\le \sqrt{\Xi_{D,\alpha,\mu_{0}}} \begin{cases} n^{-1/2},&\alpha<D/2,\\
		\sqrt{\log(en)/n},&\alpha=D/2,\\
		n^{-\beta},&D/2<\alpha<D. \end{cases} \label{eq:microscopic-l2-scaling}
\end{align}
Indeed, for $\alpha<D/2$ the shell sum is $\mathcal{O}(\ell_{\Lambda}^{D-2\alpha})$, so its square root is $\mathcal{O}(n^{1/2-\alpha/D})$, and division by $n^{\beta}=n^{1-\alpha/D}$ gives $n^{-1/2}$.
At $\alpha=D/2$ the shell sum is $\mathcal{O}(\log(en))$, so division of its square root by $n^{1/2}$ gives $\sqrt{\log(en)/n}$.
For $\alpha>D/2$ the shell sum remains uniformly bounded, and hence division by $n^{\beta}$ gives $n^{-\beta}$.
Therefore $\mathcal{L}_{\alpha}(n)\to0$, and hence $\mathcal{A}_{0}\to0$, in all three regimes.
The fixed-dimensional initialization above then gives $p_{0}>1/2$ at every microscopic site for all sufficiently large $n$.
Together with $\delta_{0,i}(i^{\mathrm{c}})\le\mathcal{A}_{0}\sqrt{q_{i}}$, this verifies the two hypotheses of Proposition~\ref{prop:robustness-fixed-point} with $\mathcal{A}=\mathcal{A}_{0}$ and $\delta=\Delta$.

\begin{theorem}
	\label{thm:single-scale-arbitrary-cut}
	Under the microscopic assumptions, for every bipartition and all sufficiently large $n$,
	\begin{align}
		q_{*}&\le C_{\rm mic} \begin{cases} n^{-1},&\alpha<D/2,\\
			(\log(en))/n,&\alpha=D/2,\\
			n^{-2\beta},&D/2<\alpha<D, \end{cases} \label{eq:single-site-q-scaling}\\
		S_{A}(\ket{\Omega})&\le C_{\mathrm{arb}} \begin{cases} \log(en),&\alpha<D/2,\\
			(\log(en))^{2},&\alpha=D/2,\\
			n^{2\alpha/D-1}\log(en),&D/2<\alpha<D. \end{cases} \label{eq:single-scale-entropy-all-regimes}
	\end{align}
	The constants may be chosen explicitly as
	\begin{align}
		C_{\rm mic} &:=16d^{4}\left(\frac{J}{\Delta}\right)^{2}\frac{\bar{g}}{\Delta}\,\Xi_{D,\alpha,\mu_{0}}, \label{eq:single-scale-deviation-constant}\\
		C_{\mathrm{arb}} &:=[1+\log(d-1)]\max\{1,C_{\rm mic}\}. \label{eq:single-scale-entropy-constant}
	\end{align}
	Here $\Xi_{D,\alpha,\mu_{0}}$ is defined in Eq.~\eqref{eq:microscopic-square-sum-constant}, and the constants $C_{\rm mic}$ and $C_{\mathrm{arb}}$ depend only on the fixed microscopic parameters and are independent of the bipartition.
	In particular, $C_{\mathrm{arb}}=\mathcal{O}_{D,\alpha,d,\mu_{0}}\left(1+\left(\frac{J}{\Delta}\right)^{2}\frac{\bar{g}}{\Delta}\right)$, with an implicit factor independent of $J$, $\bar{g}$, $\Delta$, $n$, and the bipartition.
	In the regime $D/2<\alpha<D$, the power and logarithmic scaling of the worst-case arbitrary-cut bound are attained by the dimer example in Proposition~\ref{prop:dimer-obstruction}.
\end{theorem}

\begin{proof}
	For every site $i$, the preceding estimates verify the hypotheses of Proposition~\ref{prop:robustness-fixed-point} with $\mathcal{A}=\mathcal{A}_{0}$, where $\mathcal{A}_{0}$ is defined in Eq.~\eqref{eq:A0-robustness-coefficient}, and with $\delta=\Delta$, the spectral gap of the microscopic Hamiltonian $H$.
	Therefore, Proposition~\ref{prop:robustness-fixed-point} gives
	\begin{align}
		q_{i} \le \frac{\mathcal{A}_{0}^{2}}{4\Delta^{2}} =16d^{4}\left(\frac{J}{\Delta}\right)^{2}\frac{\bar{g}}{\Delta}\,\mathcal{L}_{\alpha}(n)^{2}.
	\end{align}
	Taking the maximum over $i$ and using Eqs.~\eqref{eq:A0-robustness-coefficient} and \eqref{eq:microscopic-l2-scaling} then gives Eq.~\eqref{eq:single-site-q-scaling} with the explicit constant in Eq.~\eqref{eq:single-scale-deviation-constant}.
	In particular, for $D/2<\alpha<D$, the shell bound $\sum_{r\ge1}r^{D-1-2\alpha}\le1+(2\alpha-D)^{-1}$ gives the explicit specialization
	\begin{align}
		q_{*}\le C_{\rm mic}n^{-2\beta}, \qquad C_{\rm mic}=\frac{16d^{4}J^{2}\bar{g}}{\Delta^{3}}\left[a_{D}\left(1+\frac{1}{2\alpha-D}\right)+(3^{D}-1)(\mu_{0}^{2}-1)\right]. \label{eq:microscopic-deviation-constant}
	\end{align}
	
	We next convert the single-site occupation bound into an entropy bound.
	If $q_{i}=0$, then $\rho_{i}$ is pure and $S(\rho_{i})=0$.
	Suppose therefore that $q_{i}>0$.
	Let the eigenvalues of $\rho_{i}$ be $1-q_{i},p_{1,i},\ldots,p_{d-1,i}$, so that $\sum_{a=1}^{d-1}p_{a,i}=q_{i}$.
	Writing $p_{a,i}=q_{i}\widetilde{p}_{a,i}$ with $\sum_{a=1}^{d-1}\widetilde{p}_{a,i}=1$, and denoting the binary entropy by $\hbin(q):=-q\log q-(1-q)\log(1-q)$, we have
	\begin{align}
		S(\rho_{i}) &= -(1-q_{i})\log(1-q_{i}) -\sum_{a=1}^{d-1}p_{a,i}\log p_{a,i}\notag\\
		&= \hbin(q_{i}) + q_{i}\left( -\sum_{a=1}^{d-1}\widetilde{p}_{a,i} \log\widetilde{p}_{a,i} \right)\notag\\
		&\le \hbin(q_{i})+q_{i}\log(d-1),
	\end{align}
	since the entropy of a probability distribution on $d-1$ outcomes is at most $\log(d-1)$.
	Thus the same bound holds also when $q_{i}=0$.
	
	Now let $X$ denote the smaller side of the bipartition, so that $|X|=\min\{|A|,|A^{\mathrm{c}}|\}$.
	Since $\ket{\Omega}$ is pure, $S_{A}(\ket{\Omega})=S_{A^{\mathrm{c}}}(\ket{\Omega})=S_{X}(\ket{\Omega})$.
	By repeated subadditivity~\cite{araki1970entropy},
	\begin{align}
		S_{A}(\ket{\Omega}) = S_{X}(\ket{\Omega}) \le \sum_{i\in X}S(\rho_{i}).
	\end{align}
	For sufficiently large $n$, Eq.~\eqref{eq:single-site-q-scaling} gives $q_{i}\le q_{*}\le1/2$ for every site $i$.
	Since $q\mapsto\hbin(q)+q\log(d-1)$ is increasing on $[0,1/2]$, each single-site entropy is bounded by $\hbin(q_{*})+q_{*}\log(d-1)$.
	Hence
	\begin{align}
		S_{A}(\ket{\Omega}) \le |X| \bigl[\hbin(q_{*})+q_{*}\log(d-1)\bigr].
	\end{align}
	
	Finally, using $|X|\le n$ and $\hbin(q)\le q\log(e/q)$, we obtain
	\begin{align}
		S_{A}(\ket{\Omega}) \le nq_{*} \left[ \log\frac{e}{q_{*}}+\log(d-1) \right].
	\end{align}
	To keep the prefactor explicit, let $f_{\alpha}(n)$ denote the corresponding piecewise factor in Eq.~\eqref{eq:single-site-q-scaling}, and set $\widehat{C}_{\rm mic}:=\max\{1,C_{\rm mic}\}$.
	For all sufficiently large $n$, $q_{*}\le\widehat{C}_{\rm mic}f_{\alpha}(n)\le1/2$.
	The case $q_{*}=0$ is immediate; otherwise, monotonicity of $q\mapsto q[\log(e/q)+\log(d-1)]$ on $[0,1/2]$ gives
	\begin{align}
		S_{A}(\ket{\Omega}) &\le \widehat{C}_{\rm mic}nf_{\alpha}(n)\left[\log\frac{e}{\widehat{C}_{\rm mic}f_{\alpha}(n)}+\log(d-1)\right] \notag\\
		&\le [1+\log(d-1)]\widehat{C}_{\rm mic}nf_{\alpha}(n)\log(en) \notag\\
		&=C_{\mathrm{arb}}nf_{\alpha}(n)\log(en). \label{eq:single-scale-explicit-entropy-prefactor}
	\end{align}
	Here $f_{\alpha}(n)\ge n^{-1}$ in all three regimes, so $\log(e/(\widehat{C}_{\rm mic}f_{\alpha}(n)))\le\log(en)$, and we also used $\log(en)\ge1$ and Eq.~\eqref{eq:single-scale-entropy-constant}.
	Since $d$ is fixed, Eq.~\eqref{eq:single-site-q-scaling} gives the three regimes directly.
	If $\alpha<D/2$, then $q_{*}=\mathcal{O}(n^{-1})$, and hence $S_{A}(\ket{\Omega})=\mathcal{O}(\log(en))$.
	If $\alpha=D/2$, then $q_{*}=\mathcal{O}(\log(en)/n)$, which gives $S_{A}(\ket{\Omega})=\mathcal{O}((\log(en))^{2})$.
	Finally, if $D/2<\alpha<D$, then $q_{*}=\mathcal{O}(n^{-2\beta})$, so $nq_{*}=\mathcal{O}(n^{1-2\beta}) =\mathcal{O}(n^{2\alpha/D-1})$, and therefore $S_{A}(\ket{\Omega})=\mathcal{O}(n^{2\alpha/D-1}\log(en))$.
	This proves Eq.~\eqref{eq:single-scale-entropy-all-regimes}.
	The matching lower bound in the regime $D/2<\alpha<D$ is given by Proposition~\ref{prop:dimer-obstruction}.
\end{proof}

\section{Mean-field renormalization group for \texorpdfstring{$\alpha>D/2$}{alpha > D/2}}
\label{sec:MFRG}

The microscopic theorem shows that the ground state has a small local deviation probability, but for $\alpha>D/2$ this information alone does not give a polylogarithmic entropy bound at larger length scales.
To exploit this local smallness recursively, we now introduce a mean-field renormalization-group (MFRG) construction~\cite{kim2024quantum}, conceptually related to block-state truncation and real-space renormalization methods~\cite{kadanoff1966scaling,wilson1974renormalization,wilson1975renormalization,arad2017rigorous}.

The basic step is to partition the current lattice into blocks and, in each block, retain only the subspace containing a controlled number of deviations from the local reference vectors.
This retained subspace is then regarded as the Hilbert space of a single effective site at the next scale.
Repeating the construction produces a sequence of coarse-grained lattices and effective Hamiltonians.
The estimates developed above are designed to control the truncation error, the spectral gap, and the effective interactions as this procedure is iterated.

In this section we define the Hilbert spaces, truncation projections, and coarse-graining maps entering one MFRG step, and then specify how they are composed across successive scales.
The accuracy of the truncations and the conditions under which the required estimates remain stable are established later.

\subsection{Parent, retained, and effective spaces}

Starting from the microscopic lattice, we repeatedly group neighboring sites into blocks and replace each retained block subspace by a single effective site.
We label the successive stages of this construction by $s=0,1,2,\ldots$, with $s=0$ denoting the original microscopic system.

At stage $s$, let $\Lambda^{(s)}$ denote the effective lattice and let $\mathcal{K}_{s,i}$ be the local Hilbert space associated with the effective site $i\in\Lambda^{(s)}$.
At the microscopic stage, $\Lambda^{(0)}=\Lambda$, $\mathcal{K}_{0,i}=\mathbb{C}^{d}$, and $H^{(0)}=H$.
We write
\begin{align}
	\mathcal{H}^{(s)}:=\bigotimes_{i\in\Lambda^{(s)}}\mathcal{K}_{s,i}, \qquad n^{(s)}:=|\Lambda^{(s)}|, \qquad d^{(s)}:=\max_{i}\dim\mathcal{K}_{s,i}. \label{eq:scale-spaces}
\end{align}
Thus $\mathcal{H}^{(s)}$ is the full Hilbert space at stage $s$, $n^{(s)}$ is the number of effective sites, and $d^{(s)}$ is the largest local Hilbert-space dimension.
For $s\ge1$, the spaces $\mathcal{K}_{s,i}$ will be defined below from retained low-deviation subspaces of blocks at the preceding stage.

Let $H^{(s)}$ denote the Hamiltonian acting on $\mathcal{H}^{(s)}$.
Whenever it has a unique ground state $\ket{\Omega^{(s)}}$, let $\rho_{i}^{(s)}$ be the reduced ground state on the effective site $i$, and choose a normalized eigenvector $\ket{0_{i}^{(s)}}$ corresponding to its largest eigenvalue.
As at the microscopic scale, this vector is used as the local reference state.
Define
\begin{align}
	P_{i}^{(s)} &:=\ket{0_{i}^{(s)}}\bra{0_{i}^{(s)}}, \qquad Q_{i}^{(s)}:=\mathds{1}_{i}-P_{i}^{(s)}. \label{eq:scale-s-reference-projectors}
\end{align}
The corresponding deviation probability and its uniform bound are
\begin{align}
	q_{s,i} &:=\bra{\Omega^{(s)}}Q_{i}^{(s)}\ket{\Omega^{(s)}}, \qquad q_{s,*}:=\max_{i}q_{s,i}, \qquad \lambda_{*,s}:=\sqrt{1-q_{s,*}}. \label{eq:scale-s-reference-state}
\end{align}
If the largest eigenspace of $\rho_{i}^{(s)}$ is degenerate, any normalized vector in that eigenspace may be chosen at this stage.
The induction developed below will ensure that the relevant largest eigenvalues are subsequently nondegenerate.

To pass from stage $s$ to stage $s+1$, partition the effective lattice $\Lambda^{(s)}$ into disjoint hypercubic blocks $L_{j}^{(s)}$ of side length $\ell_{s+1}$ and volume $b_{s+1}:=\ell_{s+1}^{D}$.
For the exact block partitions used in the construction, the side length of $\Lambda^{(s)}$ is assumed to be divisible by $\ell_{s+1}$.
The required exact hierarchy is specified in Section~\ref{sec:regular-blocking-hierarchy}; the divisibility restriction for ordinary regular-cut conclusions is removed later in Section~\ref{sec:padding} by a bounded-ratio embedding.

For a block $L\subseteq\Lambda^{(s)}$, define its deviation-number operator and the corresponding low-deviation projection by
\begin{align}
	M_{L}^{(s)} &:=\sum_{i\in L}Q_{i}^{(s)}, \qquad \Pi_{L,\le z}^{(s)} :=\mathds{1}_{[0,z]}(M_{L}^{(s)}). \label{eq:scale-deviation-number}
\end{align}
Thus $\Pi_{L,\le z}^{(s)}$ retains precisely those block states in which at most $z$ effective sites deviate from their local reference vectors.

Choose a cutoff $z_{s+1}\in\mathbb{Z}_{\ge0}$.
For each block $L_{j}^{(s)}$, the retained subspace is $\operatorname{Ran}\Pi_{L_{j}^{(s)},\le z_{s+1}}^{(s)}$.
We identify this retained subspace with the local Hilbert space of a single effective site at the next stage.
Accordingly, choose $\mathcal{K}_{s+1,j}$ with $\dim\mathcal{K}_{s+1,j} =\operatorname{rank}\Pi_{L_{j}^{(s)},\le z_{s+1}}^{(s)}$ and a coisometry
\begin{align}
	V_{s+1,j}: \bigotimes_{i\in L_{j}^{(s)}}\mathcal{K}_{s,i} \longrightarrow \mathcal{K}_{s+1,j} \label{eq:local-RG-coisometry}
\end{align}
satisfying
\begin{align}
	V_{s+1,j}^{\dagger} V_{s+1,j} &= \Pi_{L_{j}^{(s)},\le z_{s+1}}^{(s)}, \qquad V_{s+1,j}V_{s+1,j}^{\dagger} = \mathds{1}_{\mathcal{K}_{s+1,j}}. \label{eq:local-RG-partial-isometry-identities}
\end{align}
Hence $V_{s+1,j}$ is unitary from the retained block subspace onto $\mathcal{K}_{s+1,j}$ and vanishes on the discarded subspace.
Its adjoint $V_{s+1,j}^{\dagger}$ embeds an effective-site state back into the retained subspace of the parent block.
Thus the effective Hilbert space contains only retained states; the discarded sector does not reappear as an artificial zero-energy sector.

Applying this construction independently to every block defines the global coarse-graining map
\begin{align}
	V_{s+1} &:=\bigotimes_{j}V_{s+1,j}, \qquad \mathsf{P}_{s\rightarrow s+1} :=V_{s+1}^{\dagger} V_{s+1} = \prod_{j}\Pi_{L_{j}^{(s)},\le z_{s+1}}^{(s)}. \label{eq:Pi-s-plus-one}
\end{align}
Thus $\mathsf{P}_{s\rightarrow s+1}$ is the projection onto the subspace in which every parent block satisfies the prescribed low-deviation cutoff.

The number of effective sites decreases by the block volume, $n^{(s+1)}=n^{(s)}/b_{s+1}$.
To bound the new local dimension, note that a retained block state with exactly $r$ deviations is obtained by choosing the $r$ deviating sites and, on each such site, choosing one of at most $d^{(s)}-1$ directions orthogonal to its reference vector.
Therefore, for $z_{s+1}\ge1$,
\begin{align}
	n^{(s+1)} &= \frac{n^{(s)}}{b_{s+1}}, \qquad d^{(s+1)} \le \sum_{r=0}^{\min\{z_{s+1},b_{s+1}\}} \binom{b_{s+1}}{r}(d^{(s)}-1)^{r} \le \bigl[1+b_{s+1}(d^{(s)}-1)\bigr]^{z_{s+1}} \le (d^{(s)}b_{s+1})^{z_{s+1}}. \label{eq:effective-dimension-recursion}
\end{align}
Here the last two inequalities use $\binom{b_{s+1}}{r}\le b_{s+1}^{r}$, $\sum_{r=0}^{z_{s+1}}x^{r}\le(1+x)^{z_{s+1}}$ for $x\ge0$, and $1+b_{s+1}(d^{(s)}-1)\le d^{(s)}b_{s+1}$.
This estimate controls the dimension of the retained effective-site space.
By itself, however, it is not an entanglement bound for a microscopic bipartition that passes through the interior of a parent block, because the decoding map $V_{s+1,j}^{\dagger}$ may itself contain nontrivial entanglement across such a cut.

\subsection{The centered effective Hamiltonian}
\label{sec:renormalized-Hamiltonian-RE}

Suppose
\begin{align}
	H^{(s)} = \sum_{i} h_{i}^{(s)} + \sum_{i<j}h_{ij}^{(s)}. \label{eq:scale-s-Hamiltonian-decomposition}
\end{align}
Using the reference projector $P_{i}^{(s)}$ at stage $s$, define the corresponding centering maps by
\begin{align}
	\mathcal{E}_{i}^{(s)}(X) &:= \bra{0_{i}^{(s)}}X\ket{0_{i}^{(s)}}\otimes\mathds{1}_{i}, \qquad \mathcal{R}_{i}^{(s)} := \mathrm{Id}-\mathcal{E}_{i}^{(s)}. \label{eq:scale-s-centering-maps}
\end{align}
We then define the centered two-site and one-site terms by
\begin{align}
	\acute{h}_{ij}^{(s)} &:= \mathcal{R}_{i}^{(s)} \mathcal{R}_{j}^{(s)} \bigl(h_{ij}^{(s)}\bigr), \qquad \acute{h}_{i}^{(s)} := \mathcal{R}_{i}^{(s)} \bigl(h_{i}^{(s)}\bigr) + \sum_{j\ne i} \mathcal{R}_{i}^{(s)} \mathcal{E}_{j}^{(s)} \bigl(h_{ij}^{(s)}\bigr). \label{eq:scale-s-centered-onsite}
\end{align}
Define the centered Hamiltonian by
\begin{align}
	\acute{H}^{(s)} := \sum_{i}\acute{h}_{i}^{(s)} + \sum_{i<j}\acute{h}_{ij}^{(s)}.
\end{align}
For $L\subseteq\Lambda^{(s)}$ and disjoint $L,L'\subseteq\Lambda^{(s)}$, define
\begin{align}
	\acute{H}_{L}^{(s)} &:= \sum_{i\in L}\acute{h}_{i}^{(s)} + \sum_{\substack{i<j\\i,j\in L}}\acute{h}_{ij}^{(s)}, \qquad \acute{H}_{L,L'}^{(s)} := \sum_{\substack{i\in L\\j\in L'}} \acute{h}_{ij}^{(s)}. \label{eq:scale-s-centered-subsystems}
\end{align}
Let $\ket{\mathbf{0}^{(s)}}:= \bigotimes_{i\in\Lambda^{(s)}}\ket{0_{i}^{(s)}}$ be the product reference state at stage $s$.
By construction, $\acute{H}^{(s)}=H^{(s)}-c_{s}\mathds{1}$, where $c_{s}=\bra{\mathbf{0}^{(s)}}H^{(s)}\ket{\mathbf{0}^{(s)}}$.

Define the effective Hamiltonian at the next stage by
\begin{align}
	H^{(s+1)} := V_{s+1}\acute{H}^{(s)}V_{s+1}^{\dagger} \quad\text{on }\mathcal{H}^{(s+1)}. \label{MFRG_transform_RE}
\end{align}
Because $V_{s+1}$ is a tensor product of the block maps $V_{s+1,j}$, the effective Hamiltonian remains two-local, with
\begin{align}
	h_{j}^{(s+1)} &:= V_{s+1,j} \acute{H}_{L_{j}^{(s)}}^{(s)} V_{s+1,j}^{\dagger}, \qquad h_{jj'}^{(s+1)} := (V_{s+1,j}\otimes V_{s+1,j'}) \acute{H}_{L_{j}^{(s)},L_{j'}^{(s)}}^{(s)} (V_{s+1,j}\otimes V_{s+1,j'})^{\dagger}. \label{Ineq:Prop:renormalized_interaction}
\end{align}
We also use $H_{i,Y}^{(s)}:=\sum_{j\in Y}h_{ij}^{(s)}$ and $\acute{H}_{i,Y}^{(s)}:=\sum_{j\in Y}\acute{h}_{ij}^{(s)}$.

For later comparison with the parent Hamiltonian, it is convenient to represent the effective Hamiltonian back on the parent Hilbert space $\mathcal{H}^{(s)}$.
Define
\begin{align}
	\widehat{H}^{(s+1)} := V_{s+1}^{\dagger} H^{(s+1)}V_{s+1} = \mathsf{P}_{s\rightarrow s+1}\acute{H}^{(s)}\mathsf{P}_{s\rightarrow s+1}, \label{eq:H-s-construction}
\end{align}
where we used $V_{s+1}^{\dagger} V_{s+1}=\mathsf{P}_{s\rightarrow s+1}$.
Thus $\widehat{H}^{(s+1)}$ acts on the parent Hilbert space $\mathcal{H}^{(s)}$, whereas $H^{(s+1)}$ acts on the effective Hilbert space $\mathcal{H}^{(s+1)}$.
On $\operatorname{Ran}\mathsf{P}_{s\rightarrow s+1}$ the two operators are unitarily equivalent through $V_{s+1}$, while $\widehat{H}^{(s+1)}$ vanishes on the discarded parent subspace $\operatorname{Ran}(\mathds{1}-\mathsf{P}_{s\rightarrow s+1})$.
Consequently, all spectral gaps, ground states, local unitaries, and reduced states at stage $s+1$ refer to $H^{(s+1)}$ on $\mathcal{H}^{(s+1)}$.

\subsection{Cumulative maps and descendants}

The one-step maps may be composed to connect any two stages of the hierarchy.
For $0\le k<s$, define
\begin{align}
	V_{k\rightarrow s} &:= V_{s}\cdots V_{k+1}, \qquad V_{k\rightarrow k}:=\mathds{1}, \qquad \mathsf{P}_{k\rightarrow s} := V_{k\rightarrow s}^{\dagger} V_{k\rightarrow s}, \qquad O^{[k\to s]} := V_{k\rightarrow s}O V_{k\rightarrow s}^{\dagger}. \label{eq:cumulative-coisometry}
\end{align}
Thus $V_{k\rightarrow s}:\mathcal{H}^{(k)}\to\mathcal{H}^{(s)}$ is the cumulative coarse-graining map from stage $k$ to stage $s$, $\mathsf{P}_{k\rightarrow s}$ projects onto the subspace of $\mathcal{H}^{(k)}$ retained through all intermediate truncations up to stage $s$, and $O^{[k\to s]}$ denotes the effective representative on $\mathcal{H}^{(s)}$ of an operator $O$ acting on $\mathcal{H}^{(k)}$.
The index order always denotes input stage $\rightarrow$ output stage.

Since each $V_{t}$ is a coisometry, their product is again a coisometry.
Indeed, successive use of $V_{t}V_{t}^{\dagger}=\mathds{1}$ gives
\begin{align}
	V_{k\rightarrow s}V_{k\rightarrow s}^{\dagger} = \mathds{1}.
\end{align}
Hence $\mathsf{P}_{k\rightarrow s}$ is an orthogonal projection.
Moreover, since $\mathsf{P}_{s\rightarrow s+1}$ is itself an orthogonal projection, $0\le \mathsf{P}_{s\rightarrow s+1}\le\mathds{1}$, and therefore
\begin{align}
	\mathsf{P}_{k\rightarrow s+1} &= V_{k\rightarrow s}^{\dagger} \mathsf{P}_{s\rightarrow s+1} V_{k\rightarrow s} \le V_{k\rightarrow s}^{\dagger} V_{k\rightarrow s} = \mathsf{P}_{k\rightarrow s}. \label{eq:cumulative-range-nesting}
\end{align}
Thus $\operatorname{Ran}\mathsf{P}_{k\rightarrow s+1}\subseteq \operatorname{Ran}\mathsf{P}_{k\rightarrow s}$: the states retained through stage $s+1$ form a subspace of those retained through stage $s$.
In particular, the microscopic retained ranges $\operatorname{Ran}\mathsf{P}_{0\rightarrow s}$ are nested in $s$.

Conversely, a state $\ket{\psi^{(s)}}\in\mathcal{H}^{(s)}$ can be embedded back into the microscopic Hilbert space by
\begin{align}
	\ket{\widehat{\psi}^{(s)}} := V_{0\rightarrow s}^{\dagger}\ket{\psi^{(s)}}.
\end{align}
Since $V_{0\rightarrow s}^{\dagger}$ is an isometry, $\ket{\widehat{\psi}^{(s)}}$ has the same norm as $\ket{\psi^{(s)}}$.
The hat therefore denotes the microscopic representative of an effective state, not a distinct physical state.

We next record which lower-stage sites are represented by a given effective site.
For $0\le k\le s$ and $v\in\Lambda^{(s)}$, let $D_{k\leftarrow s}(v)\subseteq\Lambda^{(k)}$ be the set of stage-$k$ sites that are successively grouped into the effective site $v$.
Thus $D_{s\leftarrow s}(v)=\{v\}$, while $D_{0\leftarrow s}(v)$ is the microscopic block represented by $v$.
For $Y\subseteq\Lambda^{(s)}$, define its stage-$k$ descendant region by
\begin{align}
	Y^{(k\leftarrow s)} := \bigcup_{v\in Y}D_{k\leftarrow s}(v).
\end{align}
We also use the cumulative block-length and cutoff products
\begin{align}
	\ell_{a:b} := \prod_{t=a}^{b}\ell_{t}, \qquad Z_{s} := \prod_{t=1}^{s}z_{t}. \label{eq:descendant-set-definition}
\end{align}
Here $\ell_{t}$ is the block side length used in the step from stage $t-1$ to stage $t$, while $z_{t}$ is the corresponding low-deviation cutoff.
Empty products are understood to be one.
The quantity $\ell_{a:b}$ is the cumulative linear block scale, whereas $Z_{s}$ is only a convenient notation for the product of cutoff parameters that appears in later bounds; it is not itself a cumulative deviation cutoff.

Because the block hierarchy is nested, every lower-stage site has a unique ancestor at each higher stage.
For the exact block hierarchy considered here, a single stage-$s$ site represents
\begin{align}
	N_{s} := |D_{0\leftarrow s}(v)| = \ell_{1:s}^{D} = \frac{n}{n^{(s)}} \label{eq:microscopic-block-volume}
\end{align}
microscopic sites.

The cumulative map also factorizes over the disjoint descendant blocks represented by the stage-$s$ sites.
More generally, for $0\le k<s$, under the natural tensor-product identification,
\begin{align}
	V_{k\rightarrow s} = \bigotimes_{v\in\Lambda^{(s)}} V_{k\rightarrow s}^{v}, \qquad \mathsf{P}_{k\rightarrow s}^{v} := (V_{k\rightarrow s}^{v})^{\dagger} V_{k\rightarrow s}^{v},
\end{align}
where
\begin{align}
	V_{k\rightarrow s}^{v}: \bigotimes_{i\in D_{k\leftarrow s}(v)} \mathcal{K}_{k,i} \longrightarrow \mathcal{K}_{s,v}
\end{align}
is the cumulative coarse-graining map associated with the descendants of $v$.
Thus $\mathsf{P}_{k\rightarrow s}^{v}$ is the corresponding retained projection on the stage-$k$ descendant block.
In particular, $\mathsf{P}_{0\rightarrow s}^{v}$ acts on the microscopic block $D_{0\leftarrow s}(v)$.

The microscopic representative of the stage-$s$ reference vector $\ket{0_{v}^{(s)}}$ is $(V_{0\rightarrow s}^{v})^{\dagger} \ket{0_{v}^{(s)}}$.
For a single RG step, this reduces to the parent-space representative $V_{s,j}^{\dagger}\ket{0_{j}^{(s)}}$.

\subsubsection{Local energy bounds}

To track the size of the effective Hamiltonian at each stage, choose symmetric nonnegative quantities $J_{ij}^{(s)}=J_{ji}^{(s)}$ that bound the pair interactions, and let $\bar{g}^{(s)}$ bound the total local energy scale at a single effective site:
\begin{align}
	\norm{h_{ij}^{(s)}} &\le J_{ij}^{(s)}, \qquad \max_{i}\left( \norm{h_{i}^{(s)}} + \sum_{j\ne i}J_{ij}^{(s)} \right) \le \bar{g}^{(s)}. \label{all_to_all_cond_s_th_Re}
\end{align}
Thus $J_{ij}^{(s)}$ is an upper bound on the strength of the interaction between effective sites $i$ and $j$, while $\bar{g}^{(s)}$ is a uniform upper bound on the local energy scale at stage $s$.

When the interaction bounds are chosen in the Kac form
\begin{align}
	J_{ij}^{(s)} = \frac{J^{(s)}}{[n^{(s)}]^{\beta}} \mu_{ij}^{(s)}r_{ij}^{-\alpha}, \label{RG_interaction_strength_upper_bound}
\end{align}
where $r_{ij}$ denotes the distance on the stage-$s$ lattice, define the corresponding normalized coupling row sum by
\begin{align}
	\gamma_{0}^{(s)} := \max_{i} \frac{1}{[n^{(s)}]^{\beta}} \sum_{j\ne i} \mu_{ij}^{(s)}r_{ij}^{-\alpha}. \label{eq:scale-kac-row-sum}
\end{align}
Then
\begin{align}
	\bar{g}^{(s)} := \max_{i}\norm{h_{i}^{(s)}} + J^{(s)}\gamma_{0}^{(s)}
\end{align}
is an admissible choice in Eq.~\eqref{all_to_all_cond_s_th_Re}.
Indeed,
\begin{align}
	\sum_{j\ne i}J_{ij}^{(s)} = \frac{J^{(s)}}{[n^{(s)}]^{\beta}} \sum_{j\ne i} \mu_{ij}^{(s)}r_{ij}^{-\alpha} \le J^{(s)}\gamma_{0}^{(s)}.
\end{align}
Here $J^{(s)}$ and $\bar{g}^{(s)}$ are bounding parameters, while $\gamma_{0}^{(s)}$ records the normalized row sum of the chosen Kac kernel.
Their dependence on the renormalization stage will be estimated below rather than assumed.

\section{Geometric square sums and renormalized interactions}
\label{sec:renormalized-interactions}

The projected block--block estimate is independent of dimension, but its usefulness depends on the geometric square sum of couplings.
This section supplies the counting in detail.
In particular, the constants may depend on the fixed exponent $\alpha$; no uniformity as $\alpha$ approaches a marginal exponent is implicit.

\subsection{Counting pairs near an interface}

Let $B$ be a hypercube of side $\ell$ in a lattice of side $L$.
For $r\ge1$, let $N_{B}(r):=|\{(i,j):i\in B,j\notin B,r_{ij}=r\}|$.
A site $i$ contributing when $r\le\ell$ lies within graph distance $r$ of a face of $B$.
There are at most $2Dr\ell^{D-1}$ such sites, and each has at most $a_{D}r^{D-1}$ sites at distance $r$.
For arbitrary $r$, there are at most $\ell^{D}$ choices of $i$.
Hence
\begin{align}
	N_{B}(r)\le C_{D}\min\{\ell^{D-1}r^{D},\ell^{D}r^{D-1}\}. \label{eq:interface-pair-count}
\end{align}
The same inequalities hold with open or periodic boundaries; using the faces of a lifted cube only overcounts pairs.
In particular,
\begin{align}
	\sum_{i\in B,j\notin B}r_{ij}^{-2\alpha} \le C_{D}\left[ \ell^{D-1}\sum_{r\le\ell}r^{D-2\alpha} +\ell^{D}\sum_{\ell<r\le DL}r^{D-1-2\alpha}\right]. \label{eq:interface-two-sums}
\end{align}
If $j$ is restricted to a single neighboring block of $B$, rather than the full complement of $B$, then the distance $r_{ij}$ is at most $C_{D}\ell$.
Hence the second sum in Eq.~\eqref{eq:interface-two-sums} only extends up to $C_{D}\ell$.
This distinction will be used below for block--block and block--complement estimates.

We first collect the two geometric square-sum scales that will be used below.
For two neighboring blocks of side $\ell$, define
\begin{align}
	\mathfrak{G}_{\alpha}(\ell) := \begin{cases} \ell^{D-\alpha}, &2\alpha<D+1,\\
		\ell^{(D-1)/2}\sqrt{\log(e\ell)}, &2\alpha=D+1,\\
		\ell^{(D-1)/2}, &2\alpha>D+1. \end{cases} \label{eq:critical-block-kernel}
\end{align}
For a block of side $\ell$ and its full complement in a lattice of side $L$, define
\begin{align}
	\mathfrak{F}_{\alpha}(\ell,L) := \begin{cases} \ell^{D/2}L^{D/2-\alpha}, &\alpha<D/2,\\
		\ell^{D/2}\sqrt{\log(eL)}, &\alpha=D/2,\\
		\ell^{D-\alpha}, &D/2<\alpha<(D+1)/2,\\
		\ell^{(D-1)/2}\sqrt{\log(e\ell)}, &\alpha=(D+1)/2,\\
		\ell^{(D-1)/2}, &\alpha>(D+1)/2. \end{cases} \label{eq:F-alpha-square-sum}
\end{align}
Here $L$ denotes the side length of the full lattice.

For an exact cubic partition into blocks $B_{v}$ of side $\ell$, we index the blocks by the sites $v$ of the corresponding coarse lattice.
For a coarse site $v=(v_{1},\ldots,v_{D})$, define its Moore neighborhood by
\begin{align}
	\widetilde{\mathcal{S}}_{v} := \left\{ v'\ne v: |v_{a}-v_{a}'|\le1 \text{ for every }a=1,\ldots,D \right\}.
\end{align}
Thus $\widetilde{\mathcal{S}}_{v}$ consists of the coarse blocks sharing a face, edge, or corner with $B_{v}$.
For periodic boundaries, the difference $v-v'$ is understood in the minimum-displacement convention.
For the quantitative entropy estimates in the interval $D/2<\alpha<(D+1)/2$, fix
\begin{align}
	C_{D,\alpha}:=C_{D}\left[1+D^{D-\alpha}\left(1+\frac{1}{D-\alpha}\right)+\frac{1}{2\alpha-D}+\frac{1}{D+1-2\alpha}\right]. \label{eq:geometric-constant-choice}
\end{align}
This value is used for $C_{D,\alpha}$ throughout the subsequent estimates in this interval; outside it, that symbol retains its generic geometric-constant meaning.

\begin{lemma}
	\label{gamma_L_L'_alpha}
	Recall that $\gamma_{X,Y}^{(2\alpha)}$ denotes the weighted square-sum kernel defined in Eq.~\eqref{Prop/main_ineq:renormalized_gamma_def}.
	For a cube $B$ of side $\ell$,
	\begin{align}
		\gamma_{B,B^{\mathrm{c}}}^{(2\alpha)} \le C_{D,\alpha}\mathfrak{F}_{\alpha}(\ell,L) + C_{D}(\mu_{0}-1)\ell^{(D-1)/2}. \label{eq:block-complement-kernel}
	\end{align}
	For two distinct blocks $B_{v},B_{v'}$ in the above cubic partition, let $R_{vv'}$ denote the graph distance between $v$ and $v'$ on the coarse lattice.
	Then
	\begin{align}
		\gamma_{B_{v},B_{v'}}^{(2\alpha)} \le \begin{cases} (2D)^{\alpha} \ell^{D-\alpha}R_{vv'}^{-\alpha}, &v'\notin\widetilde{\mathcal{S}}_{v},\\
			C_{D,\alpha}\mathfrak{G}_{\alpha}(\ell) + C_{D}(\mu_{0}-1)\ell^{(D-1)/2}, &v'\in\widetilde{\mathcal{S}}_{v}. \end{cases} \label{eq:block-kernel-piecewise}
	\end{align}
\end{lemma}

\begin{proof}
	By definition,
	\begin{align}
		\gamma_{B,B^{\mathrm{c}}}^{(2\alpha)} = \left(\sum_{\substack{i\in B\\j\notin B}}\mu_{ij}^{2}r_{ij}^{-2\alpha}\right)^{1/2} \le \left(\sum_{\substack{i\in B\\j\notin B}}r_{ij}^{-2\alpha}\right)^{1/2} + (\mu_{0}-1) \left(\sum_{\substack{i\in B,\ j\notin B\\j\in\mathcal{S}_{i}}}r_{ij}^{-2\alpha}\right)^{1/2}, \label{eq:block-kernel-split}
	\end{align}
	where we used $\mu_{ij}=1+(\mu_{0}-1)\mathds{1}_{\{j\in\mathcal{S}_{i}\}}$ and the triangle inequality for the Euclidean norm.
	
	For the first term, Eq.~\eqref{eq:interface-two-sums} gives
	\begin{align}
		\sum_{\substack{i\in B\\j\notin B}}r_{ij}^{-2\alpha} \le C_{D}\left[ \ell^{D-1}\sum_{r\le\ell}r^{D-2\alpha} + \ell^{D}\sum_{\ell<r\le DL}r^{D-1-2\alpha} \right]. \label{eq:block-kernel-sum-proof}
	\end{align}
	The two contributions satisfy
	\begin{align}
		\ell^{D-1}\sum_{r\le\ell}r^{D-2\alpha} &\le C_{D,\alpha} \begin{cases} \ell^{2D-2\alpha},&2\alpha<D+1,\\
			\ell^{D-1}\log(e\ell),&2\alpha=D+1,\\
			\ell^{D-1},&2\alpha>D+1, \end{cases} \label{eq:block-kernel-short-distance-power-sum}\\
		\ell^{D}\sum_{\ell<r\le DL}r^{D-1-2\alpha} &\le C_{D,\alpha} \begin{cases} \ell^{D}L^{D-2\alpha},&2\alpha<D,\\
			\ell^{D}\log(eL),&2\alpha=D,\\
			\ell^{2D-2\alpha},&2\alpha>D. \end{cases} \label{eq:block-kernel-power-sums}
	\end{align}
	Therefore, using $L\ge\ell$ and taking the square root,
	\begin{align}
		\left(\sum_{\substack{i\in B\\j\notin B}}r_{ij}^{-2\alpha}\right)^{1/2} &\le C_{D,\alpha} \begin{cases} \ell^{D/2}L^{D/2-\alpha},&\alpha<D/2,\\
			\ell^{D/2}\sqrt{\log(eL)},&\alpha=D/2,\\
			\ell^{D-\alpha},&D/2<\alpha<(D+1)/2,\\
			\ell^{(D-1)/2}\sqrt{\log(e\ell)},&\alpha=(D+1)/2,\\
			\ell^{(D-1)/2},&\alpha>(D+1)/2, \end{cases} \notag\\
		&= C_{D,\alpha}\mathfrak{F}_{\alpha}(\ell,L). \label{eq:unweighted-block-complement}
	\end{align}
	For the second term in Eq.~\eqref{eq:block-kernel-split}, a pair $(i,j)$ contributes only if $i\in B$, $j\notin B$, and $j\in\mathcal{S}_{i}$.
	Such an $i$ must lie on the boundary of $B$, and each site has at most $3^{D}-1$ Moore neighbors.
	Hence the number of contributing pairs is at most $(3^{D}-1)|\partial B|\le C_{D}\ell^{D-1}$.
	Since $r_{ij}\ge1$, this term is therefore bounded by $C_{D}\ell^{(D-1)/2}$.
	Since $r_{ij}^{-2\alpha}\le1$,
	\begin{align}
		\left(\sum_{\substack{i\in B,\ j\notin B\\j\in\mathcal{S}_{i}}}r_{ij}^{-2\alpha}\right)^{1/2} \le C_{D}\ell^{(D-1)/2}.
	\end{align}
	Combining the last two bounds with Eq.~\eqref{eq:block-kernel-split} proves Eq.~\eqref{eq:block-complement-kernel}.
	
	Now let $v'\in\widetilde{\mathcal{S}}_{v}$.
	Since $B_{v}$ and $B_{v'}$ are neighboring blocks of side $\ell$, every $i\in B_{v}$, $j\in B_{v'}$ satisfies $r_{ij}\le C_{D}\ell$.
	The same counting as in Eq.~\eqref{eq:block-kernel-sum-proof}, with the second sum truncated at $C_{D}\ell$, yields
	\begin{align}
		\left(\sum_{\substack{i\in B_{v}\\j\in B_{v'}}}r_{ij}^{-2\alpha}\right)^{1/2} &\le C_{D,\alpha} \begin{cases} \ell^{D-\alpha},&2\alpha<D+1,\\
			\ell^{(D-1)/2}\sqrt{\log(e\ell)},&2\alpha=D+1,\\
			\ell^{(D-1)/2},&2\alpha>D+1, \end{cases} \notag\\
		&= C_{D,\alpha}\mathfrak{G}_{\alpha}(\ell).
	\end{align}
	For the Moore part, a contributing pair must satisfy $i\in B_{v}$, $j\in B_{v'}$, and $j\in\mathcal{S}_{i}$.
	Hence $i$ must lie on the boundary of $B_{v}$, and the same counting as above gives at most $C_{D}\ell^{D-1}$ such pairs.
	Since $r_{ij}\ge1$, their contribution is bounded by $C_{D}(\mu_{0}-1)\ell^{(D-1)/2}$.
	This proves the neighboring-block branch of Eq.~\eqref{eq:block-kernel-piecewise}.
	
	Finally, suppose $v'\notin\widetilde{\mathcal{S}}_{v}$, and set $q:=\max_{1\le a\le D}|v_{a}-v_{a}'|$.
	Since $R_{vv'}=\sum_{a=1}^{D}|v_{a}-v_{a}'|$, we have $q\ge R_{vv'}/D$.
	Moreover, $v'\notin\widetilde{\mathcal{S}}_{v}$ implies $q\ge2$.
	Along a coordinate for which the displacement equals $q$, two blocks of side $\ell$ are separated by at least $\ell(q-1)$ microscopic lattice spacings.
	Hence, for every $i\in B_{v}$ and $j\in B_{v'}$,
	\begin{align}
		r_{ij} \ge \ell(q-1) \ge \frac{\ell q}{2} \ge \frac{\ell R_{vv'}}{2D}.
	\end{align}
	Since the two blocks are not coarse Moore neighbors, no microscopic Moore pair connects them, and therefore $\mu_{ij}=1$ for all $i\in B_{v}$ and $j\in B_{v'}$.
	Since each block contains $\ell^{D}$ sites, there are $\ell^{2D}$ such pairs.
	Consequently,
	\begin{align}
		\gamma_{B_{v},B_{v'}}^{(2\alpha)} &= \left( \sum_{\substack{i\in B_{v}\\j\in B_{v'}}} r_{ij}^{-2\alpha} \right)^{1/2} \le \left[ \ell^{2D} \left( \frac{2D}{\ell R_{vv'}} \right)^{2\alpha} \right]^{1/2} = (2D)^{\alpha} \ell^{D-\alpha}R_{vv'}^{-\alpha}.
	\end{align}
	This proves the remaining branch of Eq.~\eqref{eq:block-kernel-piecewise}.
	For $D/2<\alpha<(D+1)/2$, the estimates $\sum_{r=1}^{\ell}r^{D-2\alpha}\le\ell^{D+1-2\alpha}/(D+1-2\alpha)$ and $\sum_{r>\ell}r^{D-1-2\alpha}\le\ell^{D-2\alpha}/(2\alpha-D)$ verify the choice in Eq.~\eqref{eq:geometric-constant-choice} for both the squared sums and their square roots.
	The same value also dominates $a_{D}D^{D-\alpha}(1+(D-\alpha)^{-1})$ and $\sqrt{a_{D}(1+(2\alpha-D)^{-1})}$, which will control the normalized row sums and internal square sums.
	For periodic boundaries, the same argument applies with the coarse displacement chosen according to the minimum-displacement convention.
\end{proof}

\subsection{Three interaction recurrences}

We now determine how the effective pair-interaction strength changes from stage $s$ to stage $s+1$.
The one-step compression estimate bounds a new pair interaction by the weighted square-sum coupling between its two parent blocks.
The geometric estimates derived above show two different behaviors: non-neighboring coarse blocks retain the power-law decay with their coarse-lattice distance, whereas coarse Moore neighbors can acquire a stronger short-distance contribution.
Tracking these two contributions gives recurrences for the overall interaction strength $J^{(s)}$ and, when necessary, for the short-distance factor $\mu_{s}$.

At stage $s$, suppose
\begin{align}
	\norm{h_{ij}^{(s)}} \le J_{ij}^{(s)} := \frac{J^{(s)}}{[n^{(s)}]^{\beta}} \mu_{ij}^{(s)}r_{ij}^{-\alpha}, \qquad \mu_{ij}^{(s)} = \begin{cases} \mu_{s},&j\in\mathcal{S}_{i}^{(s)},\\
		1,&j\notin\mathcal{S}_{i}^{(s)}, \end{cases} \qquad \mu_{s}\ge1. \label{eq:scale-kac-upper-bound}
\end{align}
For disjoint $X,Y\subseteq\Lambda^{(s)}$, define the corresponding weighted square-sum kernel by
\begin{align}
	\gamma_{X,Y}^{(2\alpha,s)} := \left( \sum_{\substack{i\in X\\j\in Y}} [\mu_{ij}^{(s)}]^{2}r_{ij}^{-2\alpha} \right)^{1/2}.
\end{align}
Let $L_{j}^{(s)}$ and $L_{j'}^{(s)}$ be the two parent blocks that become sites $j$ and $j'$ at stage $s+1$.
Since the local coisometries are unitary from the retained block subspaces onto the corresponding effective sites, Proposition~\ref{Prop:renormalized_interaction} gives
\begin{align}
	\norm{h_{jj'}^{(s+1)}} \le \frac{16z_{s+1}J^{(s)}}{[n^{(s)}]^{\beta}} \gamma_{L_{j}^{(s)},L_{j'}^{(s)}}^{(2\alpha,s)}. \label{eq:one-step-interaction-norm}
\end{align}
Thus it remains to insert the geometric bounds on the block kernel and rewrite the resulting interaction bound in the form of Eq.~\eqref{eq:scale-kac-upper-bound} at stage $s+1$.

For $D/2<\alpha\le(D+1)/2$, we absorb the microscopic Moore factor into $J_{\rm K}:=\mu_{0}J$.
Since $\mu_{ij}\le\mu_{0}$, the microscopic interaction bound remains valid with $J^{(0)}=J_{\rm K}$ and with the stagewise short-distance factor initialized to one.
For $(D+1)/2<\alpha<D$, we instead retain $J^{(0)}=J$ and the original factor $\mu_{0}$.
We recall that $Z_{s}=\prod_{t=1}^{s}z_{t}$ and $\ell_{1:s}=\prod_{t=1}^{s}\ell_{t}$.
We also define
\begin{align}
	\chi_{s} := \prod_{t=1}^{s}\sqrt{\log(e\ell_{t})}.
\end{align}
The factor $\chi_{s}$ is needed only when $\alpha=(D+1)/2$.

\begin{prop}
	\label{prop:Renormalized_interactions}
	With $c_{\rm int}\ge1$ fixed by Eq.~\eqref{eq:c-int-choice}, depending only on $D$ and $\alpha$, Eq.~\eqref{eq:scale-kac-upper-bound} holds at every defined stage with the following choices.
	
	If $D/2<\alpha<(D+1)/2$, then
	\begin{align}
		\mu_{s}=1, \qquad J^{(s)} = c_{\rm int}^{s}J_{\rm K}Z_{s}. \label{eq:J-recurrence-Dhalf-alpha-Dplus1half}
	\end{align}
	If $\alpha=(D+1)/2<D$, then
	\begin{align}
		\mu_{s}=1, \qquad J^{(s)} = c_{\rm int}^{s}J_{\rm K}Z_{s}\chi_{s}. \label{eq:J-renormalization-alpha-eq-Dplus1half}
	\end{align}
	If $(D+1)/2<\alpha<D$, then
	\begin{align}
		\mu_{s} = \mu_{0} \ell_{1:s}^{\alpha-(D+1)/2}, \qquad J^{(s)} = c_{\rm int}^{s}JZ_{s}. \label{eq:mu-positive-recurrence-alpha-gt-Dplus1half}
	\end{align}
\end{prop}

\begin{proof}
	We prove the three statements by induction on the stage $s$.
	For $D/2<\alpha\le(D+1)/2$, we use at stage $0$ the reparametrized bound
	\begin{align}
		J^{(0)} = J_{\rm K} := \mu_{0}J, \qquad \mu_{ij}^{(0)} = 1.
	\end{align}
	This is valid because the original microscopic weights satisfy $\mu_{ij}\le\mu_{0}$.
	Thus the microscopic parameter $\mu_{0}$ is absorbed into $J_{\rm K}$ and is not identified with the stagewise short-distance factor.
	For $(D+1)/2<\alpha<D$, we instead retain the original stage-$0$ bound, with $J^{(0)}=J$ and the original microscopic factor $\mu_{0}$.
	Hence the three claimed formulas are initialized at stage $0$, with empty products equal to one.
	
	We first fix the constant used throughout the induction.
	In Eq.~\eqref{eq:block-kernel-piecewise}, $C_{D,\alpha}$ is the constant multiplying the ordinary square-sum contribution $\mathfrak{G}_{\alpha}$, while $C_{D}$ multiplies the additional Moore contribution.
	Set
	\begin{align}
		C_{\rm far} := 16(2D)^{\alpha}, \qquad C_{\rm nb} := 16\bigl(C_{D,\alpha}+C_{D}\bigr),
	\end{align}
	and define
	\begin{align}
		c_{\rm int} := \max\left\{ 1,\, C_{\rm far},\, D^{\alpha} C_{\rm nb} \right\}. \label{eq:c-int-choice}
	\end{align}
	This choice depends only on $D$ and $\alpha$ and is independent of $s$, $n$, and the block sizes.
	
	Assume that the corresponding claim holds at stage $s$.
	Since $n^{(s+1)}=n^{(s)}/\ell_{s+1}^{D}$ and $D\beta=D-\alpha$,
	\begin{align}
		\frac{\ell_{s+1}^{D-\alpha}}{[n^{(s)}]^{\beta}} = \frac{1}{[n^{(s+1)}]^{\beta}}. \label{eq:kac-scale-covariance}
	\end{align}
	
	We first treat interactions between blocks that are not coarse Moore neighbors.
	Let $R_{jj'}$ be the graph distance between $j$ and $j'$ on the stage-$(s+1)$ lattice.
	The first branch of Eq.~\eqref{eq:block-kernel-piecewise} gives
	\begin{align}
		\gamma_{L_{j}^{(s)},L_{j'}^{(s)}}^{(2\alpha,s)} \le (2D)^{\alpha} \ell_{s+1}^{D-\alpha}R_{jj'}^{-\alpha}.
	\end{align}
	Combining this with Eqs.~\eqref{eq:one-step-interaction-norm} and \eqref{eq:kac-scale-covariance} yields
	\begin{align}
		\norm{h_{jj'}^{(s+1)}} \le C_{\rm far}z_{s+1}J^{(s)} \frac{R_{jj'}^{-\alpha}}{[n^{(s+1)}]^{\beta}} \le c_{\rm int}z_{s+1}J^{(s)} \frac{R_{jj'}^{-\alpha}}{[n^{(s+1)}]^{\beta}}. \label{eq:distant-stage-interaction}
	\end{align}
	Thus the non-neighboring interactions retain the same power-law distance dependence at stage $s+1$.
	
	For coarse Moore neighbors, $1\le R_{jj'}\le D$.
	The factor $R_{jj'}^{-\alpha}$ is used here only to place the neighboring interactions in the same Kac form as the non-neighboring ones.
	By Eq.~\eqref{eq:c-int-choice},
	\begin{align}
		c_{\rm int}R_{jj'}^{-\alpha} \ge c_{\rm int}D^{-\alpha} \ge C_{\rm nb}. \label{eq:c-int-neighbor-absorption}
	\end{align}
	
	Suppose first that $D/2<\alpha<(D+1)/2$.
	The induction hypothesis gives $\mu_{s}=1$, so the Moore contribution in Eq.~\eqref{eq:block-kernel-piecewise} vanishes.
	Since $\mathfrak{G}_{\alpha}(\ell_{s+1}) =\ell_{s+1}^{D-\alpha}$, the neighboring-block estimate gives
	\begin{align}
		\gamma_{L_{j}^{(s)},L_{j'}^{(s)}}^{(2\alpha,s)} \le C_{D,\alpha}\ell_{s+1}^{D-\alpha}.
	\end{align}
	Hence Eqs.~\eqref{eq:one-step-interaction-norm} and \eqref{eq:kac-scale-covariance} imply
	\begin{align}
		\norm{h_{jj'}^{(s+1)}} \le 16C_{D,\alpha}z_{s+1}J^{(s)} \frac{1}{[n^{(s+1)}]^{\beta}} \le C_{\rm nb}z_{s+1}J^{(s)} \frac{1}{[n^{(s+1)}]^{\beta}}.
	\end{align}
	Using Eq.~\eqref{eq:c-int-neighbor-absorption},
	\begin{align}
		\norm{h_{jj'}^{(s+1)}} \le c_{\rm int}z_{s+1}J^{(s)} \frac{R_{jj'}^{-\alpha}}{[n^{(s+1)}]^{\beta}}.
	\end{align}
	Together with Eq.~\eqref{eq:distant-stage-interaction}, this proves the stage-$(s+1)$ interaction bound with
	\begin{align}
		\mu_{s+1} = 1, \qquad J^{(s+1)} = c_{\rm int}z_{s+1}J^{(s)}. \label{eq:one-step-recurrence-Dhalf-alpha-Dplus1half}
	\end{align}
	Using $J^{(s)}=c_{\rm int}^{s}J_{\rm K}Z_{s}$ and $Z_{s+1}=Z_{s}z_{s+1}$ gives
	\begin{align}
		J^{(s+1)} = c_{\rm int}^{s+1}J_{\rm K}Z_{s+1},
	\end{align}
	which closes the induction in this range.
	
	Next suppose that $\alpha=(D+1)/2$.
	Again the induction hypothesis gives $\mu_{s}=1$, so the Moore contribution vanishes.
	In this case
	\begin{align}
		\mathfrak{G}_{\alpha}(\ell_{s+1}) = \ell_{s+1}^{D-\alpha} \sqrt{\log(e\ell_{s+1})}.
	\end{align}
	The ordinary square-sum term in Eq.~\eqref{eq:block-kernel-piecewise} therefore gives
	\begin{align}
		\norm{h_{jj'}^{(s+1)}} \le 16C_{D,\alpha}z_{s+1}J^{(s)} \sqrt{\log(e\ell_{s+1})} \frac{1}{[n^{(s+1)}]^{\beta}}.
	\end{align}
	Since $16C_{D,\alpha}\le C_{\rm nb}$, Eq.~\eqref{eq:c-int-neighbor-absorption} gives
	\begin{align}
		\norm{h_{jj'}^{(s+1)}} \le c_{\rm int}z_{s+1}J^{(s)} \sqrt{\log(e\ell_{s+1})} \frac{R_{jj'}^{-\alpha}}{[n^{(s+1)}]^{\beta}}.
	\end{align}
	For non-neighboring pairs, Eq.~\eqref{eq:distant-stage-interaction} satisfies the same bound because $\sqrt{\log(e\ell_{s+1})}\ge1$.
	Therefore
	\begin{align}
		\mu_{s+1} = 1, \qquad J^{(s+1)} = c_{\rm int}z_{s+1} \sqrt{\log(e\ell_{s+1})}\, J^{(s)}. \label{eq:one-step-recurrence-alpha-eq-Dplus1half}
	\end{align}
	Using $J^{(s)}=c_{\rm int}^{s}J_{\rm K}Z_{s}\chi_{s}$, $Z_{s+1}=Z_{s}z_{s+1}$, and $\chi_{s+1} =\chi_{s}\sqrt{\log(e\ell_{s+1})}$ gives
	\begin{align}
		J^{(s+1)} = c_{\rm int}^{s+1}J_{\rm K}Z_{s+1}\chi_{s+1},
	\end{align}
	which closes the induction at $\alpha=(D+1)/2$.
	
	Finally, suppose $(D+1)/2<\alpha<D$.
	For coarse Moore neighbors, the ordinary square-sum term in Eq.~\eqref{eq:block-kernel-piecewise} contributes $C_{D,\alpha}\ell_{s+1}^{(D-1)/2}$, while the additional Moore term contributes $C_{D}(\mu_{s}-1)\ell_{s+1}^{(D-1)/2}$.
	Hence
	\begin{align}
		\gamma_{L_{j}^{(s)},L_{j'}^{(s)}}^{(2\alpha,s)} \le \left[ C_{D,\alpha} + C_{D}(\mu_{s}-1) \right] \ell_{s+1}^{(D-1)/2} \le \bigl(C_{D,\alpha}+C_{D}\bigr) \mu_{s}\ell_{s+1}^{(D-1)/2},
	\end{align}
	where the last inequality uses $\mu_{s}\ge1$.
	Moreover,
	\begin{align}
		\frac{\ell_{s+1}^{(D-1)/2}}{[n^{(s)}]^{\beta}} = \frac{ \ell_{s+1}^{\alpha-(D+1)/2} }{ [n^{(s+1)}]^{\beta} }.
	\end{align}
	Equation~\eqref{eq:one-step-interaction-norm} therefore gives
	\begin{align}
		\norm{h_{jj'}^{(s+1)}} \le C_{\rm nb}z_{s+1}J^{(s)} \mu_{s} \ell_{s+1}^{\alpha-(D+1)/2} \frac{1}{[n^{(s+1)}]^{\beta}}.
	\end{align}
	We therefore set
	\begin{align}
		\mu_{s+1} = \mu_{s} \ell_{s+1}^{\alpha-(D+1)/2}.
	\end{align}
	Using Eq.~\eqref{eq:c-int-neighbor-absorption},
	\begin{align}
		\norm{h_{jj'}^{(s+1)}} \le c_{\rm int}z_{s+1}J^{(s)} \mu_{s+1} \frac{R_{jj'}^{-\alpha}}{[n^{(s+1)}]^{\beta}}.
	\end{align}
	Together with Eq.~\eqref{eq:distant-stage-interaction}, the stage-$(s+1)$ interaction bound holds with
	\begin{align}
		\mu_{s+1} &= \mu_{s} \ell_{s+1}^{\alpha-(D+1)/2}, \qquad J^{(s+1)} = c_{\rm int}z_{s+1}J^{(s)}. \label{eq:one-step-recurrence-alpha-gt-Dplus1half}
	\end{align}
	Using $\mu_{s} =\mu_{0}\ell_{1:s}^{\alpha-(D+1)/2}$, $J^{(s)}=c_{\rm int}^{s}JZ_{s}$, and $Z_{s+1}=Z_{s}z_{s+1}$ gives
	\begin{align}
		\mu_{s+1} &= \mu_{0} \ell_{1:s+1}^{\alpha-(D+1)/2}, \qquad J^{(s+1)} = c_{\rm int}^{s+1}JZ_{s+1},
	\end{align}
	which closes the induction.
\end{proof}

\subsubsection{Accumulation at $\alpha=(D+1)/2$}

The factor $\sqrt{\log(e\ell_{t})}$ generated at one stage need not remain polylogarithmic after multiplication over all stages.
Suppose, for example, that $n^{(s)}\simeq n^{(1-\nu)^{s}}$ for some fixed $0<\nu<1$, and that the hierarchy terminates at a stage $s_{\mathrm{fin}}$ for which $n^{(s_{\mathrm{fin}})}$ is polylogarithmic in $n$.
Then $s_{\mathrm{fin}}=\Theta(\log\log n)$.
Moreover, $n^{(t)}=n^{(t-1)}/\ell_{t}^{D}$ gives
\begin{align}
	\log\ell_{t} = \frac{1}{D} \left( \log n^{(t-1)}-\log n^{(t)} \right) = \Theta\left( (1-\nu)^{t-1}\log n \right).
\end{align}
It follows that
\begin{align}
	\log\chi_{s_{\mathrm{fin}}} = \frac{1}{2} \sum_{t=1}^{s_{\mathrm{fin}}} \log\log(e\ell_{t}) = \Theta\left((\log\log n)^{2}\right). \label{eq:cumulative-growth-alpha-eq-Dplus1half}
\end{align}
Hence
\begin{align}
	\chi_{s_{\mathrm{fin}}} = \exp\left[ \Theta\left((\log\log n)^{2}\right) \right].
\end{align}
This is subpolynomial in $n$, since $(\log\log n)^{2}=o(\log n)$, but it is larger than every fixed power of $\log n$.
The factor $\chi_{s_{\mathrm{fin}}}$ must therefore be kept explicitly.

\section{Hierarchical centering and compressed fluctuations}
\label{sec:hierarchical-centering}

A microscopic one-site sum becomes a one-site operator on a large renormalized block.
Bounding it by the sum of microscopic norms would lose the square-sum gain.
The following centering proposition controls its fluctuating part on the retained range, even when acting on that range can leak into a discarded subspace.

\subsection{A one-sided retained-range norm estimate}

For products of cutoffs define the one-sided fluctuation factors
\begin{align}
	\eta_{k+1:s}:=\prod_{t=k+1}^{s}(2\sqrt{z_{t}}+\sqrt{z_{t}+1}), \qquad \eta_{s}:=\eta_{1:s},\qquad \eta_{0}:=1. \label{defs_zeta_s_eta_s}
\end{align}

\begin{prop}
	\label{block_renormalized_variance}
	Let $0\le k<s$, let $v\in\Lambda^{(s)}$, and let $I\subseteq D_{k\leftarrow s}(v)$.
	Define
	\begin{align}
		A^{(k)} := \sum_{i\in I}a_{i}^{(k)},
	\end{align}
	where each $a_{i}^{(k)}$ is a Hermitian scale-$k$ one-site operator.
	Then
	\begin{align}
		\inf_{c\in\mathbb{R}} \norm{(A^{(k)}-c\mathds{1})\mathsf{P}_{k\rightarrow s}^{v}} \le \eta_{k+1:s} \left( \sum_{i\in I}\norm{a_{i}^{(k)}}^{2} \right)^{1/2}. \label{eq:main_ineq_renormalized}
	\end{align}
	In particular, for $k=0$, $I=D_{0\leftarrow s}(v)$, and $A=\sum_{i\in D_{0\leftarrow s}(v)}a_{i}$, this becomes
	\begin{align}
		\inf_{c\in\mathbb{R}}\norm{(A-c\mathds{1})\mathsf{P}_{0\rightarrow s}^{v}} \le\eta_{s}\left(\sum_{i\in D_{0\leftarrow s}(v)} \norm{a_{i}}^{2}\right)^{1/2}.
	\end{align}
	No bound on the site dimensions is required.
\end{prop}

\begin{proof}
	Fix the starting scale $k$.
	We prove the assertion by induction on the number of renormalization steps above scale $k$.
	At order zero there is one scale-$k$ site, $\mathsf{P}_{k\rightarrow k}=\mathds{1}$, and choosing $c=0$ gives
	\begin{align}
		\norm{(a_{i}^{(k)}-c\mathds{1})\mathsf{P}_{k\rightarrow k}} =\norm{a_{i}^{(k)}} =\eta_{k+1:k}\norm{a_{i}^{(k)}},
	\end{align}
	where the empty product is understood as $\eta_{k+1:k}=1$.
	This proves the base case.
	
	At the first level above $k$, fix a scale-$(k+1)$ site $u\in\Lambda^{(k+1)}$, with $I\subseteq D_{k\leftarrow k+1}(u)$.
	Write $\Pi_{\le z}^{(k)}$ for the low-flip projection on the block $D_{k\leftarrow k+1}(u)$ under consideration at scale $k$, and set $c=\sum_{i}\bra{0_{i}^{(k)}}a_{i}^{(k)}\ket{0_{i}^{(k)}}$ and
	\begin{align}
		a_{i}' := a_{i}^{(k)} -\bra{0_{i}^{(k)}}a_{i}^{(k)}\ket{0_{i}^{(k)}}\mathds{1}.
	\end{align}
	Here and below, the sums may equivalently be taken over the full descendant block by setting $a_{i}^{(k)}=0$ for $i\notin I$.
	Since $P_{i}^{(k)}a'_{i}P_{i}^{(k)}=0$,
	\begin{align}
		a'_{i} &=(P_{i}^{(k)}+Q_{i}^{(k)})a'_{i}P_{i}^{(k)} +a'_{i}Q_{i}^{(k)} = Q_{i}^{(k)}a'_{i}P_{i}^{(k)} +a'_{i}Q_{i}^{(k)} = Q_{i}^{(k)}a_{i}^{(k)}P_{i}^{(k)} +a'_{i}Q_{i}^{(k)}.
	\end{align}
	Therefore
	\begin{align}
		(A^{(k)}-c\mathds{1}) \Pi_{\le z_{k+1}}^{(k)} =\sum_{i}\bigl[ \Pi_{\le z_{k+1}+1}^{(k)} Q_{i}^{(k)}a_{i}^{(k)}P_{i}^{(k)} \Pi_{\le z_{k+1}}^{(k)} +a'_{i}Q_{i}^{(k)} \Pi_{\le z_{k+1}}^{(k)} \bigr]. \label{eq:first-level-centering-decomposition}
	\end{align}
	Here $\Pi_{\le z_{k+1}+1}^{(k)}$ can be inserted in the first term because $Q_{i}^{(k)}a_{i}^{(k)}P_{i}^{(k)}$ changes only site $i$: when acting on a vector with at most $z_{k+1}$ deviations, it can increase the deviation number by at most one.
	
	We bound the two sums in Eq.~\eqref{eq:first-level-centering-decomposition} separately.
	For normalized vectors $\ket{\varphi_{k+1}}\in \operatorname{Ran}\Pi_{\le z_{k+1}}^{(k)}$ and $\ket{\psi_{k+1}}\in \operatorname{Ran}\Pi_{\le z_{k+1}+1}^{(k)}$,
	\begin{align}
		\left| \sum_{i} \bra{\psi_{k+1}} Q_{i}^{(k)}a_{i}^{(k)}P_{i}^{(k)} \ket{\varphi_{k+1}} \right| &\le \sum_{i} \norm{Q_{i}^{(k)}\ket{\psi_{k+1}}}\, \norm{a_{i}^{(k)}}\, \norm{P_{i}^{(k)}\ket{\varphi_{k+1}}}\notag\\
		&\le \sum_{i} \norm{Q_{i}^{(k)}\ket{\psi_{k+1}}} \norm{a_{i}^{(k)}}\notag\\
		&\le \left( \sum_{i}\norm{Q_{i}^{(k)}\ket{\psi_{k+1}}}^{2} \right)^{1/2} \left( \sum_{i}\norm{a_{i}^{(k)}}^{2} \right)^{1/2}.
	\end{align}
	Since $\ket{\psi_{k+1}}$ lies in the sector with at most $z_{k+1}+1$ deviations,
	\begin{align}
		\sum_{i}\norm{Q_{i}^{(k)}\ket{\psi_{k+1}}}^{2} = \bra{\psi_{k+1}}\sum_{i}Q_{i}^{(k)}\ket{\psi_{k+1}} \le z_{k+1}+1.
	\end{align}
	Hence
	\begin{align}
		\norm{ \sum_{i} \Pi_{\le z_{k+1}+1}^{(k)} Q_{i}^{(k)}a_{i}^{(k)}P_{i}^{(k)} \Pi_{\le z_{k+1}}^{(k)} } \le \sqrt{z_{k+1}+1} \left( \sum_{i}\norm{a_{i}^{(k)}}^{2} \right)^{1/2}. \label{eq:first-level-offdiagonal-bound}
	\end{align}
	
	For the second sum, let $\ket{\chi_{k+1}}\in \operatorname{Ran}\Pi_{\le z_{k+1}}^{(k)}$ be normalized.
	Since
	\begin{align}
		\norm{a'_{i}} \le \norm{a_{i}^{(k)}} + |\bra{0_{i}^{(k)}}a_{i}^{(k)}\ket{0_{i}^{(k)}}| \le2\norm{a_{i}^{(k)}},
	\end{align}
	the triangle inequality and Cauchy--Schwarz inequality give
	\begin{align}
		\norm{ \sum_{i} a'_{i}Q_{i}^{(k)}\ket{\chi_{k+1}} } &\le \sum_{i}\norm{a'_{i}} \norm{Q_{i}^{(k)}\ket{\chi_{k+1}}}\notag\\
		&\le 2\sum_{i}\norm{a_{i}^{(k)}} \norm{Q_{i}^{(k)}\ket{\chi_{k+1}}}\notag\\
		&\le 2\left( \sum_{i}\norm{a_{i}^{(k)}}^{2} \right)^{1/2} \left( \sum_{i} \norm{Q_{i}^{(k)}\ket{\chi_{k+1}}}^{2} \right)^{1/2}\notag\\
		&\le 2\sqrt{z_{k+1}} \left( \sum_{i}\norm{a_{i}^{(k)}}^{2} \right)^{1/2}. \label{eq:first-level-diagonal-bound}
	\end{align}
	Combining Eqs.~\eqref{eq:first-level-offdiagonal-bound} and \eqref{eq:first-level-diagonal-bound} yields
	\begin{align}
		\norm{ (A^{(k)}-c\mathds{1}) \Pi_{\le z_{k+1}}^{(k)} } \le \bigl( 2\sqrt{z_{k+1}}+\sqrt{z_{k+1}+1} \bigr) \left( \sum_{i}\norm{a_{i}^{(k)}}^{2} \right)^{1/2} = \eta_{k+1:k+1} \left( \sum_{i}\norm{a_{i}^{(k)}}^{2} \right)^{1/2}.
	\end{align}
	This proves the first nontrivial level.
	
	For the induction step, suppose the assertion holds through level $t$, where $k<t<s$, and fix a level-$(t+1)$ site $u\in\Lambda^{(t+1)}$.
	Split its scale-$k$ descendant block $D_{k\leftarrow t+1}(u)$ into its level-$t$ children $w\in D_{t\leftarrow t+1}(u)$.
	Write
	\begin{align}
		A^{(k)}=\sum_{w} A_{w}, \qquad A_{w}:= \sum_{i\in I\cap D_{k\leftarrow t}(w)} a_{i}^{(k)}, \qquad \zeta_{w}^{2}:= \sum_{i\in I\cap D_{k\leftarrow t}(w)} \norm{a_{i}^{(k)}}^{2}.
	\end{align}
	The induction hypothesis gives a real $c_{w}$ such that
	\begin{align}
		\norm{ (A_{w}-c_{w}\mathds{1})\mathsf{P}_{k\rightarrow t}^{w} } \le \eta_{k+1:t}\zeta_{w}. \label{eq:hierarchical-centering-IH}
	\end{align}
	
	Let $\ket{e_{w}}:=(V_{k\rightarrow t}^{w})^{\dagger} \ket{0_{w}^{(t)}}$ and $\pi_{w}:=\ket{e_{w}}\bra{e_{w}}$ on the full scale-$k$ child space.
	Since $\ket{e_{w}}\in\operatorname{Ran}\mathsf{P}_{k\rightarrow t}^{w}$, we have $\pi_{w}\le \mathsf{P}_{k\rightarrow t}^{w}$.
	Choosing $C_{w}:=\bra{e_{w}}A_{w}\ket{e_{w}}$, we first note that
	\begin{align}
		\pi_{w}(A_{w}-C_{w}\mathds{1})\pi_{w}=0.
	\end{align}
	Moreover,
	\begin{align}
		|C_{w}-c_{w}| &= \left| \bra{e_{w}}(A_{w}-c_{w}\mathds{1})\ket{e_{w}} \right| \le \norm{ (A_{w}-c_{w}\mathds{1})\mathsf{P}_{k\rightarrow t}^{w} } \le \eta_{k+1:t}\zeta_{w}.
	\end{align}
	Therefore
	\begin{align}
		\norm{ (A_{w}-C_{w}\mathds{1})\mathsf{P}_{k\rightarrow t}^{w} } &\le \norm{ (A_{w}-c_{w}\mathds{1})\mathsf{P}_{k\rightarrow t}^{w} } +|C_{w}-c_{w}| \norm{\mathsf{P}_{k\rightarrow t}^{w}} \le 2\eta_{k+1:t}\zeta_{w}. \label{eq:child-centering-bound}
	\end{align}
	
	A scale-$k$ state in $\operatorname{Ran}\mathsf{P}_{k\rightarrow t+1}^{u}$ satisfies two distinct conditions.
	First, each level-$t$ child $w\in D_{t\leftarrow t+1}(u)$ must already lie in its cumulative retained range $\mathsf{P}_{k\rightarrow t}^{w}$.
	Second, among these children, at most $z_{t+1}$ may differ from their level-$t$ reference states.
	To express the second condition on the scale-$k$ child spaces, let $\widehat{\Pi}_{\le z}$ be the spectral projection of $\sum_{w}(\mathds{1}-\pi_{w})$ onto eigenvalues at most $z$, where the sum runs over $w\in D_{t\leftarrow t+1}(u)$.
	Then
	\begin{align}
		\mathsf{P}_{k\rightarrow t+1}^{u}\le \bigotimes_{w\in D_{t\leftarrow t+1}(u)}\mathsf{P}_{k\rightarrow t}^{w}, \qquad \mathsf{P}_{k\rightarrow t+1}^{u}\le\widehat{\Pi}_{\le z_{t+1}}. \label{eq:top-cumulative-range-inclusions}
	\end{align}
	The first condition enforces retention within every child through level $t$, whereas the second counts the number of non-reference children at the next truncation step.
	
	Set $C:=\sum_{w} C_{w}$.
	Since $\pi_{w}(A_{w}-C_{w}\mathds{1})\pi_{w}=0$,
	\begin{align}
		(A_{w}-C_{w}\mathds{1})\pi_{w} = (\mathds{1}-\pi_{w}) (A_{w}-C_{w}\mathds{1})\pi_{w}.
	\end{align}
	Inserting $\mathds{1}=\pi_{w}+(\mathds{1}-\pi_{w})$ on the right of each $A_{w}-C_{w}\mathds{1}$ therefore gives
	\begin{align}
		(A^{(k)}-C\mathds{1})\mathsf{P}_{k\rightarrow t+1}^{u} = \sum_{w}\bigl[ (\mathds{1}-\pi_{w}) (A_{w}-C_{w}\mathds{1})\pi_{w}\mathsf{P}_{k\rightarrow t+1}^{u} +(A_{w}-C_{w}\mathds{1}) (\mathds{1}-\pi_{w})\mathsf{P}_{k\rightarrow t+1}^{u} \bigr].
	\end{align}
	A child-supported operator can change the number of non-reference children by at most one, since it acts on only one child.
	Because $\mathsf{P}_{k\rightarrow t+1}^{u}\le\widehat{\Pi}_{\le z_{t+1}}$, the first term has output support in $\operatorname{Ran}\widehat{\Pi}_{\le z_{t+1}+1}$.
	Hence
	\begin{align}
		(A^{(k)}-C\mathds{1})\mathsf{P}_{k\rightarrow t+1}^{u} =\sum_{w}\bigl[ \widehat{\Pi}_{\le z_{t+1}+1} (\mathds{1}-\pi_{w}) (A_{w}-C_{w}\mathds{1})\pi_{w}\mathsf{P}_{k\rightarrow t+1}^{u} +(A_{w}-C_{w}\mathds{1}) (\mathds{1}-\pi_{w})\mathsf{P}_{k\rightarrow t+1}^{u} \bigr]. \label{eq:hierarchical-centering-induction-decomposition}
	\end{align}
	
	We again bound the two sums separately.
	For the first sum, the scalar shift disappears from the off-diagonal block:
	\begin{align}
		(\mathds{1}-\pi_{w})(A_{w}-C_{w}\mathds{1})\pi_{w} = (\mathds{1}-\pi_{w})(A_{w}-c_{w}\mathds{1})\pi_{w},
	\end{align}
	because $(\mathds{1}-\pi_{w})\pi_{w}=0$.
	Since $\pi_{w}\le \mathsf{P}_{k\rightarrow t}^{w}$, Eq.~\eqref{eq:hierarchical-centering-IH} gives
	\begin{align}
		\norm{ (\mathds{1}-\pi_{w}) (A_{w}-C_{w}\mathds{1})\pi_{w} } = \norm{ (\mathds{1}-\pi_{w}) (A_{w}-c_{w}\mathds{1})\pi_{w} } \le \norm{ (A_{w}-c_{w}\mathds{1})\mathsf{P}_{k\rightarrow t}^{w} } \le \eta_{k+1:t}\zeta_{w}. \label{eq:hierarchical-first-block-bound}
	\end{align}
	Let $\ket{\varphi_{t+1}}\in \operatorname{Ran}\mathsf{P}_{k\rightarrow t+1}^{u}$ and $\ket{\psi_{t+1}}\in \operatorname{Ran}\widehat{\Pi}_{\le z_{t+1}+1}$ be normalized.
	Then
	\begin{align}
		\left| \sum_{w}\bra{\psi_{t+1}} (\mathds{1}-\pi_{w}) (A_{w}-C_{w}\mathds{1})\pi_{w} \ket{\varphi_{t+1}} \right| &\le \eta_{k+1:t} \sum_{w} \zeta_{w} \norm{(\mathds{1}-\pi_{w})\ket{\psi_{t+1}}} \notag\\
		&\le \eta_{k+1:t} \left(\sum_{w}\zeta_{w}^{2}\right)^{1/2} \left( \sum_{w} \norm{ (\mathds{1}-\pi_{w})\ket{\psi_{t+1}} }^{2} \right)^{1/2}.
	\end{align}
	Because the projectors $\mathds{1}-\pi_{w}$ act on distinct children and $\ket{\psi_{t+1}}$ belongs to the sector with at most $z_{t+1}+1$ non-reference children,
	\begin{align}
		\sum_{w} \norm{ (\mathds{1}-\pi_{w})\ket{\psi_{t+1}} }^{2} = \bra{\psi_{t+1}} \sum_{w}(\mathds{1}-\pi_{w}) \ket{\psi_{t+1}} \le z_{t+1}+1.
	\end{align}
	Therefore
	\begin{align}
		\norm{ \sum_{w} \widehat{\Pi}_{\le z_{t+1}+1} (\mathds{1}-\pi_{w}) (A_{w}-C_{w}\mathds{1}) \pi_{w}\mathsf{P}_{k\rightarrow t+1}^{u} } \le \eta_{k+1:t}\sqrt{z_{t+1}+1} \left(\sum_{w}\zeta_{w}^{2}\right)^{1/2}. \label{eq:hierarchical-first-sum-bound}
	\end{align}
	
	For the second sum, $\pi_{w}\le \mathsf{P}_{k\rightarrow t}^{w}$ implies $\mathsf{P}_{k\rightarrow t}^{w}\pi_{w} =\pi_{w}\mathsf{P}_{k\rightarrow t}^{w}=\pi_{w}$, and hence
	\begin{align}
		[\mathsf{P}_{k\rightarrow t}^{w},\mathds{1}-\pi_{w}]=0.
	\end{align}
	Moreover, the first inclusion in Eq.~\eqref{eq:top-cumulative-range-inclusions} implies that $\mathsf{P}_{k\rightarrow t}^{w}\mathsf{P}_{k\rightarrow t+1}^{u} =\mathsf{P}_{k\rightarrow t+1}^{u}$ on the child-$w$ factor.
	Thus, for a normalized $\ket{\chi_{t+1}}\in\operatorname{Ran}\mathsf{P}_{k\rightarrow t+1}^{u}$,
	\begin{align}
		\norm{ (A_{w}-C_{w}\mathds{1}) (\mathds{1}-\pi_{w})\ket{\chi_{t+1}} } &= \norm{ (A_{w}-C_{w}\mathds{1}) \mathsf{P}_{k\rightarrow t}^{w} (\mathds{1}-\pi_{w})\ket{\chi_{t+1}} }\notag\\
		&\le \norm{ (A_{w}-C_{w}\mathds{1})\mathsf{P}_{k\rightarrow t}^{w} } \norm{ (\mathds{1}-\pi_{w})\ket{\chi_{t+1}} }\notag\\
		&\le 2\eta_{k+1:t}\zeta_{w} \norm{ (\mathds{1}-\pi_{w})\ket{\chi_{t+1}} },
	\end{align}
	where Eq.~\eqref{eq:child-centering-bound} was used in the last step.
	Summing over $w$ and applying Cauchy--Schwarz inequality gives
	\begin{align}
		\norm{ \sum_{w} (A_{w}-C_{w}\mathds{1}) (\mathds{1}-\pi_{w})\ket{\chi_{t+1}} } &\le 2\eta_{k+1:t} \sum_{w} \zeta_{w} \norm{ (\mathds{1}-\pi_{w})\ket{\chi_{t+1}} } \notag\\
		&\le 2\eta_{k+1:t} \left(\sum_{w}\zeta_{w}^{2}\right)^{1/2} \left( \sum_{w} \norm{ (\mathds{1}-\pi_{w})\ket{\chi_{t+1}} }^{2} \right)^{1/2}.
	\end{align}
	Since $\ket{\chi_{t+1}}\in\operatorname{Ran}\mathsf{P}_{k\rightarrow t+1}^{u}$ and $\mathsf{P}_{k\rightarrow t+1}^{u}\le\widehat{\Pi}_{\le z_{t+1}}$,
	\begin{align}
		\sum_{w} \norm{ (\mathds{1}-\pi_{w})\ket{\chi_{t+1}} }^{2} = \bra{\chi_{t+1}} \sum_{w}(\mathds{1}-\pi_{w}) \ket{\chi_{t+1}} \le z_{t+1}.
	\end{align}
	Hence
	\begin{align}
		\norm{ \sum_{w} (A_{w}-C_{w}\mathds{1}) (\mathds{1}-\pi_{w})\mathsf{P}_{k\rightarrow t+1}^{u} } \le 2\eta_{k+1:t}\sqrt{z_{t+1}} \left(\sum_{w}\zeta_{w}^{2}\right)^{1/2}. \label{eq:hierarchical-second-sum-bound}
	\end{align}
	
	Combining Eqs.~\eqref{eq:hierarchical-first-sum-bound} and \eqref{eq:hierarchical-second-sum-bound}, we obtain
	\begin{align}
		\norm{ (A^{(k)}-C\mathds{1})\mathsf{P}_{k\rightarrow t+1}^{u} } \le \eta_{k+1:t} \left( 2\sqrt{z_{t+1}}+\sqrt{z_{t+1}+1} \right) \left(\sum_{w}\zeta_{w}^{2}\right)^{1/2} = \eta_{k+1:t+1} \left(\sum_{w}\zeta_{w}^{2}\right)^{1/2}.
	\end{align}
	Finally, the descendant blocks of the level-$t$ children are disjoint, so
	\begin{align}
		\sum_{w}\zeta_{w}^{2} = \sum_{i\in I}\norm{a_{i}^{(k)}}^{2}.
	\end{align}
	This proves the induction.
	Taking $t+1=s$ and $u=v$ proves Eq.~\eqref{eq:main_ineq_renormalized}.
\end{proof}

\subsection{The gap controls the compressed variance}
\label{sec:compressed-variance-gap}

\begin{prop}
	\label{trade_off_variance_renomarlized}
	Let $0\le k<s$, let $Y_{k}\subseteq\Lambda^{(k)}$, and let $A^{(k)}=\sum_{i\in Y_{k}}a_{i}^{(k)}$ be a sum of Hermitian scale-$k$ one-site operators.
	Write $A^{[k\to s]}:=V_{k\rightarrow s}A^{(k)}V_{k\rightarrow s}^{\dagger}$.
	Recall that $\eta_{k+1:s}:=\prod_{t=k+1}^{s}(2\sqrt{z_{t}}+\sqrt{z_{t}+1})$ and $\eta_{s}:=\eta_{1:s}$.
	Suppose that $H^{(s)}$ satisfies Eq.~\eqref{all_to_all_cond_s_th_Re} and has a unique normalized ground state $\ket{\Omega^{(s)}}$ with gap $\Delta^{(s)}$.
	With $\operatorname{Var}_{\Omega^{(s)}}(X):= \bra{\Omega^{(s)}}X^{2}\ket{\Omega^{(s)}}- \bigl(\bra{\Omega^{(s)}}X\ket{\Omega^{(s)}}\bigr)^{2}$, one has
	\begin{align}
		\Delta^{(s)}\operatorname{Var}_{\Omega^{(s)}} \left(A^{[k\to s]}\right) \le4\eta_{k+1:s}^{2}\bar{g}^{(s)} \sum_{i\in Y_{k}}\norm{a_{i}^{(k)}}^{2}. \label{eq:variance-gap-general-scale}
	\end{align}
	In particular, at $k=0$ the factor is $\eta_{s}^{2}$.
\end{prop}

\begin{proof}
	Group the summands of $A^{(k)}$ by their scale-$s$ ancestor $v$ and write
	\begin{align}
		A^{(k)}=\sum_{v} A_{v}, \qquad A_{v}:=\sum_{i\in Y_{k}\cap D_{k\leftarrow s}(v)}a_{i}^{(k)}.
	\end{align}
	Because the cumulative coisometry factors over the disjoint scale-$s$ descendant blocks, the compression of $A_{v}$ acts only on the effective site $v$.
	Define $B_{v}:= V_{k\rightarrow s}^{v} A_{v} (V_{k\rightarrow s}^{v})^{\dagger}$.
	Then
	\begin{align}
		A^{[k\to s]}=\sum_{v} B_{v}.
	\end{align}
	
	By Proposition~\ref{block_renormalized_variance}, for every $v$ there exists a real scalar $C_{v}$ such that
	\begin{align}
		\norm{(A_{v}-C_{v}\mathds{1})\mathsf{P}_{k\rightarrow s}^{v}} \le\eta_{k+1:s}\left( \sum_{i\in Y_{k}\cap D_{k\leftarrow s}(v)} \norm{a_{i}^{(k)}}^{2} \right)^{1/2}. \label{eq:local-parent-centering}
	\end{align}
	Since $V_{k\rightarrow s}^{v}$ is a coisometry,
	\begin{align}
		(V_{k\rightarrow s}^{v})^{\dagger} V_{k\rightarrow s}^{v} =\mathsf{P}_{k\rightarrow s}^{v}, \qquad V_{k\rightarrow s}^{v} (V_{k\rightarrow s}^{v})^{\dagger} =\mathds{1}.
	\end{align}
	In particular, $V_{k\rightarrow s}^{v} =V_{k\rightarrow s}^{v}\mathsf{P}_{k\rightarrow s}^{v}$, and therefore
	\begin{align}
		B_{v}-C_{v}\mathds{1} = V_{k\rightarrow s}^{v} (A_{v}-C_{v}\mathds{1}) (V_{k\rightarrow s}^{v})^{\dagger} = V_{k\rightarrow s}^{v} (A_{v}-C_{v}\mathds{1})\mathsf{P}_{k\rightarrow s}^{v} (V_{k\rightarrow s}^{v})^{\dagger}.
	\end{align}
	Hence
	\begin{align}
		\norm{B_{v}-C_{v}\mathds{1}} \le \norm{(A_{v}-C_{v}\mathds{1})\mathsf{P}_{k\rightarrow s}^{v}} \le\eta_{k+1:s}\left( \sum_{i\in Y_{k}\cap D_{k\leftarrow s}(v)} \norm{a_{i}^{(k)}}^{2} \right)^{1/2}. \label{eq:local-compressed-centering}
	\end{align}
	
	Subtracting a scalar does not change a variance, so
	\begin{align}
		\operatorname{Var}_{\Omega^{(s)}} \left(A^{[k\to s]}\right) = \operatorname{Var}_{\Omega^{(s)}} \left( \sum_{v}(B_{v}-C_{v}\mathds{1}) \right).
	\end{align}
	The operators $B_{v}-C_{v}\mathds{1}$ are Hermitian one-site operators on the stage-$s$ lattice.
	Therefore Lemma~\ref{lem:basic-variance-gap}, applied to $H^{(s)}$ with $\delta=\Delta^{(s)}$ and $g=\bar{g}^{(s)}$, gives
	\begin{align}
		\Delta^{(s)} \operatorname{Var}_{\Omega^{(s)}} \left(A^{[k\to s]}\right) &\le 4\bar{g}^{(s)} \sum_{v}\norm{B_{v}-C_{v}\mathds{1}}^{2} \le 4\eta_{k+1:s}^{2}\bar{g}^{(s)} \sum_{v} \sum_{i\in Y_{k}\cap D_{k\leftarrow s}(v)} \norm{a_{i}^{(k)}}^{2}.
	\end{align}
	The descendant sets $D_{k\leftarrow s}(v)$ are disjoint and partition the stage-$k$ sites according to their stage-$s$ ancestors.
	Hence
	\begin{align}
		\sum_{v} \sum_{i\in Y_{k}\cap D_{k\leftarrow s}(v)} \norm{a_{i}^{(k)}}^{2} = \sum_{i\in Y_{k}}\norm{a_{i}^{(k)}}^{2},
	\end{align}
	which proves Eq.~\eqref{eq:variance-gap-general-scale}.
\end{proof}

Proposition~\ref{trade_off_variance_renomarlized} controls the variance of the compressed observable $A^{[0\to s]}$, not that of the uncompressed microscopic operator $A$ itself.
The latter may contain an additional contribution from leakage outside the retained range.
Accordingly, all fluctuation estimates below are formulated in terms of compressed observables.

\section{Dimension-free robustness on an effective site}
\label{sec:robustness}

We now bound the energy cost of reflecting a single effective site about its leading reference vector.
The main difficulty is that the local Hilbert-space dimension at a renormalized site may grow rapidly with the RG scale, so expanding an effective interaction in a local operator basis would introduce an unwanted dependence on this dimension.
We avoid this by tracing the interaction back to its microscopic descendants and expanding only the fixed $d$-dimensional microscopic operator basis.
This yields a bound independent of the renormalized local dimension.
We first give a direct-compression proof, followed by a scale-resolved derivation that keeps track of the intermediate centered terms.

\subsection{The target and an indexwise square-sum inequality}

Fix $v\in\Lambda^{(s)}$ and $Y\subseteq v^{\mathrm{c}}$, and define the local reflection $u_{v}^{(s)} := P_{v}^{(s)}-Q_{v}^{(s)}$.
Define
\begin{align}
	\delta_{s,v}(Y):= \left| \bra{\Omega^{(s)}} (u_{v}^{(s)})^{\dagger} H_{v,Y}^{(s)}u_{v}^{(s)} \ket{\Omega^{(s)}} -\bra{\Omega^{(s)}}H_{v,Y}^{(s)}\ket{\Omega^{(s)}} \right|. \label{Quantity_robustness}
\end{align}
Since $u_{v}^{(s)}=(u_{v}^{(s)})^{\dagger}$ and
\begin{align}
	u_{v}^{(s)}Xu_{v}^{(s)}-X = -2\left( P_{v}^{(s)}XQ_{v}^{(s)} + Q_{v}^{(s)}XP_{v}^{(s)} \right)
\end{align}
for every Hermitian $X$, we obtain
\begin{align}
	\delta_{s,v}(Y) &= 4\left| \operatorname{Re} \bra{\Omega^{(s)}} P_{v}^{(s)}H_{v,Y}^{(s)}Q_{v}^{(s)} \ket{\Omega^{(s)}} \right| \le 4\left| \bra{\Omega^{(s)}} P_{v}^{(s)}H_{v,Y}^{(s)}Q_{v}^{(s)} \ket{\Omega^{(s)}} \right|. \label{Quantity_robustness_rewrite}
\end{align}
For $Y=v^{\mathrm{c}}$, this quantity is the total nonnegative energy excess generated by the reflection $u_{v}^{(s)}$.
Indeed, terms supported entirely on $v^{\mathrm{c}}$ are unchanged, while $[u_{v}^{(s)},\rho_{v}^{(s)}]=0$ implies
\begin{align}
	\bra{\Omega^{(s)}} (u_{v}^{(s)})^{\dagger} h_{v}^{(s)}u_{v}^{(s)} \ket{\Omega^{(s)}} = \bra{\Omega^{(s)}}h_{v}^{(s)}\ket{\Omega^{(s)}}.
\end{align}
Hence
\begin{align}
	\delta_{s,v}(v^{\mathrm{c}}) = \bra{\Omega^{(s)}} (u_{v}^{(s)})^{\dagger} H^{(s)}u_{v}^{(s)} \ket{\Omega^{(s)}} - \bra{\Omega^{(s)}}H^{(s)}\ket{\Omega^{(s)}} \ge0,
\end{align}
where the last inequality follows from the variational principle.

\begin{lemma}
	\label{sum_reduction_lemma}
	Let $\ket{\omega}=\sum_{j}\lambda_{j}\ket{j}\ket{\phi_{j}}$ be a normalized Schmidt decomposition, $P=\ket{0}\bra{0}$, and $Q=\mathds{1}-P$ on the first factor.
	Suppose $A_{p}$ act on the first factor and $B_{p}=B_{p}^{\dagger}$ act on the second factor.
	Suppose further that, for every complex coefficient family $(c_{p})_{p}$,
	\begin{align}
		\left\|\sum_{p} c_{p} A_{p}\right\| \le \kappa\left(\sum_{p}|c_{p}|^{2}\right)^{1/2}. \label{sum_reduction_lemma_assmp}
	\end{align}
	Then
	\begin{align}
		\left| \sum_{p} \bra{\omega}(PA_{p}Q)\otimes B_{p}\ket{\omega} \right| \le \kappa\norm{Q\ket{\omega}} \left( \sum_{p}\norm{B_{p}\ket{\omega}}^{2} \right)^{1/2}. \label{sum_reduction_lemma_main}
	\end{align}
\end{lemma}

\begin{proof}
	If $\lambda_{0}=0$, the left-hand side of Eq.~\eqref{sum_reduction_lemma_main} vanishes.
	Assume henceforth that $\lambda_{0}>0$.
	Expanding the Schmidt decomposition gives
	\begin{align}
		\sum_{p} \bra{\omega}(PA_{p}Q)\otimes B_{p}\ket{\omega} &= \lambda_{0}\sum_{j>0}\lambda_{j} \sum_{p} \bra{0} A_{p}\ket{j}\, \bra{\phi_{0}}B_{p}\ket{\phi_{j}}\notag\\
		&= \lambda_{0}\sum_{j>0}\lambda_{j} \bra{0} \left( \sum_{p} \bra{\phi_{0}}B_{p}\ket{\phi_{j}}A_{p} \right) \ket{j}.
	\end{align}
	For each fixed $j$, apply Eq.~\eqref{sum_reduction_lemma_assmp} with $c_{p}=\bra{\phi_{0}}B_{p}\ket{\phi_{j}}$.
	Then
	\begin{align}
		\left\| \sum_{p} \bra{\phi_{0}}B_{p}\ket{\phi_{j}}A_{p} \right\| \le \kappa \left( \sum_{p} \left| \bra{\phi_{0}}B_{p}\ket{\phi_{j}} \right|^{2} \right)^{1/2}.
	\end{align}
	Using this bound for each fixed $j$, we obtain
	\begin{align}
		\left| \sum_{p} \bra{\omega}(PA_{p}Q)\otimes B_{p}\ket{\omega} \right| &\le \kappa\lambda_{0} \sum_{j>0}\lambda_{j} \left( \sum_{p} \left| \bra{\phi_{0}}B_{p}\ket{\phi_{j}} \right|^{2} \right)^{1/2}.
	\end{align}
	Applying Cauchy--Schwarz inequality to the sum over $j$ then gives
	\begin{align}
		\left| \sum_{p} \bra{\omega}(PA_{p}Q)\otimes B_{p}\ket{\omega} \right| &\le \kappa\lambda_{0} \left(\sum_{j>0}\lambda_{j}^{2}\right)^{1/2} \left[ \sum_{j>0} \left( \sum_{p} \left| \bra{\phi_{0}}B_{p}\ket{\phi_{j}} \right|^{2} \right) \right]^{1/2}\notag\\
		&= \kappa\lambda_{0} \left(\sum_{j>0}\lambda_{j}^{2}\right)^{1/2} \left( \sum_{p,j>0} \left| \bra{\phi_{0}}B_{p}\ket{\phi_{j}} \right|^{2} \right)^{1/2}\notag\\
		&= \kappa\lambda_{0} \norm{Q\ket{\omega}} \left( \sum_{p,j>0} \left| \bra{\phi_{0}}B_{p}\ket{\phi_{j}} \right|^{2} \right)^{1/2}. \label{eq:indexwise-reduction-CS}
	\end{align}
	Since each $B_{p}$ is Hermitian, Bessel's inequality gives
	\begin{align}
		\sum_{j>0} \left| \bra{\phi_{0}}B_{p}\ket{\phi_{j}} \right|^{2} \le \norm{B_{p}\ket{\phi_{0}}}^{2}. \label{eq:indexwise-reduction-Bessel}
	\end{align}
	Moreover,
	\begin{align}
		\norm{B_{p}\ket{\omega}}^{2} = \sum_{j}\lambda_{j}^{2}\norm{B_{p}\ket{\phi_{j}}}^{2} \ge \lambda_{0}^{2}\norm{B_{p}\ket{\phi_{0}}}^{2}. \label{eq:indexwise-reduction-Schmidt-norm}
	\end{align}
	Combining Eqs.~\eqref{eq:indexwise-reduction-Bessel} and \eqref{eq:indexwise-reduction-Schmidt-norm}, we obtain
	\begin{align}
		\sum_{j>0} \left| \bra{\phi_{0}}B_{p}\ket{\phi_{j}} \right|^{2} \le \frac{\norm{B_{p}\ket{\omega}}^{2}}{\lambda_{0}^{2}}. \label{eq:indexwise-reduction-Bessel-final}
	\end{align}
	Substituting Eq.~\eqref{eq:indexwise-reduction-Bessel-final} into Eq.~\eqref{eq:indexwise-reduction-CS} proves Eq.~\eqref{sum_reduction_lemma_main}.
\end{proof}

\subsection{A shorter direct-compression proof}
\label{sec:direct-compression-robustness}

There is a useful shortcut before resolving the interaction into all its intermediate centered pieces.
For fixed final sites, the difference between their effective interaction and the compressed microscopic interaction consists only of operators supported on one side.

\begin{lemma}
	\label{lem:interaction-modulo-local}
	Let $B=D_{0\leftarrow s}(v)$ and $Y_{0}=Y^{(0\leftarrow s)} = \bigcup_{w \in Y} D_{0 \leftarrow s}(w)$.
	There are an operator $L_{v}$ supported on $v$ and an operator $L_{Y}$ supported on $Y$ such that
	\begin{align}
		H_{v,Y}^{(s)} = (H_{B,Y_{0}})^{[0\to s]} +L_{v}\otimes\mathds{1} +\mathds{1}_{v}\otimes L_{Y}. \label{eq:interaction-modulo-local}
	\end{align}
	Scalar terms may be included in either local summand.
\end{lemma}

\begin{proof}
	The key observation is already visible in a single RG step.
	For two disjoint parent blocks $L,L'$, Eq.~\eqref{Decompose_RE_vs_original} can be rewritten as
	\begin{align}
		\acute{H}_{L,L'} = H_{L,L'} -\bra{\mathbf{0}}_{L'}H_{L,L'}\ket{\mathbf{0}}_{L'}\otimes\mathds{1}_{L'} -\mathds{1}_{L}\otimes \bra{\mathbf{0}}_{L}H_{L,L'}\ket{\mathbf{0}}_{L} +\bra{\mathbf{0}}_{L\cup L'}H_{L,L'} \ket{\mathbf{0}}_{L\cup L'}\mathds{1}_{L\cup L'}. \label{eq:one-step-interaction-modulo-local}
	\end{align}
	The second term on the right is supported only on $L$, the third only on $L'$, and the last term is a scalar.
	Thus centering changes a two-block interaction only by one-sided terms.
	
	Let $V_{L}$ and $V_{L'}$ be the coisometries associated with these two blocks.
	Since
	\begin{align}
		V_{L}V_{L}^{\dagger}=\mathds{1}, \qquad V_{L'}V_{L'}^{\dagger}=\mathds{1},
	\end{align}
	tensor-product compression preserves the one-sided form.
	Namely, an operator $K_{L}\otimes\mathds{1}$ supported only on $L$ satisfies
	\begin{align}
		(V_{L}\otimes V_{L'}) (K_{L}\otimes\mathds{1}) (V_{L}\otimes V_{L'})^{\dagger} = (V_{L}K_{L}V_{L}^{\dagger})\otimes\mathds{1},
	\end{align}
	and similarly for an operator supported only on $L'$.
	Hence, after one RG step, the compressed centered interaction differs from the compressed uncentered interaction only by operators supported on the two resulting effective sites.
	
	Now fix two final sites $v,w\in\Lambda^{(s)}$ and write
	\begin{align}
		B_{v}:=D_{0\leftarrow s}(v), \qquad B_{w}:=D_{0\leftarrow s}(w).
	\end{align}
	We iterate the preceding one-step observation through the hierarchy.
	At each step, the two-sided part is compressed to the next scale, while centering contributes only terms supported on one of the two ancestors.
	Compression preserves this one-sided support.
	Consequently, after all $s$ steps there exist operators $L_{v,w}^{(v)}$ supported on $v$ and $L_{v,w}^{(w)}$ supported on $w$ such that
	\begin{align}
		H_{v,w}^{(s)} = (H_{B_{v},B_{w}})^{[0\to s]} + L_{v,w}^{(v)}\otimes\mathds{1} + \mathds{1}_{v}\otimes L_{v,w}^{(w)}. \label{eq:pair-interaction-modulo-local}
	\end{align}
	Indeed, the assertion is true after the first step by Eq.~\eqref{eq:one-step-interaction-modulo-local}, and the preceding compression argument shows that the same form is preserved at every subsequent step.
	
	Finally, sum Eq.~\eqref{eq:pair-interaction-modulo-local} over $w\in Y$.
	Since the microscopic descendant blocks $D_{0\leftarrow s}(w)$ are disjoint,
	\begin{align}
		\sum_{w\in Y}H_{B_{v},B_{w}} = H_{B_{v},Y^{(0\leftarrow s)}}.
	\end{align}
	Using $B=B_{v}=D_{0\leftarrow s}(v)$ and $Y_{0}=Y^{(0\leftarrow s)}$, and collecting all terms supported on $v$ into $L_{v}$ and all terms supported on $Y$ into $L_{Y}$, we obtain
	\begin{align}
		H_{v,Y}^{(s)} = (H_{B,Y_{0}})^{[0\to s]} + L_{v}\otimes\mathds{1} + \mathds{1}_{v}\otimes L_{Y},
	\end{align}
	which proves Eq.~\eqref{eq:interaction-modulo-local}.
\end{proof}

The one-sided correction terms in Eq.~\eqref{eq:interaction-modulo-local} do not contribute to the off-diagonal matrix element.
Indeed,
\begin{align}
	\bra{\Omega^{(s)}}P_{v}^{(s)}L_{v}Q_{v}^{(s)}\ket{\Omega^{(s)}}=0, \qquad P_{v}^{(s)}(\mathds{1}_{v}\otimes L_{Y})Q_{v}^{(s)}=0. \label{eq:one-sided-corrections-vanish}
\end{align}
because $[P_{v}^{(s)},\rho_{v}^{(s)}]=0$ in the first identity, while $P_{v}^{(s)}Q_{v}^{(s)}=0$ gives the second.

\begin{prop}
	\label{prop:direct-robustness}
	Let $v\in\Lambda^{(s)}$ and $Y\subseteq v^{\mathrm{c}}$.
	Let $B:=D_{0\leftarrow s}(v)$ and $Y_{0}:=Y^{(0\leftarrow s)} = \bigcup_{w \in Y} D_{0 \leftarrow s}(w)$, and define the microscopic crossing strength
	\begin{align}
		\mathfrak{C}_{0}(v,Y) := \left( \sum_{i\in B,j\in Y_{0}} [J_{ij}^{(0)}]^{2} \right)^{1/2}.
	\end{align}
	Here $J_{ij}^{(0)}$ is the stage-$0$ specialization of Eq.~\eqref{RG_interaction_strength_upper_bound}.
	Suppose that $H^{(s)}$ satisfies Eq.~\eqref{eq:scale-s-Hamiltonian-decomposition} and has a unique normalized ground state $\ket{\Omega^{(s)}}$ with gap $\Delta^{(s)}$.
	Then
	\begin{align}
		\delta_{s,v}(Y) \le 8\sqrt{2}\,d^{2}\eta_{s}^{2} \sqrt{\frac{\bar{g}^{(s)}}{\Delta^{(s)}}} \,\mathfrak{C}_{0}(v,Y)\sqrt{q_{s,v}}. \label{eq:direct-reflection-bound}
	\end{align}
	Here $\eta_{s}$ is defined in Eq.~\eqref{defs_zeta_s_eta_s}, $q_{s,v}$ in Eq.~\eqref{eq:scale-s-reference-state}, and $\bar{g}^{(s)}$ in Eq.~\eqref{all_to_all_cond_s_th_Re}.
\end{prop}

\begin{proof}
	By Eq.~\eqref{Quantity_robustness_rewrite}, it suffices to bound the off-diagonal matrix element $\left| \bra{\Omega^{(s)}} P_{v}^{(s)}H_{v,Y}^{(s)}Q_{v}^{(s)} \ket{\Omega^{(s)}} \right|$.
	We first reduce this effective interaction to the compression of the corresponding microscopic interaction, and then apply Lemma~\ref{sum_reduction_lemma} using the hierarchical centering and compressed-variance estimates.
	
	By Lemma~\ref{lem:interaction-modulo-local} and Eq.~\eqref{eq:one-sided-corrections-vanish},
	\begin{align}
		\bra{\Omega^{(s)}} P_{v}^{(s)}H_{v,Y}^{(s)}Q_{v}^{(s)} \ket{\Omega^{(s)}} = \bra{\Omega^{(s)}} P_{v}^{(s)}(H_{B,Y_{0}})^{[0\to s]}Q_{v}^{(s)} \ket{\Omega^{(s)}}. \label{eq:direct-robustness-microscopic-reduction}
	\end{align}
	Using Eqs.~\eqref{eq:basis-reconstruction-constant} and \eqref{eq:microscopic-d2-expansion}, and summing over $i\in B$, we write
	\begin{align}
		H_{B,Y_{0}} &= \sum_{i\in B}\sum_{a=1}^{d^{2}} \mathcal{X}_{i,a}\otimes B_{i,a}, \qquad B_{i,a}:=\sum_{j\in Y_{0}}b_{i,j,a}, \label{eq:direct-robustness-basis-expansion}
	\end{align}
	where each $b_{i,j,a}$ is Hermitian, supported on $j$, and satisfies
	\begin{align}
		\norm{b_{i,j,a}} \le \norm{h_{ij}^{(0)}} \le J_{ij}^{(0)}. \label{eq:direct-robustness-microscopic-coefficient}
	\end{align}
	For each $i\in B$ and $1\le a\le d^{2}$, set $A_{i,a} := P_{v}^{(s)} (\mathcal{X}_{i,a})^{[0\to s]} Q_{v}^{(s)}$.
	Because the cumulative compression factors over the disjoint scale-$s$ descendant blocks, Eq.~\eqref{eq:direct-robustness-microscopic-reduction} becomes
	\begin{align}
		\bra{\Omega^{(s)}} P_{v}^{(s)}H_{v,Y}^{(s)}Q_{v}^{(s)} \ket{\Omega^{(s)}} = \sum_{a=1}^{d^{2}}\sum_{i\in B} \bra{\Omega^{(s)}} A_{i,a}\otimes B_{i,a}^{[0\to s]} \ket{\Omega^{(s)}}. \label{eq:direct-robustness-expanded}
	\end{align}
	
	Fix $a$.
	To apply Lemma~\ref{sum_reduction_lemma} with the index $p=i$, we first verify its coefficient bound for the family $\{A_{i,a}\}_{i\in B}$.
	For arbitrary real coefficients $(\xi_{i})$ and any $c\in\mathbb{R}$, the definition of $A_{i,a}$ and $P_{v}^{(s)}Q_{v}^{(s)}=0$ give
	\begin{align}
		\sum_{i\in B}\xi_{i}A_{i,a} &= P_{v}^{(s)} \left( \sum_{i\in B}\xi_{i}(\mathcal{X}_{i,a})^{[0\to s]} -c\mathds{1} \right) Q_{v}^{(s)}.
	\end{align}
	Since all $i\in B=D_{0\leftarrow s}(v)$ belong to the same scale-$s$ descendant block, the local cumulative coisometry gives
	\begin{align}
		\sum_{i\in B}\xi_{i}(\mathcal{X}_{i,a})^{[0\to s]} -c\mathds{1} &= V_{0\rightarrow s}^{v} \left( \sum_{i\in B}\xi_{i}\mathcal{X}_{i,a} -c\mathds{1} \right) (V_{0\rightarrow s}^{v})^{\dagger} = V_{0\rightarrow s}^{v} \left( \sum_{i\in B}\xi_{i}\mathcal{X}_{i,a} -c\mathds{1} \right) \mathsf{P}_{0\rightarrow s}^{v} (V_{0\rightarrow s}^{v})^{\dagger},
	\end{align}
	where $(V_{0\rightarrow s}^{v})^{\dagger} V_{0\rightarrow s}^{v} =\mathsf{P}_{0\rightarrow s}^{v}$ was used in the second equality.
	Therefore
	\begin{align}
		\left\| \sum_{i\in B}\xi_{i}A_{i,a} \right\| &\le \inf_{c\in\mathbb{R}} \left\| \left( \sum_{i\in B}\xi_{i}\mathcal{X}_{i,a} -c\mathds{1} \right) \mathsf{P}_{0\rightarrow s}^{v} \right\| \le \eta_{s} \left( \sum_{i\in B}\xi_{i}^{2}\norm{\mathcal{X}_{i,a}}^{2} \right)^{1/2} \le \eta_{s} \left( \sum_{i\in B}\xi_{i}^{2} \right)^{1/2}, \label{eq:direct-robustness-real-coefficient-bound}
	\end{align}
	where Proposition~\ref{block_renormalized_variance} was used in the second inequality and $\norm{\mathcal{X}_{i,a}}\le1$ in the last inequality.
	For arbitrary complex coefficients $(\xi_{i})$, applying Eq.~\eqref{eq:direct-robustness-real-coefficient-bound} separately to the real and imaginary parts gives
	\begin{align}
		\left\| \sum_{i\in B}\xi_{i}A_{i,a} \right\| &\le \eta_{s} \left[ \left( \sum_{i\in B}(\operatorname{Re}\xi_{i})^{2} \right)^{1/2} + \left( \sum_{i\in B}(\operatorname{Im}\xi_{i})^{2} \right)^{1/2} \right] \le \sqrt{2}\,\eta_{s} \left( \sum_{i\in B}|\xi_{i}|^{2} \right)^{1/2}. \label{eq:direct-robustness-complex-coefficient-bound}
	\end{align}
	
	The other input to Lemma~\ref{sum_reduction_lemma} is the square-summed fluctuation of the corresponding operators $B_{i,a}^{[0\to s]}$.
	For each $i\in B$, Proposition~\ref{trade_off_variance_renomarlized} applied to the microscopic one-site sum $B_{i,a}=\sum_{j\in Y_{0}}b_{i,j,a}$ gives
	\begin{align}
		\inf_{c_{i,a}\in\mathbb{R}} \norm{ (B_{i,a}-c_{i,a}\mathds{1})^{[0\to s]} \ket{\Omega^{(s)}} }^{2} \le 4\eta_{s}^{2} \frac{\bar{g}^{(s)}}{\Delta^{(s)}} \sum_{j\in Y_{0}} \norm{b_{i,j,a}}^{2}. \label{eq:direct-robustness-single-complement-variance}
	\end{align}
	Summing over $i\in B$, using Eq.~\eqref{eq:direct-robustness-microscopic-coefficient}, and taking a square root yields
	\begin{align}
		\left( \sum_{i\in B} \inf_{c_{i,a}\in\mathbb{R}} \norm{ (B_{i,a}-c_{i,a}\mathds{1})^{[0\to s]} \ket{\Omega^{(s)}} }^{2} \right)^{1/2} &\le 2\eta_{s} \sqrt{\frac{\bar{g}^{(s)}}{\Delta^{(s)}}} \left( \sum_{i\in B,j\in Y_{0}} \norm{b_{i,j,a}}^{2} \right)^{1/2} \notag\\
		&\le 2\eta_{s} \sqrt{\frac{\bar{g}^{(s)}}{\Delta^{(s)}}} \,\mathfrak{C}_{0}(v,Y). \label{eq:direct-robustness-complement-fluctuation}
	\end{align}
	
	For each $i\in B$, choose a real scalar $c_{i,a}$ attaining the corresponding infimum in Eq.~\eqref{eq:direct-robustness-single-complement-variance}.
	Subtracting these scalars from $B_{i,a}^{[0\to s]}$ does not change Eq.~\eqref{eq:direct-robustness-expanded}, because
	\begin{align}
		\bra{\Omega^{(s)}}A_{i,a}\ket{\Omega^{(s)}} &= \operatorname{Tr}\left( \rho_{v}^{(s)} P_{v}^{(s)} (\mathcal{X}_{i,a})^{[0\to s]} Q_{v}^{(s)} \right)\notag\\
		&= \operatorname{Tr}\left( Q_{v}^{(s)} \rho_{v}^{(s)} P_{v}^{(s)} (\mathcal{X}_{i,a})^{[0\to s]} \right)\notag\\
		&= \operatorname{Tr}\left( Q_{v}^{(s)} P_{v}^{(s)} \rho_{v}^{(s)} (\mathcal{X}_{i,a})^{[0\to s]} \right) =0,
	\end{align}
	Hence Lemma~\ref{sum_reduction_lemma}, applied for this fixed $a$ with $p=i$ and $\kappa=\sqrt{2}\,\eta_{s}$, together with Eq.~\eqref{eq:direct-robustness-complement-fluctuation} and $\norm{Q_{v}^{(s)}\ket{\Omega^{(s)}}}=\sqrt{q_{s,v}}$, gives
	\begin{align}
		&\left| \sum_{i\in B} \bra{\Omega^{(s)}} A_{i,a}\otimes (B_{i,a}-c_{i,a}\mathds{1})^{[0\to s]} \ket{\Omega^{(s)}} \right| \le 2\sqrt{2}\,\eta_{s}^{2} \sqrt{\frac{\bar{g}^{(s)}}{\Delta^{(s)}}} \,\mathfrak{C}_{0}(v,Y)\sqrt{q_{s,v}}. \label{eq:direct-robustness-fixed-basis-index}
	\end{align}
	Summing this estimate over $a=1,\ldots,d^{2}$ in Eq.~\eqref{eq:direct-robustness-expanded} yields
	\begin{align}
		\left| \bra{\Omega^{(s)}} P_{v}^{(s)}H_{v,Y}^{(s)}Q_{v}^{(s)} \ket{\Omega^{(s)}} \right| \le 2\sqrt{2}\,d^{2}\eta_{s}^{2} \sqrt{\frac{\bar{g}^{(s)}}{\Delta^{(s)}}} \,\mathfrak{C}_{0}(v,Y)\sqrt{q_{s,v}}.
	\end{align}
	Finally, Eq.~\eqref{Quantity_robustness_rewrite} gives Eq.~\eqref{eq:direct-reflection-bound}.
\end{proof}

For $D/2<\alpha<(D+1)/2$, we use the stage-$0$ parametrization introduced above, namely $J^{(0)}=J_{\rm K}:=\mu_{0}J$ and $\mu_{ij}^{(0)}=1$.
Hence Eq.~\eqref{RG_interaction_strength_upper_bound} gives
\begin{align}
	J_{ij}^{(0)} = \frac{J_{\rm K}}{n^{\beta}}r_{ij}^{-\alpha}.
\end{align}
Moreover, $B=D_{0\leftarrow s}(v)$ is a microscopic cube of side $\ell_{1:s}$, while $Y_{0}=Y^{(0\leftarrow s)}\subseteq B^{\mathrm{c}}$.
Therefore, using the definition of $\mathfrak{C}_{0}(v,Y)$ and Eq.~\eqref{eq:block-complement-kernel},
\begin{align}
	\mathfrak{C}_{0}(v,Y) &= \left( \sum_{i\in B,j\in Y_{0}} [J_{ij}^{(0)}]^{2} \right)^{1/2} \le \frac{J_{\rm K}}{n^{\beta}} \left( \sum_{i\in B,j\notin B} r_{ij}^{-2\alpha} \right)^{1/2} \le C_{D,\alpha}J_{\rm K} \frac{\ell_{1:s}^{D-\alpha}}{n^{\beta}}.
\end{align}
Here $\mathfrak{F}_{\alpha}(\ell_{1:s},n^{1/D}) =\ell_{1:s}^{D-\alpha}$ for $D/2<\alpha<(D+1)/2$.
Since $D\beta=D-\alpha$ and $\ell_{1:s}^{D}=n/n^{(s)}$ by Eq.~\eqref{eq:microscopic-block-volume},
\begin{align}
	\frac{\ell_{1:s}^{D-\alpha}}{n^{\beta}} = \frac{\ell_{1:s}^{D\beta}}{n^{\beta}} = [n^{(s)}]^{-\beta},
\end{align}
and hence
\begin{align}
	\mathfrak{C}_{0}(v,Y) \le C_{D,\alpha}J_{\rm K}[n^{(s)}]^{-\beta}. \label{eq:direct-crossing-Dhalf-alpha-Dplus1half-bound}
\end{align}

The factor $\eta_{s}$ can be bounded directly from Eq.~\eqref{defs_zeta_s_eta_s}.
For $z_{t}\ge1$,
\begin{align}
	2\sqrt{z_{t}}+\sqrt{z_{t}+1} \le (2+\sqrt{2})\sqrt{z_{t}},
\end{align}
and therefore
\begin{align}
	\eta_{s}^{2} \le (2+\sqrt{2})^{2s}\prod_{t=1}^{s}z_{t} = (6+4\sqrt{2})^{s} Z_{s}. \label{eq:direct-eta-bound}
\end{align}
Combining Eqs.~\eqref{eq:direct-reflection-bound}, \eqref{eq:direct-crossing-Dhalf-alpha-Dplus1half-bound}, and \eqref{eq:direct-eta-bound} gives the scale dependence required in the closed induction.

At $\alpha=(D+1)/2<D$, the same calculation instead uses
\begin{align}
	\mathfrak{F}_{\alpha}(\ell_{1:s},n^{1/D}) = \ell_{1:s}^{(D-1)/2} \sqrt{\log(e\ell_{1:s})},
\end{align}
so that the direct reflection bound acquires the additional factor $\sqrt{\log(e\ell_{1:s})}$.
The interaction and local-energy recurrences, however, still contain the cumulative factor $\chi_{s}$.

\subsection{The exact interaction identity to be iterated}

Fix $v\in\Lambda^{(s)}$ and $Y\subseteq v^{\mathrm{c}}$.
For $0\le k<s$, let $Y_{k}:=Y^{(k\leftarrow s)}$.
Thus $D_{k\leftarrow s}(v)\cap Y_{k}=\varnothing$, and for $i\in D_{k\leftarrow s}(v)$, $H_{i,Y_{k}}^{(k)}=\sum_{j\in Y_{k}}h_{ij}^{(k)}$ describes the interaction between the descendant $i$ of $v$ and the descendants of $Y$.
We rewrite the RE decomposition to identify which terms vanish in the final $P_{v}^{(s)}(\cdot)Q_{v}^{(s)}$ matrix element and which terms remain to be bounded.

Define
\begin{align}
	\mathcal{W}_{i}^{(k)}(X) &:=X-P_{i}^{(k)}XP_{i}^{(k)} =P_{i}^{(k)}XQ_{i}^{(k)}+Q_{i}^{(k)}X. \label{eq:W-map-definition}
\end{align}
The map $\mathcal{W}_{i}^{(k)}$ preserves Hermiticity, though it need not preserve positivity.
Also define
\begin{align}
	B_{i}^{(k)} :=\bra{0_{i}^{(k)}}H_{i,Y_{k}}^{(k)}\ket{0_{i}^{(k)}}, \qquad D_{i}^{(k)} :=\sum_{j\in Y_{k}} \bra{0_{j}^{(k)}}h_{ij}^{(k)}\ket{0_{j}^{(k)}}, \qquad c_{i}^{(k)} :=\bra{0_{i}^{(k)}}D_{i}^{(k)}\ket{0_{i}^{(k)}}. \label{eq:one-sided-centering-terms}
\end{align}
Here $B_{i}^{(k)}$ acts on $Y_{k}$, $D_{i}^{(k)}$ acts on $i$, and $c_{i}^{(k)}$ is a scalar.

By these definitions,
\begin{align}
	\mathcal{W}_{i}^{(k)}(H_{i,Y_{k}}^{(k)}) = H_{i,Y_{k}}^{(k)}-P_{i}^{(k)}\otimes B_{i}^{(k)}, \qquad \mathcal{W}_{i}^{(k)}(D_{i}^{(k)}) = D_{i}^{(k)}-c_{i}^{(k)}P_{i}^{(k)}.
\end{align}
Therefore,
\begin{align}
	&\mathcal{W}_{i}^{(k)}(H_{i,Y_{k}}^{(k)}) -\mathcal{W}_{i}^{(k)}(D_{i}^{(k)})\otimes\mathds{1}_{Y_{k}} -Q_{i}^{(k)}\otimes \left( B_{i}^{(k)}-c_{i}^{(k)}\mathds{1}_{Y_{k}} \right)\notag\\
	&\qquad= H_{i,Y_{k}}^{(k)} -(P_{i}^{(k)}+Q_{i}^{(k)})\otimes B_{i}^{(k)} -D_{i}^{(k)}\otimes\mathds{1}_{Y_{k}} +c_{i}^{(k)}(P_{i}^{(k)}+Q_{i}^{(k)})\otimes\mathds{1}_{Y_{k}}\notag\\
	&\qquad= H_{i,Y_{k}}^{(k)} -\mathds{1}_{i}\otimes B_{i}^{(k)} -D_{i}^{(k)}\otimes\mathds{1}_{Y_{k}} +c_{i}^{(k)}\mathds{1}_{i\cup Y_{k}},
\end{align}
where we used $P_{i}^{(k)}+Q_{i}^{(k)}=\mathds{1}_{i}$.
The last expression is exactly the previously defined RE decomposition, and hence
\begin{align}
	\acute{H}_{i,Y_{k}}^{(k)} = \mathcal{W}_{i}^{(k)}(H_{i,Y_{k}}^{(k)}) -\mathcal{W}_{i}^{(k)}(D_{i}^{(k)})\otimes\mathds{1}_{Y_{k}} -Q_{i}^{(k)}\otimes \left( B_{i}^{(k)}-c_{i}^{(k)}\mathds{1}_{Y_{k}} \right). \label{eq:RE-one-step-robustness}
\end{align}
For any $b_{i}\in\mathbb{R}$,
\begin{align}
	-Q_{i}^{(k)}\otimes \left( B_{i}^{(k)}-c_{i}^{(k)}\mathds{1}_{Y_{k}} \right) = -Q_{i}^{(k)}\otimes \left( B_{i}^{(k)}-b_{i}\mathds{1}_{Y_{k}} \right) +(c_{i}^{(k)}-b_{i}) Q_{i}^{(k)}\otimes\mathds{1}_{Y_{k}}. \label{eq:complement-fluctuation-splitting}
\end{align}
The terms $\mathcal{W}_{i}^{(k)}(D_{i}^{(k)})$ and $(c_{i}^{(k)}-b_{i})Q_{i}^{(k)}$ remain supported only on the descendant block of $v$ and therefore give zero in the final $P_{v}^{(s)}(\cdot)Q_{v}^{(s)}$ matrix element.
The remaining terms are the continuing interaction $\mathcal{W}_{i}^{(k)}(H_{i,Y_{k}}^{(k)})$ and the complement fluctuation $-Q_{i}^{(k)}\otimes\left(B_{i}^{(k)}-b_{i}\mathds{1}_{Y_{k}}\right)$.

\subsection{Transport maps with their spaces made explicit}

The complement fluctuation generated at scale $k$ must still be transported through the remaining RG steps before it contributes to the final $P_{v}^{(s)}(\cdot)Q_{v}^{(s)}$ matrix element.
We therefore introduce a transport map for its local factor.

Recall that $V_{k+1}=\bigotimes_{j}V_{k+1,j}$ is the one-step coarse-graining coisometry from scale $k$ to scale $k+1$.
Define
\begin{align}
	\mathsf{C}_{k\to k+1}(X):=V_{k+1}XV_{k+1}^{\dagger}.
\end{align}
For a fixed descendant $i=i_{k}$ of $v=i_{s}$, let $i_{k+1},\ldots,i_{s}$ denote its unique ancestors.
For a local operator at scale $k$, define
\begin{align}
	\mathsf{T}_{s-1\rightarrow s}^{i_{s-1}}(X) :=\mathsf{C}_{s-1\to s}(X),\qquad \mathsf{T}_{k\rightarrow s}^{i_{k}}(X) :=\mathsf{T}_{k+1\rightarrow s}^{i_{k+1}} \left( \mathcal{W}_{i_{k+1}}^{(k+1)} \bigl(\mathsf{C}_{k\to k+1}(X)\bigr) \right) \quad(k<s-1). \label{eq:bare-seed-transport-map}
\end{align}
Thus $\mathsf{T}_{k\rightarrow s}^{i}$ compresses a local scale-$k$ operator up to scale $s$, inserting the $\mathcal{W}$ map at each intermediate ancestor.
Identities outside its local support are implicit, and the result is an operator on the final site $v$.

We will use the transported local factors
\begin{align}
	E_{i}^{k\rightarrow s} :=\mathsf{T}_{k\rightarrow s}^{i}(Q_{i}^{(k)}), \qquad A_{i,v_{0}}^{0\rightarrow s} :=\mathsf{T}_{0\rightarrow s}^{i} \bigl(\mathcal{W}_{i}^{(0)}(\mathcal{X}_{i,v_{0}})\bigr). \label{eq:transported-families}
\end{align}
The distinction between the two definitions reflects their origin.
The factor $Q_{i}^{(k)}$ has already been produced by the fluctuation term in Eq.~\eqref{eq:complement-fluctuation-splitting}, so its transport starts with compression.
By contrast, the microscopic seed comes from the $\mathcal{W}_{i}^{(0)}(H_{i,Y_{0}}^{(0)})$ term and therefore contains $\mathcal{W}_{i}^{(0)}$ before the first compression.

Since $Y_{k}$ is disjoint from the descendant block of $v$, every coisometry acts separately on the two factors.
Hence a product of a local descendant operator and an operator on $Y_{k}$ is transported as the product of the corresponding site-$v$ transport and the ordinary compression of the $Y_{k}$ operator.
This uses the tensor-product structure of the RG blocks and does not assert multiplicativity of compression for general operators.

By the definition of the effective interaction,
\begin{align}
	H_{v,Y}^{(s)} = \sum_{i\in D_{s-1\leftarrow s}(v)} (\acute{H}_{i,Y_{s-1}}^{(s-1)})^{[s-1\to s]}. \label{eq:effective-interaction-compression}
\end{align}
Applying Eq.~\eqref{eq:RE-one-step-robustness} at scale $s-1$, the one-sided terms remain supported on the descendant block of $v$ and, after transport to scale $s$, become operators supported only on the final site $v$.
Such a term, denoted by $L_{v}$, gives
\begin{align}
	\bra{\Omega^{(s)}} P_{v}^{(s)}L_{v}Q_{v}^{(s)} \ket{\Omega^{(s)}} = \operatorname{Tr}\left( \rho_{v}^{(s)}P_{v}^{(s)}L_{v}Q_{v}^{(s)} \right) =0,
\end{align}
because $[P_{v}^{(s)},\rho_{v}^{(s)}]=0$.
Therefore only the continuing interaction part and the complement fluctuation contribute to the final $P_{v}^{(s)}(\cdot)Q_{v}^{(s)}$ matrix element, and Eq.~\eqref{eq:complement-fluctuation-splitting} gives
\begin{align}
	\bra{\Omega^{(s)}} P_{v}^{(s)}H_{v,Y}^{(s)}Q_{v}^{(s)} \ket{\Omega^{(s)}} &= \sum_{i\in D_{s-1\leftarrow s}(v)} \bra{\Omega^{(s)}} P_{v}^{(s)} \left( \mathcal{W}_{i}^{(s-1)} (H_{i,Y_{s-1}}^{(s-1)}) \right)^{[s-1\to s]} Q_{v}^{(s)} \ket{\Omega^{(s)}}\notag\\
	&\qquad - \sum_{i\in D_{s-1\leftarrow s}(v)} \bra{\Omega^{(s)}} P_{v}^{(s)}E_{i}^{s-1\rightarrow s}Q_{v}^{(s)} \otimes \left( B_{i}^{(s-1)}-b_{i}\mathds{1}_{Y_{s-1}} \right)^{[s-1\to s]} \ket{\Omega^{(s)}}. \label{eq:first-level-robustness-expansion}
\end{align}
The first term is the interaction contribution that must still be expanded.
For each $i_{s-1}\in D_{s-1\leftarrow s}(v)$,
\begin{align}
	H_{i_{s-1},Y_{s-1}}^{(s-1)} = \sum_{i_{s-2}\in D_{s-2\leftarrow s-1}(i_{s-1})} (\acute{H}_{i_{s-2},Y_{s-2}}^{(s-2)})^{[s-2\to s-1]}.
\end{align}
Applying Eq.~\eqref{eq:RE-one-step-robustness} once more to these scale-$(s-2)$ centered interactions therefore gives
\begin{align}
	&\bra{\Omega^{(s)}} P_{v}^{(s)}H_{v,Y}^{(s)}Q_{v}^{(s)} \ket{\Omega^{(s)}}\notag\\
	&\qquad = \sum_{i_{s-2}\in D_{s-2\leftarrow s}(v)} \bra{\Omega^{(s)}} P_{v}^{(s)} \mathsf{C}_{s-1\to s} \left( \mathcal{W}_{i_{s-1}}^{(s-1)} \left( \mathsf{C}_{s-2\to s-1} \left( \mathcal{W}_{i_{s-2}}^{(s-2)} (H_{i_{s-2},Y_{s-2}}^{(s-2)}) \right) \right) \right) Q_{v}^{(s)} \ket{\Omega^{(s)}}\notag\\
	&\qquad \qquad - \sum_{i\in D_{s-2\leftarrow s}(v)} \bra{\Omega^{(s)}} P_{v}^{(s)}E_{i}^{s-2\rightarrow s}Q_{v}^{(s)} \otimes \left( B_{i}^{(s-2)}-b_{i}\mathds{1}_{Y_{s-2}} \right)^{[s-2\to s]} \ket{\Omega^{(s)}}\notag\\
	&\qquad \qquad - \sum_{i\in D_{s-1\leftarrow s}(v)} \bra{\Omega^{(s)}} P_{v}^{(s)}E_{i}^{s-1\rightarrow s}Q_{v}^{(s)} \otimes \left( B_{i}^{(s-1)}-b_{i}\mathds{1}_{Y_{s-1}} \right)^{[s-1\to s]} \ket{\Omega^{(s)}},
\end{align}
where $i_{s-1}$ in the first sum is the unique scale-$(s-1)$ ancestor of $i_{s-2}$.
Thus the first term is again the interaction contribution that remains to be expanded, while the second and third terms are the complement fluctuations generated at scales $s-2$ and $s-1$, respectively.
At every subsequent step, the one-sided terms are discarded for the same reason as above: after transport they act only on the final site $v$, and their $P_{v}^{(s)}(\cdot)Q_{v}^{(s)}$ matrix elements vanish because $[P_{v}^{(s)},\rho_{v}^{(s)}]=0$.

Repeating the same substitution generates one complement-fluctuation term at each scale, while the remaining interaction is expanded down to scale $0$.
At that scale there is no lower level, so we use the fixed microscopic basis expansion
\begin{align}
	H_{i,Y_{0}}^{(0)} = \sum_{v_{0}=1}^{d^{2}} \mathcal{X}_{i,v_{0}}\otimes B_{i,v_{0}}^{(0)}.
\end{align}
The surviving local factor is $\mathcal{W}_{i}^{(0)}(\mathcal{X}_{i,v_{0}})$, whose transport to scale $s$ is $A_{i,v_{0}}^{0\rightarrow s}$.
Since $A_{i,v_{0}}^{0\rightarrow s}$ acts only on the final site $v$, we also have
\begin{align}
	\bra{\Omega^{(s)}} P_{v}^{(s)}A_{i,v_{0}}^{0\rightarrow s}Q_{v}^{(s)} \ket{\Omega^{(s)}} = \operatorname{Tr}\left( \rho_{v}^{(s)} P_{v}^{(s)}A_{i,v_{0}}^{0\rightarrow s}Q_{v}^{(s)} \right) =0.
\end{align}
Consequently, an arbitrary real scalar may be subtracted from each microscopic complement coefficient without changing the matrix element:
\begin{align}
	B_{i,v_{0}}^{(0)} \longrightarrow B_{i,v_{0}}^{(0)}-b_{i,v_{0}}\mathds{1}_{Y_{0}}, \qquad b_{i,v_{0}}\in\mathbb{R}.
\end{align}
Hence
\begin{align}
	\bra{\Omega^{(s)}} P_{v}^{(s)}H_{v,Y}^{(s)}Q_{v}^{(s)} \ket{\Omega^{(s)}} &= \sum_{i\in D_{0\leftarrow s}(v)} \sum_{v_{0}=1}^{d^{2}} \bra{\Omega^{(s)}} P_{v}^{(s)}A_{i,v_{0}}^{0\rightarrow s}Q_{v}^{(s)} \otimes \left( B_{i,v_{0}}^{(0)}-b_{i,v_{0}}\mathds{1}_{Y_{0}} \right)^{[0\to s]} \ket{\Omega^{(s)}}\notag\\
	&\qquad - \sum_{k=0}^{s-1} \sum_{i\in D_{k\leftarrow s}(v)} \bra{\Omega^{(s)}} P_{v}^{(s)}E_{i}^{k\rightarrow s}Q_{v}^{(s)} \otimes \left( B_{i}^{(k)}-b_{i}\mathds{1}_{Y_{k}} \right)^{[k\to s]} \ket{\Omega^{(s)}}. \label{eq:exact-indexed-robustness-expansion}
\end{align}
Here $B_{i,v_{0}}^{(0)}$ are the microscopic complement coefficients in the fixed-basis expansion of $H_{i,Y_{0}}^{(0)}$, and all $b$'s are arbitrary real scalars.
The first sum is the interaction contribution remaining when the recursion reaches the microscopic scale, while the second sum collects the complement fluctuations generated at all preceding scales.

\subsection{Coefficient bounds from the low-flip constraint}

\begin{lemma}
	\label{lemm:upp_bar_zeta^(0)}
	For $s\ge1$, $0\le k<s$, every complex coefficient family, and every fixed $v_{0}\in\{1,\ldots,d^{2}\}$, one has
	\begin{align}
		\left\|\sum_{i\in D_{k\leftarrow s}(v)}c_{i}E_{i}^{k\rightarrow s}\right\| \le\kappa_{k} \left( \sum_{i\in D_{k\leftarrow s}(v)}|c_{i}|^{2} \right)^{1/2}, \qquad \kappa_{k} :=2^{s-k-1}\sqrt{\prod_{t=k+1}^{s}z_{t}}. \label{eq:kappa-E-bound}
	\end{align}
	Moreover,
	\begin{align}
		\left\|\sum_{i\in D_{0\leftarrow s}(v)}c_{i}A_{i,v_{0}}^{0\rightarrow s}\right\| \le\kappa_{0}^{\rm seed} \left( \sum_{i\in D_{0\leftarrow s}(v)}|c_{i}|^{2} \right)^{1/2}, \qquad \kappa_{0}^{\rm seed} :=2^{s}\sqrt{Z_{s}}. \label{eq:kappa-seed-bound}
	\end{align}
\end{lemma}

\begin{proof}
	We first record the one-step estimate that will be used for both bounds.
	Consider one parent block whose retained subspace contains at most $z$ excited children, and let $V$ be the corresponding coisometry, so that
	\begin{align}
		V^{\dagger} V=\Pi_{\le z}, \qquad VV^{\dagger}=\mathds{1}.
	\end{align}
	By Eq.~\eqref{eq:W-map-definition},
	\begin{align}
		\mathcal{W}_{i}(X_{i}) = P_{i}X_{i}Q_{i}+Q_{i}X_{i}.
	\end{align}
	Let $\ket{\psi},\ket{\varphi}\in\operatorname{Ran}\Pi_{\le z}$ be normalized.
	Since the retained subspace contains at most $z$ excitations,
	\begin{align}
		\sum_{i}\norm{Q_{i}\ket{\psi}}^{2} = \bra{\psi}\sum_{i}Q_{i}\ket{\psi} \le z, \qquad \sum_{i}\norm{Q_{i}\ket{\varphi}}^{2} = \bra{\varphi}\sum_{i}Q_{i}\ket{\varphi} \le z.
	\end{align}
	Therefore Cauchy--Schwarz inequality gives
	\begin{align}
		\left| \sum_{i} \bra{\psi}P_{i}X_{i}Q_{i}\ket{\varphi} \right| &\le \sum_{i} \norm{X_{i}}\, \norm{Q_{i}\ket{\varphi}} \le \sqrt{z} \left(\sum_{i}\norm{X_{i}}^{2}\right)^{1/2}, \label{eq:kappa-one-step-PXQ}\\
		\left| \sum_{i} \bra{\psi}Q_{i}X_{i}\ket{\varphi} \right| &\le \sum_{i} \norm{Q_{i}\ket{\psi}}\, \norm{X_{i}} \le \sqrt{z} \left(\sum_{i}\norm{X_{i}}^{2}\right)^{1/2}. \label{eq:kappa-one-step-QX}
	\end{align}
	If $\ket{\widehat{\psi}}$ and $\ket{\widehat{\varphi}}$ are normalized vectors in the effective space, then $V^{\dagger}\ket{\widehat{\psi}}$ and $V^{\dagger}\ket{\widehat{\varphi}}$ are normalized vectors in $\operatorname{Ran}\Pi_{\le z}$.
	Hence the preceding bounds imply
	\begin{align}
		\left\| V\sum_{i}\mathcal{W}_{i}(X_{i})V^{\dagger} \right\| \le 2\sqrt{z} \left(\sum_{i}\norm{X_{i}}^{2}\right)^{1/2}. \label{eq:kappa-one-step}
	\end{align}
	
	For the first compression of the $E$ family we use a stronger estimate.
	Since the projections $Q_{i}$ commute, the operator $\sum_{i} c_{i}Q_{i}$ is diagonal in a joint occupation basis.
	On $\operatorname{Ran}\Pi_{\le z}$, every joint occupation eigenvector has an excited set $S$ with $|S|\le z$, and the corresponding eigenvalue of $\sum_{i} c_{i}Q_{i}$ is $\sum_{i\in S}c_{i}$.
	Thus
	\begin{align}
		\left| \sum_{i\in S}c_{i} \right| \le \sqrt{|S|} \left(\sum_{i\in S}|c_{i}|^{2}\right)^{1/2} \le \sqrt{z} \left(\sum_{i}|c_{i}|^{2}\right)^{1/2}.
	\end{align}
	Since $V$ is unitary from $\operatorname{Ran}\Pi_{\le z}$ onto the effective space, it follows that
	\begin{align}
		\left\| V\sum_{i} c_{i}Q_{i}V^{\dagger} \right\| \le \sqrt{z} \left(\sum_{i}|c_{i}|^{2}\right)^{1/2}. \label{eq:kappa-base-Q}
	\end{align}
	
	We now prove Eq.~\eqref{eq:kappa-E-bound}.
	Fix $0\le k<s$.
	At the first step $k\to k+1$, the definition $E_{i}^{k\rightarrow k+1} =\mathsf{T}_{k\rightarrow k+1}^{i}(Q_{i}^{(k)})$ and Eq.~\eqref{eq:bare-seed-transport-map} give
	\begin{align}
		E_{i}^{k\rightarrow k+1} = V_{k+1,u}Q_{i}^{(k)}V_{k+1,u}^{\dagger}
	\end{align}
	for $i\in D_{k\leftarrow k+1}(u)$.
	Therefore Eq.~\eqref{eq:kappa-base-Q}, with $z=z_{k+1}$, gives
	\begin{align}
		\left\| \sum_{i\in D_{k\leftarrow k+1}(u)} c_{i}E_{i}^{k\rightarrow k+1} \right\| \le \sqrt{z_{k+1}} \left( \sum_{i\in D_{k\leftarrow k+1}(u)} |c_{i}|^{2} \right)^{1/2}. \label{eq:kappa-E-base}
	\end{align}
	Suppose now that for some $r$ with $k<r<s$, every scale-$r$ site $w$ satisfies
	\begin{align}
		\left\| \sum_{i\in D_{k\leftarrow r}(w)} c_{i}E_{i}^{k\rightarrow r} \right\| \le 2^{r-k-1} \sqrt{\prod_{t=k+1}^{r}z_{t}} \left( \sum_{i\in D_{k\leftarrow r}(w)} |c_{i}|^{2} \right)^{1/2}. \label{eq:kappa-E-induction}
	\end{align}
	For a scale-$(r+1)$ site $u$, let $w\in D_{r\leftarrow r+1}(u)$ denote its scale-$r$ children and set
	\begin{align}
		X_{w} := \sum_{i\in D_{k\leftarrow r}(w)} c_{i}E_{i}^{k\rightarrow r}.
	\end{align}
	By the recursive definition Eq.~\eqref{eq:bare-seed-transport-map}, transporting these operators from scale $r$ to scale $r+1$ gives
	\begin{align}
		\sum_{i\in D_{k\leftarrow r+1}(u)} c_{i}E_{i}^{k\rightarrow r+1} = V_{r+1,u} \left( \sum_{w\in D_{r\leftarrow r+1}(u)} \mathcal{W}_{w}^{(r)}(X_{w}) \right) V_{r+1,u}^{\dagger}. \label{eq:kappa-E-one-step-recursion}
	\end{align}
	Applying Eq.~\eqref{eq:kappa-one-step} with $z=z_{r+1}$ and then using Eq.~\eqref{eq:kappa-E-induction}, we obtain
	\begin{align}
		\left\| \sum_{i\in D_{k\leftarrow r+1}(u)} c_{i}E_{i}^{k\rightarrow r+1} \right\| &\le 2\sqrt{z_{r+1}} \left( \sum_{w\in D_{r\leftarrow r+1}(u)} \norm{X_{w}}^{2} \right)^{1/2}\notag\\
		&\le 2^{r-k} \sqrt{\prod_{t=k+1}^{r+1}z_{t}} \left( \sum_{w\in D_{r\leftarrow r+1}(u)} \sum_{i\in D_{k\leftarrow r}(w)} |c_{i}|^{2} \right)^{1/2}\notag\\
		&= 2^{r-k} \sqrt{\prod_{t=k+1}^{r+1}z_{t}} \left( \sum_{i\in D_{k\leftarrow r+1}(u)} |c_{i}|^{2} \right)^{1/2}. \label{eq:kappa-E-induction-step}
	\end{align}
	The last equality follows because the sets $D_{k\leftarrow r}(w)$ associated with distinct children $w$ are disjoint and their union is $D_{k\leftarrow r+1}(u)$.
	Thus the induction closes.
	Taking $r+1=s$ and $u=v$ proves Eq.~\eqref{eq:kappa-E-bound}.
	
	We next prove Eq.~\eqref{eq:kappa-seed-bound}.
	Fix $v_{0}\in\{1,\ldots,d^{2}\}$.
	At the first step, the definition of the microscopic seed and Eq.~\eqref{eq:transported-families} give
	\begin{align}
		A_{i,v_{0}}^{0\rightarrow1} = V_{1,u} \mathcal{W}_{i}^{(0)}(\mathcal{X}_{i,v_{0}}) V_{1,u}^{\dagger}
	\end{align}
	for $i\in D_{0\leftarrow1}(u)$.
	Applying Eq.~\eqref{eq:kappa-one-step} with $X_{i}=c_{i}\mathcal{X}_{i,v_{0}}$ and using $\norm{\mathcal{X}_{i,v_{0}}}\le1$, we obtain
	\begin{align}
		\left\| \sum_{i\in D_{0\leftarrow1}(u)} c_{i}A_{i,v_{0}}^{0\rightarrow1} \right\| \le 2\sqrt{z_{1}} \left( \sum_{i\in D_{0\leftarrow1}(u)} |c_{i}|^{2} \right)^{1/2}. \label{eq:kappa-seed-base}
	\end{align}
	
	Unlike the $E$ family, whose first compression is controlled by the stronger estimate Eq.~\eqref{eq:kappa-base-Q}, the seed family already contains the map $\mathcal{W}_{i}^{(0)}$ before its first compression.
	At every subsequent scale $r\to r+1$, its transported operators obey the same recursion as in Eq.~\eqref{eq:kappa-E-one-step-recursion}, and hence each additional scale contributes a factor $2\sqrt{z_{r+1}}$ through Eq.~\eqref{eq:kappa-one-step}.
	Iterating from scale $0$ to scale $s$ therefore gives
	\begin{align}
		\left\| \sum_{i\in D_{0\leftarrow s}(v)} c_{i}A_{i,v_{0}}^{0\rightarrow s} \right\| &\le \prod_{t=1}^{s}(2\sqrt{z_{t}}) \left(\sum_{i}|c_{i}|^{2}\right)^{1/2} = 2^{s}\sqrt{Z_{s}} \left(\sum_{i}|c_{i}|^{2}\right)^{1/2},
	\end{align}
	which proves Eq.~\eqref{eq:kappa-seed-bound}.
\end{proof}

\subsection{Compressed complement fluctuations}

For $k<s$, define
\begin{align}
	\Theta_{i}^{k\rightarrow s} :=\inf_{b\in\mathbb{R}} \norm{ \left( B_{i}^{(k)}-b\mathds{1}_{Y_{k}} \right)^{[k\to s]} \ket{\Omega^{(s)}} }, \qquad \overline{\Theta}_{v}^{k\rightarrow s} := \left( \sum_{i\in D_{k\leftarrow s}(v)} \left(\Theta_{i}^{k\rightarrow s}\right)^{2} \right)^{1/2}. \label{eq:theta-definitions}
\end{align}
The microscopic-basis quantities $\Theta_{i,v_{0}}^{0\rightarrow s}$ and $\overline{\Theta}_{v;v_{0}}^{0\rightarrow s}$ are defined in the same way from $B_{i,v_{0}}^{(0)}$.

We also define the coupling square sum
\begin{align}
	\mathfrak{C}_{k}(v,Y) := \left( \sum_{i\in D_{k\leftarrow s}(v)} \sum_{j\in Y_{k}} [J_{ij}^{(k)}]^{2} \right)^{1/2}. \label{eq:coupling-square-sum}
\end{align}
Here $J_{ij}^{(k)}$ is the stage-$k$ interaction upper bound introduced in Eq.~\eqref{RG_interaction_strength_upper_bound}.

\begin{lemma}
	\label{lemma:bar_Theta_k-to-s}
	For all $0\le k<s$,
	\begin{align}
		\overline{\Theta}_{v}^{k\rightarrow s} &\le2\eta_{k+1:s}\sqrt{\frac{\bar{g}^{(s)}}{\Delta^{(s)}}} \,\mathfrak{C}_{k}(v,Y), \label{eq:theta-square-sum-general-scale}\\
		\overline{\Theta}_{v;v_{0}}^{0\rightarrow s} &\le2\eta_{s}\sqrt{\frac{\bar{g}^{(s)}}{\Delta^{(s)}}} \,\mathfrak{C}_{0}(v,Y). \label{eq:theta-square-sum-bound}
	\end{align}
	Under the Kac interaction bound,
	\begin{align}
		\mathfrak{C}_{k}(v,Y)\le \frac{J^{(k)}}{[n^{(k)}]^{\beta}} \left[C_{D,\alpha}\mathfrak{F}_{\alpha} (\ell_{k+1:s},[n^{(k)}]^{1/D}) +C_{D}(\mu_{k}-1)\ell_{k+1:s}^{(D-1)/2}\right]. \label{bar_Theta_k-to-s/main_bound/main}
	\end{align}
\end{lemma}

\begin{proof}
	From Eq.~\eqref{eq:one-sided-centering-terms}, $B_{i}^{(k)}$ is a one-site sum over the sites $j\in Y_{k}$:
	\begin{align}
		B_{i}^{(k)} = \sum_{j\in Y_{k}} \bra{0_{i}^{(k)}}h_{ij}^{(k)}\ket{0_{i}^{(k)}}.
	\end{align}
	Each summand is Hermitian and has norm at most $J_{ij}^{(k)}$.
	Since the compression is unital,
	\begin{align}
		\left( B_{i}^{(k)}-b\mathds{1}_{Y_{k}} \right)^{[k\to s]} = (B_{i}^{(k)})^{[k\to s]}-b\mathds{1},
	\end{align}
	and therefore
	\begin{align}
		(\Theta_{i}^{k\rightarrow s})^{2} = \inf_{b\in\mathbb{R}} \norm{ \left( (B_{i}^{(k)})^{[k\to s]}-b\mathds{1} \right) \ket{\Omega^{(s)}} }^{2} = \operatorname{Var}_{\Omega^{(s)}} \left( (B_{i}^{(k)})^{[k\to s]} \right).
	\end{align}
	Applying Proposition~\ref{trade_off_variance_renomarlized} to the scale-$k$ one-site sum defining $B_{i}^{(k)}$ gives
	\begin{align}
		(\Theta_{i}^{k\rightarrow s})^{2} &\le 4\eta_{k+1:s}^{2} \frac{\bar{g}^{(s)}}{\Delta^{(s)}} \sum_{j\in Y_{k}} \norm{ \bra{0_{i}^{(k)}}h_{ij}^{(k)}\ket{0_{i}^{(k)}} }^{2} \notag\\
		&\le 4\eta_{k+1:s}^{2} \frac{\bar{g}^{(s)}}{\Delta^{(s)}} \sum_{j\in Y_{k}} \norm{h_{ij}^{(k)}}^{2} \notag\\
		&\le 4\eta_{k+1:s}^{2} \frac{\bar{g}^{(s)}}{\Delta^{(s)}} \sum_{j\in Y_{k}}[J_{ij}^{(k)}]^{2},
	\end{align}
	where we used $\norm{\bra{0_{i}^{(k)}}h_{ij}^{(k)}\ket{0_{i}^{(k)}}} \le\norm{h_{ij}^{(k)}}$ and $\norm{h_{ij}^{(k)}}\le J_{ij}^{(k)}$.
	Summing over $i\in D_{k\leftarrow s}(v)$ and taking a square root yields
	\begin{align}
		\overline{\Theta}_{v}^{k\rightarrow s} \le 2\eta_{k+1:s} \sqrt{\frac{\bar{g}^{(s)}}{\Delta^{(s)}}} \,\mathfrak{C}_{k}(v,Y).
	\end{align}
	
	For the microscopic-basis coefficients, Eqs.~\eqref{eq:basis-reconstruction-constant} and \eqref{eq:microscopic-d2-expansion} give
	\begin{align}
		B_{i,v_{0}}^{(0)} = \sum_{j\in Y_{0}}b_{i,j,v_{0}}^{(0)}, \qquad \norm{b_{i,j,v_{0}}^{(0)}} \le \norm{h_{ij}^{(0)}} \le J_{ij}^{(0)},
	\end{align}
	where each $b_{i,j,v_{0}}^{(0)}$ is Hermitian and supported on $j$.
	Applying the same compressed variance estimate with $k=0$ gives
	\begin{align}
		(\Theta_{i,v_{0}}^{0\rightarrow s})^{2} \le 4\eta_{s}^{2} \frac{\bar{g}^{(s)}}{\Delta^{(s)}} \sum_{j\in Y_{0}}[J_{ij}^{(0)}]^{2}.
	\end{align}
	Summing over $i\in D_{0\leftarrow s}(v)$ and taking a square root therefore gives
	\begin{align}
		\overline{\Theta}_{v;v_{0}}^{0\rightarrow s} \le 2\eta_{s} \sqrt{\frac{\bar{g}^{(s)}}{\Delta^{(s)}}} \,\mathfrak{C}_{0}(v,Y).
	\end{align}
	This proves Eqs.~\eqref{eq:theta-square-sum-general-scale} and \eqref{eq:theta-square-sum-bound}.
	
	It remains to estimate $\mathfrak{C}_{k}(v,Y)$.
	Under the Kac interaction bound in Eq.~\eqref{RG_interaction_strength_upper_bound},
	\begin{align}
		\mathfrak{C}_{k}(v,Y) &\le \frac{J^{(k)}}{[n^{(k)}]^{\beta}} \left( \sum_{i\in D_{k\leftarrow s}(v)} \sum_{j\in Y_{k}} [\mu_{ij}^{(k)}]^{2}r_{ij}^{-2\alpha} \right)^{1/2}.
	\end{align}
	The descendant block $D_{k\leftarrow s}(v)$ is a cube of side $\ell_{k+1:s}$ in the scale-$k$ lattice, while $Y_{k}$ is contained in its complement.
	Hence Eq.~\eqref{eq:block-complement-kernel} gives
	\begin{align}
		\mathfrak{C}_{k}(v,Y) \le \frac{J^{(k)}}{[n^{(k)}]^{\beta}} \left[ C_{D,\alpha}\mathfrak{F}_{\alpha} (\ell_{k+1:s},[n^{(k)}]^{1/D}) +C_{D}(\mu_{k}-1)\ell_{k+1:s}^{(D-1)/2} \right],
	\end{align}
	which proves Eq.~\eqref{bar_Theta_k-to-s/main_bound/main}.
\end{proof}

\subsection{The dimension-free reflection proposition}

\begin{prop}
	\label{prop:renormalized-robustness}
	At any defined scale $s\ge1$, suppose that $H^{(s)}$ has a unique gapped ground state.
	Then, for every $v\in\Lambda^{(s)}$ and $Y\subseteq v^{\mathrm{c}}$,
	\begin{align}
		\delta_{s,v}(Y)\le\mathcal{A}_{s}(v,Y)\sqrt{q_{s,v}}, \label{eq:renormalized-refined-robustness}
	\end{align}
	where one may take
	\begin{align}
		\mathcal{A}_{s}(v,Y):=8\sqrt{\frac{\bar{g}^{(s)}}{\Delta^{(s)}}} \left[ d^{2}\kappa_{0}^{\rm seed}\eta_{s}\mathfrak{C}_{0}(v,Y) +\sum_{k=0}^{s-1}\kappa_{k}\eta_{k+1:s}\mathfrak{C}_{k}(v,Y) \right]. \label{eq:exact-As-coefficient}
	\end{align}
	No factor involving $d^{(s)}$ occurs.
	If $D/2<\alpha<(D+1)/2$ and $z_{t}\ge1$, then the constants $C_{\rm rob},c_{\rm rob}\ge1$ specified in Eq.~\eqref{eq:robustness-constant-choice}, depending only on $D,\alpha,d$, satisfy
	\begin{align}
		\mathcal{A}_{s}(v,Y)\le C_{\rm rob}c_{\rm rob}^{s}J_{\rm K}Z_{s} \frac{1}{[n^{(s)}]^{\beta}} \sqrt{\frac{\bar{g}^{(s)}}{\Delta^{(s)}}}. \label{eq:As-robustness-coefficient}
	\end{align}
	If $\alpha=(D+1)/2<D$, the right-hand side of Eq.~\eqref{eq:As-robustness-coefficient} acquires an additional factor $\chi_{s}$.
	If $(D+1)/2<\alpha<D$, the same bound holds with $J_{\rm K}$ replaced by $J\mu_{0}\ell_{1:s}^{\alpha-(D+1)/2}$.
\end{prop}

\begin{proof}
	Choose the scalar minimizers in Eq.~\eqref{eq:exact-indexed-robustness-expansion}.
	For each $0\le k<s$, Lemma~\ref{sum_reduction_lemma} together with Eq.~\eqref{eq:kappa-E-bound} gives
	\begin{align}
		&\left| \sum_{i\in D_{k\leftarrow s}(v)} \bra{\Omega^{(s)}} P_{v}^{(s)}E_{i}^{k\rightarrow s}Q_{v}^{(s)} \otimes \left( B_{i}^{(k)}-b_{i}\mathds{1}_{Y_{k}} \right)^{[k\to s]} \ket{\Omega^{(s)}} \right| \le \kappa_{k}\sqrt{q_{s,v}}\, \overline{\Theta}_{v}^{k\rightarrow s}.
	\end{align}
	Here $\kappa_{k}$ is defined in Eq.~\eqref{eq:kappa-E-bound}, $q_{s,v}$ in Eq.~\eqref{eq:scale-s-reference-state}, and $\overline{\Theta}_{v}^{k\rightarrow s}$ in Eq.~\eqref{eq:theta-definitions}.
	Similarly, for each fixed $v_{0}\in\{1,\ldots,d^{2}\}$, Lemma~\ref{sum_reduction_lemma} and Eq.~\eqref{eq:kappa-seed-bound} give
	\begin{align}
		&\left| \sum_{i\in D_{0\leftarrow s}(v)} \bra{\Omega^{(s)}} P_{v}^{(s)}A_{i,v_{0}}^{0\rightarrow s}Q_{v}^{(s)} \otimes \left( B_{i,v_{0}}^{(0)} -b_{i,v_{0}}\mathds{1}_{Y_{0}} \right)^{[0\to s]} \ket{\Omega^{(s)}} \right| \le \kappa_{0}^{\rm seed}\sqrt{q_{s,v}}\, \overline{\Theta}_{v;v_{0}}^{0\rightarrow s}.
	\end{align}
	Here $\kappa_{0}^{\rm seed}$ is defined in Eq.~\eqref{eq:kappa-seed-bound} and $\overline{\Theta}_{v;v_{0}}^{0\rightarrow s}$ in Eq.~\eqref{eq:theta-definitions}.
	Hence Eq.~\eqref{eq:exact-indexed-robustness-expansion} and the triangle inequality yield
	\begin{align}
		\left| \bra{\Omega^{(s)}} P_{v}^{(s)}H_{v,Y}^{(s)}Q_{v}^{(s)} \ket{\Omega^{(s)}} \right| \le \sqrt{q_{s,v}} \left[ \kappa_{0}^{\rm seed} \sum_{v_{0}=1}^{d^{2}} \overline{\Theta}_{v;v_{0}}^{0\rightarrow s} + \sum_{k=0}^{s-1} \kappa_{k}\overline{\Theta}_{v}^{k\rightarrow s} \right].
	\end{align}
	Applying Lemma~\ref{lemma:bar_Theta_k-to-s} gives
	\begin{align}
		\left| \bra{\Omega^{(s)}} P_{v}^{(s)}H_{v,Y}^{(s)}Q_{v}^{(s)} \ket{\Omega^{(s)}} \right| \le 2\sqrt{\frac{\bar{g}^{(s)}}{\Delta^{(s)}}} \sqrt{q_{s,v}} \left[ d^{2}\kappa_{0}^{\rm seed}\eta_{s}\mathfrak{C}_{0}(v,Y) + \sum_{k=0}^{s-1} \kappa_{k}\eta_{k+1:s}\mathfrak{C}_{k}(v,Y) \right].
	\end{align}
	Combining this with Eq.~\eqref{Quantity_robustness_rewrite} proves Eq.~\eqref{eq:exact-As-coefficient}.
	
	Suppose now that $D/2<\alpha<(D+1)/2$.
	Then $\mu_{k}=1$, while Eq.~\eqref{eq:F-alpha-square-sum} gives $\mathfrak{F}_{\alpha}(\ell_{k+1:s},[n^{(k)}]^{1/D}) =\ell_{k+1:s}^{D-\alpha}$.
	Hence Eq.~\eqref{bar_Theta_k-to-s/main_bound/main} yields
	\begin{align}
		\mathfrak{C}_{k}(v,Y) \le C_{D,\alpha}J^{(k)} \frac{\ell_{k+1:s}^{D-\alpha}}{[n^{(k)}]^{\beta}} = \frac{C_{D,\alpha}J^{(k)}}{[n^{(s)}]^{\beta}},
	\end{align}
	where the equality follows from $n^{(s)}=n^{(k)}/\ell_{k+1:s}^{D}$ and $D\beta=D-\alpha$.
	For $z_{t}\ge1$, $\eta_{k+1:s}\le(2+\sqrt{2})^{s-k} \sqrt{\prod_{t=k+1}^{s}z_{t}}$, and therefore, using $J^{(k)}=c_{\rm int}^{k}J_{\rm K}Z_{k}$ from Eq.~\eqref{eq:J-recurrence-Dhalf-alpha-Dplus1half},
	\begin{align}
		\kappa_{k}\eta_{k+1:s}J^{(k)} \le \frac{1}{2}(4+2\sqrt{2})^{s-k} c_{\rm int}^{k}J_{\rm K}Z_{s}, \qquad \kappa_{0}^{\rm seed}\eta_{s}J^{(0)} \le (4+2\sqrt{2})^{s}J_{\rm K}Z_{s}.
	\end{align}
	Define
	\begin{align}
		C_{\rm rob}:=8(C_{D,\alpha}+C_{D})(d^{2}+1), \qquad c_{\rm rob}:=2\max\left\{4+2\sqrt{2},c_{\rm int}\right\}. \label{eq:robustness-constant-choice}
	\end{align}
	Indeed, the sum of the preceding scalar coefficients is at most $(d^{2}+s/2)\max\{4+2\sqrt{2},c_{\rm int}\}^{s}\le(d^{2}+1)c_{\rm rob}^{s}$.
	The factor $8(C_{D,\alpha}+C_{D})$ includes the prefactor in Eq.~\eqref{eq:exact-As-coefficient} and also covers the geometric coefficient needed in the final regime below.
	Substitution into Eq.~\eqref{eq:exact-As-coefficient} proves Eq.~\eqref{eq:As-robustness-coefficient}.
	
	Suppose next that $\alpha=(D+1)/2<D$.
	Then $\mu_{k}=1$, and Eq.~\eqref{eq:F-alpha-square-sum} gives $\mathfrak{F}_{\alpha}(\ell_{k+1:s},[n^{(k)}]^{1/D}) =\ell_{k+1:s}^{(D-1)/2}\sqrt{\log(e\ell_{k+1:s})}$.
	Since $D-\alpha=(D-1)/2$, Eq.~\eqref{bar_Theta_k-to-s/main_bound/main} therefore gives
	\begin{align}
		\mathfrak{C}_{k}(v,Y) \le \frac{C_{D,\alpha}J^{(k)}}{[n^{(s)}]^{\beta}} \sqrt{\log(e\ell_{k+1:s})}.
	\end{align}
	Writing $\xi_{k+1:s}:=\sqrt{\log(e\ell_{k+1:s})}$ and using $J^{(k)}=c_{\rm int}^{k}J_{\rm K}Z_{k}\chi_{k}$ from Eq.~\eqref{eq:J-renormalization-alpha-eq-Dplus1half}, we have $\chi_{k}\xi_{k+1:s}\le\chi_{s}$, because $1+\sum_{t=k+1}^{s}\log\ell_{t} \le\prod_{t=k+1}^{s}(1+\log\ell_{t})$.
	The preceding coefficient summation therefore gives Eq.~\eqref{eq:As-robustness-coefficient} with one additional factor $\chi_{s}$.
	
	Finally, suppose that $(D+1)/2<\alpha<D$.
	In this case Eq.~\eqref{eq:F-alpha-square-sum} gives $\mathfrak{F}_{\alpha}(\ell_{k+1:s},[n^{(k)}]^{1/D}) =\ell_{k+1:s}^{(D-1)/2}$, and hence Eq.~\eqref{bar_Theta_k-to-s/main_bound/main} implies
	\begin{align}
		\mathfrak{C}_{k}(v,Y) &\le \frac{J^{(k)}}{[n^{(k)}]^{\beta}} \left[C_{D,\alpha}+C_{D}(\mu_{k}-1)\right] \ell_{k+1:s}^{(D-1)/2} \le (C_{D,\alpha}+C_{D})J^{(k)} \frac{\mu_{k}\ell_{k+1:s}^{(D-1)/2}}{[n^{(k)}]^{\beta}}.
	\end{align}
	Here the second inequality uses $\mu_{k}\ge1$.
	Using $\mu_{k}=\mu_{0}\ell_{1:k}^{\alpha-(D+1)/2}$, $n^{(s)}=n^{(k)}/\ell_{k+1:s}^{D}$, and $\alpha-(D+1)/2+D\beta=(D-1)/2$, we obtain
	\begin{align}
		\frac{\mu_{k}\ell_{k+1:s}^{(D-1)/2}}{[n^{(k)}]^{\beta}} = \frac{\mu_{0}\ell_{1:s}^{\alpha-(D+1)/2}}{[n^{(s)}]^{\beta}}.
	\end{align}
	Together with $J^{(k)}=c_{\rm int}^{k}JZ_{k}$, the same coefficient summation as above proves the stated bound with $J_{\rm K}$ replaced by $J\mu_{0}\ell_{1:s}^{\alpha-(D+1)/2}$.
\end{proof}

\subsection{Inheriting the bootstrap seed instead of assuming it}
\label{sec:inherited-reference}

A reflection bound alone does not guarantee that the leading local eigenvalue exceeds $1/2$.
To apply Proposition~\ref{prop:robustness-fixed-point} recursively, we therefore need to show that this bootstrap seed is inherited from the preceding scale.
An accurate low-flip compression provides precisely this information.

\begin{lemma}
	\label{lem:inherited-reference}
	Suppose a step $s\to s+1$ has normalized ground states satisfying
	\begin{align}
		\norm{\ket{\Omega^{(s)}}-V_{s+1}^{\dagger}\ket{\Omega^{(s+1)}}} \le\epsilon_{s}. \label{eq:step-state-error}
	\end{align}
	Then the largest eigenvalue $p_{0,j}^{(s+1)}:=1-q_{s+1,j}$ at every effective site $j$ satisfies
	\begin{align}
		p_{0,j}^{(s+1)} \ge 1-b_{s+1}q_{s,*}-\epsilon_{s}. \label{eq:inherited-leading-eigenvalue}
	\end{align}
	Moreover, for every child $i\in L_{j}^{(s)}$ and any normalized largest-eigenvalue vector $\ket{0_{j}^{(s+1)}}$,
	\begin{align}
		\bra{0_{j}^{(s+1)}} V_{s+1,j}Q_{i}^{(s)}V_{s+1,j}^{\dagger} \ket{0_{j}^{(s+1)}} \le \frac{q_{s,i}+\epsilon_{s}}{p_{0,j}^{(s+1)}}. \label{eq:reference-child-occupation-general}
	\end{align}
	If, in addition,
	\begin{align}
		b_{s+1}q_{s,*}+\epsilon_{s}<\frac{1}{2}, \label{eq:inherited-reference-smallness}
	\end{align}
	then the bootstrap seed is inherited at scale $s+1$:
	\begin{align}
		p_{0,j}^{(s+1)} > \frac{1}{2}. \label{eq:inherited-bootstrap-seed}
	\end{align}
	In particular, this largest eigenvalue is nondegenerate, and every child $i\in L_{j}^{(s)}$ satisfies
	\begin{align}
		\bra{0_{j}^{(s+1)}} V_{s+1,j}Q_{i}^{(s)}V_{s+1,j}^{\dagger} \ket{0_{j}^{(s+1)}} \le 2(q_{s,*}+\epsilon_{s}). \label{eq:reference-child-occupation}
	\end{align}
\end{lemma}

\begin{proof}
	The product of the child reference vectors belongs to the retained subspace of $L_{j}^{(s)}$, even when the cutoff is zero.
	Its normalized effective image is
	\begin{align}
		\ket{e_{j}} := V_{s+1,j} \bigotimes_{i\in L_{j}^{(s)}}\ket{0_{i}^{(s)}}.
	\end{align}
	Since the projections $P_{i}^{(s)}$ act on distinct children,
	\begin{align}
		\mathds{1} - \prod_{i\in L_{j}^{(s)}}P_{i}^{(s)} \le \sum_{i\in L_{j}^{(s)}}Q_{i}^{(s)}.
	\end{align}
	Therefore
	\begin{align}
		\bra{\Omega^{(s)}} \prod_{i\in L_{j}^{(s)}}P_{i}^{(s)} \ket{\Omega^{(s)}} &\ge 1- \sum_{i\in L_{j}^{(s)}} \bra{\Omega^{(s)}}Q_{i}^{(s)}\ket{\Omega^{(s)}} = 1- \sum_{i\in L_{j}^{(s)}}q_{s,i} \ge 1-b_{s+1}q_{s,*}. \label{eq:parent-reference-product-weight}
	\end{align}
	
	The two normalized parent-space states $\ket{\Omega^{(s)}}$ and $V_{s+1}^{\dagger}\ket{\Omega^{(s+1)}}$ are within vector distance $\epsilon_{s}$ by Eq.~\eqref{eq:step-state-error}.
	For normalized pure states $\ket{\psi}$ and $\ket{\varphi}$, $\frac{1}{2}\norm{\ket{\psi}\bra{\psi}-\ket{\varphi}\bra{\varphi}}_{1} =\sqrt{1-|\braket{\psi|\varphi}|^{2}} \le\norm{\ket{\psi}-\ket{\varphi}}$.
	Moreover, for every effect $0\le X\le\mathds{1}$, $|\operatorname{Tr}[X(\rho-\sigma)]| \le\frac{1}{2}\norm{\rho-\sigma}_{1}$.
	Hence Eq.~\eqref{eq:step-state-error} implies that the expectation value of any effect in these two parent-space states differs by at most $\epsilon_{s}$.
	Applying this to the product-reference projector on $L_{j}^{(s)}$ and using Eq.~\eqref{eq:parent-reference-product-weight}, we obtain
	\begin{align}
		\bra{\Omega^{(s+1)}} V_{s+1} \left( \prod_{i\in L_{j}^{(s)}}P_{i}^{(s)} \right) V_{s+1}^{\dagger} \ket{\Omega^{(s+1)}} \ge \bra{\Omega^{(s)}} \prod_{i\in L_{j}^{(s)}}P_{i}^{(s)} \ket{\Omega^{(s)}} -\epsilon_{s} \ge 1-b_{s+1}q_{s,*}-\epsilon_{s}. \label{eq:inherited-product-reference-bound}
	\end{align}
	Since $\ket{e_{j}} =V_{s+1,j}\bigotimes_{i\in L_{j}^{(s)}}\ket{0_{i}^{(s)}}$ and the global coisometry factorizes over the parent blocks,
	\begin{align}
		V_{s+1} \left( \prod_{i\in L_{j}^{(s)}}P_{i}^{(s)} \right) V_{s+1}^{\dagger} = \ket{e_{j}}\bra{e_{j}}\otimes\mathds{1}_{j^{\mathrm{c}}}.
	\end{align}
	Therefore the left-hand side of Eq.~\eqref{eq:inherited-product-reference-bound} is $\bra{e_{j}}\rho_{j}^{(s+1)}\ket{e_{j}}$, where $\rho_{j}^{(s+1)} =\operatorname{Tr}_{j^{\mathrm{c}}} \ket{\Omega^{(s+1)}}\bra{\Omega^{(s+1)}}$ is the reduced ground state on the effective site $j$.
	Since $p_{0,j}^{(s+1)}$ is the largest eigenvalue of $\rho_{j}^{(s+1)}$, we conclude that
	\begin{align}
		p_{0,j}^{(s+1)} \ge \bra{e_{j}}\rho_{j}^{(s+1)}\ket{e_{j}} \ge 1-b_{s+1}q_{s,*}-\epsilon_{s},
	\end{align}
	which proves Eq.~\eqref{eq:inherited-leading-eigenvalue}.
	
	We next bound the excitation $\bra{0_{j}^{(s+1)}}V_{s+1,j}Q_{i}^{(s)}V_{s+1,j}^{\dagger}\ket{0_{j}^{(s+1)}}$ of a scale-$s$ child $i\in L_{j}^{(s)}$ inside the scale-$(s+1)$ reference state.
	For $i\in L_{j}^{(s)}$, define
	\begin{align}
		X_{i} := V_{s+1,j}Q_{i}^{(s)}V_{s+1,j}^{\dagger}.
	\end{align}
	Since $0\le Q_{i}^{(s)}\le\mathds{1}$ and $V_{s+1,j}V_{s+1,j}^{\dagger}=\mathds{1}$,
	\begin{align}
		0\le X_{i}\le\mathds{1}.
	\end{align}
	Moreover, the expectation of $X_{i}$ in the scale-$(s+1)$ reduced state equals the expectation of $Q_{i}^{(s)}$ in the embedded scale-$(s+1)$ ground state.
	Hence Eq.~\eqref{eq:step-state-error}, applied to the effect $Q_{i}^{(s)}$, gives
	\begin{align}
		\operatorname{Tr}\left(\rho_{j}^{(s+1)}X_{i}\right) = \bra{\Omega^{(s+1)}} V_{s+1}Q_{i}^{(s)}V_{s+1}^{\dagger} \ket{\Omega^{(s+1)}} \le \bra{\Omega^{(s)}}Q_{i}^{(s)}\ket{\Omega^{(s)}} +\epsilon_{s} = q_{s,i}+\epsilon_{s}. \label{eq:inherited-child-upper-expectation}
	\end{align}
	On the other hand, since $\ket{0_{j}^{(s+1)}}$ is a largest-eigenvalue eigenvector of $\rho_{j}^{(s+1)}$ with eigenvalue $p_{0,j}^{(s+1)}$,
	\begin{align}
		\rho_{j}^{(s+1)} \geq \bra{0_{j}^{(s+1)}} \rho_{j}^{(s+1)} \ket{0_{j}^{(s+1)}} \ket{0_{j}^{(s+1)}}\bra{0_{j}^{(s+1)}} = p_{0,j}^{(s+1)} \ket{0_{j}^{(s+1)}}\bra{0_{j}^{(s+1)}}.
	\end{align}
	Because $X_{i}\ge0$, it follows that
	\begin{align}
		\operatorname{Tr}\left(\rho_{j}^{(s+1)}X_{i}\right) \ge p_{0,j}^{(s+1)} \bra{0_{j}^{(s+1)}}X_{i}\ket{0_{j}^{(s+1)}}.
	\end{align}
	Combining this with Eq.~\eqref{eq:inherited-child-upper-expectation} yields
	\begin{align}
		p_{0,j}^{(s+1)} \bra{0_{j}^{(s+1)}} V_{s+1,j}Q_{i}^{(s)}V_{s+1,j}^{\dagger} \ket{0_{j}^{(s+1)}} \le q_{s,i}+\epsilon_{s}.
	\end{align}
	Since $p_{0,j}^{(s+1)}>0$, division by $p_{0,j}^{(s+1)}$ proves Eq.~\eqref{eq:reference-child-occupation-general}.
	
	Finally, under Eq.~\eqref{eq:inherited-reference-smallness}, Eq.~\eqref{eq:inherited-leading-eigenvalue} gives directly
	\begin{align}
		p_{0,j}^{(s+1)} \ge 1-b_{s+1}q_{s,*}-\epsilon_{s} > \frac{1}{2},
	\end{align}
	which proves Eq.~\eqref{eq:inherited-bootstrap-seed}.
	An eigenvalue strictly larger than $1/2$ is necessarily nondegenerate.
	Moreover, $1/p_{0,j}^{(s+1)}<2$ and $q_{s,i}\le q_{s,*}$, so Eq.~\eqref{eq:reference-child-occupation-general} gives
	\begin{align}
		\bra{0_{j}^{(s+1)}} V_{s+1,j}Q_{i}^{(s)}V_{s+1,j}^{\dagger} \ket{0_{j}^{(s+1)}} \le 2(q_{s,*}+\epsilon_{s}),
	\end{align}
	which proves Eq.~\eqref{eq:reference-child-occupation}.
\end{proof}

Equation~\eqref{eq:inherited-bootstrap-seed} is the bootstrap input needed at the next scale.
Thus, once the step error $\epsilon_{s}$ and $b_{s+1}q_{s,*}$ are sufficiently small, the leading eigenvalue condition required by Proposition~\ref{prop:robustness-fixed-point} is inherited rather than imposed as a new assumption.
Combining this inherited seed with Proposition~\ref{prop:renormalized-robustness} and Proposition~\ref{prop:robustness-fixed-point} then yields the refined bound on $q_{s+1,*}$.
Equation~\eqref{eq:reference-child-occupation} will also be used to control the entanglement hidden inside the decoding map.

\section{Closed MFRG induction for \texorpdfstring{$D/2<\alpha<(D+1)/2$}{D/2 < alpha < (D+1)/2}}
\label{sec:closed-induction}

The one-step estimates must now be combined into a closed induction.
We make two choices for this purpose.
First, at each scale, the error introduced by the auxiliary low-flip cutoff on the complement is chosen in terms of the current number of effective sites, rather than the original microscopic system size $n$.
Second, we stop the RG once the number of effective sites becomes polylogarithmic in $n$.
This keeps the factors accumulated over the RG steps under control while the Kac normalization continues to suppress the effective interactions.
As explained in Section~\ref{sec:cut-entropy}, this is sufficient for the entropy bound; the final decoded state need not approximate the microscopic ground state with inverse-polynomial accuracy in $n$.

\subsection{A regular exact blocking hierarchy and its depth}
\label{sec:regular-blocking-hierarchy}

Throughout this section, assume $D/2<\alpha<(D+1)/2$.
Recall from Eq.~\eqref{eq:preMFRG-beta} that $\beta=1-\alpha/D$, so that $(D-1)/(2D)<\beta<1/2$ in this range.
Also recall from Eq.~\eqref{eq:scale-spaces} that $n^{(s)}=|\Lambda^{(s)}|$ is the number of effective sites at scale $s$.
The auxiliary complement cutoff $m_{0}$ introduced in Eq.~\eqref{choice_of_m_0} will be denoted $m_{\mathrm{c},s}$ when applied at scale $s$.

Recall also that the step $s\to s+1$ groups a block of $b_{s+1}=\ell_{s+1}^{D}$ scale-$s$ sites into one scale-$(s+1)$ effective site, with $n^{(s+1)}=n^{(s)}/b_{s+1}$ as in Eq.~\eqref{eq:effective-dimension-recursion}, and that $z_{s+1}$ denotes the corresponding low-deviation cutoff entering Eq.~\eqref{eq:local-RG-partial-isometry-identities}.

Choose fixed parameters
\begin{align}
	0<\nu<\min\{\beta,4\beta^{2}\}, \qquad p\ge3, \qquad \kappa:=\max\left\{\frac{1}{2},1-2\beta\right\}, \qquad \sigma_{0}:=\beta(1-\kappa)-\frac{\nu}{2}>0. \label{eq:blocking-parameters-Dhalf-alpha-Dplus1half}
\end{align}
The choice of $\nu$ ensures $\sigma_{0}>0$.
Fix the integer cutoff
\begin{align}
	z:=\max\left\{1, \left\lceil\frac{16(p+5)}{\sigma_{0}}\right\rceil\right\}. \label{eq:constant-cutoff-choice}
\end{align}
We use the same cutoff at every RG step, namely $z_{s}:=z$ for all $s\ge1$.
A regular exact blocking hierarchy means that each step is an exact cubic tiling with
\begin{align}
	c_{b}[n^{(s)}]^{\nu}\le b_{s+1}=\ell_{s+1}^{D}\le [n^{(s)}]^{\nu}, \qquad n^{(s+1)}=\frac{n^{(s)}}{b_{s+1}}, \label{eq:regular-exact-blocking-hierarchy}
\end{align}
where $c_{b}>0$ is independent of $n,s$.
In particular,
\begin{align}
	[n^{(s)}]^{1-\nu} \le n^{(s+1)} \le c_{b}^{-1}[n^{(s)}]^{1-\nu}.
\end{align}
For example, if the microscopic side is a power of two, take $\ell_{s+1}=2^{\lfloor\nu\log_{2}([n^{(s)}]^{1/D})\rfloor}$.
Every tiling is then exact and $c_{b}=2^{-D}$ works.
Thus the blocking-hierarchy assumption is realized on dyadic cubic boxes with either boundary condition.
More generally, every exactly divisible hierarchy satisfying Eq.~\eqref{eq:regular-exact-blocking-hierarchy} is covered.

Let $L_{n}:=\log(en)$ and $T_{n}:=L_{n}^{K}$, with a fixed exponent $K$ to be chosen, and stop at the first $s_{\mathrm{fin}}$ for which $n^{(s_{\mathrm{fin}})}\le T_{n}$.
Since $n^{(s_{\mathrm{fin}}-1)}>T_{n}$ and Eq.~\eqref{eq:regular-exact-blocking-hierarchy} gives $n^{(s_{\mathrm{fin}})} \ge[n^{(s_{\mathrm{fin}}-1)}]^{1-\nu}$, we have
\begin{align}
	T_{n}^{1-\nu}<n^{(s_{\mathrm{fin}})}\le T_{n}. \label{eq:stopping-size}
\end{align}
For the depth, the other side of Eq.~\eqref{eq:regular-exact-blocking-hierarchy} gives $\log n^{(s+1)} \le(1-\nu)\log n^{(s)}+\log(c_{b}^{-1})$.
At every active step $s<s_{\mathrm{fin}}$, one has $n^{(s)}>T_{n}$.
Thus, for all sufficiently large $n$, $\log(c_{b}^{-1})\le(\nu/2)\log n^{(s)}$, and hence
\begin{align}
	\log n^{(s+1)} \le\left(1-\frac{\nu}{2}\right)\log n^{(s)}, \qquad \log T_{n} < \log n^{(s_{\mathrm{fin}}-1)} \le \left(1-\frac{\nu}{2}\right)^{s_{\mathrm{fin}}-1}\log n.
\end{align}
Taking logarithms of the last inequality and using $\log T_{n}=K\log\log(en)$ gives
\begin{align}
	s_{\mathrm{fin}}-1 < \frac{ \log \left(\log n/\log T_{n}\right) }{ -\log(1-\nu/2) } \le \frac{\log\log(en)}{-\log(1-\nu/2)}.
\end{align}
Therefore, with the explicit constant $C_{\nu}:=1+[-\log(1-\nu/2)]^{-1}$ depending only on $\nu$,
\begin{align}
	s_{\mathrm{fin}}\le C_{\nu}\log\log(en) \label{eq:stopping-size-and-depth}
\end{align}
for all sufficiently large $n$.

This choice of stopping scale allows any fixed exponential-in-depth and polylogarithmic prefactor to be absorbed into an arbitrarily small power of the current system size.
More precisely, for any fixed $A,C\ge1$, $B\ge0$, and $\varepsilon>0$, one can choose $K$ so large that
\begin{align}
	A C^{s} L_{n}^{B}\le [n^{(s)}]^{\varepsilon} \quad(0\le s\le s_{\mathrm{fin}}) \label{eq:polylog-absorption}
\end{align}
for all sufficiently large $n$.
Indeed, Eq.~\eqref{eq:stopping-size-and-depth} gives $s_{\mathrm{fin}}\le C_{\nu}\log\log(en)$, and hence, for $s\le s_{\mathrm{fin}}$,
\begin{align}
	A C^{s} L_{n}^{B} &\le A C^{s_{\mathrm{fin}}}L_{n}^{B} \le L_{n}^{B+C_{\nu}\log C+1}, \label{eq:polylog-prefactor-bound}
\end{align}
where the last inequality also uses $A\le L_{n}$ for all sufficiently large $n$.
On the other hand, $n^{(s)}\ge n^{(s_{\mathrm{fin}})}$ for $s\le s_{\mathrm{fin}}$, and Eq.~\eqref{eq:stopping-size} gives
\begin{align}
	[n^{(s)}]^{\varepsilon} \ge [n^{(s_{\mathrm{fin}})}]^{\varepsilon} > L_{n}^{K(1-\nu)\varepsilon}.
\end{align}
Thus Eq.~\eqref{eq:polylog-absorption} follows whenever
\begin{align}
	K(1-\nu)\varepsilon > B+C_{\nu}\log C+1.
\end{align}
Only finitely many such absorption bounds are needed below.
Therefore, after the cutoff $z$ and all exponential-in-$s$ constants have been fixed, a single fixed $K$ can be chosen large enough to satisfy all of them.

\subsection{The on-site recurrence with the stronger conditional bound}

We now apply the block estimates proved above at scale $s$.
For $D/2<\alpha<(D+1)/2$, Eq.~\eqref{eq:J-recurrence-Dhalf-alpha-Dplus1half} gives $\mu_{s}=1$.
Hence, by the definition Eq.~\eqref{eq:preMFRG-internal-gamma} and the lattice counting bound Eq.~\eqref{eq:preMFRG-lattice-counting},
\begin{align}
	\bar{\gamma}_{L}^{(s)} &= \max_{i\in L} \left( \sum_{j\in L\setminus\{i\}}r_{ij}^{-2\alpha} \right)^{1/2} \le C_{D,\alpha},
\end{align}
uniformly in the current block $L$, since $2\alpha>D$.
Likewise, Eqs.~\eqref{eq:scale-kac-row-sum} and \eqref{eq:preMFRG-lattice-counting}, together with $\beta=1-\alpha/D$, give $\gamma_{0}^{(s)}\le C_{D,\alpha}$.

Provided $q_{s,*}\le1/2$, the scale-$s$ specialization of Eq.~\eqref{eq:dimension-free-onsite-simplification} gives
\begin{align}
	\delta_{*}^{(s)} \le (2+17\sqrt{2})\bar{g}^{(s)}\sqrt{q_{s,*}}. \label{eq:scale-onsite-improvement}
\end{align}
Applying Eq.~\eqref{Prop/main_ineq:renormalized_interaction_onsite} to a parent block $L$ of size $|L|=b_{s+1}$, with the fixed cutoff $z$ from Eq.~\eqref{eq:constant-cutoff-choice}, therefore gives
\begin{align}
	\norm{h_{j}^{(s+1)}} \le \frac{8C_{D,\alpha}J^{(s)}z\sqrt{b_{s+1}}}{[n^{(s)}]^{\beta}} +2z\bar{g}^{(s)} +2(2+17\sqrt{2})\bar{g}^{(s)} \sqrt{z b_{s+1}q_{s,*}}.
\end{align}
For the pair interactions at scale $s+1$, Eq.~\eqref{all_to_all_cond_s_th_Re} gives $\norm{h_{ij}^{(s+1)}}\le J_{ij}^{(s+1)}$.
Hence, using Eqs.~\eqref{RG_interaction_strength_upper_bound} and \eqref{eq:scale-kac-row-sum},
\begin{align}
	\sum_{j\ne i}\norm{h_{ij}^{(s+1)}} \le \sum_{j\ne i}J_{ij}^{(s+1)} \le J^{(s+1)}\gamma_{0}^{(s+1)} \le C_{D,\alpha}J^{(s+1)}.
\end{align}
Combining this interaction row-sum bound with the preceding estimate for the new on-site term gives
\begin{align}
	\bar{g}^{(s+1)} \le \frac{8C_{D,\alpha}J^{(s)}z\sqrt{b_{s+1}}}{[n^{(s)}]^{\beta}} +2z\bar{g}^{(s)} +2(2+17\sqrt{2})\bar{g}^{(s)} \sqrt{z b_{s+1}q_{s,*}} +C_{D,\alpha}J^{(s+1)}. \label{eq:gbar-recurrence}
\end{align}
If $b_{s+1}q_{s,*}\le1/16$, then the third term is at most $(1+17/\sqrt{2})\sqrt{z}\,\bar{g}^{(s)}$.
Moreover, Eq.~\eqref{eq:regular-exact-blocking-hierarchy} gives $b_{s+1}\le[n^{(s)}]^{\nu}$, while Eq.~\eqref{eq:blocking-parameters-Dhalf-alpha-Dplus1half} gives $\nu<\beta$, and hence
\begin{align}
	\frac{\sqrt{b_{s+1}}}{[n^{(s)}]^{\beta}} \le [n^{(s)}]^{\nu/2-\beta} \le1.
\end{align}
Therefore the first term in Eq.~\eqref{eq:gbar-recurrence} is at most $8C_{D,\alpha}zJ^{(s)}$.
Since $z$ is fixed and Eq.~\eqref{eq:J-recurrence-Dhalf-alpha-Dplus1half} gives $J^{(s+1)}=c_{\rm int}zJ^{(s)}$, the recurrence reduces to
\begin{align}
	\bar{g}^{(s+1)} \le C_{g}\bar{g}^{(s)}+C_{g}J^{(s)} \label{eq:gbar-linear-recurrence}
\end{align}
with
\begin{align}
	C_{g}:=\max\left\{1,\,2z+\left(1+\frac{17}{\sqrt{2}}\right)\sqrt{z},\,C_{D,\alpha}z(8+c_{\rm int})\right\}. \label{eq:onsite-recurrence-constant}
\end{align}
Thus the recurrence preserves an exponential-in-$s$ bound, with no additional positive power of $b_{s+1}$ or $n^{(s)}$, as quantified by the auxiliary envelope in the closed induction below.

\subsection{Uniform shell hopping at a current scale}

Fix an active scale $s$ and suppose that the following inductive bounds hold with $C_{1}>0$ carrying energy units and dimensionless $C_{2},C_{3},C_{4}\ge1$:
\begin{align}
	J^{(s)}+\bar{g}^{(s)} &\le C_{1}C_{2}^{s}, \qquad q_{s,*} \le C_{3}C_{4}^{s}[n^{(s)}]^{-2\beta}, \qquad \Delta^{(s)} \ge\Delta/2. \label{eq:inductive-size-envelopes}
\end{align}
We show that these bounds imply a uniform low-deviation tail for every parent block at scale $s$.

For a parent block $L$, choose the auxiliary complement tolerance
\begin{align}
	\epsilon_{\mathrm{c},s} := [n^{(s)}]^{-(p+6)}. \label{eq:auxiliary-complement-tolerance}
\end{align}
Applying the cutoff estimate Eq.~\eqref{eq:preMFRG-cutoff-scaling} at scale $s$, with $G_{s}=2\bar{g}^{(s)}n^{(s)}$, gives
\begin{align}
	m_{\mathrm{c},s} &\le 2n^{(s)}q_{s,*} + \sqrt{\frac{2\bar{g}^{(s)}n^{(s)}}{\Delta^{(s)}}} \log\frac{2}{\epsilon_{\mathrm{c},s}} +2 \le C_{5}C_{6}^{s}L_{n}[n^{(s)}]^{\kappa}, \label{eq:scale-global-cutoff}
\end{align}
with the cutoff capped by the complement size, where
\begin{align}
	C_{5}:=2C_{3}+2(p+7)\sqrt{\frac{C_{1}}{\Delta}}+2, \qquad C_{6}:=\max\left\{1,C_{4},\sqrt{C_{2}}\right\}. \label{eq:cutoff-envelope-constants}
\end{align}
Indeed, Eq.~\eqref{eq:inductive-size-envelopes} bounds the first contribution by $2C_{3}C_{4}^{s}[n^{(s)}]^{1-2\beta}$ and the second by $2(p+7)\sqrt{C_{1}/\Delta}(\sqrt{C_{2}})^{s}L_{n}[n^{(s)}]^{1/2}$.
Here we used $\log(2/\epsilon_{\mathrm{c},s})\le(p+7)L_{n}$ and $\Delta^{(s)}\ge\Delta/2$.
Thus $\kappa=\max\{1/2,1-2\beta\}$ from Eq.~\eqref{eq:blocking-parameters-Dhalf-alpha-Dplus1half} controls both terms.

We next apply Lemma~\ref{lem:preMFRG-compression} to this auxiliary complement cutoff.
Denoting its energy defect by $e_{\mathrm{c},s}$, the compression lemma gives
\begin{align}
	\bar{\Delta}_{\mathrm{c},s} &\ge \Delta^{(s)}-e_{\mathrm{c},s}, \qquad e_{\mathrm{c},s} \le \frac{G_{s}\epsilon_{\mathrm{c},s}^{2}}{1-\epsilon_{\mathrm{c},s}^{2}},
\end{align}
with $G_{s}=2\bar{g}^{(s)}n^{(s)}$.
Using Eq.~\eqref{eq:inductive-size-envelopes} and $\epsilon_{\mathrm{c},s}=[n^{(s)}]^{-(p+6)}$, we obtain
\begin{align}
	e_{\mathrm{c},s} \le \frac{ 2C_{1}C_{2}^{s}[n^{(s)}]^{-2p-11} }{ 1-[n^{(s)}]^{-2p-12} }.
\end{align}
Applying Eq.~\eqref{eq:polylog-absorption} with $A=\max\{1,C_{1}/\Delta\}$, $C=C_{2}$, $B=0$, and $\varepsilon=1/2$ gives $(C_{1}/\Delta)C_{2}^{s}\le[n^{(s)}]^{1/2}$ for all sufficiently large $n$.
The preceding energy-defect bound then implies $e_{\mathrm{c},s}\le\Delta/4\le\Delta^{(s)}/2$, and hence
\begin{align}
	\bar{\Delta}_{\mathrm{c},s} \ge \frac{\Delta^{(s)}}{2}.
\end{align}
Moreover, Eq.~\eqref{eq:preMFRG-complement-error} gives
\begin{align}
	\eta_{\mathrm{c},s} \le \sqrt{ \frac{2G_{s}\epsilon_{\mathrm{c},s}^{2}}{\Delta^{(s)}(1-\epsilon_{\mathrm{c},s}^{2})} } \le [n^{(s)}]^{-(p+5)} \label{eq:scale-complement-error}
\end{align}
for all sufficiently large $n$.

The same inductive occupation bound also controls the starting shell.
Equation~\eqref{eq:regular-exact-blocking-hierarchy} gives $b_{s+1}\le[n^{(s)}]^{\nu}$, and hence
\begin{align}
	b_{s+1}q_{s,*} \le C_{3}C_{4}^{s}[n^{(s)}]^{-(2\beta-\nu)}.
\end{align}
Since $2\beta-\nu>0$, Eq.~\eqref{eq:polylog-absorption} allows $K$ to be chosen so that $b_{s+1}q_{s,*}\le1/16$ throughout the active hierarchy for all sufficiently large $n$.
Together with Eq.~\eqref{eq:scale-complement-error}, this gives $\sqrt{b_{s+1}q_{s,*}}+\eta_{\mathrm{c},s}<1/\sqrt{2}$.
Therefore Eq.~\eqref{eq:preMFRG-zero-median-criterion} implies that the median shell of the auxiliary ground state is
\begin{align}
	m_{*}=0.
\end{align}

It remains to control hopping across the finite shell interval $0\le x\le z$.
Apply the scale-$s$ version of Eq.~\eqref{Parameter_bar_J_m_0/L} with $m_{0}=m_{\mathrm{c},s}$ and $|L|=b_{s+1}$.
For $D/2<\alpha<(D+1)/2$, Eq.~\eqref{eq:J-recurrence-Dhalf-alpha-Dplus1half} gives $\mu_{s}=1$, so the excess Moore-neighborhood term vanishes.
Moreover, Eq.~\eqref{eq:F-alpha-definition} gives $F_{\alpha}(n^{(s)},m_{\mathrm{c},s})=m_{\mathrm{c},s}^{\beta}$.
Using Eqs.~\eqref{eq:regular-exact-blocking-hierarchy}, \eqref{eq:scale-onsite-improvement}, and \eqref{eq:scale-global-cutoff}, the remaining terms obey
\begin{align}
	\bar{J}_{m_{\mathrm{c},s},L}^{(s)}[0,z] \le C_{\mathrm{hop}}c_{\mathrm{hop}}^{s}L_{n}^{B_{\mathrm{hop}}} \left( [n^{(s)}]^{\nu/2-\beta} + [n^{(s)}]^{-\beta+\beta\kappa+\nu/2} \right), \label{eq:three-hopping-exponents}
\end{align}
with the explicit constants
\begin{align}
	C_{\mathrm{hop}} &:=16(z+2)C_{D,\alpha}C_{1}+4(2+17\sqrt{2})\sqrt{z+2}\,C_{1}\sqrt{C_{3}}+32C_{\mathrm{occ},1}C_{\mathrm{occ},2}\sqrt{z+2}\,C_{1}C_{5}^{\beta}, \label{eq:shell-hopping-amplitude-constant}\\
	c_{\mathrm{hop}} &:=\max\left\{1,C_{2},C_{2}\sqrt{C_{4}},C_{2}C_{6}^{\beta}\right\}, \label{eq:shell-hopping-growth-constant} \\
	B_{\mathrm{hop}} &:=1. \label{eq:shell-hopping-constants}
\end{align}
The three contributions are respectively the internal interaction, the on-site off-diagonal term, and the power-law crossing term in Eq.~\eqref{Parameter_bar_J_m_0/L}.
For the last term, use $m_{\mathrm{c},s}^{\beta}\le C_{5}^{\beta}C_{6}^{\beta s}L_{n}^{\beta}[n^{(s)}]^{\beta\kappa}$ and $\sqrt{\log(en^{(s)})}\le L_{n}^{1/2}$.
Since $\beta<1/2$, the resulting logarithmic factor satisfies $L_{n}^{\beta+1/2}\le L_{n}$, which verifies $B_{\mathrm{hop}}=1$.
By the definition $\sigma_{0}=\beta(1-\kappa)-\nu/2$ in Eq.~\eqref{eq:blocking-parameters-Dhalf-alpha-Dplus1half}, the second power is exactly $[n^{(s)}]^{-\sigma_{0}}$, while $\beta-\nu/2=\sigma_{0}+\beta\kappa\ge\sigma_{0}$, so the first term decays at least as fast.
Since $\bar{\Delta}_{\mathrm{c},s}\ge\Delta^{(s)}/2\ge\Delta/4$, Eq.~\eqref{eq:polylog-absorption}, with $A=\max\{1,8C_{\mathrm{hop}}/\Delta\}$, $C=c_{\mathrm{hop}}$, $B=B_{\mathrm{hop}}$, and absorption exponent $\sigma_{0}/2$, therefore gives
\begin{align}
	\frac{ \bar{J}_{m_{\mathrm{c},s},L}^{(s)}[0,z] }{ \bar{\Delta}_{\mathrm{c},s} } \le [n^{(s)}]^{-\sigma_{0}/2} <1. \label{eq:uniform-small-hopping}
\end{align}

Since $m_{*}=0$, Eq.~\eqref{eq:scale-covariant-small-hopping-tail} now gives for the auxiliary ground state
\begin{align}
	\norm{ \Pi_{L,>z}^{(s)} \ket{\widehat{\bar{\Omega}}^{(s)}} } \le [n^{(s)}]^{-\sigma_{0} z/16} \le [n^{(s)}]^{-(p+5)},
\end{align}
where the last inequality follows from Eq.~\eqref{eq:constant-cutoff-choice}.
Finally, Eq.~\eqref{eq:tail-after-global-truncation}, together with Eq.~\eqref{eq:scale-complement-error}, yields the desired tail for the actual scale-$s$ ground state:
\begin{align}
	\norm{ \Pi_{L,>z}^{(s)} \ket{\Omega^{(s)}} } \le 2[n^{(s)}]^{-(p+5)} =:\tau_{s}. \label{eq:uniform-local-block-tail}
\end{align}
Thus the inductive bounds in Eq.~\eqref{eq:inductive-size-envelopes} imply a uniform low-deviation tail on every parent block, with the small-hopping condition verified over the entire shell interval $0\le x\le z$.

\subsection{Global errors and the closed induction}

Write $E_{0}^{(s)}:=E_{0}(H^{(s)})$ for the ground energy at scale $s$.
The global retained projection for the step $s\to s+1$ is $\mathsf{P}_{s\rightarrow s+1}$ from Eq.~\eqref{eq:Pi-s-plus-one}.
Define its discarded ground-state weight and the corresponding energy-weighted compression error by
\begin{align}
	w_{s} := \norm{ (\mathds{1}-\mathsf{P}_{s\rightarrow s+1}) \ket{\Omega^{(s)}} }^{2}, \qquad e_{s} := \frac{ \bra{\Omega^{(s)}} (\mathds{1}-\mathsf{P}_{s\rightarrow s+1}) (H^{(s)}-E_{0}^{(s)}) (\mathds{1}-\mathsf{P}_{s\rightarrow s+1}) \ket{\Omega^{(s)}} }{ 1-w_{s} }. \label{eq:step-energy-error}
\end{align}
There are $n^{(s)}/b_{s+1}$ parent blocks at this step.
Since Eq.~\eqref{eq:uniform-local-block-tail} gives a discarded norm at most $\tau_{s}=2[n^{(s)}]^{-(p+5)}$ on each block, the union bound Eq.~\eqref{eq:preMFRG-block-union} yields
\begin{align}
	w_{s} &\le \frac{n^{(s)}}{b_{s+1}}\tau_{s}^{2} \le 4[n^{(s)}]^{-2p-9}. \label{eq:global-discarded-weight}
\end{align}
Moreover, Eq.~\eqref{all_to_all_cond_s_th_Re} implies $\norm{H^{(s)}-E_{0}^{(s)}}\le2\bar{g}^{(s)}n^{(s)}$.
Hence, once $w_{s}\le1/2$,
\begin{align}
	e_{s} &\le \frac{2\bar{g}^{(s)}n^{(s)}w_{s}}{1-w_{s}} \le 16\bar{g}^{(s)}[n^{(s)}]^{-2p-8}. \label{eq:global-step-leakage}
\end{align}
Applying Lemma~\ref{lem:preMFRG-compression} to the retained space then gives
\begin{align}
	\Delta^{(s+1)} &\ge \Delta^{(s)}-e_{s}, \qquad \epsilon_{s} := \norm{ \ket{\Omega^{(s)}}- V_{s+1}^{\dagger}\ket{\Omega^{(s+1)}} } \le \sqrt{\frac{2e_{s}}{\Delta^{(s)}}}. \label{eq:step-gap-and-error}
\end{align}
Under the inductive bounds below, Eq.~\eqref{eq:polylog-absorption} makes the last quantity at most $[n^{(s)}]^{-p}$ for all sufficiently large $n$.
The scalar centering used in the definition of $H^{(s+1)}$ does not affect either the gap or the ground state.

\begin{prop}
	\label{prop:closed-induction-Dhalf-alpha-Dplus1half}
	Assume the microscopic hypotheses and $D/2<\alpha<(D+1)/2$.
	Fix $\nu,p,z$ as in Eqs.~\eqref{eq:blocking-parameters-Dhalf-alpha-Dplus1half} and \eqref{eq:constant-cutoff-choice}, and use a regular exact blocking hierarchy satisfying Eq.~\eqref{eq:regular-exact-blocking-hierarchy}.
	With $C_{\mathrm{ind}}$ and $K$ defined in Eqs.~\eqref{eq:closed-induction-constant} and \eqref{eq:explicit-stopping-exponent}, there exists $n_{0}$ such that, for $n\ge n_{0}$, the MFRG construction is well defined through $s_{\mathrm{fin}}$ and the following bounds hold at every defined scale.
	For the interaction strength $J^{(s)}$ in Eq.~\eqref{RG_interaction_strength_upper_bound}, the local energy scale $\bar{g}^{(s)}$ in Eq.~\eqref{all_to_all_cond_s_th_Re}, and the gap $\Delta^{(s)}$ of $H^{(s)}$, one has
	\begin{align}
		J^{(s)} &= c_{\rm int}^{s}J_{\rm K}z^{s}, \qquad \bar{g}^{(s)} \le \Delta C_{\mathrm{ind}}^{s+1}, \qquad \Delta^{(s)} \ge \frac{\Delta}{2}. \label{eq:closed-induction-energy}
	\end{align}
	For the deviation probabilities defined in Eq.~\eqref{eq:scale-s-reference-state} and the effective local dimension in Eq.~\eqref{eq:scale-spaces},
	\begin{align}
		q_{s,*} &\le C_{\mathrm{ind}}^{s+1}[n^{(s)}]^{-2\beta}, \qquad p_{0,i}^{(s)} > \frac{1}{2}, \qquad \log d^{(s)} \le C_{\mathrm{ind}}^{s+1}\log(en). \label{eq:closed-induction-reference}
	\end{align}
	The one-step error $\epsilon_{s}$ of Eq.~\eqref{eq:step-state-error} and the cumulative map $V_{0\rightarrow s}$ of Eq.~\eqref{eq:cumulative-coisometry} satisfy
	\begin{align}
		\epsilon_{s} &\le [n^{(s)}]^{-p} \quad(0\le s<s_{\mathrm{fin}}), \qquad \norm{ \ket{\Omega} - V_{0\rightarrow s}^{\dagger} \ket{\Omega^{(s)}} } \le \sum_{t=0}^{s-1}\epsilon_{t}. \label{eq:closed-induction-error}
	\end{align}
	In addition, for every step $s\to s+1$, the deviation weight of each child $i$ inside a scale-$(s+1)$ reference vector satisfies
	\begin{align}
		\bra{0_{j}^{(s+1)}} V_{s+1,j}Q_{i}^{(s)}V_{s+1,j}^{\dagger} \ket{0_{j}^{(s+1)}} \le 2(q_{s,*}+\epsilon_{s}) \le 4C_{\mathrm{ind}}^{s+1}[n^{(s)}]^{-2\beta}. \label{eq:uniform-reference-child-bound}
	\end{align}
\end{prop}

\begin{proof}
	At scale $0$, Eq.~\eqref{eq:microscopic-deviation-constant} gives $q_{0,*}\le C_{\rm mic}n^{-2\beta}$ with the explicit microscopic constant defined there.
	Hence, for all sufficiently large $n$, $q_{0,i}\le q_{0,*}<1/2$ and therefore $p_{0,i}^{(0)}=1-q_{0,i}>1/2$.
	The interaction bound is initialized with $J^{(0)}=J_{\rm K}:=\mu_{0}J$ and $\mu_{ij}^{(0)}=1$ as in the parametrization preceding Eq.~\eqref{eq:J-recurrence-Dhalf-alpha-Dplus1half}.
	Using the admissible row-sum choice of the local energy scale gives $\bar{g}^{(0)}=g_{0}+J_{\rm K}\gamma_{0}^{(0)}\le g_{0}+\mu_{0}J\gamma_{0}\le\mu_{0}\bar{g}$, while $\Delta^{(0)}=\Delta$.
	
	To keep the choice of constants noncircular, we first fix an auxiliary envelope for the local energy scale independently of the common induction constant below.
	Recall that $z$ is the fixed cutoff chosen in Eq.~\eqref{eq:constant-cutoff-choice}, $c_{\rm int}$ is the fixed interaction-renormalization constant chosen in Eq.~\eqref{eq:c-int-choice}, and $C_{\rm rob},c_{\rm rob}$ are the fixed constants appearing in Eq.~\eqref{eq:As-robustness-coefficient}.
	Let $C_{g}$ be the fixed constant in Eq.~\eqref{eq:gbar-linear-recurrence}.
	Define fixed dimensionless constants
	\begin{align}
		A_{g}:=\max\left\{1,\frac{\mu_{0}\bar{g}}{\Delta},\frac{J_{\rm K}}{\Delta}\right\}, \qquad R_{g}:=\max\{2C_{g},c_{\rm int}z,1\}. \label{eq:closed-induction-auxiliary-constants}
	\end{align}
	Also define
	\begin{align}
		Q:=\max\left\{1,C_{\rm mic},2C_{\rm rob}^{2}\left(\frac{J_{\rm K}}{\Delta}\right)^{2}A_{g}\right\}, \qquad R_{q}:=\max\left\{1,(c_{\rm rob}z)^{2}R_{g}\right\}. \label{eq:closed-induction-probability-constants}
	\end{align}
	Define the common induction constant, independently of $K$ and $n$, by
	\begin{align}
		C_{\mathrm{ind}}:=\max\{2,A_{g},R_{g},Q,R_{q},2z,\log d\}. \label{eq:closed-induction-constant}
	\end{align}
	All entries in this maximum are dimensionless, and the same choice also covers the fixed scale-zero dimension bound.
	For the shell-hopping estimates above, fix the input envelopes by
	\begin{align}
		C_{1}:=2\Delta C_{\mathrm{ind}}, \qquad C_{2}:=C_{\mathrm{ind}}, \qquad C_{3}:=C_{\mathrm{ind}}, \qquad C_{4}:=C_{\mathrm{ind}}. \label{eq:closed-induction-envelope-choice}
	\end{align}
	Whenever the inductive bounds hold, $J^{(s)}\le J_{\rm K}(c_{\rm int}z)^{s}\le\Delta C_{\mathrm{ind}}^{s+1}$ and $\bar{g}^{(s)}\le\Delta C_{\mathrm{ind}}^{s+1}$, so these envelopes are admissible in Eq.~\eqref{eq:inductive-size-envelopes}.
	Equations~\eqref{eq:cutoff-envelope-constants}, \eqref{eq:shell-hopping-growth-constant}, and \eqref{eq:shell-hopping-constants} then give $C_{6}=C_{\mathrm{ind}}$, $c_{\mathrm{hop}}=C_{\mathrm{ind}}^{3/2}$, and $B_{\mathrm{hop}}=1$, since $\beta<1/2$.
	Now fix the stopping exponent by
	\begin{align}
		K:=1+\max\left\{\frac{2(C_{\nu}\log C_{\mathrm{ind}}+1)}{(1-\nu)(2\beta-\nu)},\,\frac{2(B_{\mathrm{hop}}+C_{\nu}\log c_{\mathrm{hop}}+1)}{(1-\nu)\sigma_{0}},\,\frac{C_{\nu}\log C_{\mathrm{ind}}+1}{1-\nu}\right\}. \label{eq:explicit-stopping-exponent}
	\end{align}
	The three entries correspond to occupation smallness, shell hopping, and the local-energy absorption, respectively.
	Moreover, $(2\beta-\nu)/2<1/2$ implies that the first entry also suffices to absorb any fixed multiple of $C_{2}^{s}$ into $[n^{(s)}]^{1/2}$, as required for the auxiliary-compression error.
	All fixed dimensionless amplitudes, including $C_{1}/\Delta$ and $C_{\mathrm{hop}}/\Delta$, are covered by the fixed-amplitude part of Eq.~\eqref{eq:polylog-absorption} for sufficiently large $n$.
	Then Eqs.~\eqref{eq:closed-induction-energy} and \eqref{eq:closed-induction-reference} hold at scale $0$.
	We shall also prove simultaneously by induction the auxiliary bound
	\begin{align}
		\bar{g}^{(s)}\le\Delta A_{g}R_{g}^{s}.
	\end{align}
	
	Suppose now that Eqs.~\eqref{eq:closed-induction-energy} and \eqref{eq:closed-induction-reference}, together with the auxiliary bound $\bar{g}^{(s)}\le\Delta A_{g}R_{g}^{s}$, hold through scale $s<s_{\mathrm{fin}}$, and that Eq.~\eqref{eq:closed-induction-error} holds through the preceding steps.
	Since the cutoff is fixed, Eq.~\eqref{eq:J-recurrence-Dhalf-alpha-Dplus1half} gives directly
	\begin{align}
		J^{(s+1)}=c_{\rm int}zJ^{(s)}=c_{\rm int}^{s+1}J_{\rm K}z^{s+1}.
	\end{align}
	Moreover, Eq.~\eqref{eq:regular-exact-blocking-hierarchy} and Eq.~\eqref{eq:closed-induction-reference} give
	\begin{align}
		b_{s+1}q_{s,*}\le C_{\mathrm{ind}}^{s+1}[n^{(s)}]^{-(2\beta-\nu)}.
	\end{align}
	Because $2\beta-\nu>0$, applying Eq.~\eqref{eq:polylog-absorption} with $\varepsilon=(2\beta-\nu)/2$ gives $b_{s+1}q_{s,*}\le[n^{(s)}]^{-(2\beta-\nu)/2}$ throughout the active hierarchy with the choice of $K$ in Eq.~\eqref{eq:explicit-stopping-exponent}.
	Since $n^{(s)}>T_{n}$ for $s<s_{\mathrm{fin}}$, this is at most $1/16$ for all sufficiently large $n$.
	
	Using Eq.~\eqref{eq:gbar-linear-recurrence}, $\bar{g}^{(s+1)}\le C_{g}\bar{g}^{(s)}+C_{g}J^{(s)}$.
	By the auxiliary induction hypothesis, $J^{(s)}=J_{\rm K}(c_{\rm int}z)^{s}$, $J_{\rm K}\le\Delta A_{g}$, $c_{\rm int}z\le R_{g}$, and $2C_{g}\le R_{g}$, and therefore
	\begin{align}
		\bar{g}^{(s+1)}\le\Delta C_{g}A_{g}R_{g}^{s}+C_{g}J_{\rm K}(c_{\rm int}z)^{s}\le2\Delta C_{g}A_{g}R_{g}^{s}\le\Delta A_{g}R_{g}^{s+1}\le\Delta C_{\mathrm{ind}}^{s+2}.
	\end{align}
	Thus both the auxiliary bound $\bar{g}^{(s+1)}\le\Delta A_{g}R_{g}^{s+1}$ and the stated local-energy bound in Eq.~\eqref{eq:closed-induction-energy} close at scale $s+1$.
	
	The uniform shell-hopping estimates Eqs.~\eqref{eq:uniform-small-hopping} and \eqref{eq:uniform-local-block-tail} therefore apply at scale $s$.
	By Eq.~\eqref{eq:global-step-leakage}, $e_{s}\le16\bar{g}^{(s)}[n^{(s)}]^{-2p-8}$.
	Using the inductive bound $\bar{g}^{(s)}\le\Delta C_{\mathrm{ind}}^{s+1}$ and applying Eq.~\eqref{eq:polylog-absorption} with $A=C_{\mathrm{ind}}$, $C=C_{\mathrm{ind}}$, $B=0$, and $\varepsilon=1$, the choice in Eq.~\eqref{eq:explicit-stopping-exponent} gives $\bar{g}^{(s)}\le\Delta C_{\mathrm{ind}}^{s+1}\le\Delta n^{(s)}$ throughout the active hierarchy for all sufficiently large $n$.
	Since $n^{(s)}>T_{n}$ for $s<s_{\mathrm{fin}}$,
	\begin{align}
		e_{s}\le16\Delta[n^{(s)}]^{-2p-7}\le16\Delta T_{n}^{-2p-7}.
	\end{align}
	Consequently, for every $s<s_{\mathrm{fin}}$,
	\begin{align}
		\sum_{t=0}^{s}e_{t}\le16\Delta(s+1)T_{n}^{-2p-7}\le16\Delta C_{\nu}\log\log(en)\,T_{n}^{-2p-7}\le\frac{\Delta}{2},
	\end{align}
	where Eq.~\eqref{eq:stopping-size-and-depth} was used in the second inequality, and the last inequality holds for all sufficiently large $n$.
	Applying Lemma~\ref{lem:preMFRG-compression} at each physical RG step therefore gives
	\begin{align}
		\Delta^{(s+1)}\ge\Delta-\sum_{t=0}^{s}e_{t}\ge\frac{\Delta}{2}.
	\end{align}
	Thus the physical RG gap is controlled by the accumulated additive compression losses rather than by a multiplicative loss at each step.
	
	By Eq.~\eqref{eq:step-gap-and-error}, $\epsilon_{s}\le\sqrt{2e_{s}/\Delta^{(s)}}$.
	Using the preceding bound $e_{s}\le16\Delta[n^{(s)}]^{-2p-7}$ and $\Delta^{(s)}\ge\Delta/2$, we obtain
	\begin{align}
		\epsilon_{s}\le8[n^{(s)}]^{-p-7/2}\le[n^{(s)}]^{-p}
	\end{align}
	for all sufficiently large $n$.
	This proves $\epsilon_{s}\le[n^{(s)}]^{-p}$ in Eq.~\eqref{eq:closed-induction-error}.
	
	We next verify the condition required for Eq.~\eqref{eq:inherited-bootstrap-seed}.
	Increasing $n_{0}$ if necessary gives $\epsilon_{s}\le[n^{(s)}]^{-p}\le1/16$ at the current active step.
	Together with $b_{s+1}q_{s,*}\le1/16$, this implies
	\begin{align}
		b_{s+1}q_{s,*}+\epsilon_{s}\le\frac{1}{8}<\frac{1}{2}.
	\end{align}
	Thus Eq.~\eqref{eq:inherited-reference-smallness} holds, and Lemma~\ref{lem:inherited-reference} gives
	\begin{align}
		p_{0,j}^{(s+1)}>\frac{1}{2}
	\end{align}
	as stated in Eq.~\eqref{eq:inherited-bootstrap-seed}, together with the child-deviation bound Eq.~\eqref{eq:reference-child-occupation}.
	
	Since $p_{0,j}^{(s+1)}>1/2$, we may apply Proposition~\ref{prop:robustness-fixed-point} at scale $s+1$ using the reflection estimate in Eq.~\eqref{eq:As-robustness-coefficient}.
	With $Z_{s+1}=z^{s+1}$, the latter gives, uniformly in the effective site $j$,
	\begin{align}
		\mathcal{A}_{s+1}(j,j^{\mathrm{c}})\le C_{\rm rob}c_{\rm rob}^{s+1}J_{\rm K}z^{s+1}[n^{(s+1)}]^{-\beta}\sqrt{\frac{\bar{g}^{(s+1)}}{\Delta^{(s+1)}}}.
	\end{align}
	Equation~\eqref{eq:robustness-fixed-point-result} therefore implies
	\begin{align}
		q_{s+1,*}&\le\frac{C_{\rm rob}^{2}c_{\rm rob}^{2(s+1)}J_{\rm K}^{2}z^{2(s+1)}\bar{g}^{(s+1)}}{4[\Delta^{(s+1)}]^{3}[n^{(s+1)}]^{2\beta}} \notag\\
		&\le2C_{\rm rob}^{2}\left(\frac{J_{\rm K}}{\Delta}\right)^{2}A_{g}\left[(c_{\rm rob}z)^{2}R_{g}\right]^{s+1}[n^{(s+1)}]^{-2\beta} \notag\\
		&\le QR_{q}^{s+1}[n^{(s+1)}]^{-2\beta} \notag\\
		&\le C_{\mathrm{ind}}^{s+2}[n^{(s+1)}]^{-2\beta}.
	\end{align}
	Here we used the auxiliary bound $\bar{g}^{(s+1)}\le\Delta A_{g}R_{g}^{s+1}$ and $\Delta^{(s+1)}\ge\Delta/2$.
	Thus the reference-state bound in Eq.~\eqref{eq:closed-induction-reference} closes at scale $s+1$.
	
	Equation~\eqref{eq:reference-child-occupation} further gives, for every child $i\in L_{j}^{(s)}$,
	\begin{align}
		\bra{0_{j}^{(s+1)}}V_{s+1,j}Q_{i}^{(s)}V_{s+1,j}^{\dagger}\ket{0_{j}^{(s+1)}}\le2(q_{s,*}+\epsilon_{s}).
	\end{align}
	Here Eq.~\eqref{eq:closed-induction-reference} gives $q_{s,*}\le C_{\mathrm{ind}}^{s+1}[n^{(s)}]^{-2\beta}$, while Eq.~\eqref{eq:closed-induction-error} gives $\epsilon_{s}\le[n^{(s)}]^{-p}$.
	Since $p\ge3>2\beta$, it follows that
	\begin{align}
		2(q_{s,*}+\epsilon_{s})\le 4C_{\mathrm{ind}}^{s+1}[n^{(s)}]^{-2\beta},
	\end{align}
	where we used $C_{\mathrm{ind}}\ge1$ and $[n^{(s)}]^{-p}\le[n^{(s)}]^{-2\beta}$.
	This proves Eq.~\eqref{eq:uniform-reference-child-bound}.
	
	Finally, taking logarithms in Eq.~\eqref{eq:effective-dimension-recursion} and using $z_{s+1}=z$ gives
	\begin{align}
		\log d^{(s+1)}\le z\log d^{(s)}+z\log b_{s+1}.
	\end{align}
	By the induction hypothesis, $\log d^{(s)}\le C_{\mathrm{ind}}^{s+1}\log(en)$, while $b_{s+1}\le n$.
	Therefore
	\begin{align}
		\log d^{(s+1)}\le\left(zC_{\mathrm{ind}}^{s+1}+z\right)\log(en)\le C_{\mathrm{ind}}^{s+2}\log(en),
	\end{align}
	since the fixed constant $C_{\mathrm{ind}}$ was chosen so that $C_{\mathrm{ind}}\ge2z$.
	This completes Eq.~\eqref{eq:closed-induction-reference} at scale $s+1$.
	
	For the cumulative state error, Eq.~\eqref{eq:cumulative-coisometry} and the fact that $V_{0\rightarrow s}^{\dagger}$ is an isometry give
	\begin{align}
		\norm{V_{0\rightarrow s+1}^{\dagger}\ket{\Omega^{(s+1)}}-V_{0\rightarrow s}^{\dagger}\ket{\Omega^{(s)}}}\le\epsilon_{s}.
	\end{align}
	Telescoping over the preceding scales therefore yields
	\begin{align}
		\norm{\ket{\Omega}-V_{0\rightarrow s+1}^{\dagger}\ket{\Omega^{(s+1)}}}\le\sum_{t=0}^{s}\epsilon_{t},
	\end{align}
	which proves the cumulative part of Eq.~\eqref{eq:closed-induction-error}.
	
	The constants have been chosen noncircularly in the following order.
	After fixing the cutoff $z$, first fix the auxiliary local-energy constants $A_{g},R_{g}$ from the linear recurrence Eq.~\eqref{eq:gbar-linear-recurrence}, and then fix $Q,R_{q}$ from the robustness estimate.
	These constants depend only on the fixed microscopic parameters and on $z$, and not on the common induction constant $C_{\mathrm{ind}}$.
	Next define $C_{\mathrm{ind}}$ by Eq.~\eqref{eq:closed-induction-constant}, which dominates $A_{g},R_{g},Q,R_{q}$ and the fixed dimension-recursion constants.
	With these constants fixed, use Eq.~\eqref{eq:closed-induction-envelope-choice} in the shell-hopping estimates and define $K$ by Eq.~\eqref{eq:explicit-stopping-exponent}, which satisfies the finite collection of applications of Eq.~\eqref{eq:polylog-absorption}.
	Finally choose $n_{0}$ sufficiently large so that all fixed numerical thresholds above hold.
	This closes the induction.
\end{proof}

The cumulative RG error need not be inverse-polynomial in the microscopic size $n$, because the hierarchy is stopped once the effective system size becomes polylogarithmic in $n$.
This is sufficient for the entropy bound below, which uses the nested retained spaces and the intermediate effective ground states rather than a single inverse-polynomially accurate final approximation.
The construction is therefore structural and does not by itself provide an efficient ground-state algorithm.

\section{Entropy bounds for geometric cuts}
\label{sec:cut-entropy}

A small final effective space does not by itself control the entanglement across a microscopic cut, since the decoding maps may generate entanglement across that cut.
We therefore use the full RG hierarchy rather than only the final effective state.
We first bound the entropy of states in the nested retained spaces using the child-deviation estimates proved above, and then decompose the exact ground state according to these nested spaces.

\subsection{Crossing effective sites and the geometric condition}

For a microscopic bipartition $A|A^{\mathrm{c}}$, let $\mathcal{C}_{s}(A)$ denote the set of scale-$s$ effective sites whose microscopic descendant regions intersect both sides of the cut:
\begin{align}
	\mathcal{C}_{s}(A) := \left\{ v\in\Lambda^{(s)}: D_{0\leftarrow s}(v)\cap A\ne\varnothing,\ D_{0\leftarrow s}(v)\cap A^{\mathrm{c}}\ne\varnothing \right\}. \label{eq:crossing-block-definition}
\end{align}
We assume that
\begin{align}
	|\mathcal{C}_{s}(A)| \le C_{A}[n^{(s)}]^{\theta}, \qquad 0\le\theta<1, \qquad 1\le s\le s_{\mathrm{fin}}, \label{eq:regular-cut-count}
\end{align}
where $C_{A}$ is independent of $n$.

For the regular blocking hierarchy, the microscopic descendant region $D_{0\leftarrow s}(v)$ of a scale-$s$ effective site has linear size $\ell_{1:s}$.
Since the full system has linear size $\ell_{\Lambda}$,
\begin{align}
	n^{(s)} = \left( \frac{\ell_{\Lambda}}{\ell_{1:s}} \right)^{D}.
\end{align}
A boundary consisting of a fixed number of regular $(D-1)$-dimensional pieces can therefore intersect only
\begin{align}
	|\mathcal{C}_{s}(A)| \le C_{\rm geo} \left( \frac{\ell_{\Lambda}}{\ell_{1:s}}+1 \right)^{D-1} \le C_{\rm geo}' [n^{(s)}]^{(D-1)/D}
\end{align}
scale-$s$ descendant regions.
Thus, for an ordinary codimension-one regular cut, Eq.~\eqref{eq:regular-cut-count} holds with $\theta = (D-1)/D$.
This includes planar half-system cuts, axis-aligned boxes, and, in one dimension, intervals with a fixed number of endpoints.
By contrast, a highly fragmented cut such as an even--odd partition may have $|\mathcal{C}_{s}(A)|=\Theta(n^{(s)})$ and therefore need not satisfy Eq.~\eqref{eq:regular-cut-count} with any $\theta<1$.

The entropy bound below uses only the crossing-site estimate Eq.~\eqref{eq:regular-cut-count}; no stronger geometric assumption on the cut is required.

\subsection{Two elementary entropy estimates}

Let $\rho$ be a density matrix on a $d'$-dimensional space, and let $P=\ket{0}\bra{0}$ be a rank-one reference projector.
Define
\begin{align}
	q := \operatorname{Tr}\bigl((\mathds{1}-P)\rho\bigr),
\end{align}
so that $q$ is the total weight of $\rho$ outside the reference state $\ket{0}$.
Removing the coherence between $\ket{0}$ and its orthogonal complement can only increase the entropy.
The resulting state has weight $1-q$ on the one-dimensional reference space and weight $q$ on a space of dimension at most $d'-1$.
Hence
\begin{align}
	S(\rho) \le h_{2}(q)+q\log(d'-1), \label{eq:entropy-arbitrary-reference}
\end{align}
where $h_{2}(q):=-q\log q-(1-q)\log(1-q)$ is the binary entropy.
This estimate does not require $\ket{0}$ to be an eigenvector of $\rho$.

For any $y\in(0,1/2]$, positivity of the binary relative entropy gives
\begin{align}
	\hbin(q) \le -\log(1-y) + q\log\frac{1-y}{y}.
\end{align}
Combining this with Eq.~\eqref{eq:entropy-arbitrary-reference}, and using $-\log(1-y)\le2y$ and $(d'-1)(1-y)\le d'$, gives
\begin{align}
	S(\rho) \le 2y+q\log(d'/y). \label{eq:affine-entropy-upper-bound}
\end{align}
If $d'=1$, then $S(\rho)=0$ and the estimate is unnecessary.

The second estimate separates the contribution of a reference vector from its orthogonal complement.

\begin{lemma}
	\label{lem:positive-reference-bound}
	For $X\ge0$, a normalized vector $\ket{0}$, and $Q=\mathds{1}-\ket{0}\bra{0}$,
	\begin{align}
		X \le 2\bra{0}X\ket{0}\,\mathds{1} + 2\norm{X}\,Q. \label{eq:positive-reference-bound}
	\end{align}
\end{lemma}

\begin{proof}
	For $\ket{\psi}=a\ket{0}+\ket{\phi}$ with $\ket{\phi}\perp\ket{0}$,
	\begin{align}
		\bra{\psi} X\ket{\psi} &= \norm{aX^{1/2}\ket{0}+X^{1/2}\ket{\phi}}^{2} \notag\\
		&\le 2|a|^{2}\bra{0}X\ket{0} + 2\norm{X}\norm{\ket{\phi}}^{2} \notag\\
		&\le 2\bra{0}X\ket{0}\norm{\ket{\psi}}^{2} + 2\norm{X}\bra{\psi} Q\ket{\psi}.
	\end{align}
	Since this holds for every $\ket{\psi}$, it proves Eq.~\eqref{eq:positive-reference-bound}.
\end{proof}

\subsection{Entropy of an arbitrary vector in a retained range}

Use the fixed cutoff $z_{s}=z\ge1$ of the closed hierarchy, and fix a scale $0\le t\le s_{\mathrm{fin}}$.
For every $s<t$, Eq.~\eqref{eq:reference-child-occupation} gives, for each child $i$ of a scale-$(s+1)$ reference vector,
\begin{align}
	\bra{0_{j}^{(s+1)}} V_{s+1,j}Q_{i}^{(s)}V_{s+1,j}^{\dagger} \ket{0_{j}^{(s+1)}} \le 2(q_{s,*}+\epsilon_{s}).
\end{align}
This bound concerns only the reference vectors $\ket{0_{j}^{(s+1)}}$.
We do not assume that an arbitrary state $\ket{\psi}\in\operatorname{Ran}\mathsf{P}_{0\rightarrow t}$ has small occupation outside the child reference spaces.
The proof instead uses the small child occupations of the reference vectors through the decoding hierarchy.

\begin{lemma}
	\label{lem:weighted-tree-entropy}
	Fix $0\le t\le s_{\mathrm{fin}}$.
	For $0\le s\le t$, define
	\begin{align}
		y_{s} := \frac{1}{2}[n^{(s)}]^{-2\beta}, \qquad \omega_{s} := \log\left(2d^{(s)}[n^{(s)}]^{2\beta}\right), \qquad W_{t} := \max_{0\le s\le t}\omega_{s}, \qquad B_{t} := (2z)^{t}W_{t}. \label{eq:entropy-tree-weights}
	\end{align}
	Then every normalized pure state $\ket{\psi}\in\operatorname{Ran}\mathsf{P}_{0\rightarrow t}$ satisfies
	\begin{align}
		S_{A}(\ket{\psi}) \le (B_{t}+2y_{t})n^{(t)} + 2\sum_{s=0}^{t-1} \left( 2B_{t}(q_{s,*}+\epsilon_{s})+y_{s} \right) b_{s+1}|\mathcal{C}_{s+1}(A)|. \label{eq:weighted-tree-entropy-bound}
	\end{align}
	The estimate is uniform over the state at scale $t$.
\end{lemma}

\begin{proof}
	For a scale-$s$ node $v$ with $s<t$, let $\operatorname{par}(v)$ denote its unique parent at scale $s+1$.
	Define
	\begin{align}
		\mathcal{S}_{t,s}(A) &:= \left\{ v\in\Lambda^{(s)}: D_{0\leftarrow s}(v)\subseteq A,\ D_{0\leftarrow s+1}(\operatorname{par}(v))\not\subseteq A \right\}, \qquad 0\le s<t, \label{eq:selected-cut-nodes-nonterminal}\\
		\mathcal{S}_{t,t}(A) &:= \left\{ v\in\Lambda^{(t)}: D_{0\leftarrow t}(v)\subseteq A \right\}. \label{eq:selected-cut-nodes}
	\end{align}
	The sets $\mathcal{S}_{t,s}(A)$ contain the maximal hierarchy nodes, up to scale $t$, whose microscopic descendant regions lie entirely in $A$.
	These regions are disjoint and partition $A$:
	\begin{align}
		A = \bigsqcup_{s=0}^{t} \bigsqcup_{v\in\mathcal{S}_{t,s}(A)} D_{0\leftarrow s}(v).
	\end{align}
	For $v\in\mathcal{S}_{t,s}(A)$ with $s<t$, the parent $\operatorname{par}(v)$ intersects both $A$ and $A^{\mathrm{c}}$, and hence belongs to $\mathcal{C}_{s+1}(A)$.
	Since each scale-$(s+1)$ parent has $b_{s+1}$ scale-$s$ children,
	\begin{align}
		|\mathcal{S}_{t,s}(A)| \le b_{s+1}|\mathcal{C}_{s+1}(A)|, \qquad s<t.
	\end{align}
	At scale $t$, we simply have $|\mathcal{S}_{t,t}(A)|\le n^{(t)}$.
	
	Since $\ket{\psi}\in\operatorname{Ran}\mathsf{P}_{0\rightarrow t}$ and $\mathsf{P}_{0\rightarrow t} =V_{0\rightarrow t}^{\dagger} V_{0\rightarrow t}$ by Eq.~\eqref{eq:cumulative-coisometry}, there is a normalized scale-$t$ state $\ket{\chi}$ such that
	\begin{align}
		\ket{\psi} = V_{0\rightarrow t}^{\dagger}\ket{\chi}.
	\end{align}
	Because the cumulative retained ranges are nested, for $v\in\mathcal{S}_{t,s}(A)$ the state $V_{0\rightarrow s}\ket{\psi}$ is normalized and belongs to $\mathcal{H}^{(s)}$.
	Define its reduced state on the scale-$s$ effective site $v$ by
	\begin{align}
		\rho_{v} := \operatorname{Tr}_{\Lambda^{(s)}\setminus\{v\}} \left( V_{0\rightarrow s}\ket{\psi}\bra{\psi} V_{0\rightarrow s}^{\dagger} \right).
	\end{align}
	For $v\in\mathcal{S}_{t,s}(A)$, let $\rho_{D_{0\leftarrow s}(v)}^{\psi}$ denote the reduced state of $\ket{\psi}$ on the microscopic descendant region $D_{0\leftarrow s}(v)$.
	The cumulative decoding map factorizes over the scale-$s$ descendant regions, and $D_{0\leftarrow s}(v)\subseteq A$.
	Hence the decoder associated with $v$ acts entirely inside $A$ and is isometric on the retained space, so
	\begin{align}
		S\left(\rho_{D_{0\leftarrow s}(v)}^{\psi}\right) = S(\rho_{v}).
	\end{align}
	
	Define $q_{v} := \operatorname{Tr}\left( Q_{v}^{(s)}\rho_{v} \right)$.
	Since the regions $D_{0\leftarrow s}(v)$ with $v\in\mathcal{S}_{t,s}(A)$ form a disjoint partition of $A$, subadditivity gives
	\begin{align}
		S_{A}(\ket{\psi}) \le \sum_{s=0}^{t} \sum_{v\in\mathcal{S}_{t,s}(A)} S\left(\rho_{D_{0\leftarrow s}(v)}^{\psi}\right) = \sum_{s=0}^{t} \sum_{v\in\mathcal{S}_{t,s}(A)} S(\rho_{v}).
	\end{align}
	Applying Eq.~\eqref{eq:affine-entropy-upper-bound} to each $\rho_{v}$ with $d'=d^{(s)}$, $q=q_{v}$, and $y=y_{s}$, and using $\log(d^{(s)}/y_{s}) =\log(2d^{(s)}[n^{(s)}]^{2\beta}) =\omega_{s}$, gives
	\begin{align}
		S_{A}(\ket{\psi}) \le 2\sum_{s=0}^{t}\sum_{v\in\mathcal{S}_{t,s}(A)}y_{s} + \sum_{s=0}^{t}\sum_{v\in\mathcal{S}_{t,s}(A)}\omega_{s}q_{v}. \label{eq:antichain-entropy-sum}
	\end{align}
	Using $|\mathcal{S}_{t,s}(A)| \le b_{s+1}|\mathcal{C}_{s+1}(A)|$ for $s<t$ and $|\mathcal{S}_{t,t}(A)|\le n^{(t)}$,
	\begin{align}
		2\sum_{s=0}^{t}\sum_{v\in\mathcal{S}_{t,s}(A)}y_{s} \leq 2\sum_{s=0}^{t-1} y_{s} b_{s+1}|\mathcal{C}_{s+1}(A)| + 2y_{t}n^{(t)}. \label{eq:tree-affine-contribution}
	\end{align}
	
	We next control the remaining occupation-weighted sum $\sum_{s=0}^{t}\sum_{v\in\mathcal{S}_{t,s}(A)}\omega_{s}q_{v}$.
	Associate an operator $K_{u}$ with every hierarchy node $u$ whose descendant branch contains a selected node, and define these operators recursively from lower to higher scales.
	If $u\in\mathcal{S}_{t,s}(A)$ is a selected node, set
	\begin{align}
		K_{u} := \omega_{s}Q_{u}^{(s)}.
	\end{align}
	Otherwise, if $u$ is at scale $s+1$, define
	\begin{align}
		K_{u} := V_{u} \left( \sum_{c:\,\mathrm{child\ of}\,u}K_{c} \right) V_{u}^{\dagger},
	\end{align}
	where $K_{c}:=0$ if the descendant branch of $c$ contains no selected node.
	Thus each selected-node operator is coarse-grained along its ancestors to scale $t$, and contributions from different branches are added when they reach the same parent.
	Consequently,
	\begin{align}
		\sum_{s=0}^{t}\sum_{v\in\mathcal{S}_{t,s}(A)}\omega_{s}q_{v} = \bra{\chi} \left( \sum_{r\in\Lambda^{(t)}}K_{r} \right) \ket{\chi}. \label{eq:tree-root-expectation}
	\end{align}
	
	We recursively choose nonnegative scalar coefficients $A_{v},B_{v}\in\mathbb{R}_{\ge0}$ such that
	\begin{align}
		K_{v} \le A_{v}\mathds{1} + B_{v}Q_{v}.
	\end{align}
	For a selected node $v\in\mathcal{S}_{t,s}(A)$, the definition $K_{v}=\omega_{s}Q_{v}^{(s)}$ gives $A_{v}=0$ and $B_{v}=\omega_{s}$, while for a branch containing no selected node we have $K_{v}=0$ and take $A_{v}=B_{v}=0$.
	Now consider a node $v$ at scale $s+1$ lying above selected nodes.
	Using the recursively constructed bounds for the scale-$s$ children,
	\begin{align}
		\sum_{c:\,\mathrm{child\ of}\,v}K_{c} \le \left(\sum_{c}A_{c}\right)\mathds{1} + \sum_{c}B_{c}Q_{c}^{(s)}.
	\end{align}
	Applying the parent coisometry and using $V_{v}V_{v}^{\dagger}=\mathds{1}$ gives
	\begin{align}
		K_{v} \le \left(\sum_{c}A_{c}\right)\mathds{1} + X_{v}, \qquad X_{v} := V_{v} \left( \sum_{c}B_{c}Q_{c}^{(s)} \right) V_{v}^{\dagger} \ge0.
	\end{align}
	Since, by Eq.~\eqref{eq:scale-deviation-number}, $\sum_{c}B_{c}Q_{c}^{(s)} \le(\max_{c}B_{c})M_{L_{v}}^{(s)}$, and the retained block satisfies $M_{L_{v}}^{(s)}\le z$, we have
	\begin{align}
		\norm{X_{v}} \le z\max_{c}B_{c}. \label{eq:tree-X-norm}
	\end{align}
	Moreover, Eq.~\eqref{eq:reference-child-occupation} gives
	\begin{align}
		\bra{0_{v}^{(s+1)}}X_{v}\ket{0_{v}^{(s+1)}} = \sum_{c}B_{c} \bra{0_{v}^{(s+1)}} V_{v}Q_{c}^{(s)}V_{v}^{\dagger} \ket{0_{v}^{(s+1)}} \le 2(q_{s,*}+\epsilon_{s})\sum_{c}B_{c}. \label{eq:tree-X-reference}
	\end{align}
	Applying Lemma~\ref{lem:positive-reference-bound} to $X_{v}$, together with Eqs.~\eqref{eq:tree-X-norm} and \eqref{eq:tree-X-reference}, we may therefore choose
	\begin{align}
		A_{v} = \sum_{c}A_{c} + 4(q_{s,*}+\epsilon_{s})\sum_{c}B_{c}, \qquad B_{v} = 2z\max_{c}B_{c}. \label{eq:tree-cost-recursion}
	\end{align}
	
	Along any branch there are at most $t$ recursive steps, and the initial coefficients satisfy $B_{v}\le W_{t}$.
	Since each step gives $B_{v}=2z\max_{c}B_{c}$,
	\begin{align}
		B_{v} \le (2z)^{t}W_{t} = B_{t}
	\end{align}
	at every node.
	This gives a uniform bound on the $B$-coefficients throughout the hierarchy.
	We next bound the accumulated $A$-coefficients.
	In Eq.~\eqref{eq:tree-cost-recursion}, the term $\sum_{c}A_{c}$ only carries previously generated scalar contributions to the parent, while $4(q_{s,*}+\epsilon_{s})\sum_{c}B_{c}$ is the new scalar contribution generated at scale $s+1$.
	For each scale $s<t$, the scale-$(s+1)$ nodes that receive contributions from selected nodes belong to $\mathcal{C}_{s+1}(A)$, and each has at most $b_{s+1}$ children.
	Since every child coefficient satisfies $B_{c}\le B_{t}$, the new scalar contributions generated at scale $s+1$ satisfy
	\begin{align}
		\sum_{v\in\mathcal{C}_{s+1}(A)} 4(q_{s,*}+\epsilon_{s}) \sum_{c:\,\mathrm{child\ of}\,v}B_{c} \le 4B_{t}(q_{s,*}+\epsilon_{s}) b_{s+1}|\mathcal{C}_{s+1}(A)|.
	\end{align}
	Since the $\sum_{c}A_{c}$ terms merely transport these scalar contributions upward without increasing them, summing over all scales bounds the total scalar coefficient at the scale-$t$ roots:
	\begin{align}
		\sum_{r\in\Lambda^{(t)}}A_{r} \le 4B_{t} \sum_{s=0}^{t-1} (q_{s,*}+\epsilon_{s}) b_{s+1}|\mathcal{C}_{s+1}(A)|.
	\end{align}
	
	At scale $t$, the recursively constructed operators satisfy $K_{r}\le A_{r}\mathds{1}+B_{r}Q_{r}$ for $r\in\Lambda^{(t)}$.
	Hence Eq.~\eqref{eq:tree-root-expectation} gives
	\begin{align}
		\sum_{s=0}^{t} \sum_{v\in\mathcal{S}_{t,s}(A)} \omega_{s}q_{v} &\le \sum_{r\in\Lambda^{(t)}}A_{r} + \sum_{r\in\Lambda^{(t)}} B_{r}\bra{\chi}Q_{r}\ket{\chi} \notag\\
		&\le 4B_{t} \sum_{s=0}^{t-1} (q_{s,*}+\epsilon_{s}) b_{s+1}|\mathcal{C}_{s+1}(A)| + B_{t}n^{(t)}. \label{eq:tree-occupation-contribution}
	\end{align}
	where the last term follows from $B_{r}\le B_{t}$, $Q_{r}\le\mathds{1}$, and $|\Lambda^{(t)}|=n^{(t)}$.
	
	Substituting Eqs.~\eqref{eq:tree-affine-contribution} and \eqref{eq:tree-occupation-contribution} into Eq.~\eqref{eq:antichain-entropy-sum}, we obtain
	\begin{align}
		S_{A}(\ket{\psi}) &\le 2\sum_{s=0}^{t-1} y_{s} b_{s+1}|\mathcal{C}_{s+1}(A)| + 2y_{t}n^{(t)} + 4B_{t} \sum_{s=0}^{t-1} (q_{s,*}+\epsilon_{s}) b_{s+1}|\mathcal{C}_{s+1}(A)| + B_{t}n^{(t)} \notag\\
		&= (B_{t}+2y_{t})n^{(t)} + 2\sum_{s=0}^{t-1} \left( 2B_{t}(q_{s,*}+\epsilon_{s})+y_{s} \right) b_{s+1}|\mathcal{C}_{s+1}(A)|.
	\end{align}
	This proves Eq.~\eqref{eq:weighted-tree-entropy-bound}.
\end{proof}

We now combine the retained-range entropy bound Eq.~\eqref{eq:weighted-tree-entropy-bound} with the geometric cut condition.
By Eq.~\eqref{eq:regular-cut-count} and $n^{(s+1)}=n^{(s)}/b_{s+1}$,
\begin{align}
	b_{s+1}|\mathcal{C}_{s+1}(A)| \le C_{A}b_{s+1}[n^{(s+1)}]^{\theta} = C_{A}[n^{(s)}]^{\theta} b_{s+1}^{1-\theta} \le C_{A}[n^{(s)}]^{\theta+\nu(1-\theta)}. \label{eq:antichain-count-scaling}
\end{align}
The factor $b_{s+1}^{1-\theta}$ cannot in general be omitted, since a single node in $\mathcal{C}_{s+1}(A)$ may have many children whose descendant regions lie entirely on one side of the cut.

For the entropy estimate, assume that the fixed block exponent $\nu$ was chosen from the outset so that, in addition to Eq.~\eqref{eq:blocking-parameters-Dhalf-alpha-Dplus1half},
\begin{align}
	\lambda_{\rm cut} := 2\beta-\theta-\nu(1-\theta) > 0. \label{eq:cut-exponent-condition}
\end{align}
Such a choice is possible whenever $\theta<2\beta$.
For an ordinary codimension-one regular cut, $\theta=(D-1)/D$, and therefore
\begin{align}
	\lambda_{\rm cut} = \frac{D+1-2\alpha-\nu}{D}.
\end{align}
Hence a positive choice of $\nu$ satisfying Eq.~\eqref{eq:cut-exponent-condition} exists exactly when $\alpha<(D+1)/2$, with the additional restriction $\nu<D+1-2\alpha$.
By Eq.~\eqref{eq:closed-induction-reference}, $q_{s,*}\le C_{\mathrm{ind}}^{s+1}[n^{(s)}]^{-2\beta}$, with the fixed induction constant $C_{\mathrm{ind}}>1$, while Eq.~\eqref{eq:closed-induction-error} gives $\epsilon_{s}\le[n^{(s)}]^{-p}$.
Since $p\ge3>2\beta$, and $y_{s}=\tfrac{1}{2}[n^{(s)}]^{-2\beta}$,
\begin{align}
	2(q_{s,*}+\epsilon_{s})+y_{s} \le \frac{9}{2}C_{\mathrm{ind}}^{s+1}[n^{(s)}]^{-2\beta}.
\end{align}
It therefore follows from Eq.~\eqref{eq:antichain-count-scaling} that
\begin{align}
	\left( 2(q_{s,*}+\epsilon_{s})+y_{s} \right) b_{s+1}|\mathcal{C}_{s+1}(A)| \le \frac{9}{2}C_{A} C_{\mathrm{ind}}^{s+1} [n^{(s)}]^{-\lambda_{\rm cut}}. \label{eq:tree-level-entropy-decay}
\end{align}

By the definition of $\omega_{s}$, Eq.~\eqref{eq:closed-induction-reference}, and $n^{(s)}\le n$,
\begin{align}
	\omega_{s} &= \log 2+\log d^{(s)}+2\beta\log n^{(s)} \le 2C_{\mathrm{ind}}^{s+1}\log(en),
\end{align}
since $2\beta<1$, $\log 2+2\beta\log n^{(s)}\le\log(en)$, and $C_{\mathrm{ind}}\ge2$.
Hence $W_{t}=\max_{0\le s\le t}\omega_{s}\le2C_{\mathrm{ind}}^{t+1}\log(en)$, and since $B_{t}=(2z)^{t}W_{t}$,
\begin{align}
	B_{t} \le 2(2z)^{t}C_{\mathrm{ind}}^{t+1}\log(en) \le C_{B}^{t+1}\log(en) \label{eq:tree-weight-polylog}
\end{align}
with the explicit choice $C_{B}:=2zC_{\mathrm{ind}}>1$, using $z\ge1$.

Since $t\le s_{\mathrm{fin}}\le C_{\nu}\log\log(en)$, Eq.~\eqref{eq:tree-weight-polylog} gives
\begin{align}
	B_{t} &\le C_{B}^{t+1}\log(en) \le C_{B}\exp\left( C_{\nu}\log C_{B}\,\log\log(en) \right)\log(en) = C_{B}[\log(en)]^{1+C_{\nu}\log C_{B}}.
\end{align}
Thus, uniformly for $0\le t\le s_{\mathrm{fin}}$,
\begin{align}
	B_{t}+1 \le (C_{B}+1)[\log(en)]^{\gamma_{B}}, \qquad \gamma_{B} := 1+C_{\nu}\log C_{B}. \label{eq:B-polylog-bound}
\end{align}

Similarly, since $\lambda_{\rm cut}>0$, $n^{(s)}\ge1$, and the fixed constant $C_{\mathrm{ind}}$ in Eq.~\eqref{eq:closed-induction-reference} satisfies $C_{\mathrm{ind}}>1$,
\begin{align}
	\sum_{s=0}^{t-1} C_{\mathrm{ind}}^{s+1}[n^{(s)}]^{-\lambda_{\rm cut}} \le \sum_{s=0}^{t-1}C_{\mathrm{ind}}^{s+1} \le \frac{C_{\mathrm{ind}}}{C_{\mathrm{ind}}-1}C_{\mathrm{ind}}^{t} \le \frac{C_{\mathrm{ind}}}{C_{\mathrm{ind}}-1} \exp\left( C_{\nu}\log C_{\mathrm{ind}}\,\log\log(en) \right) = \frac{C_{\mathrm{ind}}}{C_{\mathrm{ind}}-1} [\log(en)]^{\gamma_{\Sigma}}, \label{eq:scale-sum-polylog-bound}
\end{align}
where $\gamma_{\Sigma}:=C_{\nu}\log C_{\mathrm{ind}}$.
Thus both $B_{t}+1$ and the accumulated scale contribution are bounded uniformly throughout the active hierarchy by fixed powers of $\log(en)$.

Since $y_{t}\le1/2$, $B_{t}+2y_{t}\le B_{t}+1$, and
\begin{align}
	2B_{t}(q_{s,*}+\epsilon_{s})+y_{s} \le (B_{t}+1)\left( 2(q_{s,*}+\epsilon_{s})+y_{s} \right).
\end{align}
Therefore Eqs.~\eqref{eq:weighted-tree-entropy-bound} and \eqref{eq:tree-level-entropy-decay} give
\begin{align}
	S_{A}(\ket{\psi}) \le (B_{t}+1)n^{(t)} + 9C_{A}(B_{t}+1) \sum_{s=0}^{t-1} C_{\mathrm{ind}}^{s+1}[n^{(s)}]^{-\lambda_{\rm cut}}.
\end{align}
Using Eqs.~\eqref{eq:B-polylog-bound} and \eqref{eq:scale-sum-polylog-bound}, we obtain
\begin{align}
	S_{A}(\ket{\psi}) &\le (C_{B}+1)[\log(en)]^{\gamma_{B}} \left( n^{(t)} + \frac{9C_{A}C_{\mathrm{ind}}}{C_{\mathrm{ind}}-1} [\log(en)]^{\gamma_{\Sigma}} \right) \notag\\
	&\le C_{\mathrm{pre}}[\log(en)]^{\gamma_{B}+\gamma_{\Sigma}} \bigl(n^{(t)}+1\bigr)
\end{align}
uniformly for $0\le t\le s_{\mathrm{fin}}$, where
\begin{align}
	C_{\mathrm{pre}}:=(2zC_{\mathrm{ind}}+1)\left(1+\frac{9C_{A}C_{\mathrm{ind}}}{C_{\mathrm{ind}}-1}\right). \label{eq:prefix-entropy-constant}
\end{align}
Consequently,
\begin{align}
	\sup_{\substack{ \psi\in\operatorname{Ran}\mathsf{P}_{0\rightarrow t}\\ \norm{\ket{\psi}}=1 }} S_{A}(\ket{\psi}) \le K_{n}\bigl(n^{(t)}+1\bigr), \qquad K_{n}:=C_{\mathrm{pre}}[\log(en)]^{\gamma_{B}+\gamma_{\Sigma}}, \qquad 0\le t\le s_{\mathrm{fin}}, \label{eq:uniform-prefix-entropy}
\end{align}
where $\gamma_{B}+\gamma_{\Sigma}=1+C_{\nu}\log(2zC_{\mathrm{ind}}^{2})$ and all constants are independent of $n$.

\subsection{Recovering the entropy of the exact ground state}

We next need a bound for a coherent sum, not just for a mixture.
\begin{lemma}
	\label{lem:orthogonal-superposition-entropy}
	Let normalized mutually orthogonal vectors $\ket{\psi_{j}}$ and probabilities $p_{j}$ satisfy $\ket{\psi}=\sum_{j=1}^{r}\sqrt{p_{j}}\ket{\psi_{j}}$.
	For any bipartition,
	\begin{align}
		S_{A}(\ket{\psi})\le r\left[H(p)+\sum_{j=1}^{r}p_{j}S_{A}(\ket{\psi_{j}})\right]. \label{eq:orthogonal-superposition-entropy}
	\end{align}
	Zero-probability terms may be omitted.
\end{lemma}

\begin{proof}
	Introduce a flag system $F$ on side $A$ and define
	\begin{align}
		\ket{\Xi} := \sum_{j=1}^{r}\sqrt{p_{j}}\ket{j}_{F}\ket{\psi_{j}}_{AB}.
	\end{align}
	Writing $\rho_{B,j}:=\operatorname{Tr}_{A}\ket{\psi_{j}}\bra{\psi_{j}}$, the reduced state on $B$ is $\sum_{j}p_{j}\rho_{B,j}$.
	Since $\ket{\Xi}$ is pure across $FA:B$, and using the standard entropy-of-a-mixture bound~\cite{nielsen2010quantum}, we obtain
	\begin{align}
		S_{FA}(\ket{\Xi}) = S\left(\sum_{j=1}^{r}p_{j}\rho_{B,j}\right) \le H(p)+\sum_{j=1}^{r}p_{j}S_{A}(\ket{\psi_{j}}). \label{eq:flagged-superposition-entropy}
	\end{align}
	
	Now measure $F$ in an orthonormal basis containing $\ket{+}=r^{-1/2}\sum_{j=1}^{r}\ket{j}$.
	The corresponding unnormalized post-measurement state is
	\begin{align}
		\frac{1}{\sqrt{r}} \sum_{j=1}^{r}\sqrt{p_{j}}\ket{\psi_{j}} = \frac{1}{\sqrt{r}}\ket{\psi}.
	\end{align}
	Thus the $+$ outcome occurs with probability $1/r$ and yields $\ket{\psi}$.
	Since a local measurement cannot increase entanglement on average~\cite{vidal2000entanglement},
	\begin{align}
		\frac{1}{r} S_{A}(\ket{\psi}) \le S_{FA}(\ket{\Xi}).
	\end{align}
	Combining this with Eq.~\eqref{eq:flagged-superposition-entropy} gives Eq.~\eqref{eq:orthogonal-superposition-entropy}.
\end{proof}

The cumulative projector $\mathsf{P}_{0\rightarrow t} = V_{0\rightarrow t}^{\dagger} V_{0\rightarrow t}$ defined in Eq.~\eqref{eq:cumulative-coisometry} projects onto the microscopic subspace retained through all truncations up to scale $t$.
These retained spaces decrease with the RG scale: by Eq.~\eqref{eq:cumulative-range-nesting},
\begin{align}
	\mathsf{P}_{0\rightarrow t+1} \le \mathsf{P}_{0\rightarrow t}, \qquad \operatorname{Ran}\mathsf{P}_{0\rightarrow t+1} \subseteq \operatorname{Ran}\mathsf{P}_{0\rightarrow t}.
\end{align}
Therefore, for $0\le t<s_{\mathrm{fin}}$, $\mathsf{P}_{0\rightarrow t}-\mathsf{P}_{0\rightarrow t+1}$ is the orthogonal projection onto the part retained through scale $t$ but discarded at scale $t+1$.
Define
\begin{align}
	\mathfrak{p}_{t} &:= \norm{ (\mathsf{P}_{0\rightarrow t} -\mathsf{P}_{0\rightarrow t+1})\ket{\Omega} }^{2}, \qquad \ket{\omega_{t}} := \frac{ (\mathsf{P}_{0\rightarrow t} -\mathsf{P}_{0\rightarrow t+1})\ket{\Omega} }{ \sqrt{\mathfrak{p}_{t}} },
\end{align}
whenever $\mathfrak{p}_{t}>0$.
At the terminal scale, define
\begin{align}
	\mathfrak{p}_{s_{\mathrm{fin}}} &:= \norm{ \mathsf{P}_{0\rightarrow s_{\mathrm{fin}}}\ket{\Omega} }^{2}, \qquad \ket{\omega_{s_{\mathrm{fin}}}} := \frac{ \mathsf{P}_{0\rightarrow s_{\mathrm{fin}}}\ket{\Omega} }{ \sqrt{\mathfrak{p}_{s_{\mathrm{fin}}}} },
\end{align}
whenever $\mathfrak{p}_{s_{\mathrm{fin}}}>0$.
Zero-probability terms are omitted.

Since $V_{0\rightarrow0}=\mathds{1}$ in Eq.~\eqref{eq:cumulative-coisometry}, $\mathsf{P}_{0\rightarrow0}=\mathds{1}$, and telescoping the successive projections gives
\begin{align}
	\ket{\Omega} = \sum_{t=0}^{s_{\mathrm{fin}}-1} (\mathsf{P}_{0\rightarrow t}-\mathsf{P}_{0\rightarrow t+1})\ket{\Omega} +\mathsf{P}_{0\rightarrow s_{\mathrm{fin}}}\ket{\Omega} = \sum_{t=0}^{s_{\mathrm{fin}}} \sqrt{\mathfrak{p}_{t}}\ket{\omega_{t}}. \label{eq:nested-projection-shell-decomposition}
\end{align}
The range inclusion in Eq.~\eqref{eq:cumulative-range-nesting} also implies that the projections $\mathsf{P}_{0\rightarrow t}-\mathsf{P}_{0\rightarrow t+1}$, together with $\mathsf{P}_{0\rightarrow s_{\mathrm{fin}}}$, have mutually orthogonal ranges.
Hence the nonzero vectors $\ket{\omega_{t}}$ are mutually orthogonal and $\sum_{t}\mathfrak{p}_{t}=1$.
Moreover, $\ket{\omega_{t}}\in\operatorname{Ran}\mathsf{P}_{0\rightarrow t}$ for every nonzero shell.

For $t<s_{\mathrm{fin}}$, the shell amplitude satisfies
\begin{align}
	\sqrt{\mathfrak{p}_{t}} &= \norm{ (\mathsf{P}_{0\rightarrow t} -\mathsf{P}_{0\rightarrow t+1})\ket{\Omega} } \le \norm{ (\mathds{1}-\mathsf{P}_{0\rightarrow t+1})\ket{\Omega} }.
\end{align}
The microscopic representative $V_{0\rightarrow t+1}^{\dagger}\ket{\Omega^{(t+1)}}$ belongs to $\operatorname{Ran}\mathsf{P}_{0\rightarrow t+1}$ by Eq.~\eqref{eq:cumulative-coisometry}.
Therefore $(\mathds{1}-\mathsf{P}_{0\rightarrow t+1})\ket{\Omega}$ belongs to $(\operatorname{Ran}\mathsf{P}_{0\rightarrow t+1})^{\perp}$, whereas $\mathsf{P}_{0\rightarrow t+1}\ket{\Omega} - V_{0\rightarrow t+1}^{\dagger}\ket{\Omega^{(t+1)}}$ belongs to $\operatorname{Ran}\mathsf{P}_{0\rightarrow t+1}$.
These two vectors are therefore orthogonal, and
\begin{align}
	\ket{\Omega} - V_{0\rightarrow t+1}^{\dagger}\ket{\Omega^{(t+1)}} &= (\mathds{1}-\mathsf{P}_{0\rightarrow t+1})\ket{\Omega} + \mathsf{P}_{0\rightarrow t+1}\ket{\Omega} - V_{0\rightarrow t+1}^{\dagger}\ket{\Omega^{(t+1)}}.
\end{align}
It follows that
\begin{align}
	\sqrt{\mathfrak{p}_{t}} &\le \norm{ (\mathds{1}-\mathsf{P}_{0\rightarrow t+1})\ket{\Omega} } \le \norm{ \ket{\Omega} - V_{0\rightarrow t+1}^{\dagger}\ket{\Omega^{(t+1)}} } \le \sum_{s=0}^{t}\epsilon_{s} \le (t+1)[n^{(t)}]^{-p}. \label{eq:projection-shell-probabilities}
\end{align}
Here the cumulative approximation estimate in Eq.~\eqref{eq:closed-induction-error}, applied at scale $t+1$, gives the penultimate bound.
The same equation gives $\epsilon_{s}\le[n^{(s)}]^{-p}$ for $s<s_{\mathrm{fin}}$, and since $n^{(s)}\ge n^{(t)}$ for $s\le t$,
\begin{align}
	\sum_{s=0}^{t}\epsilon_{s} \le \sum_{s=0}^{t}[n^{(s)}]^{-p} \le (t+1)[n^{(t)}]^{-p},
\end{align}
which gives the last bound above.
Finally, since $\mathsf{P}_{0\rightarrow s_{\mathrm{fin}}}$ is an orthogonal projection and $\ket{\Omega}$ is normalized, $\mathfrak{p}_{s_{\mathrm{fin}}}\le1$.
Thus the small weights of the nonterminal shells are controlled by the approximation errors at all preceding RG scales, which is why the full hierarchy of approximations is needed rather than only the final norm error.

\begin{theorem}
	\label{thm:regular-cut-polylog-entropy}
	Under the microscopic assumptions and a regular exact hierarchy, let the cut satisfy Eq.~\eqref{eq:regular-cut-count} with $0\le\theta<2\beta$.
	Suppose $D/2<\alpha<(D+1)/2$ and choose
	\begin{align}
		0<\nu<\min\left\{\beta,4\beta^{2},\frac{2\beta-\theta}{1-\theta}\right\}. \label{eq:full-regular-cut-blocking-condition}
	\end{align}
	Fix $p\ge3$, and take $z$, $C_{\mathrm{ind}}$, and $K$ from Eqs.~\eqref{eq:constant-cutoff-choice}, \eqref{eq:closed-induction-constant}, and \eqref{eq:explicit-stopping-exponent}.
	Then, for all sufficiently large $n$, the bound holds with the following explicit constants, independent of $n$ and of the cut apart from its uniform geometric constant:
	\begin{align}
		S_{A}(\ket{\Omega})\le C_{\mathrm{ent}}[\log(en)]^{\mathfrak{a}}. \label{eq:regular-cut-polylog-entropy}
	\end{align}
	Here
	\begin{align}
		C_{\mathrm{ent}} &:=(C_{\nu}+1)\left[C_{\nu}+1+2(1+C_{\nu}^{3})(2zC_{\mathrm{ind}}+1)\left(1+\frac{9C_{A}C_{\mathrm{ind}}}{C_{\mathrm{ind}}-1}\right)\right], \label{eq:explicit-entropy-prefactor}\\
		\mathfrak{a} &:=K+2+C_{\nu}\log(2zC_{\mathrm{ind}}^{2}). \label{eq:explicit-entropy-exponent}
	\end{align}
	In particular, this holds for codimension-one regular cuts with $\theta=(D-1)/D$ throughout $D/2<\alpha<(D+1)/2$, subject to the stated exact-hierarchy convention.
	
	For a microscopic-parameter summary, one may fix the admissible choices
	\begin{align}
		p:=3, \qquad \nu:=\frac{1}{2} \min\left\{\beta,4\beta^{2},\frac{2\beta-\theta}{1-\theta}\right\}. \label{eq:entropy-microscopic-parameter-choice}
	\end{align}
	With this choice and the same exact-hierarchy convention, the explicit constants above satisfy
	\begin{align}
		C_{\mathrm{ent}}
		&=\mathcal{O}_{D,\alpha,d,\mu_{0},\theta}\left(
		(1+C_{A})
		\left[
		1+
		\frac{\bar{g}}{\Delta}
		+
		\frac{\mu_{0}J}{\Delta}
		+
		\left(\frac{J}{\Delta}\right)^{2}\frac{\bar{g}}{\Delta}
		+
		\left(\frac{\mu_{0}J}{\Delta}\right)^{2}
		\max\left\{
		1,
		\frac{\bar{g}}{\Delta},
		\frac{\mu_{0}J}{\Delta}
		\right\}
		\right]
		\right),
		\label{eq:entropy-microscopic-prefactor-dependence}\\
		\mathfrak{a}
		&=\mathcal{O}_{D,\alpha,d,\mu_{0},\theta}\left(
		1+
		\log\left[
		1+
		\frac{\bar{g}}{\Delta}
		+
		\frac{\mu_{0}J}{\Delta}
		+
		\left(\frac{J}{\Delta}\right)^{2}\frac{\bar{g}}{\Delta}
		+
		\left(\frac{\mu_{0}J}{\Delta}\right)^{2}
		\max\left\{
		1,
		\frac{\bar{g}}{\Delta},
		\frac{\mu_{0}J}{\Delta}
		\right\}
		\right]
		\right).
		\label{eq:entropy-microscopic-exponent-dependence}
	\end{align}
	The implicit factors depend only on the displayed subscripts, and not on $J$, $\bar{g}$, $\Delta$, $C_{A}$, or $n$.
	For a codimension-one regular cut, $\theta=(D-1)/D$ is fixed by $D$ and may be omitted from these subscripts.
\end{theorem}

\begin{proof}
	Fix $p\ge3$ and the cutoff $z$ as in Eqs.~\eqref{eq:blocking-parameters-Dhalf-alpha-Dplus1half} and \eqref{eq:constant-cutoff-choice}.
	Use the exponent $K$ specified in Eq.~\eqref{eq:explicit-stopping-exponent} and the stopping threshold $T_{n}=L_{n}^{K}$, where $L_{n}:=\log(en)$.
	The additional constraint on $\nu$ is precisely Eq.~\eqref{eq:cut-exponent-condition}.
	Hence, under the geometric cut condition Eq.~\eqref{eq:regular-cut-count}, all the hypotheses leading to Eq.~\eqref{eq:uniform-prefix-entropy} are satisfied.
	
	Apply Lemma~\ref{lem:orthogonal-superposition-entropy} to Eq.~\eqref{eq:nested-projection-shell-decomposition}.
	The number $r$ of nonzero shell components satisfies $r\le s_{\mathrm{fin}}+1$, and therefore $H(\mathfrak{p})\le\log(s_{\mathrm{fin}}+1)$.
	Moreover, Eq.~\eqref{eq:projection-shell-probabilities} gives, for $t<s_{\mathrm{fin}}$,
	\begin{align}
		\mathfrak{p}_{t} \le (t+1)^{2}[n^{(t)}]^{-2p},
	\end{align}
	while $\mathfrak{p}_{s_{\mathrm{fin}}}\le1$.
	Since every $\ket{\omega_{t}}$ belongs to $\operatorname{Ran}\mathsf{P}_{0\rightarrow t}$, Eq.~\eqref{eq:uniform-prefix-entropy} gives $S_{A}(\ket{\omega_{t}})\le K_{n}(n^{(t)}+1)$.
	Hence
	\begin{align}
		S_{A}(\ket{\Omega}) &\le r\left[ H(\mathfrak{p}) + \sum_{t=0}^{s_{\mathrm{fin}}} \mathfrak{p}_{t} S_{A}(\ket{\omega_{t}}) \right] \notag\\
		&\le (s_{\mathrm{fin}}+1) \left[ \log(s_{\mathrm{fin}}+1) + K_{n}\sum_{t=0}^{s_{\mathrm{fin}}} \mathfrak{p}_{t}(n^{(t)}+1) \right] \notag\\
		&= (s_{\mathrm{fin}}+1) \left[ \log(s_{\mathrm{fin}}+1) + K_{n}\mathfrak{p}_{s_{\mathrm{fin}}} (n^{(s_{\mathrm{fin}})}+1) + K_{n}\sum_{t=0}^{s_{\mathrm{fin}}-1} \mathfrak{p}_{t}(n^{(t)}+1) \right] \notag\\
		&\le (s_{\mathrm{fin}}+1)\left[ \log(s_{\mathrm{fin}}+1) + K_{n}(n^{(s_{\mathrm{fin}})}+1) + K_{n}\sum_{t=0}^{s_{\mathrm{fin}}-1} (t+1)^{2}[n^{(t)}]^{-2p}(n^{(t)}+1) \right].
	\end{align}
	By the stopping rule, $n^{(s_{\mathrm{fin}})}\le T_{n}$, whereas $n^{(t)}>T_{n}\ge1$ for every $t<s_{\mathrm{fin}}$.
	Since $p\ge3$,
	\begin{align}
		[n^{(t)}]^{-2p}(n^{(t)}+1) \le 2[n^{(t)}]^{1-2p} \le 2T_{n}^{1-2p},
	\end{align}
	and $\sum_{t=0}^{s_{\mathrm{fin}}-1}(t+1)^{2} \le s_{\mathrm{fin}}^{3}$.
	Therefore
	\begin{align}
		S_{A}(\ket{\Omega}) \le (s_{\mathrm{fin}}+1) \left[ \log(s_{\mathrm{fin}}+1) + K_{n}(T_{n}+1) + 2K_{n}s_{\mathrm{fin}}^{3}T_{n}^{1-2p} \right]. \label{eq:final-entropy-sum}
	\end{align}
	By Eq.~\eqref{eq:uniform-prefix-entropy}, $K_{n}=C_{\mathrm{pre}}L_{n}^{1+C_{\nu}\log(2zC_{\mathrm{ind}}^{2})}$.
	Also, $s_{\mathrm{fin}}\le C_{\nu}\log L_{n}$ and $L_{n}\ge1$ give $s_{\mathrm{fin}}+1\le(C_{\nu}+1)L_{n}$, $\log(s_{\mathrm{fin}}+1)\le(C_{\nu}+1)L_{n}$, and $s_{\mathrm{fin}}^{3}\le C_{\nu}^{3}L_{n}^{3}$.
	Since $T_{n}+1\le2L_{n}^{K}$, Eq.~\eqref{eq:final-entropy-sum} yields
	\begin{align}
		S_{A}(\ket{\Omega}) &\le (C_{\nu}+1)^{2}L_{n}^{2}+2(C_{\nu}+1)C_{\mathrm{pre}}L_{n}^{K+2+C_{\nu}\log(2zC_{\mathrm{ind}}^{2})}  + 2(C_{\nu}+1)C_{\nu}^{3}C_{\mathrm{pre}}L_{n}^{5+K(1-2p)+C_{\nu}\log(2zC_{\mathrm{ind}}^{2})}\notag\\
		&\le (C_{\nu}+1)\left[C_{\nu}+1+2(1+C_{\nu}^{3})C_{\mathrm{pre}}\right]L_{n}^{K+2+C_{\nu}\log(2zC_{\mathrm{ind}}^{2})} \notag \\
		&=C_{\mathrm{ent}}L_{n}^{\mathfrak{a}}.
	\end{align}
	Here $K\ge1$ and $p\ge3$ imply $5+K(1-2p)\le K+2$, and $K+2+C_{\nu}\log(2zC_{\mathrm{ind}}^{2})\ge2$ also controls the first term.
	Substitution of Eq.~\eqref{eq:prefix-entropy-constant} gives the prefactor in Eq.~\eqref{eq:explicit-entropy-prefactor}, while the resulting power of $L_{n}$ is exactly the exponent in Eq.~\eqref{eq:explicit-entropy-exponent}.
	This proves Eq.~\eqref{eq:regular-cut-polylog-entropy}.
	
	It remains to verify the microscopic-parameter dependence in Eqs.~\eqref{eq:entropy-microscopic-prefactor-dependence} and \eqref{eq:entropy-microscopic-exponent-dependence}.
	With Eq.~\eqref{eq:entropy-microscopic-parameter-choice}, the quantities $\nu$, $\kappa$, $\sigma_{0}$, $z$, and $C_{\nu}$ depend only on $D,\alpha,\theta$.
	Equations~\eqref{eq:geometric-constant-choice}, \eqref{eq:c-int-choice}, \eqref{eq:robustness-constant-choice}, and \eqref{eq:onsite-recurrence-constant} then fix $c_{\rm int}$, $C_{\rm rob}$, $c_{\rm rob}$, $C_{g}$, $R_{g}$, and $R_{q}$ independently of $J$, $\bar{g}$, and $\Delta$.
	Since
	\begin{align}
		A_{g}=\max\left\{1,\frac{\mu_{0}\bar{g}}{\Delta},\frac{\mu_{0}J}{\Delta}\right\},
	\end{align}
	and Eq.~\eqref{eq:microscopic-deviation-constant} gives
	\begin{align}
		C_{\rm mic}
		=
		\mathcal{O}_{D,\alpha,d,\mu_{0}}\left(
		\left(\frac{J}{\Delta}\right)^{2}
		\frac{\bar{g}}{\Delta}
		\right),
	\end{align}
	Eqs.~\eqref{eq:closed-induction-probability-constants} and \eqref{eq:closed-induction-constant} imply
	\begin{align}
		C_{\mathrm{ind}}
		=
		\mathcal{O}_{D,\alpha,d,\mu_{0},\theta}\left(
		1+
		\frac{\bar{g}}{\Delta}
		+
		\frac{\mu_{0}J}{\Delta}
		+
		\left(\frac{J}{\Delta}\right)^{2}\frac{\bar{g}}{\Delta}
		+
		\left(\frac{\mu_{0}J}{\Delta}\right)^{2}
		\max\left\{
		1,
		\frac{\bar{g}}{\Delta},
		\frac{\mu_{0}J}{\Delta}
		\right\}
		\right).
	\end{align}
	Furthermore, Eq.~\eqref{eq:closed-induction-envelope-choice} gives $c_{\mathrm{hop}}=C_{\mathrm{ind}}^{3/2}$ and $B_{\mathrm{hop}}=1$, so Eq.~\eqref{eq:explicit-stopping-exponent} yields
	\begin{align}
		K=\mathcal{O}_{D,\alpha,\theta}(1+\log C_{\mathrm{ind}}).
	\end{align}
	Finally, $C_{\mathrm{ind}}/(C_{\mathrm{ind}}-1)\le2$ and Eqs.~\eqref{eq:explicit-entropy-prefactor} and \eqref{eq:explicit-entropy-exponent} give, respectively,
	\begin{align}
		C_{\mathrm{ent}}
		&=\mathcal{O}_{D,\alpha,\theta}\bigl((1+C_{A})C_{\mathrm{ind}}\bigr), \label{eq:entropy-prefactor-from-induction}\\
		\mathfrak{a}
		&=\mathcal{O}_{D,\alpha,\theta}(1+\log C_{\mathrm{ind}}). \label{eq:entropy-exponent-from-induction}
	\end{align}
	Combining the bound on $C_{\mathrm{ind}}$ with Eq.~\eqref{eq:entropy-prefactor-from-induction} proves Eq.~\eqref{eq:entropy-microscopic-prefactor-dependence}.
	Combining the same bound with Eq.~\eqref{eq:entropy-exponent-from-induction} proves Eq.~\eqref{eq:entropy-microscopic-exponent-dependence}.
	The fixed hierarchy constant $c_{b}$ affects only how large $n$ must be for the depth estimate and the absorption bounds, not the displayed entropy coefficients.
\end{proof}

\subsection{Removing exact-divisibility restrictions by bounded-ratio padding}
\label{sec:padding}

The regular exact blocking hierarchy used above requires the successive lattice side lengths to satisfy exact divisibility conditions.
This is an arithmetic requirement of the blocking construction rather than an additional assumption on the Hamiltonian or its ground state.
We now remove it by embedding an arbitrary cubic box into a dyadic cubic box, namely one whose side length is a power of two, whose side length differs from the original one by at most a fixed factor.
The enlargement preserves the Kac-normalized interaction bound and the regular-cut geometry up to fixed constants, while the added sites are placed in a product state so that the entanglement entropy across the corresponding cut is unchanged.
Theorem~\ref{thm:regular-cut-polylog-entropy} can therefore be applied to the enlarged system and transferred back to the original one.

\begin{corollary}
	\label{cor:arbitrary-side-lengths}
	Let the microscopic side length be an arbitrary integer $L$, with open or periodic boundary conditions.
	Associate to $A$ the union of its unit lattice cells, and suppose that, for every $1\le r\le L$, its boundary can be covered by at most $C_{A}(L/r)^{D\theta}$ axis-aligned cubes of side $r$, with fixed $0\le\theta<1$ and $C_{A}$.
	If $D/2<\alpha<(D+1)/2$ and $\theta<2\beta$, then the polylogarithmic entropy bound of Theorem~\ref{thm:regular-cut-polylog-entropy} holds without an exact-divisibility assumption on the original box.
	In particular, finite-face regular boundaries are covered with $\theta=(D-1)/D$.
\end{corollary}

\begin{proof}
	Let $L':=2^{\lceil\log_{2}L\rceil}$ and $n':=(L')^{D}$.
	Then $L\le L'<2L$ and $n\le n'<2^{D}n$, so the enlargement changes the system size by at most a dimension-dependent constant factor.
	On each coordinate define
	\begin{align}
		f(k):=\left\lfloor\frac{L'}{L}k\right\rfloor, \qquad 0\le k\le L.
	\end{align}
	Since every increment $f(k+1)-f(k)$ is either one or two, the coordinatewise embedding $i\mapsto f(i)$ has uniformly bounded distance distortion:
	\begin{align}
		r'_{f(i),f(j)} \le 2r_{ij}.
	\end{align}
	The same bound holds for periodic boundaries, including shortest paths that wrap around the periodic boundary.
	
	Place every original Hamiltonian term on the corresponding image sites, and add the independent on-site term $\Delta(\mathds{1}-\ket{0}\bra{0})$ at each additional site.
	The enlarged system then has ground state $\ket{\Omega'}=\ket{\Omega}\otimes\ket{\mathbf{0}}_{\rm extra}$ and gap at least $\Delta$.
	Its local energy scale remains uniformly bounded.
	Moreover, using the distance bound above,
	\begin{align}
		\norm{h_{ij}} &\le \frac{J\mu_{0}}{n^{\beta}}r_{ij}^{-\alpha} \le \frac{2^{\alpha} J\mu_{0}(n'/n)^{\beta}}{(n')^{\beta}} (r'_{f(i),f(j)})^{-\alpha} \le \frac{2^{D} J\mu_{0}}{(n')^{\beta}} (r'_{f(i),f(j)})^{-\alpha},
	\end{align}
	where $\alpha+D\beta=D$ and $n'/n<2^{D}$ were used.
	Hence the enlarged Hamiltonian satisfies the same Kac-normalized interaction assumption, with only the fixed interaction constant modified.
	
	It remains to transfer the cut geometry and its entropy.
	For each original lattice cell $i$, assign the enlarged rectangular region $\prod_{a=1}^{D} \{f(i_{a}),\ldots,f(i_{a}+1)-1\}$ to $A'$ exactly when $i\in A$.
	These regions partition the enlarged box.
	Since all additional spins are in product states,
	\begin{align}
		S_{A'}(\ket{\Omega'}) = S_{A}(\ket{\Omega}).
	\end{align}
	
	Extend $f$ linearly on each unit interval.
	The resulting coordinatewise map has slopes between one and two.
	Therefore a boundary-cover cube of side $\ell/2$ in the original box is mapped into a box of side at most $\ell$, which intersects at most a dimension-dependent number of dyadic blocks of side $\ell$.
	Applying the assumed boundary-cover estimate at $\ell=\ell_{1:s}$ gives
	\begin{align}
		|\mathcal{C}_{s}(A')| &\le C_{D}'C_{A} \left( \frac{L'}{\ell_{1:s}} \right)^{D\theta} = C_{D}'C_{A}[n'^{(s)}]^{\theta}.
	\end{align}
	The same conclusion holds for periodic boundaries after splitting any covering cube crossing a seam into a bounded number of pieces.
	Thus $A'$ satisfies Eq.~\eqref{eq:regular-cut-count} with the same exponent $\theta$.
	
	Since $L'$ is dyadic, the enlarged lattice admits the regular exact blocking hierarchy required in Theorem~\ref{thm:regular-cut-polylog-entropy}.
	Applying that theorem to $\ket{\Omega'}$, and using $n'=\mathcal{O}(n)$ together with $S_{A'}(\ket{\Omega'})=S_{A}(\ket{\Omega})$, yields the claimed polylogarithmic entropy bound for the original system.
\end{proof}

\section{The critical exponent \texorpdfstring{$\alpha=(D+1)/2<D$}{alpha = (D+1)/2 < D}}
\label{sec:critical-MFRG}

At $\alpha=(D+1)/2<D$, the one-step geometric, compression, centering, variance, and reflection estimates remain valid at every defined scale satisfying their stated gap hypotheses.
In particular, Eq.~\eqref{eq:J-renormalization-alpha-eq-Dplus1half} gives
\begin{align}
	\mu_{s}=1, \qquad J^{(s)}=c_{\mathrm{int}}^{s}J_{\mathrm{K}}Z_{s}\chi_{s}, \qquad \chi_{s}=\prod_{t=1}^{s}\sqrt{\log(e\ell_{t})},
\end{align}
while the critical case of Proposition~\ref{prop:renormalized-robustness} gives
\begin{align}
	\mathcal{A}_{s}(v,Y) \le C_{\mathrm{rob}}c_{\mathrm{rob}}^{s}J_{\mathrm{K}}Z_{s}\chi_{s}[n^{(s)}]^{-\beta}\sqrt{\frac{\bar{g}^{(s)}}{\Delta^{(s)}}}.
\end{align}
Combining Eqs.~\eqref{eq:direct-reflection-bound}, \eqref{eq:F-alpha-square-sum}, and \eqref{eq:direct-eta-bound} instead gives a direct-compression reflection estimate with the factor $\sqrt{\log(e\ell_{1:s})}$ rather than $\chi_{s}$.
This improvement concerns the reflection estimate only and does not remove $\chi_{s}$ from the interaction recurrence or independently control $\bar{g}^{(s)}$.

The closed induction leading to a polylogarithmic effective system size does not apply unchanged at the critical exponent.
First, Eq.~\eqref{eq:cumulative-growth-alpha-eq-Dplus1half} shows that $\chi_{s}$ need not be bounded by a fixed power of $\log n$ when the depth grows with $n$.
Second, for a codimension-one cut, $2\beta=(D-1)/D=\theta$, so the exponent in Eq.~\eqref{eq:cut-exponent-condition} becomes $-\nu(1-\theta)<0$ for every fixed $\nu>0$.
Neither obstruction prevents a finite-depth construction stopped at a positive power of $n$.
We use this construction to prove a subpolynomial entropy bound, without assuming the conclusion of the subcritical closed induction.

\subsection{A subpolynomial entropy theorem}

\begin{theorem}
	\label{thm:critical-subpolynomial-entropy}
	Assume the microscopic hypotheses with $D\ge2$ and $\alpha=(D+1)/2$.
	Let $A|A^{\mathrm{c}}$ be a regular codimension-one cut in the following uniform sense.
	Identifying $A$ with the union of its unit lattice cells, suppose that its boundary can be covered by at most $C_A(\ell_\Lambda/r)^{D-1}$ axis-aligned cubes of side $r$ for every $1 \le r \le \ell_\Lambda$.
	For periodic boundaries, use the corresponding periodic covering convention.
	Then, for every $\varepsilon>0$, there are constants $C_{\varepsilon}$ and $n_{0}(\varepsilon)$ such that
	\begin{align}
		S_{A}(\ket{\Omega}) \le C_{\varepsilon}n^{\varepsilon} \qquad \bigl(n\ge n_{0}(\varepsilon)\bigr). \label{eq:critical-subpolynomial-entropy}
	\end{align}
	The constants depend only on $\varepsilon$, $C_{A}$, and the fixed microscopic parameters, and not on $n$.
\end{theorem}

\begin{proof}
	The bound $S_{A}(\ket{\Omega})\le n\log d$ handles $\varepsilon\ge1$, so fix $0<\varepsilon<1$.
	We first work on a dyadic cubic box and remove this restriction at the end.
	All geometric constants in this proof are taken from the critical branches of the geometric estimates, not from the subcritical expression in Eq.~\eqref{eq:geometric-constant-choice}.
	Set
	\begin{align}
		\theta:=\frac{D-1}{D}, \qquad \beta=\frac{\theta}{2}, \qquad L_{n}:=\log(en), \qquad t:=\frac{\varepsilon}{2}, \qquad a:=\frac{\varepsilon}{4}, \label{eq:critical-parameters-one}
	\end{align}
	\begin{align}
		\nu:=\frac{1}{4}\min\{\beta,D\varepsilon,1\}, \qquad \sigma:=\frac{\beta-\nu}{2}>0, \qquad p:=3. \label{eq:critical-parameters-two}
	\end{align}
	Use the dyadic hierarchy of Section~\ref{sec:regular-blocking-hierarchy}, with
	\begin{align}
		\ell_{s+1}:=2^{\left\lfloor\nu\log_{2}\bigl([n^{(s)}]^{1/D}\bigr)\right\rfloor}, \qquad b_{s+1}:=\ell_{s+1}^{D}, \qquad n^{(s+1)}=\frac{n^{(s)}}{b_{s+1}}. \label{eq:critical-dyadic-hierarchy}
	\end{align}
	Every tiling is exact, and
	\begin{align}
		2^{-D}[n^{(s)}]^{\nu}\le b_{s+1}\le[n^{(s)}]^{\nu}, \qquad [n^{(s)}]^{1-\nu}\le n^{(s+1)}\le2^{D}[n^{(s)}]^{1-\nu}. \label{eq:critical-hierarchy-size}
	\end{align}
	Let $s_{\mathrm{crit}}$ be the first stage for which $n^{(s_{\mathrm{crit}})}\le n^{t}$.
	The final step cannot be treated as an equality of powers, but Eq.~\eqref{eq:critical-hierarchy-size} gives the sufficient bounds
	\begin{align}
		n^{a}<n^{t(1-\nu)}<n^{(s_{\mathrm{crit}})}\le n^{t}, \qquad n^{(s)}\ge n^{a} \quad (0\le s\le s_{\mathrm{crit}}). \label{eq:critical-terminal-size}
	\end{align}
	Here $\nu<1/2$ and $a=t/2$.
	At every active stage, $n^{(s)}>n^{t}$, so for sufficiently large $n$,
	\begin{align}
		\log n^{(s+1)}\le\left(1-\frac{\nu}{2}\right)\log n^{(s)}.
	\end{align}
	Consequently, the depth is bounded independently of $n$ by
	\begin{align}
		s_{\mathrm{crit}}\le s_{\mathrm{crit},*}:=1+\left\lceil\frac{\log(1/t)}{-\log(1-\nu/2)}\right\rceil. \label{eq:critical-finite-depth}
	\end{align}
	In particular, $\chi_{s}\le L_{n}^{s/2}\le L_{n}^{s_{\mathrm{crit},*}/2}$ throughout this hierarchy.
	Choose the same deviation cutoff at every step by setting
	\begin{align}
		z_{s}:=z, \qquad z:=\max\left\{1,\left\lceil\frac{16(p+3)}{a\sigma}\right\rceil\right\}. \label{eq:critical-cutoff}
	\end{align}
	All these parameters are fixed independently of $n$.
	
	We now close the finite-depth induction.
	We shall obtain constants $G_{s},Q_{s},D_{s}$, independent of $n$, such that
	\begin{align}
		J^{(s)}\le j_{s}L_{n}^{s/2}, \qquad \bar{g}^{(s)}\le G_{s}L_{n}^{s/2}, \qquad j_{s}:=J_{\mathrm{K}}(c_{\mathrm{int}}z)^{s}, \label{eq:critical-energy-envelope} 
	\end{align}
	\begin{align}
		q_{s,*}\le Q_{s}L_{n}^{3s/2}[n^{(s)}]^{-\theta}, \qquad \log d^{(s)}\le D_{s}L_{n}, \qquad p_{0,i}^{(s)}>\frac{1}{2}, \label{eq:critical-occupation-envelope}
	\end{align}
	\begin{align}
		\Delta^{(s)}\ge\frac{\Delta}{2}, \qquad \epsilon_{s}:=\norm{\ket{\Omega^{(s)}}-V_{s+1}^{\dagger}\ket{\Omega^{(s+1)}}}\le n^{-p} \quad (s<s_{\mathrm{crit}}). \label{eq:critical-gap-error-envelope}
	\end{align}
	The quantities $j_{s}$ and $G_{s}$ have energy units, whereas $Q_{s}$ and $D_{s}$ are dimensionless.
	At $s=0$, Theorem~\ref{thm:single-scale-arbitrary-cut} gives
	$q_{0,*}\le C_{\mathrm{mic}}n^{-2\beta}=C_{\mathrm{mic}}n^{-\theta}$,
	which initializes the occupation envelope and also implies
	$p_{0,i}^{(0)}>1/2$ for sufficiently large $n$.
	Equation~\eqref{eq:J-renormalization-alpha-eq-Dplus1half}, together with $\chi_{s}\le L_{n}^{s/2}$ and $z_{s}=z$, gives the first bound in Eq.~\eqref{eq:critical-energy-envelope}.
	
	Here is a noncircular choice of the remaining envelopes.
	Let $C_{\mathrm{geo}}\ge1$ bound both the normalized coupling row sum $\gamma_{0}^{(s)}$ in Eq.~\eqref{eq:scale-kac-row-sum} and the internal square-sum row norm when $\mu_{s}=1$.
	Such a constant exists since $\alpha<D$ and $2\alpha>D$.
	Choose
	\begin{align}
		G_{0}:=\mu_{0}\bar{g}+J_{\mathrm{K}}+\Delta, \qquad C_{z}:=1+2z+2(2+17\sqrt{2})\sqrt{z}+C_{\mathrm{geo}}(8z+1), \label{eq:critical-energy-constants}
	\end{align}
	\begin{align}
		G_{s+1}:=C_{z}(G_{s}+j_{s}+j_{s+1}), \qquad D_{0}:=\log d, \qquad D_{s+1}:=z(D_{s}+1), \label{eq:critical-envelope-recursions}
	\end{align}
	\begin{align}
		Q_{s}:=1+C_{\mathrm{mic}}+2C_{\mathrm{rob}}^{2}(c_{\mathrm{rob}}z)^{2s}\frac{J_{\mathrm{K}}^{2}G_{s}}{\Delta^{3}}. \label{eq:critical-occupation-constants}
	\end{align}
	Here $C_{\mathrm{mic}}$ is the microscopic deviation constant in Eq.~\eqref{eq:microscopic-deviation-constant}, while $C_{\mathrm{rob}}$ and $c_{\mathrm{rob}}$ are the fixed constants in the critical case of Proposition~\ref{prop:renormalized-robustness}.
	For fixed $\varepsilon$, these arrays are defined for all $0\le s\le s_{\mathrm{crit},*}$ independently of $n$.
	
	Suppose the induction has been established through an active stage $s$.
	The current block has a small mean deviation number because
	\begin{align}
		b_{s+1}q_{s,*}\le Q_{s}L_{n}^{3s/2}[n^{(s)}]^{\nu-\theta}=o(1). \label{eq:critical-block-mean}
	\end{align}
	This is uniform over the finitely many stages, since $n^{(s)}\ge n^{a}$ and $\nu<\theta$.
	For each parent block, choose the auxiliary complement norm-tail tolerance
	\begin{align}
		\epsilon_{\mathrm{c},s}:=n^{-(p+4)}. \label{eq:critical-complement-tolerance}
	\end{align}
	Applying Eq.~\eqref{choice_of_m_0} at the current scale produces a complement cutoff $m_{\mathrm{c},s}$ with
	\begin{align}
		m_{\mathrm{c},s}\le2n^{(s)}q_{s,*}+\sqrt{\frac{2\bar{g}^{(s)}n^{(s)}}{\Delta^{(s)}}}\log\frac{2}{\epsilon_{\mathrm{c},s}}+2\le C_{s}L_{n}^{3s/2+1}[n^{(s)}]^{1/2}. \label{eq:critical-complement-cutoff}
	\end{align}
	The cutoff is capped by the complement size.
	The last bound uses $1-\theta=1/D\le1/2$, together with the inductive energy and gap estimates.
	Throughout this proof, $C_{s}$ denotes a finite constant depending only on the fixed parameters and the previously defined envelope constants, but not on $n$ or the block.
	
	Set $G_{\mathrm{aux},s}:=2\bar{g}^{(s)}n^{(s)}$.
	Lemma~\ref{lem:preMFRG-compression}, applied to the auxiliary complement projection, gives
	\begin{align}
		e_{\mathrm{c},s}\le\frac{G_{\mathrm{aux},s}\epsilon_{\mathrm{c},s}^{2}}{1-\epsilon_{\mathrm{c},s}^{2}}, \qquad
		\eta_{\mathrm{c},s}:=\norm{\ket{\Omega^{(s)}}-\ket{\widehat{\bar{\Omega}}^{(s)}}}\le\sqrt{\frac{2G_{\mathrm{aux},s}\epsilon_{\mathrm{c},s}^{2}}{\Delta^{(s)}(1-\epsilon_{\mathrm{c},s}^{2})}}, \label{eq:critical-auxiliary-compression}
	\end{align}
	where $\ket{\widehat{\bar{\Omega}}^{(s)}}$ denotes the ground state of the Hamiltonian after the auxiliary complement cutoff, embedded back into the scale-$s$ Hilbert space.
	Since $n^{(s)}\le n$ and $s\le s_{\mathrm{crit},*}$, these estimates imply, for sufficiently large $n$,
	\begin{align}
		\bar{\Delta}_{\mathrm{c},s} \ge \Delta^{(s)} - e_{\mathrm{c},s} \ge \frac{\Delta^{(s)}}{2} \ge \frac{\Delta}{4}, \qquad \eta_{\mathrm{c},s}\le n^{-(p+3)}. \label{eq:critical-auxiliary-gap-error}
	\end{align}
	Since Eq.~\eqref{eq:critical-block-mean} gives $b_{s+1}q_{s,*}=o(1)$ and Eq.~\eqref{eq:critical-auxiliary-gap-error} gives $\eta_{\mathrm{c},s}=o(1)$, the zero-median criterion in Eq.~\eqref{eq:preMFRG-zero-median-criterion} yields $m_{*}=0$ for sufficiently large $n$.
	
	We now verify that the hopping between deviation-number shells is uniformly small on the finite interval $0\le x\le z$.
	Applying Eq.~\eqref{Parameter_bar_J_m_0/L} at scale $s$ with $m_{0}=m_{\mathrm{c},s}$ and $|L|=b_{s+1}$, and using $\mu_{s}=1$, $F_{\alpha}(n^{(s)},m_{\mathrm{c},s})=m_{\mathrm{c},s}^{\beta}$, $\bar{\gamma}_{L}^{(s)}\le C_{\mathrm{geo}}$, and Eq.~\eqref{eq:dimension-free-onsite-simplification}, which gives $\delta_{*}^{(s)} \le (2+17\sqrt{2})\bar{g}^{(s)}\sqrt{q_{s,*}}$, yields
	\begin{align}
		\bar{J}_{m_{\mathrm{c},s},L}^{(s)}[0,z] &\le C_{s,z}\left[J^{(s)}[n^{(s)}]^{-\beta}\sqrt{b_{s+1}}+\bar{g}^{(s)}\sqrt{b_{s+1}q_{s,*}}+J^{(s)}[n^{(s)}]^{-\beta}m_{\mathrm{c},s}^{\beta}\sqrt{b_{s+1}\log(en^{(s)})}\right] \notag \\
		&\le C_{s,z}L_{n}^{2s+2}\left[[n^{(s)}]^{\nu/2-\beta}+[n^{(s)}]^{\nu/2-\beta/2}\right] \notag \\
		&\le2C_{s,z}L_{n}^{2s+2}[n^{(s)}]^{-\sigma}. \label{eq:critical-shell-hopping}
	\end{align}
	The logarithmic power in the middle line is an upper bound for all three terms; in particular, the crossing term has power $s/2+\beta(3s/2+1)+1/2\le2s+2$.
	Because $n^{(s)}\ge n^{a}$ and the depth is bounded, Eq.~\eqref{eq:critical-auxiliary-gap-error} implies
	\begin{align}
		\rho_{s}:=\frac{\bar{J}_{m_{\mathrm{c},s},L}^{(s)}[0,z]}{\bar{\Delta}_{\mathrm{c},s}}\le n^{-a\sigma/2}<1 \label{eq:critical-hopping-ratio}
	\end{align}
	for all sufficiently large $n$, uniformly in the block and the active stage.
	Lemma~\ref{lem:preMFRG-small-hopping}, followed by Eq.~\eqref{eq:tail-after-global-truncation}, now gives
	\begin{align}
		\norm{\Pi_{L,>z}^{(s)}\ket{\Omega^{(s)}}}\le\eta_{\mathrm{c},s}+\rho_{s}^{z/8}\le n^{-(p+3)}+n^{-a\sigma z/16}\le2n^{-(p+3)}. \label{eq:critical-local-tail}
	\end{align}
	When the hopping envelope is zero, the auxiliary tail vanishes and the same estimate holds.
	
	Let $w_{s}$ and $e_{s}$ be the discarded weight and energy defect from Eq.~\eqref{eq:step-energy-error}.
	There are at most $n$ parent blocks, so Eq.~\eqref{eq:preMFRG-block-union} gives
	\begin{align}
		w_{s}\le4n^{-2p-5}, \qquad \frac{e_{s}}{\Delta}\le\frac{2\bar{g}^{(s)}n^{(s)}}{\Delta}\frac{w_{s}}{1-w_{s}}\le C_{s}L_{n}^{s/2}n^{-2p-4}. \label{eq:critical-step-defect}
	\end{align}
	Lemma~\ref{lem:preMFRG-compression}, applied to the retained projection at the step $s\to s+1$, consequently yields
	\begin{align}
		\Delta^{(s+1)}\ge\Delta^{(s)}-e_{s}, \qquad \epsilon_{s}\le\sqrt{\frac{2e_{s}}{\Delta^{(s)}}}\le n^{-p}. \label{eq:critical-step-compression}
	\end{align}
	The new Hamiltonian acts on the retained Hilbert space, as in Eq.~\eqref{MFRG_transform_RE}; scalar centering does not affect these conclusions.
	Since $s_{\mathrm{crit}}\le s_{\mathrm{crit},*}$ and the right-hand side of Eq.~\eqref{eq:critical-step-defect} tends to zero uniformly over these finitely many stages, $n$ can be chosen so that every step has $e_{s}\le\Delta/[4(s_{\mathrm{crit},*}+1)]$.
	Inductively, $\Delta^{(s)}\ge\Delta-\sum_{r<s}e_{r}\ge3\Delta/4$, which in particular maintains the weaker gap envelope used above.
	This also proves existence and uniqueness of the effective ground state at every step.
	
	Equations~\eqref{eq:critical-block-mean} and \eqref{eq:critical-step-compression} imply $b_{s+1}q_{s,*}+\epsilon_{s}<1/2$.
	Lemma~\ref{lem:inherited-reference} therefore gives $p_{0,j}^{(s+1)}>1/2$ and
	\begin{align}
		\bra{0_{j}^{(s+1)}}V_{s+1,j}Q_{i}^{(s)}V_{s+1,j}^{\dagger}\ket{0_{j}^{(s+1)}}\le2(q_{s,*}+\epsilon_{s}). \label{eq:critical-reference-child}
	\end{align}
	We next verify the new energy envelope without using the new occupation estimate.
	Proposition~\ref{Prop:renormalized_interaction/on-site}, together with $\sum_{j\ne i}\norm{h_{ij}^{(s+1)}} \le J^{(s+1)}\gamma_{0}^{(s+1)}\le C_{\mathrm{geo}}J^{(s+1)}$, gives
	\begin{align}
		\bar{g}^{(s+1)} \le 8C_{\mathrm{geo}} zJ^{(s)}[n^{(s)}]^{-\beta}\sqrt{b_{s+1}}+2z\bar{g}^{(s)} + 2(2+17\sqrt{2})\bar{g}^{(s)}\sqrt{zb_{s+1}q_{s,*}} + C_{\mathrm{geo}}J^{(s+1)}. \label{eq:critical-energy-step}
	\end{align}
	For sufficiently large $n$, $b_{s+1}q_{s,*}\le1$ and $[n^{(s)}]^{-\beta}\sqrt{b_{s+1}}\le1$.
	Equations~\eqref{eq:critical-energy-constants} and \eqref{eq:critical-envelope-recursions} therefore give
	$\bar{g}^{(s+1)}\le G_{s+1}L_{n}^{(s+1)/2}$.
	Since Lemma~\ref{lem:inherited-reference} gives $p_{0,j}^{(s+1)}>1/2$ at every scale-$(s+1)$ site, the critical case of Proposition~\ref{prop:renormalized-robustness}, followed by Proposition~\ref{prop:robustness-fixed-point}, gives
	\begin{align}
		q_{s+1,*}
		\le
		\frac{C_{\mathrm{rob}}^{2}(c_{\mathrm{rob}}z)^{2(s+1)}J_{\mathrm{K}}^{2}\chi_{s+1}^{2}\bar{g}^{(s+1)}}{4[\Delta^{(s+1)}]^{3}}
		[n^{(s+1)}]^{-\theta}
		\le
		Q_{s+1}L_{n}^{3(s+1)/2}[n^{(s+1)}]^{-\theta}.
		\label{eq:critical-occupation-step}
	\end{align}
	Finally, Eq.~\eqref{eq:effective-dimension-recursion} gives
	\begin{align}
		\log d^{(s+1)}\le z\bigl(\log d^{(s)}+\log b_{s+1}\bigr)\le z(D_{s}+1)L_{n}=D_{s+1}L_{n}.
	\end{align}
	This closes Eqs.~\eqref{eq:critical-energy-envelope}--\eqref{eq:critical-gap-error-envelope}.
	There is no circular choice of constants: $a,t,\nu,\sigma,s_{\mathrm{crit},*},z$ are fixed first, the envelope arrays are fixed next, and only then is $n_{0}(\varepsilon)$ increased to meet the finitely many smallness conditions.
	
	We turn to the entropy, including the entanglement generated by decoding across the microscopic cut.
	The covering hypothesis gives Eq.~\eqref{eq:regular-cut-count} with $\theta=(D-1)/D$ and a modified dimension-dependent multiple of $C_{A}$, which we absorb into $C_{A}$.
	Lemma~\ref{lem:weighted-tree-entropy} is applicable because Eqs.~\eqref{eq:local-RG-partial-isometry-identities} and \eqref{eq:Pi-s-plus-one} define the hierarchy of low-deviation retained spaces, while Eq.~\eqref{eq:critical-reference-child} provides the required child-reference occupation bound.  
	For its weights, $y_{s}=\frac{1}{2}[n^{(s)}]^{-\theta}$ and $B_{s_{\mathrm{crit}}}\le C_{\varepsilon}L_{n}$ because $s_{\mathrm{crit}}\le s_{\mathrm{crit},*}$, while $z$ and the coefficients $D_s$ are bounded over this finite range.
	Moreover,
	\begin{align}
		b_{s+1}|\mathcal{C}_{s+1}(A)|\le C_{A}[n^{(s)}]^{\theta}b_{s+1}^{1-\theta}\le C_{A}[n^{(s)}]^{\theta+\nu/D}. \label{eq:critical-cut-count}
	\end{align}
	Using Eqs.~\eqref{eq:critical-occupation-envelope}, \eqref{eq:critical-gap-error-envelope}, and \eqref{eq:critical-cut-count}, together with $B_{s_{\mathrm{crit}}}\le C_{\varepsilon}L_{n}$ and $y_{s}=\frac{1}{2}[n^{(s)}]^{-\theta}$, Eq.~\eqref{eq:weighted-tree-entropy-bound} gives, for every normalized $\ket{\psi}\in\operatorname{Ran}\mathsf{P}_{0\rightarrow s_{\mathrm{crit}}}$,
	\begin{align}
		S_{A}(\ket{\psi}) &\le C_{\varepsilon}L_{n}n^{(s_{\mathrm{crit}})}+C_{\varepsilon}L_{n}^{3s_{\mathrm{crit},*}/2+1}\sum_{s=0}^{s_{\mathrm{crit}}-1}[n^{(s)}]^{\nu/D} \notag\\
		&\le C_{\varepsilon}L_{n}^{3s_{\mathrm{crit},*}/2+1}\bigl(n^{t}+n^{\nu/D}\bigr) \notag\\
		&\le C_{\varepsilon}n^{\varepsilon/2}L_{n}^{3s_{\mathrm{crit},*}/2+1}. \label{eq:critical-retained-entropy}
	\end{align}
	Here $\nu/D\le\varepsilon/4<t$.
	Thus, the positive power lost in the boundary count has been bounded rather than incorrectly discarded.
	
	Finally, we transfer the retained-space entropy bound to the exact ground state.
	Telescoping the one-step errors in Eq.~\eqref{eq:critical-step-compression} gives
	\begin{align}
		\norm{\ket{\Omega}-V_{0\rightarrow s_{\mathrm{crit}}}^{\dagger}\ket{\Omega^{(s_{\mathrm{crit}})}}}\le s_{\mathrm{crit},*}n^{-p}.
	\end{align}
	Since the decoded effective ground state belongs to $\operatorname{Ran}\mathsf{P}_{0\rightarrow s_{\mathrm{crit}}}$,
	\begin{align}
		\norm{(\mathds{1}-\mathsf{P}_{0\rightarrow s_{\mathrm{crit}}})\ket{\Omega}}^{2}\le s_{\mathrm{crit},*}^{2}n^{-2p}. \label{eq:finite-depth-final-discarded-weight}
	\end{align}
	For $0<\norm{(\mathds{1}-\mathsf{P}_{0\rightarrow s_{\mathrm{crit}}})\ket{\Omega}}^{2}<1$, write the orthogonal decomposition
	\begin{align}
		\ket{\Omega}
		=\sqrt{1-\norm{(\mathds{1}-\mathsf{P}_{0\rightarrow s_{\mathrm{crit}}})\ket{\Omega}}^{2}}\ket{\psi}
		+\norm{(\mathds{1}-\mathsf{P}_{0\rightarrow s_{\mathrm{crit}}})\ket{\Omega}}\ket{\xi},
	\end{align}
	where
	\begin{align}
		\ket{\psi}
		:=\frac{\mathsf{P}_{0\rightarrow s_{\mathrm{crit}}}\ket{\Omega}}
		{\sqrt{1-\norm{(\mathds{1}-\mathsf{P}_{0\rightarrow s_{\mathrm{crit}}})\ket{\Omega}}^{2}}}, \qquad
		\ket{\xi}
		:=\frac{(\mathds{1}-\mathsf{P}_{0\rightarrow s_{\mathrm{crit}}})\ket{\Omega}}
		{\norm{(\mathds{1}-\mathsf{P}_{0\rightarrow s_{\mathrm{crit}}})\ket{\Omega}}}.
	\end{align}
	Lemma~\ref{lem:orthogonal-superposition-entropy} with two components, together with $S_{A}(\ket{\xi})\le n\log d$, gives
	\begin{align}
		S_{A}(\ket{\Omega})\le2\hbin\left(\norm{(\mathds{1}-\mathsf{P}_{0\rightarrow s_{\mathrm{crit}}})\ket{\Omega}}^{2}\right)+2S_{A}(\ket{\psi})+2\norm{(\mathds{1}-\mathsf{P}_{0\rightarrow s_{\mathrm{crit}}})\ket{\Omega}}^{2}n\log d. \label{eq:finite-depth-exact-entropy-transfer}
	\end{align}
	The case $\norm{(\mathds{1}-\mathsf{P}_{0\rightarrow s_{\mathrm{crit}}})\ket{\Omega}}=0$ follows directly from Eq.~\eqref{eq:critical-retained-entropy}.
	Otherwise, $\ket{\psi}\in\operatorname{Ran}\mathsf{P}_{0\rightarrow s_{\mathrm{crit}}}$ by construction, so Eq.~\eqref{eq:critical-retained-entropy} applies to $S_{A}(\ket{\psi})$.
	Since $p=3$, Eq.~\eqref{eq:finite-depth-final-discarded-weight} gives $\norm{(\mathds{1}-\mathsf{P}_{0\rightarrow s_{\mathrm{crit}}})\ket{\Omega}}^{2}\le s_{\mathrm{crit},*}^{2}n^{-6}$.
	Hence the binary-entropy term in Eq.~\eqref{eq:finite-depth-exact-entropy-transfer} is $\mathcal{O}_{\varepsilon}(n^{-6}\log n)$, while its last term is $\mathcal{O}_{\varepsilon}(n^{-5})$.
	Combining Eq.~\eqref{eq:finite-depth-exact-entropy-transfer} with Eq.~\eqref{eq:critical-retained-entropy}, and using $L_{n}^{3s_{\mathrm{crit},*}/2+1}\le n^{\varepsilon/2}$ for sufficiently large $n$, yields $S_{A}(\ket{\Omega})\le C_{\varepsilon}n^{\varepsilon}$.
	This proves Eq.~\eqref{eq:critical-subpolynomial-entropy} on the dyadic box.
	
	For an arbitrary microscopic side length, use the embedding constructed in the proof of Corollary~\ref{cor:arbitrary-side-lengths}.
	The embedding itself applies for every $\alpha<D$, independently of the exponent restriction in the statement of that earlier corollary.
	It gives $n\le n'<2^{D}n$, a unique ground state $\ket{\Omega'}=\ket{\Omega}\otimes\ket{\mathbf{0}}_{\mathrm{extra}}$, gap at least $\Delta$, and a Kac interaction bound with modified fixed parameters.
	The covering assumption is preserved up to fixed constants, and $S_{A'}(\ket{\Omega'})=S_{A}(\ket{\Omega})$.
	Applying the dyadic result just proved to the enlarged system therefore proves the theorem on the original box, for either boundary condition.
\end{proof}

The theorem is an entropy upper bound for every fixed positive exponent, with constants and a starting size that may depend on that exponent.
It does not assert a bound by one fixed power of $\log n$.
In particular, substituting $\varepsilon=0$ or making $\varepsilon$ depend on $n$ in this proof is not justified by the stated estimates.

\section{The range \texorpdfstring{$(D+1)/2<\alpha<D$}{(D+1)/2 < alpha < D}}
\label{sec:supercritical-MFRG}

For $(D+1)/2<\alpha<D$, Proposition~\ref{prop:Renormalized_interactions}, in particular Eq.~\eqref{eq:mu-positive-recurrence-alpha-gt-Dplus1half}, gives
\begin{align}
	\mu_{s}=\mu_{0}\ell_{1:s}^{\alpha-(D+1)/2}, \qquad J^{(s)}=c_{\mathrm{int}}^{s}JZ_{s}.
\end{align}
The scale dependence of $\mu_{s}$ comes from the neighboring-block square sum.
In fact, Eq.~\eqref{eq:critical-block-kernel} gives
\begin{align}
	\mathfrak{G}_{\alpha}(\ell_{s+1})=\ell_{s+1}^{(D-1)/2}
\end{align}
in this range.
Since $D\beta=D-\alpha$,
\begin{align}
	\frac{\ell_{s+1}^{(D-1)/2}}{[n^{(s)}]^{\beta}}=\ell_{s+1}^{\alpha-(D+1)/2}\frac{1}{[n^{(s+1)}]^{\beta}}.
\end{align}
Thus each RG step leaves an additional short-distance factor $\ell_{s+1}^{\alpha-(D+1)/2}$, whose iteration produces $\mu_{s}$.

For $D/2<\alpha<(D+1)/2$, Eq.~\eqref{eq:critical-block-kernel} instead gives $\mathfrak{G}_{\alpha}(\ell_{s+1})=\ell_{s+1}^{D-\alpha}$, and Eq.~\eqref{eq:kac-scale-covariance} absorbs this block-size factor entirely into the Kac normalization.
Consequently, Eq.~\eqref{eq:J-recurrence-Dhalf-alpha-Dplus1half} gives $\mu_{s}=1$ at every scale in that interval.
For $(D+1)/2<\alpha<D$, by contrast, the additional short-distance factor grows with the cumulative block length.
Using $N_{s}=\ell_{1:s}^{D}=n/n^{(s)}$, the Moore-neighbor interaction envelope satisfies
\begin{align}
	\frac{J^{(s)}\mu_{s}}{[n^{(s)}]^{\beta}}=J\mu_{0}c_{\mathrm{int}}^{s}Z_{s}\frac{N_{s}^{(D-1)/(2D)}}{n^{\beta}}. \label{eq:neighbor-upper-bound-alpha-gt-Dplus1half}
\end{align}
Thus the induction leading to a polylogarithmic effective system size is no longer preserved.
Nevertheless, stopping earlier leaves a strictly positive smallness margin and yields a subvolume entropy theorem.

\subsection{The finite stopping scale and the entropy theorem}

Set
\begin{align}
	\theta:=\frac{D-1}{D}, \qquad \eta_{\alpha}:=1-\frac{2\beta}{\theta}=\frac{2\alpha-D-1}{D-1}, \qquad \lambda_{\alpha}:=1-\eta_{\alpha}=\frac{2(D-\alpha)}{D-1}. \label{eq:supercritical-exponents}
\end{align}
Then $0<\eta_{\alpha}<1$, and the scalar quantity
\begin{align}
	\mathfrak{r}_{s}:=\left(\frac{n^{\eta_{\alpha}}}{n^{(s)}}\right)^{\theta/2}=\frac{N_{s}^{\theta/2}}{n^{\beta}} \label{eq:supercritical-smallness-parameter}
\end{align}
records the normalized growth of the short-distance interaction.
In particular,
\begin{align}
	\mu_{s}[n^{(s)}]^{-\beta}=\mu_{0}\mathfrak{r}_{s}, \qquad \ell_{1:s}^{\alpha-(D+1)/2}[n^{(s)}]^{-\beta}=\mathfrak{r}_{s}. \label{eq:supercritical-scale-identities}
\end{align}
The second identity shows that the same scalar enters the supercritical reflection bound.
The natural volume scale is therefore $N_{s}\sim n^{\lambda_{\alpha}}$, corresponding to $n^{(s)}\sim n^{\eta_{\alpha}}$.
Here $N_{s}$ is the number of microscopic sites represented by one effective site, not its linear side length, which is $N_{s}^{1/D}$.
The following theorem gives an entropy bound with exponent arbitrarily close to the natural scale exponent $\eta_{\alpha}$ from above.

\begin{theorem}
	\label{thm:supercritical-subvolume-entropy}
	Assume the microscopic hypotheses with $D\ge2$ and $(D+1)/2<\alpha<D$.
	Let the cut satisfy the uniform codimension-one covering hypothesis in Theorem~\ref{thm:critical-subpolynomial-entropy}.
	Then, for every $\varepsilon>0$, there are constants $C_{\varepsilon}$ and $n_{0}(\varepsilon)$ such that
	\begin{align}
		S_{A}(\ket{\Omega}) \le C_{\varepsilon}n^{\eta_{\alpha}+\varepsilon} = C_{\varepsilon}n^{\frac{2\alpha-D-1}{D-1}+\varepsilon} \qquad \bigl(n\ge n_{0}(\varepsilon)\bigr). \label{eq:supercritical-subvolume-entropy}
	\end{align}
	The constants depend only on $\varepsilon$, $C_{A}$, and the fixed microscopic parameters, and not on $n$.
	As in the critical theorem, on an exact hierarchy it is enough to assume Eq.~\eqref{eq:regular-cut-count} with $\theta=(D-1)/D$.
	For $0<\varepsilon<1-\eta_{\alpha}$, the bound is subvolume.  
\end{theorem}

\begin{proof}
	It suffices to prove the assertion for $0<\varepsilon<1-\eta_{\alpha}$, since larger values follow from $S_{A}(\ket{\Omega})\le n\log d$.
	We first consider a dyadic cubic box.
	Keep $\theta$ and $\eta_{\alpha}$ as in Eq.~\eqref{eq:supercritical-exponents}, and define
	\begin{align}
		t:=\eta_{\alpha}+\frac{\varepsilon}{2}, \qquad a:=\eta_{\alpha}+\frac{\varepsilon}{4}, \qquad \delta:=\theta\left(1-\frac{\eta_{\alpha}}{a}\right)>0, \qquad \kappa:=\max\left\{\frac{1}{2},1-\delta\right\}, \label{eq:supercritical-parameters-one}
	\end{align}
	\begin{align}
		\nu := \frac{1}{4} \min \left\{1,\frac{\varepsilon}{t},\delta,\beta,2\beta\delta,D\varepsilon\right\}, \qquad \sigma := \min\left\{\frac{\delta-\nu}{2},\beta(1-\kappa)-\frac{\nu}{2}\right\}>0. \label{eq:supercritical-parameters-two}
	\end{align}
	To check the last inequality, observe that $1-\kappa=\min\{1/2,\delta\}$, $\nu\le\delta/4$, $\nu\le\beta/4$, and $\nu\le\beta\delta/2$.
	Use the dyadic hierarchy in Eq.~\eqref{eq:critical-dyadic-hierarchy} with this value of $\nu$, and let $s_{\mathrm{super}}$ be the first stage for which $n^{(s_{\mathrm{super}})}\le n^{t}$.
	Equation~\eqref{eq:critical-hierarchy-size} gives
	\begin{align}
		n^{a}\le n^{t(1-\nu)}<n^{(s_{\mathrm{super}})}\le n^{t}, \qquad n^{(s)}\ge n^{a} \quad (0\le s\le s_{\mathrm{super}}), \label{eq:supercritical-terminal-size}
	\end{align}
	because $t\nu\le\varepsilon/4$.
	The same depth argument as in Eq.~\eqref{eq:critical-finite-depth} gives
	\begin{align}
		s_{\mathrm{super}}\le s_{\mathrm{super},*}:=1+\left\lceil\frac{\log(1/t)}{-\log(1-\nu/2)}\right\rceil. \label{eq:supercritical-finite-depth}
	\end{align}
	In particular, $s_{\mathrm{super},*}$ is independent of $n$.
	Fix
	\begin{align}
		p:=3, \qquad z_{s}:=z, \qquad z:=\max\left\{1,\left\lceil\frac{16(p+3)}{a\sigma}\right\rceil\right\}. \label{eq:supercritical-cutoff}
	\end{align}
	By Eq.~\eqref{eq:supercritical-terminal-size}, the smallness parameter in Eq.~\eqref{eq:supercritical-smallness-parameter} satisfies
	\begin{align}
		\mathfrak{r}_{s}^{2}=n^{\theta\eta_{\alpha}}[n^{(s)}]^{-\theta}\le[n^{(s)}]^{-\delta}, \qquad 0\le s\le s_{\mathrm{super}}. \label{eq:supercritical-uniform-smallness}
	\end{align}
	Thus, $\mathfrak{r}_{s}^{2}$ remains polynomially small in $n^{(s)}$ at every defined stage, including the terminal stage.
	
	We now close the finite-depth induction.
	We shall obtain constants $G_{s},Q_{s},D_{s}$, independent of $n$, such that
	\begin{align}
		J^{(s)}=j_{s}:=J(c_{\mathrm{int}}z)^{s}, \qquad \bar{g}^{(s)}\le G_{s}, \qquad \Delta^{(s)}\ge\frac{\Delta}{2}, \label{eq:supercritical-energy-envelope}
	\end{align}
	\begin{align}
		q_{s,*}\le Q_{s}\mathfrak{r}_{s}^{2}, \qquad p_{0,i}^{(s)}>\frac{1}{2}, \qquad \log d^{(s)}\le D_{s}\log(en), \label{eq:supercritical-occupation-envelope}
	\end{align}
	\begin{align}
		\epsilon_{s}:=\norm{\ket{\Omega^{(s)}}-V_{s+1}^{\dagger}\ket{\Omega^{(s+1)}}}\le n^{-p} \qquad (s<s_{\mathrm{super}}). \label{eq:supercritical-error-envelope}
	\end{align}
	At $s=0$, Theorem~\ref{thm:single-scale-arbitrary-cut} and Eq.~\eqref{eq:supercritical-exponents} give
	\begin{align}
		q_{0,*}\le C_{\mathrm{mic}}n^{-2\beta}=C_{\mathrm{mic}}\mathfrak{r}_{0}^{2}, \qquad \mathfrak{r}_{0}^{2}=n^{\theta\eta_{\alpha}-\theta}=n^{-2\beta},
	\end{align}
	which initializes the occupation envelope and also implies $p_{0,i}^{(0)}>1/2$ for sufficiently large $n$.
	Equation~\eqref{eq:mu-positive-recurrence-alpha-gt-Dplus1half}, together with $z_{s}=z$, gives the stated value of $j_{s}$.
	
	Unlike in the critical case, $\mu_{s}$ is not uniformly bounded, so it must not be dropped from the local-energy and hopping estimates.
	For the normalized coupling row sum in Eq.~\eqref{eq:scale-kac-row-sum} and the internal square-sum strength in Eq.~\eqref{eq:preMFRG-internal-gamma}, the lattice estimates give a constant $C_{\mathrm{geo}}\ge1$, depending only on $D$ and $\alpha$, such that
	\begin{align}
		\gamma_{0}^{(s)} \le C_{\mathrm{geo}}\left(1+\mu_{s}[n^{(s)}]^{-\beta}\right)\le C_{\mathrm{geo}}(1+\mu_{0}), \qquad
		\bar{\gamma}_{L}^{(s)} \le C_{\mathrm{geo}}\mu_{s}. \label{eq:supercritical-row-bounds}
	\end{align}
	The first bound follows by separating the unweighted power-law row sum from the bounded number of excess Moore-neighbor terms.
	Equations~\eqref{eq:supercritical-scale-identities} and \eqref{eq:supercritical-uniform-smallness} give $\mu_{s}[n^{(s)}]^{-\beta}=\mu_{0}\mathfrak{r}_{s}\le\mu_{0}$, which yields the second inequality in the first bound.
	The bound on $\bar{\gamma}_{L}^{(s)}$ follows from $2\alpha>D$ and the same bounded Moore degree.
	Thus, the normalized row sum $\gamma_{0}^{(s)}$ remains uniformly bounded despite the growth of $\mu_{s}$, whereas the internal square-sum strength $\bar{\gamma}_{L}^{(s)}$ grows proportionally to $\mu_{s}$.
	
	Here is a noncircular choice of the remaining envelopes.
	Choose
	\begin{align}
		G_{0}:=\bar{g}+J+\Delta, \qquad C_{z}:=1+2z+2(2+17\sqrt{2})\sqrt{z}+C_{\mathrm{geo}}(8\mu_{0}z+1+\mu_{0}), \label{eq:supercritical-energy-constants}
	\end{align}
	\begin{align}
		G_{s+1}:=C_{z}(G_{s}+j_{s}+j_{s+1}), \qquad D_{0}:=\log d, \qquad D_{s+1}:=z(D_{s}+1), \label{eq:supercritical-envelope-recursions}
	\end{align}
	\begin{align}
		Q_{s}:=1+C_{\mathrm{mic}}+2C_{\mathrm{rob}}^{2}(c_{\mathrm{rob}}z)^{2s}\frac{(J\mu_{0})^{2}G_{s}}{\Delta^{3}}, \qquad 0\le s\le s_{\mathrm{super},*}. \label{eq:supercritical-occupation-constants}
	\end{align}
	Here $C_{\mathrm{mic}}$ is the microscopic deviation constant in Eq.~\eqref{eq:microscopic-deviation-constant}, while $C_{\mathrm{rob}}$ and $c_{\mathrm{rob}}$ are the fixed constants in the supercritical case of Proposition~\ref{prop:renormalized-robustness}.
	For fixed $\varepsilon$, these arrays are defined for all $0\le s\le s_{\mathrm{super},*}$ independently of $n$.
	The choices above also initialize the remaining envelopes because $\bar{g}^{(0)}\le\bar{g}\le G_{0}$, $\Delta^{(0)}=\Delta$, and $\log d^{(0)}=\log d\le D_{0}\log(en)$.
	
	Suppose the induction has been established through an active stage $s$.
	Equations~\eqref{eq:supercritical-uniform-smallness}, \eqref{eq:supercritical-occupation-envelope}, and $b_{s+1}\le[n^{(s)}]^{\nu}$ give
	\begin{align}
		b_{s+1}q_{s,*}\le Q_{s}[n^{(s)}]^{\nu-\delta}=o(1), \label{eq:supercritical-block-mean}
	\end{align}
	uniformly over the finitely many stages, since $n^{(s)}\ge n^{a}$ and $\nu<\delta$.
	For each parent block, choose the auxiliary complement norm-tail tolerance
	\begin{align}
		\epsilon_{\mathrm{c},s}:=n^{-(p+4)}. \label{eq:supercritical-complement-tolerance}
	\end{align}
	Applying Eq.~\eqref{eq:preMFRG-cutoff-scaling} at the current scale produces a complement cutoff $m_{\mathrm{c},s}$ with
	\begin{align}
		m_{\mathrm{c},s}
		\le2n^{(s)}q_{s,*}+\sqrt{\frac{2\bar{g}^{(s)}n^{(s)}}{\Delta^{(s)}}}\log\frac{2}{\epsilon_{\mathrm{c},s}}+2  \le C_{s}\left([n^{(s)}]^{1-\delta}+[n^{(s)}]^{1/2}\log(en)+1\right)
		\le C_{s}\log(en)[n^{(s)}]^{\kappa}. \label{eq:supercritical-complement-cutoff}
	\end{align}
	The last bound uses Eq.~\eqref{eq:supercritical-uniform-smallness}, the inductive energy and gap estimates, and $\kappa=\max\{1/2,1-\delta\}$.
	Throughout this proof, $C_{s}$ denotes a finite constant depending only on the fixed parameters and the previously defined envelope constants, but not on $n$ or the block.
	
	Lemma~\ref{lem:preMFRG-compression}, applied to the auxiliary complement projection, gives
	\begin{align}
		e_{\mathrm{c},s}
		&\le\frac{2\bar{g}^{(s)}n^{(s)}\epsilon_{\mathrm{c},s}^{2}}{1-\epsilon_{\mathrm{c},s}^{2}}, \qquad
		\eta_{\mathrm{c},s}:=\norm{\ket{\Omega^{(s)}}-\ket{\widehat{\bar{\Omega}}^{(s)}}}
		\le\sqrt{\frac{4\bar{g}^{(s)}n^{(s)}\epsilon_{\mathrm{c},s}^{2}}{\Delta^{(s)}(1-\epsilon_{\mathrm{c},s}^{2})}}. \label{eq:supercritical-auxiliary-compression}
	\end{align}
	Here $\ket{\widehat{\bar{\Omega}}^{(s)}}$ denotes the ground state of the Hamiltonian after the auxiliary complement cutoff, embedded back into the scale-$s$ Hilbert space.
	Since $n^{(s)}\le n$ and $s\le s_{\mathrm{super},*}$, these estimates imply, for sufficiently large $n$,
	\begin{align}
		\bar{\Delta}_{\mathrm{c},s}\ge\Delta^{(s)}-e_{\mathrm{c},s}\ge\frac{\Delta^{(s)}}{2}\ge\frac{\Delta}{4}, \qquad \eta_{\mathrm{c},s}\le n^{-(p+3)}. \label{eq:supercritical-auxiliary-gap-error}
	\end{align}
	Since Eq.~\eqref{eq:supercritical-block-mean} gives $b_{s+1}q_{s,*}=o(1)$ and Eq.~\eqref{eq:supercritical-auxiliary-gap-error} gives $\eta_{\mathrm{c},s}=o(1)$, the zero-median criterion in Eq.~\eqref{eq:preMFRG-zero-median-criterion} yields $m_{*}=0$ for sufficiently large $n$.
	
	We now verify that the hopping between deviation-number shells is uniformly small on the finite interval $0\le x\le z$.
	Applying Eq.~\eqref{Parameter_bar_J_m_0/L} at scale $s$ with $m_{0}=m_{\mathrm{c},s}$ and $|L|=b_{s+1}$, and using Eq.~\eqref{eq:supercritical-row-bounds}, Eq.~\eqref{eq:dimension-free-onsite-simplification}, $F_{\alpha}(n^{(s)},m_{\mathrm{c},s})=m_{\mathrm{c},s}^{\beta}$, and $|\partial L|\le2Db_{s+1}^{\theta}$, gives
	\begin{align}
		\bar{J}_{m_{\mathrm{c},s},L}^{(s)}[0,z]
		&\le C_{s}\left[\mathfrak{r}_{s}\sqrt{b_{s+1}}+[n^{(s)}]^{-\beta}m_{\mathrm{c},s}^{\beta}\sqrt{b_{s+1}\log(en^{(s)})}+\mathfrak{r}_{s}b_{s+1}^{\theta/2}\right] \notag\\
		&\le C_{s}[\log(en)]^{\beta+1/2}\left[[n^{(s)}]^{-(\delta-\nu)/2}+[n^{(s)}]^{-\beta(1-\kappa)+\nu/2}+[n^{(s)}]^{-(\delta-\nu\theta)/2}\right] \notag\\
		&\le3C_{s}\log(en)[n^{(s)}]^{-\sigma}. \label{eq:supercritical-shell-hopping}
	\end{align}
	Here the first term contains both the internal and on-site contributions, while the last term is the excess Moore contribution.
	Since $\beta<1/2$ and $\theta<1$, the last Moore exponent is at least $(\delta-\nu)/2$.
	The definition of $\sigma$ in Eq.~\eqref{eq:supercritical-parameters-two} therefore makes every exponent strictly positive.
	Because $n^{(s)}\ge n^{a}$ and the depth is bounded, Eqs.~\eqref{eq:supercritical-shell-hopping} and \eqref{eq:supercritical-auxiliary-gap-error} imply
	\begin{align}
		\rho_{s}:=\frac{\bar{J}_{m_{\mathrm{c},s},L}^{(s)}[0,z]}{\bar{\Delta}_{\mathrm{c},s}}\le n^{-a\sigma/2}<1 \label{eq:supercritical-hopping-ratio}
	\end{align}
	for all sufficiently large $n$, uniformly in the block and the active stage.
	Lemma~\ref{lem:preMFRG-small-hopping}, followed by Eq.~\eqref{eq:tail-after-global-truncation}, now gives
	\begin{align}
		\norm{\Pi_{L,>z}^{(s)}\ket{\Omega^{(s)}}}
		\le\eta_{\mathrm{c},s}+\rho_{s}^{z/8}
		\le n^{-(p+3)}+n^{-a\sigma z/16}
		\le2n^{-(p+3)}. \label{eq:supercritical-local-tail}
	\end{align}
	If the hopping envelope vanishes, Lemma~\ref{lem:preMFRG-small-hopping} gives a zero auxiliary tail directly, so the same bound follows from $\eta_{\mathrm{c},s}\le n^{-(p+3)}$.
	
	Let $w_{s}$ and $e_{s}$ be the discarded weight and energy defect from Eq.~\eqref{eq:step-energy-error}.
	There are $n^{(s)}/b_{s+1}\le n$ parent blocks, so Eq.~\eqref{eq:preMFRG-block-union} and the shifted-norm bound following from Eq.~\eqref{all_to_all_cond_s_th_Re} give
	\begin{align}
		w_{s}\le4n^{-2p-5}, \qquad
		\frac{e_{s}}{\Delta} \le \frac{2\bar{g}^{(s)}n^{(s)}}{\Delta} \frac{w_{s}}{1-w_{s}} \le C_{s}n^{-2p-4}. \label{eq:supercritical-step-defect}
	\end{align}
	Lemma~\ref{lem:preMFRG-compression}, applied to the retained projection at the step $s\to s+1$, consequently yields
	\begin{align}
		\Delta^{(s+1)} \ge \Delta^{(s)}-e_{s}, \qquad
		\epsilon_{s} \le \sqrt{\frac{2e_{s}}{\Delta^{(s)}}}\le n^{-p}. \label{eq:supercritical-step-compression}
	\end{align}
	Since $s_{\mathrm{super}}\le s_{\mathrm{super},*}$ and the right-hand side of Eq.~\eqref{eq:supercritical-step-defect} tends to zero uniformly over these finitely many stages, $n$ can be chosen so that every step has $e_{s}\le\Delta/[4(s_{\mathrm{super},*}+1)]$.
	Inductively,
	\begin{align}
		\Delta^{(s)} \ge \Delta-\sum_{r<s}e_{r} \ge \frac{3\Delta}{4},
	\end{align}
	which in particular maintains the weaker gap envelope used above and proves existence and uniqueness of the effective ground state at every step.
	
	Equations~\eqref{eq:supercritical-block-mean} and \eqref{eq:supercritical-step-compression} imply $b_{s+1}q_{s,*}+\epsilon_{s}<1/2$.
	Lemma~\ref{lem:inherited-reference} therefore gives $p_{0,j}^{(s+1)}>1/2$ and
	\begin{align}
		\bra{0_{j}^{(s+1)}}V_{s+1,j}Q_{i}^{(s)}V_{s+1,j}^{\dagger}\ket{0_{j}^{(s+1)}}\le2(q_{s,*}+\epsilon_{s}). \label{eq:supercritical-reference-child}
	\end{align}
	
	We next verify the new energy envelope without using the new occupation estimate.
	Proposition~\ref{Prop:renormalized_interaction/on-site}, together with Eq.~\eqref{eq:supercritical-row-bounds} and
	$\sum_{j'\ne j}\norm{h_{jj'}^{(s+1)}}\le J^{(s+1)}\gamma_{0}^{(s+1)}\le C_{\mathrm{geo}}(1+\mu_{0})J^{(s+1)}$, gives
	\begin{align}
		\bar{g}^{(s+1)}
		\le8C_{\mathrm{geo}}\mu_{0}zJ^{(s)}\mathfrak{r}_{s}\sqrt{b_{s+1}}
		+2z\bar{g}^{(s)}
		+2(2+17\sqrt{2})\bar{g}^{(s)}\sqrt{zb_{s+1}q_{s,*}}
		+C_{\mathrm{geo}}(1+\mu_{0})J^{(s+1)}. \label{eq:supercritical-energy-step}
	\end{align}
	Equations~\eqref{eq:supercritical-uniform-smallness} and \eqref{eq:critical-hierarchy-size} give
	$\mathfrak{r}_{s}\sqrt{b_{s+1}}\le[n^{(s)}]^{-(\delta-\nu)/2}\le1$.
	For sufficiently large $n$, Eq.~\eqref{eq:supercritical-block-mean} also gives $b_{s+1}q_{s,*}\le1$.
	Equations~\eqref{eq:supercritical-energy-constants} and \eqref{eq:supercritical-envelope-recursions} therefore give $\bar{g}^{(s+1)}\le G_{s+1}$.
	
	Since Lemma~\ref{lem:inherited-reference} gives $p_{0,j}^{(s+1)}>1/2$ at every scale-$(s+1)$ site and the preceding gap estimate gives a unique gapped ground state with gap $\Delta^{(s+1)}$, Proposition~\ref{prop:renormalized-robustness}, together with Eq.~\eqref{eq:supercritical-scale-identities}, followed by Proposition~\ref{prop:robustness-fixed-point}, gives
	\begin{align}
		q_{s+1,*}
		\le
		\frac{C_{\mathrm{rob}}^{2}(c_{\mathrm{rob}}z)^{2(s+1)}(J\mu_{0})^{2}\bar{g}^{(s+1)}}{4[\Delta^{(s+1)}]^{3}}
		\mathfrak{r}_{s+1}^{2}
		\le
		Q_{s+1}\mathfrak{r}_{s+1}^{2}. \label{eq:supercritical-occupation-step}
	\end{align}
	Finally, Eq.~\eqref{eq:effective-dimension-recursion} gives
	\begin{align}
		\log d^{(s+1)}
		\le z\bigl(\log d^{(s)}+\log b_{s+1}\bigr)
		\le z(D_{s}+1)\log(en)
		=D_{s+1}\log(en).
	\end{align}
	This closes Eqs.~\eqref{eq:supercritical-energy-envelope}--\eqref{eq:supercritical-error-envelope}.
	There is no circular choice of constants: $a,t,\delta,\kappa,\nu,\sigma,s_{\mathrm{super},*},z$ are fixed first, the envelope arrays are fixed next, and only then is $n_{0}(\varepsilon)$ increased to meet the finitely many smallness conditions.
	
	We turn to the entropy, including the entanglement generated by decoding across the microscopic cut.
	Lemma~\ref{lem:weighted-tree-entropy} is applicable because Eqs.~\eqref{eq:local-RG-partial-isometry-identities} and \eqref{eq:Pi-s-plus-one} define the hierarchy of low-deviation retained spaces, while Eq.~\eqref{eq:supercritical-reference-child} provides the required child-reference occupation bound.
	For its weights, $B_{s_{\mathrm{super}}}\le C_{\varepsilon}\log(en)$ because $s_{\mathrm{super}}\le s_{\mathrm{super},*}$, while $z$ and the coefficients $D_{s}$ are bounded over this finite range.
	The auxiliary affine weight satisfies
	\begin{align}
		y_{s}=\frac{1}{2}[n^{(s)}]^{-2\beta}\le\frac{1}{2}\mathfrak{r}_{s}^{2}, \label{eq:supercritical-affine-weight}
	\end{align}
	since $[n^{(s)}]^{-2\beta}/\mathfrak{r}_{s}^{2}=(n^{(s)}/n)^{\theta\eta_{\alpha}}\le1$.
	Equation~\eqref{eq:supercritical-error-envelope} also gives $\epsilon_{s}\le n^{-p}\le n^{-2\beta}\le\mathfrak{r}_{s}^{2}$.
	The covering hypothesis and Eq.~\eqref{eq:regular-cut-count} give
	\begin{align}
		\mathfrak{r}_{s}^{2}b_{s+1}|\mathcal{C}_{s+1}(A)|
		\le C_{A}n^{\theta\eta_{\alpha}}b_{s+1}^{1-\theta}
		\le C_{A}n^{\theta\eta_{\alpha}+\nu/D}
		\le C_{A}n^{\eta_{\alpha}+\varepsilon/4}. \label{eq:supercritical-weighted-cut-count}
	\end{align}
	The last inequality uses $\nu/D\le\varepsilon/4$ and $\theta\eta_{\alpha}\le\eta_{\alpha}$.
	Using Eqs.~\eqref{eq:supercritical-occupation-envelope}, \eqref{eq:supercritical-error-envelope}, \eqref{eq:supercritical-affine-weight}, and \eqref{eq:supercritical-weighted-cut-count}, together with $B_{s_{\mathrm{super}}}\le C_{\varepsilon}\log(en)$, Eq.~\eqref{eq:weighted-tree-entropy-bound} gives, for every normalized $\ket{\psi}\in\operatorname{Ran}\mathsf{P}_{0\rightarrow s_{\mathrm{super}}}$,
	\begin{align}
		S_{A}(\ket{\psi})
		\le C_{\varepsilon}\log(en)\left(n^{(s_{\mathrm{super}})}+n^{\theta\eta_{\alpha}+\nu/D}\right) \le C_{\varepsilon}n^{\eta_{\alpha}+\varepsilon/2}\log(en). \label{eq:supercritical-retained-entropy}
	\end{align}
	Here $n^{(s_{\mathrm{super}})}\le n^{t}=n^{\eta_{\alpha}+\varepsilon/2}$, while $\theta\eta_{\alpha}+\nu/D\le\eta_{\alpha}+\varepsilon/4$.
	
	Finally, we transfer the retained-space entropy bound to the exact ground state.
	Telescoping the one-step errors in Eq.~\eqref{eq:supercritical-step-compression} gives
	\begin{align}
		\norm{\ket{\Omega}-V_{0\rightarrow s_{\mathrm{super}}}^{\dagger}\ket{\Omega^{(s_{\mathrm{super}})}}}
		\le s_{\mathrm{super},*}n^{-p}.
	\end{align}
	Since the decoded effective ground state belongs to $\operatorname{Ran}\mathsf{P}_{0\rightarrow s_{\mathrm{super}}}$,
	\begin{align}
		\norm{(\mathds{1}-\mathsf{P}_{0\rightarrow s_{\mathrm{super}}})\ket{\Omega}}^{2}
		\le s_{\mathrm{super},*}^{2}n^{-2p}. \label{eq:supercritical-final-discarded-weight}
	\end{align}
	For $0<\norm{(\mathds{1}-\mathsf{P}_{0\rightarrow s_{\mathrm{super}}})\ket{\Omega}}^{2}<1$, write the orthogonal decomposition
	\begin{align}
		\ket{\Omega}
		=\sqrt{1-\norm{(\mathds{1}-\mathsf{P}_{0\rightarrow s_{\mathrm{super}}})\ket{\Omega}}^{2}}\ket{\psi}
		+\norm{(\mathds{1}-\mathsf{P}_{0\rightarrow s_{\mathrm{super}}})\ket{\Omega}}\ket{\xi},
	\end{align}
	where
	\begin{align}
		\ket{\psi}
		:=\frac{\mathsf{P}_{0\rightarrow s_{\mathrm{super}}}\ket{\Omega}}
		{\sqrt{1-\norm{(\mathds{1}-\mathsf{P}_{0\rightarrow s_{\mathrm{super}}})\ket{\Omega}}^{2}}}, \qquad
		\ket{\xi}
		:=\frac{(\mathds{1}-\mathsf{P}_{0\rightarrow s_{\mathrm{super}}})\ket{\Omega}}
		{\norm{(\mathds{1}-\mathsf{P}_{0\rightarrow s_{\mathrm{super}}})\ket{\Omega}}}.
	\end{align}
	Lemma~\ref{lem:orthogonal-superposition-entropy} with two components, together with $S_{A}(\ket{\xi})\le n\log d$, gives
	\begin{align}
		S_{A}(\ket{\Omega})
		\le2\hbin\left(\norm{(\mathds{1}-\mathsf{P}_{0\rightarrow s_{\mathrm{super}}})\ket{\Omega}}^{2}\right)
		+2S_{A}(\ket{\psi})
		+2\norm{(\mathds{1}-\mathsf{P}_{0\rightarrow s_{\mathrm{super}}})\ket{\Omega}}^{2}n\log d. \label{eq:supercritical-exact-entropy-transfer}
	\end{align}
	The case $\norm{(\mathds{1}-\mathsf{P}_{0\rightarrow s_{\mathrm{super}}})\ket{\Omega}}=0$ follows directly from Eq.~\eqref{eq:supercritical-retained-entropy}.
	Otherwise, $\ket{\psi}\in\operatorname{Ran}\mathsf{P}_{0\rightarrow s_{\mathrm{super}}}$ by construction, so Eq.~\eqref{eq:supercritical-retained-entropy} applies to $S_{A}(\ket{\psi})$.
	Since $p=3$, Eq.~\eqref{eq:supercritical-final-discarded-weight} gives $\norm{(\mathds{1}-\mathsf{P}_{0\rightarrow s_{\mathrm{super}}})\ket{\Omega}}^{2}\le s_{\mathrm{super},*}^{2}n^{-6}$.
	Hence the binary-entropy term in Eq.~\eqref{eq:supercritical-exact-entropy-transfer} is $\mathcal{O}_{\varepsilon}(n^{-6}\log n)$, while its last term is $\mathcal{O}_{\varepsilon}(n^{-5})$.
	Combining Eq.~\eqref{eq:supercritical-exact-entropy-transfer} with Eq.~\eqref{eq:supercritical-retained-entropy}, and using $\log(en)\le n^{\varepsilon/2}$ for sufficiently large $n$, yields
	\begin{align}
		S_{A}(\ket{\Omega})\le C_{\varepsilon}n^{\eta_{\alpha}+\varepsilon}.
	\end{align}
	This proves Eq.~\eqref{eq:supercritical-subvolume-entropy} on the dyadic box.
	
	For an arbitrary microscopic side length, use the embedding constructed in the proof of Corollary~\ref{cor:arbitrary-side-lengths}.
	The embedding applies for the present $\alpha<D$ and gives $n\le n'<2^{D}n$, a unique ground state $\ket{\Omega'}=\ket{\Omega}\otimes\ket{\mathbf{0}}_{\mathrm{extra}}$, gap at least $\Delta$, and a Kac interaction bound with modified fixed parameters.
	The covering assumption is preserved up to fixed constants, and $S_{A'}(\ket{\Omega'})=S_{A}(\ket{\Omega})$.
	Applying the dyadic result just proved to the enlarged system and using $n'<2^{D}n$ therefore proves Eq.~\eqref{eq:supercritical-subvolume-entropy} on the original box, for either boundary condition.
\end{proof}

The stopping rule also makes explicit how close the construction reaches to the natural block-volume scale.
Indeed, Eq.~\eqref{eq:supercritical-terminal-size} and $N_{s}=n/n^{(s)}$ give
\begin{align}
	n^{\lambda_{\alpha}-\varepsilon/2}
	\le N_{s_{\mathrm{super}}}
	<n^{\lambda_{\alpha}-\varepsilon/4},
	\qquad
	\lambda_{\alpha}=1-\eta_{\alpha}
	=\frac{2(D-\alpha)}{D-1}. \label{eq:supercritical-proved-volume-window}
\end{align}
Thus, for every fixed $\varepsilon>0$, the construction reaches a block volume whose exponent is arbitrarily close to the natural exponent $\lambda_{\alpha}$ from below as $\varepsilon$ is chosen smaller.
The theorem does not assert control at the limiting scale $N_{s}\asymp n^{\lambda_{\alpha}}$, nor does it prove an entropy bound of the form $n^{\eta_{\alpha}}[\log(en)]^{c}$ with one fixed $c$.
Such stronger conclusions would require control of the RG estimates as the positive smallness margin above the scale $n^{\eta_{\alpha}}$ vanishes, which is not needed here.

\section{Optimality of the regular-cut threshold}
\label{sec:regular-cut-sharpness}

Proposition~\ref{prop:dimer-obstruction} shows that $\alpha=D/2$ is the optimal threshold for arbitrary bipartitions under the microscopic assumptions of this work.
For regular codimension-one cuts, the preceding results instead single out the second threshold $\alpha=(D+1)/2$.
Theorem~\ref{thm:regular-cut-polylog-entropy} gives a polylogarithmic entropy bound for $D/2<\alpha<(D+1)/2$, while Theorem~\ref{thm:critical-subpolynomial-entropy} and Theorem~\ref{thm:supercritical-subvolume-entropy} treat $\alpha=(D+1)/2$ and $(D+1)/2<\alpha<D$, respectively.
The purpose of this section is to show that the location of the regular-cut threshold is genuine and is not an artifact of the MFRG estimates.

We construct an explicit Kac-normalized qubit Hamiltonian on a cubic lattice whose entanglement across a flat codimension-one cut admits lower bounds with three different scalings, separated by $2\alpha=D+1$.
For $0\le\alpha<D/2$, the construction gives a logarithmic lower bound, which matches the logarithmic upper bound of Theorem~\ref{thm:single-scale-arbitrary-cut} and therefore shows that this upper bound is optimal in its system-size dependence.
At $\alpha=(D+1)/2$, the construction gives a $(\log n)^{2}$ lower bound.
For $(D+1)/2<\alpha<D$, it gives a positive power of $n$ times a logarithm.
Moreover, throughout an open interval extending above $(D+1)/2$, the ground state is unique and the spectral gap is bounded below uniformly in the system size.
Thus, the threshold separating the universally polylogarithmic regime from a regime in which polynomial regular-cut entanglement is possible cannot be shifted above $(D+1)/2$ without additional assumptions.

The result therefore establishes both the optimal logarithmic scaling for $0\le\alpha<D/2$ and optimality of the regular-cut threshold at $\alpha=(D+1)/2$.

\subsection{A Kac-normalized spin family and the main statements}

Fix an integer $D\ge1$ and a real number $0\le\alpha<D$.
Assume that $\ell_{\Lambda}$ is even, so that $n=\ell_{\Lambda}^{D}$, and define
\begin{align}
	\Lambda :=\{0,\ldots,\ell_{\Lambda}-1\}^{D},\qquad A:= \left\{x\in\Lambda:x_{1}<\frac{\ell_{\Lambda}}{2}\right\},\qquad B:=\Lambda\setminus A, \qquad |A|=|B|=:N_{\mathrm{h}}=\frac{n}{2}. \label{eq:sharpness-flat-cut}
\end{align}
We use open boundary conditions in this construction.
Enumerate the two halves in lexicographic order, writing their sites as $x_{A}(a)$ and $x_{B}(b)$, with $a,b=0,\ldots,N_{\mathrm{h}}-1$.
Each site carries a qubit, with $\sigma_{x}^{+}:=  \ket{1}_{x} \bra{0}_{x}$, $\sigma_{x}^{-}:= \ket{0}_{x} \bra{1}_{x}$, and $n_{x}:=\sigma_{x}^{+} \sigma_{x}^{-}$.
Define
\begin{align}
	H_{F}&:=\sum_{x\in\Lambda}n_{x}
	+\sum_{a,b=0}^{N_{\mathrm{h}}-1}\left(
	F_{ab}\sigma_{x_{A}(a)}^{+}\sigma_{x_{B}(b)}^{+}
	+\overline{F_{ab}}\sigma_{x_{B}(b)}^{-}\sigma_{x_{A}(a)}^{-}
	\right), \label{eq:sharpness-spin-model}\\
	F_{ab}&:=\ell_{\Lambda}^{\alpha-D}\frac{e^{2\pi iab/N_{\mathrm{h}}}}{r_{ab}^{\alpha}},\qquad
	r_{ab}:=\norm{x_{A}(a)-x_{B}(b)}_{1}. \label{eq:sharpness-coupling-matrix}
\end{align}
No factor $N_{\mathrm{h}}^{-1/2}$ is inserted in the Fourier phase.  
For a normalized pure state $\ket{\psi}$, we write $S_{A}(\ket{\psi})$ as in Eq.~\eqref{eq:entanglement-entropy-definition}.

Since $\ell_{\Lambda}^{\alpha-D}=n^{-\beta}$ with $\beta=1-\alpha/D$, every nonzero pair term in Eq.~\eqref{eq:sharpness-spin-model} has norm $|F_{ab}|=n^{-\beta}r_{ab}^{-\alpha}$.
All pair terms within $A$ or within $B$ vanish.
Thus, the Hamiltonian satisfies Eq.~\eqref{eq:long-range-interaction} with $J=\mu_{0}=1$, while $\max_{x}\norm{n_{x}}=1$.
The row-sum estimate in Eq.~\eqref{eq:preMFRG-local-and-global-norm} therefore gives a local energy scale bounded independently of $n$.
Hence, this family belongs to the same microscopic class considered throughout the preceding sections whenever its ground state is unique and uniformly gapped.
Here and below, let $E_{0}(H_{F})\le E_{1}(H_{F})$ denote the two lowest eigenvalues of $H_{F}$, counted with multiplicity, and define the spectral gap by $\Delta_{F}:=E_{1}(H_{F})-E_{0}(H_{F})$.

\begin{prop}
	\label{prop:regular-cut-obstruction}
	For every fixed $D\ge1$ and $0\le\alpha<D$, there exist constants $c_{D,\alpha}>0$ and $\ell_{0}(D,\alpha)$ such that every normalized pure ground state $\ket{\Omega_{F}}$ of Eq.~\eqref{eq:sharpness-spin-model}, for even $\ell_{\Lambda}\ge \ell_{0}(D,\alpha)$, satisfies
	\begin{align}
		S_{A}(\ket{\Omega_{F}})\ge c_{D,\alpha}
		\begin{cases}
			\log n,&0\le\alpha<(D+1)/2,\quad D\ge1,\\
			(\log n)^{2},&\alpha=(D+1)/2,\quad D\ge2,\\
			n^{(2\alpha-D-1)/D}\log n,&(D+1)/2<\alpha<D,\quad D\ge2.
		\end{cases} \label{eq:sharpness-main-lower}
	\end{align} 
	Moreover, if $0 \le \alpha < (3D+1)/4$, then, for all sufficiently large even $\ell_{\Lambda}$, the ground state is unique and the spectral gap satisfies
	\begin{align}
		\Delta_{F}\ge1-o(1). \label{eq:sharpness-asymptotic-gap}
	\end{align}
\end{prop}

\begin{corollary}
	\label{cor:regular-cut-threshold-sharp}
	For every $D\ge1$ and $0\le\alpha<D/2$, the family in Eq.~\eqref{eq:sharpness-spin-model} has, for all sufficiently large even $\ell_{\Lambda}$, a unique ground state with a uniformly positive spectral gap and
	\begin{align}
		S_{A}(\ket{\Omega_{F}})\ge c_{D,\alpha}\log n.
	\end{align}
	Together with Theorem~\ref{thm:single-scale-arbitrary-cut}, this shows that the logarithmic arbitrary-cut entropy bound is optimal in its system-size dependence, even for regular codimension-one cuts.
	
	Assume now that $D\ge2$.
	At $\alpha=(D+1)/2$, the same family has, for all sufficiently large even $\ell_{\Lambda}$, a unique ground state with a uniformly positive spectral gap and
	\begin{align}
		S_{A}(\ket{\Omega_{F}})\ge c_{D}(\log n)^{2}. \label{eq:sharpness-critical-lower}
	\end{align}
	Furthermore, for every $(D+1)/2<\alpha<(3D+1)/4$, the same family has, for all sufficiently large even $\ell_{\Lambda}$, a unique ground state with a uniformly positive spectral gap and
	\begin{align}
		S_{A}(\ket{\Omega_{F}})\ge c_{D,\alpha}\,n^{(2\alpha-D-1)/D}\log n. \label{eq:sharpness-polynomial-lower}
	\end{align}
	Consequently, the threshold $\alpha=(D+1)/2$ in the regular codimension-one entropy problem is optimal in the following sense: under the microscopic assumptions of this work, a universal polylogarithmic entropy bound cannot hold throughout any right neighborhood of $\alpha=(D+1)/2$.
\end{corollary}

For $D=1$, the uniqueness regime in $0 \le \alpha < (3D+1)/4$ is the entire interval $0\le\alpha<1$.
For $D\ge2$, the argument below does not establish uniqueness for $(3D+1)/4\le\alpha<D$; the entropy lower bound in Proposition~\ref{prop:regular-cut-obstruction} nevertheless continues to hold for every pure ground state in that range.

The proof is organized in four steps.
First, we estimate the operator norm, Hilbert--Schmidt weight, and collision weight of the coupling matrix $F$.
Second, we embed the spin system isometrically into a quadratic bosonic pairing problem and compare their ground energies without identifying the spin ground state with a Gaussian state.
Third, the resulting energy defect is converted into an entanglement lower bound by extracting many weakly squeezed mode blocks and using squashed entanglement.
Finally, a hard-core projection estimate transfers the bosonic gap to the spin Hamiltonian in the regime in $0 \le \alpha < (3D+1)/4$ .

\subsection{Coupling-matrix estimates}

For the coupling matrix $F$ in Eq.~\eqref{eq:sharpness-coupling-matrix}, set
\begin{align}
	f :=\norm{F}, \qquad W:=\operatorname{Tr}(FF^{\dagger}), \qquad 
	\mathfrak{R} :=
	\sum_{a}\left(\sum_{b}|F_{ab}|^{2}\right)^{2}
	+\sum_{b}\left(\sum_{a}|F_{ab}|^{2}\right)^{2}. \label{eq:sharpness-matrix-quantities}
\end{align}
Constants implicit in $\lesssim$ and $\asymp$ may depend on $D$ and $\alpha$, but not on $\ell_{\Lambda}$.
We use $C_{\mathrm{sh},D},c_{\mathrm{sh},D}>0$ for fixed dimension-dependent constants in the geometric estimates below.

\begin{lemma}
	\label{lem:sharpness-matrix}
	For $0\le\alpha<D$,
	\begin{align}
		f&\lesssim \ell_{\Lambda}^{-(D-\alpha)/2},\qquad
		W\ge\frac{1}{4D^{2\alpha}}, \label{eq:sharpness-op-bound}
	\end{align}
	and
	\begin{align}
		W\asymp
		\begin{cases}
			1,&0\le\alpha<(D+1)/2,\quad D\ge1,\\
			\log \ell_{\Lambda},&\alpha=(D+1)/2,\quad D\ge2,\\
			\ell_{\Lambda}^{2\alpha-D-1},&(D+1)/2<\alpha<D,\quad D\ge2.
		\end{cases} \label{eq:sharpness-HS-bound}
	\end{align}
	Furthermore,
	\begin{align}
		\mathfrak{R}\lesssim
		\begin{cases}
			\ell_{\Lambda}^{-D},&0\le\alpha<(2D+1)/4,\\
			\ell_{\Lambda}^{-D}\log \ell_{\Lambda},&\alpha=(2D+1)/4,\\
			\ell_{\Lambda}^{4\alpha-3D-1},&(2D+1)/4<\alpha<D.
		\end{cases} \label{eq:sharpness-collision-bound}
	\end{align}
	In particular, $\mathfrak{R}=o(1)$ whenever $\alpha<(3D+1)/4$.
\end{lemma}

\begin{proof}
	We first bound the operator norm $f=\norm{F}$.
	Let $M_{ab}:=e^{2\pi iab/N_{\mathrm{h}}}$.
	Since $M$ is the unnormalized discrete Fourier matrix of size $N_{\mathrm{h}}$, one has $\norm{M}=\sqrt{N_{\mathrm{h}}}$.
	For $t>0$, define
	\begin{align}
		(K_{t})_{ab}:=e^{-tr_{ab}}
		=e^{-t\norm{x_{A}(a)-x_{B}(b)}_{1}}.
	\end{align}
	For $0<\alpha<D$, using the Laplace representation $r^{-\alpha} = \Gamma(\alpha)^{-1}\int_{0}^{\infty}t^{\alpha-1}e^{-tr}\,dt$ gives
	\begin{align}
		F = \frac{\ell_{\Lambda}^{\alpha-D}}{\Gamma(\alpha)}  \int_{0}^{\infty}t^{\alpha-1}(M\circ K_{t})\,dt, \label{eq:sharpness-Laplace-F}
	\end{align}
	where $\circ$ denotes the entrywise, or Hadamard, product, so that $(M\circ K_{t})_{ab}=M_{ab}(K_{t})_{ab}=e^{2\pi iab/N_{\mathrm{h}}}e^{-tr_{ab}}$.
	Thus, the problem reduces to controlling the operator norm of the exponentially regularized Fourier matrix $M\circ K_{t}$.

	We obtain two complementary bounds on $\norm{M\circ K_{t}}$.
	The first uses the Fourier structure of $M$.
	Let $\mathbb{T}:=\mathbb{R}/(2\pi\mathbb{Z})$ and define the probability measure $\nu_{t}^{(1)}$ on $\mathbb{T}$ by
	\begin{align}
		d\nu_{t}^{(1)}(\theta)
		=\frac{1-e^{-2t}}{2\pi\left(1-2e^{-t}\cos\theta+e^{-2t}\right)}\,d\theta.
	\end{align}
	Its Fourier coefficients satisfy $\int_{\mathbb{T}}e^{im\theta}\,d\nu_{t}^{(1)}(\theta)=e^{-t|m|}$ for every $m\in\mathbb{Z}$.
	Therefore, with $\nu_{t}:=(\nu_{t}^{(1)})^{\otimes D}$,
	\begin{align}
		e^{-t\norm{x-y}_{1}}
		=\int_{\mathbb{T}^{D}}e^{i\theta\cdot x}e^{-i\theta\cdot y}\,d\nu_{t}(\theta). \label{eq:sharpness-exponential-kernel}
	\end{align} 
	For $\theta\in\mathbb{T}^{D}$, define the diagonal unitaries by $(D_{A}(\theta))_{aa}:=e^{i\theta\cdot x_{A}(a)}$ and $(D_{B}(\theta))_{bb}:=e^{i\theta\cdot x_{B}(b)}$.
	By Eq.~\eqref{eq:sharpness-exponential-kernel},
	\begin{align}
		M\circ K_{t} = \int_{\mathbb{T}^{D}}D_{A}(\theta)M D_{B}(\theta)^{\dagger}\,d\nu_{t}(\theta).
	\end{align}
	Hence $M\circ K_{t}$ is an average of matrices obtained from $M$ by left and right multiplication by unitaries.
	Since the operator norm is unitarily invariant and $\nu_{t}$ is a probability measure,
	\begin{align}
		\norm{M\circ K_{t}} \le \norm{M}  = \sqrt{N_{\mathrm{h}}}. \label{eq:sharpness-fourier-bound}
	\end{align}
	
	The second bound uses the spatial decay of $K_{t}$.
	Since $|M_{ab}|=1$, one has $\sum_{b}|(M\circ K_{t})_{ab}|=\sum_{b}e^{-tr_{ab}}\le\sum_{z\in\mathbb{Z}^{D}\setminus\{0\}}e^{-t\norm{z}_{1}}$, because for fixed $a$ the displacement vectors $x_{B}(b)-x_{A}(a)$ form a subset of $\mathbb{Z}^{D}\setminus\{0\}$.
	The same bound holds for every column.
	Using $\norm{z}_{1}=\sum_{\mu=1}^{D}|z_{\mu}|$, the infinite-lattice sum factorizes as
	\begin{align}
		\sum_{z\in\mathbb{Z}^{D}\setminus\{0\}}e^{-t\norm{z}_{1}}  =\prod_{\mu=1}^{D}\left(\sum_{m\in\mathbb{Z}}e^{-t|m|}\right)-1  =\left(\frac{1+e^{-t}}{1-e^{-t}}\right)^{D}-1 \le C_{\mathrm{sh},D}t^{-D}.
	\end{align}
	The Schur test, $\norm{A}\le\sqrt{(\sup_{a}\sum_{b}|A_{ab}|)(\sup_{b}\sum_{a}|A_{ab}|)}$, therefore gives
	\begin{align}
		\norm{M\circ K_{t}}\le C_{\mathrm{sh},D}t^{-D}. \label{eq:sharpness-decay-bound}
	\end{align}
	Combining Eqs.~\eqref{eq:sharpness-fourier-bound} and \eqref{eq:sharpness-decay-bound} yields
	\begin{align}
		\norm{M\circ K_{t}} \le \min\left\{\sqrt{N_{\mathrm{h}}},C_{\mathrm{sh},D}t^{-D}\right\}. \label{eq:sharpness-fourier-kernel-bound}
	\end{align}
	
	We now insert Eq.~\eqref{eq:sharpness-fourier-kernel-bound} into the Laplace representation in Eq.~\eqref{eq:sharpness-Laplace-F}.
	The two bounds in Eq.~\eqref{eq:sharpness-fourier-kernel-bound} become comparable at the scale $t\asymp\ell_{\Lambda}^{-1/2}$, since $\sqrt{N_{\mathrm{h}}}\asymp\ell_{\Lambda}^{D/2}$.
	Splitting the integral at this scale gives
	\begin{align}
		f \le \frac{\ell_{\Lambda}^{\alpha-D}}{\Gamma(\alpha)}
		\left[
		\sqrt{N_{\mathrm{h}}}\int_{0}^{\ell_{\Lambda}^{-\frac{1}{2}}}t^{\alpha-1}\,dt
		+C_{\mathrm{sh},D}\int_{\ell_{\Lambda}^{-\frac{1}{2}}}^{\infty}t^{\alpha-D-1}\,dt
		\right] \lesssim \ell_{\Lambda}^{-(D-\alpha)/2}. \label{eq:sharpness-operator-norm-proof}
	\end{align}
	For $\alpha=0$, the same estimate follows directly from $F=\ell_{\Lambda}^{-D}M$, and hence $f=\sqrt{N_{\mathrm{h}}}/\ell_{\Lambda}^{D}$.
	For every $0\le\alpha<D$, since $r_{ab}\le D\ell_{\Lambda}$ for every pair across the cut,
	\begin{align}
		W = \sum_{a,b}|F_{ab}|^{2}
		\ge N_{\mathrm{h}}^{2}\ell_{\Lambda}^{2\alpha-2D}(D\ell_{\Lambda})^{-2\alpha}
		=\frac{1}{4D^{2\alpha}}.
	\end{align}
	This proves Eq.~\eqref{eq:sharpness-op-bound}.

	We next determine the scaling of the Hilbert--Schmidt weight $W$ by grouping the cross-cut pairs according to their Manhattan distance.
	Let $C(r)$ denote the number of pairs $(x,y)\in A\times B$ with $\norm{x-y}_{1}=r$.
	For a displacement $h=y-x$ with $h_{1}>0$ and $\norm{h}_{1}=r\le\ell_{\Lambda}/4$, the number of pairs across the cut realizing this displacement is $h_{1}\prod_{j=2}^{D}(\ell_{\Lambda}-|h_{j}|)$, and the product lies between $(3\ell_{\Lambda}/4)^{D-1}$ and $\ell_{\Lambda}^{D-1}$.
	Restricting to nonnegative tangential displacements gives
	\begin{align}
		\sum_{\substack{h_{1}\ge1,\ h_{2},\ldots,h_{D}\ge0\\h_{1}+\cdots+h_{D}=r}}h_{1} = \binom{r+D-1}{D} \ge \frac{r^{D}}{D!}.
	\end{align}
	This yields the lower bound $C(r) \ge c_{\mathrm{sh},D} \ell_{\Lambda}^{D-1}r^{D}$ for $1 \le r \le \ell_{\Lambda}/4$.
	For the upper bound, each displacement $h$ with $\norm{h}_{1}=r$ contributes at most $r\ell_{\Lambda}^{D-1}$ pairs across the cut, while the number of such displacement vectors is at most $a_{D}r^{D-1}$ by Eq.~\eqref{eq:preMFRG-lattice-counting}.
	Therefore,
	\begin{align}
		C(r)&\le C_{\mathrm{sh},D}\ell_{\Lambda}^{D-1}r^{D}\qquad \left(1\le r\le D\ell_{\Lambda}\right), \label{eq:sharpness-pair-count-upper}\\
		C(r)&\ge c_{\mathrm{sh},D}\ell_{\Lambda}^{D-1}r^{D}\qquad \left(1\le r\le \frac{\ell_{\Lambda}}{4}\right). \label{eq:sharpness-pair-count-lower}
	\end{align}
	Using these estimates in the definition of $W$ gives
	\begin{align}
		W =\sum_{a,b}|F_{ab}|^{2} =\ell_{\Lambda}^{2\alpha-2D}\sum_{a,b}r_{ab}^{-2\alpha} =\ell_{\Lambda}^{2\alpha-2D}\sum_{r=1}^{D\ell_{\Lambda}}C(r)r^{-2\alpha} \asymp \ell_{\Lambda}^{2\alpha-D-1}\sum_{r=1}^{\ell_{\Lambda}}r^{D-2\alpha}. \label{eq:sharpness-HS-sum}
	\end{align}
	The three cases in Eq.~\eqref{eq:sharpness-HS-bound} now follow from the elementary power-sum estimates.

	Finally, we control the collision weight $\mathfrak{R}$.
	This quantity measures the concentration of the squared coupling weight within individual rows and columns of $F$.
	Fix a site $x_{A}(a)$ at depth $d\ge1$ from the flat interface.
	Every site $x_{B}(b)$ on the opposite side then satisfies $r_{ab}\ge d$, while Eq.~\eqref{eq:preMFRG-lattice-counting} implies that the number of sites at Manhattan distance $r$ from $x_{A}(a)$ is at most $a_{D}r^{D-1}$.
	Using $|F_{ab}|^{2}=\ell_{\Lambda}^{2\alpha-2D}r_{ab}^{-2\alpha}$ and grouping the sites in $B$ according to their distance from $x_{A}(a)$ therefore gives
	\begin{align}
		\sum_{b}|F_{ab}|^{2}
		\lesssim \ell_{\Lambda}^{2\alpha-2D}
		\sum_{r=d}^{D\ell_{\Lambda}}r^{D-1-2\alpha}. \label{eq:sharpness-row-square-sum}
	\end{align}
	The same estimate holds for the column sums, and there are $\ell_{\Lambda}^{D-1}$ sites at each depth on either side of the cut.
	If $2\alpha<D$, the sum in Eq.~\eqref{eq:sharpness-row-square-sum} is $\mathcal{O}(\ell_{\Lambda}^{D-2\alpha})$, so each row or column weight is $\mathcal{O}(\ell_{\Lambda}^{-D})$; after squaring and summing over all sites, this gives $\mathfrak{R}\lesssim\ell_{\Lambda}^{-D}$.
	At $2\alpha=D$, the row and column weights are bounded by $C\ell_{\Lambda}^{-D}[1+\log(\ell_{\Lambda}/d)]$, and
	\begin{align}
		\sum_{d\le\ell_{\Lambda}}[1+\log(\ell_{\Lambda}/d)]^{2}
		=\mathcal{O}(\ell_{\Lambda})
	\end{align}
	again gives $\mathfrak{R}\lesssim\ell_{\Lambda}^{-D}$.
	If $2\alpha>D$, Eq.~\eqref{eq:sharpness-row-square-sum} gives $\sum_{b}|F_{ab}|^{2}\lesssim\ell_{\Lambda}^{2\alpha-2D}d^{D-2\alpha}$.
	Squaring this estimate and summing over the $\ell_{\Lambda}^{D-1}$ sites at each depth on the two sides yields
	\begin{align}
		\mathfrak{R}
		\lesssim \ell_{\Lambda}^{4\alpha-3D-1}
		\sum_{d=1}^{\ell_{\Lambda} / 2} d^{2D-4\alpha}. \label{eq:sharpness-collision-depth-sum}
	\end{align}
	If $4\alpha<2D+1$, the sum grows as $\ell_{\Lambda}^{2D-4\alpha+1}$; if $4\alpha=2D+1$, it grows logarithmically; and if $4\alpha>2D+1$, it remains uniformly bounded.
	This proves Eq.~\eqref{eq:sharpness-collision-bound}.
	In particular, $\mathfrak{R}=o(1)$ whenever $\alpha<(3D+1)/4$.
\end{proof}

\subsection{Exact auxiliary compression}

We next introduce an auxiliary bosonic Fock-space representation through a local isometric embedding.
This representation provides a convenient quadratic reference problem for comparing ground-state energies and extracting entanglement bounds.

Let $a_{a}$ and $b_{b}$ be bosonic annihilation operators associated with the sites in $A$ and $B$, respectively, and define
\begin{align}
	H_{\mathrm{b}}
	:=\sum_{a}a_{a}^{\dagger}a_{a}
	+\sum_{b}b_{b}^{\dagger}b_{b}
	+\sum_{a,b}\left(
	F_{ab}a_{a}^{\dagger}b_{b}^{\dagger}
	+\overline{F_{ab}}b_{b}a_{a}
	\right). \label{eq:sharpness-bosonic-Hamiltonian}
\end{align}
The expression is initially defined on the finite-particle domain; its semibounded self-adjoint realization for $\norm{F}<1$ is specified by the diagonalization below.
For occupation sets $R\subseteq A$ and $T\subseteq B$, let $\ket{R}_{A}:=\prod_{a:\,x_{A}(a)\in R}\sigma_{x_{A}(a)}^{+}\ket{0}_{A}$ and $\ket{T}_{B}:=\prod_{b:\,x_{B}(b)\in T}\sigma_{x_{B}(b)}^{+}\ket{0}_{B}$, where $\ket{0}_{A}:=\bigotimes_{x\in A}\ket{0}_{x}$ and $\ket{0}_{B}:=\bigotimes_{x\in B}\ket{0}_{x}$.
Define the local occupation-basis isometries by
\begin{align}
	\mathsf{J}_{A}\ket{R}_{A}&:=\prod_{a:\,x_{A}(a)\in R}a_{a}^{\dagger}\ket{\mathrm{vac}}_{A},\qquad
	\mathsf{J}_{B}\ket{T}_{B}:=\prod_{b:\,x_{B}(b)\in T}b_{b}^{\dagger}\ket{\mathrm{vac}}_{B}.
\end{align}
With $\mathsf{J}:=\mathsf{J}_{A}\otimes \mathsf{J}_{B}$, compression to the hard-core subspaces gives the exact identities
\begin{align}
	\mathsf{J}^{\dagger}H_{\mathrm{b}}\mathsf{J}&=H_{F},\qquad
	S_{A}(\mathsf{J}\ket{\psi})=S_{A}(\ket{\psi}) \label{eq:sharpness-exact-compression}
\end{align}
for every spin state $\ket{\psi}$.
The entropy identity holds because $\mathsf{J}_{A}$ and $\mathsf{J}_{B}$ are local isometries with respect to the bipartition $A|B$.

For the auxiliary estimates below, we continue to write $f:=\norm{F}$ and $W:=\operatorname{Tr}(FF^{\dagger})$.
By Lemma~\ref{lem:sharpness-matrix}, $f\le1/2$ for all sufficiently large $\ell_{\Lambda}$.
Take a singular-value decomposition $F=U_{F}\operatorname{diag}(s_{k})V_{F}^{\dagger}$ and define
\begin{align}
	\widehat{a}_{k}&:=\sum_{a}\overline{(U_{F})_{ak}}a_{a},\qquad
	\widehat{b}_{k}:=\sum_{b}(V_{F})_{bk}b_{b}. \label{eq:sharpness-mode-rotations}
\end{align}
The absence of complex conjugation in the definition of $\widehat{b}_{k}$ is required by the pairing term.
Indeed, $a_{a}^{\dagger}=\sum_{k}\overline{(U_{F})_{ak}}\widehat{a}_{k}^{\dagger}$ and $b_{b}^{\dagger}=\sum_{k}(V_{F})_{bk}\widehat{b}_{k}^{\dagger}$, so its coefficient matrix becomes $U_{F}^{\dagger}FV_{F}=\operatorname{diag}(s_{k})$.
Both transformations preserve the canonical commutation relations and total particle number on their respective sides.
Their number-preserving unitary implementations are local with respect to $A|B$.
Thus Eq.~\eqref{eq:sharpness-bosonic-Hamiltonian} becomes a sum of independent blocks
\begin{align}
	h_{k}
	:=\widehat{a}_{k}^{\dagger}\widehat{a}_{k}
	+\widehat{b}_{k}^{\dagger}\widehat{b}_{k}
	+s_{k}\left(
	\widehat{a}_{k}^{\dagger}\widehat{b}_{k}^{\dagger}
	+\widehat{b}_{k}\widehat{a}_{k}
	\right). \label{eq:sharpness-bosonic-block}
\end{align}
Define $\varepsilon_{k}:=\sqrt{1-s_{k}^{2}}$, $r_{k}:=s_{k}/(1+\varepsilon_{k})$, and $\zeta_{k}:=r_{k}^{2}$.
The canonical annihilation operators
\begin{align}
	d_{k}:=\frac{\widehat{a}_{k}+r_{k}\widehat{b}_{k}^{\dagger}}{\sqrt{1-r_{k}^{2}}},\qquad
	e_{k}:=\frac{\widehat{b}_{k}+r_{k}\widehat{a}_{k}^{\dagger}}{\sqrt{1-r_{k}^{2}}}
\end{align}
give
\begin{align}
	h_{k} = (\varepsilon_{k}-1) + \varepsilon_{k}\left(d_{k}^{\dagger}d_{k}+e_{k}^{\dagger}e_{k}\right). \label{eq:sharpness-bosonic-diagonal-block}
\end{align}
The canonical transformation above is implemented by the two-mode squeezing unitary
\begin{align}
	\mathsf{S}_{k}:=\exp\left[\tau_{k}\left(\widehat{a}_{k}\widehat{b}_{k}-\widehat{a}_{k}^{\dagger}\widehat{b}_{k}^{\dagger}\right)\right],\qquad
	\tau_{k}:=\operatorname{artanh}r_{k}.
\end{align}
It satisfies $d_{k}=\mathsf{S}_{k}\widehat{a}_{k}\mathsf{S}_{k}^{\dagger}$ and $e_{k}=\mathsf{S}_{k}\widehat{b}_{k}\mathsf{S}_{k}^{\dagger}$.
Thus, Eq.~\eqref{eq:sharpness-bosonic-diagonal-block} is unitarily equivalent to two number operators with spacing $\varepsilon_{k}>0$ and energy shift $\varepsilon_{k} - 1$.
The product of these finitely many squeezing unitaries therefore gives a unitary diagonalization of $H_{\mathrm{b}}$, with a discrete spectrum.
Hence, the unique ground state of $H_{\mathrm{b}}$ is
\begin{align}
	\ket{\Phi}&:=\bigotimes_{k}\ket{\phi_{k}},\qquad
	\ket{\phi_{k}}:=\sqrt{1-r_{k}^{2}}\sum_{m=0}^{\infty}(-r_{k})^{m}\ket{m,m}. \label{eq:sharpness-bosonic-ground-state}
\end{align}
Here $\ket{m,m}:=\ket{m}_{\widehat{a}_{k}}\otimes\ket{m}_{\widehat{b}_{k}}$ denotes the occupation-number state of the rotated mode pair $(\widehat{a}_{k},\widehat{b}_{k})$.
The corresponding ground energy is
\begin{align}
	E_{\mathrm{b}}:=\sum_{k}(\varepsilon_{k}-1). \label{eq:sharpness-bosonic-ground-energy}
\end{align}
The vector $\ket{\phi_{k}}$ is the normalized common vacuum of $d_{k}$ and $e_{k}$, equivalently $\mathsf{S}_{k}\ket{0,0}_{\widehat{a}_{k},\widehat{b}_{k}}$.
Since there are finitely many modes and $r_{k}<1$, $\ket{\Phi}$ has finite moments of every order in total particle number.
The gap of $H_{\mathrm{b}}$ is $\varepsilon_{*}:=\sqrt{1-f^{2}}$.
Let $Q_{k}^{\mathrm{b}}:=\mathds{1}-\ket{\phi_{k}}\bra{\phi_{k}}$.
Equation~\eqref{eq:sharpness-bosonic-diagonal-block} implies
\begin{align}
	H_{\mathrm{b}}-E_{\mathrm{b}}\mathds{1}
	\ge\varepsilon_{*}\sum_{k}Q_{k}^{\mathrm{b}}. \label{eq:sharpness-reference-gap}
\end{align}

\subsection{A uniform energy-defect bound}

We next compare the ground-state energies of the spin Hamiltonian $H_{F}$ and the auxiliary bosonic Hamiltonian $H_{\mathrm{b}}$.
The identity $\mathsf{J}^{\dagger}H_{\mathrm{b}}\mathsf{J}=H_{F}$ provides one side of the comparison, while the other will be obtained by a variational argument in the spin Hilbert space.

\begin{lemma}
	\label{lem:sharpness-defect}
	Let $E_{\mathrm{s}}$ denote the ground-state energy of $H_{F}$, and recall that $W:=\operatorname{Tr}(FF^{\dagger})$.
	If $f = \norm{F} \le 1/2$, then
	\begin{align}
		0\le\varepsilon_{\mathrm{def}}:=E_{\mathrm{s}}-E_{\mathrm{b}}\le2f^{2}W. \label{eq:sharpness-defect}
	\end{align}
	Consequently, for any normalized spin ground state $\ket{\Omega_{F}}$,
	\begin{align}
		\sum_{k}\bra{\Omega_{F}}\mathsf{J}^{\dagger}Q_{k}^{\mathrm{b}}\mathsf{J}\ket{\Omega_{F}}\le3f^{2}W. \label{eq:sharpness-defect-projectors}
	\end{align}
\end{lemma}

\begin{proof}
	The lower bound in Eq.~\eqref{eq:sharpness-defect} follows directly from the compression identity $\mathsf{J}^{\dagger}H_{\mathrm{b}}\mathsf{J}=H_{F}$ and the variational principle.
	For the upper bound, work entirely in the spin Hilbert space and write
	\begin{align}
		\mathcal{P}^{\dagger}&:=\sum_{a,b}F_{ab}\sigma_{a}^{+}\sigma_{b}^{+},\qquad \mathcal{N}_{A}:=\sum_{a\in A}n_{a},\qquad 	\mathcal{N}_{B}:=\sum_{b\in B}n_{b},  \qquad \mathcal{N} :=\mathcal{N}_{A}+\mathcal{N}_{B},\qquad
		\mathcal{G}:=\mathcal{P}^{\dagger}-\mathcal{P}. \label{eq:sharpness-spin-generators}
	\end{align}
	Here $a$ and $b$ label the corresponding physical sites in $A$ and $B$.
	Set $Z_{x}:=\mathds{1}-2n_{x}$, $X_{b}:=\sum_{a}\overline{F_{ab}}\sigma_{a}^{-}$, and $Y_{a}:=\sum_{b}\overline{F_{ab}}\sigma_{b}^{-}$.
	Direct spin commutation gives
	\begin{align}
		[\mathcal{P},\mathcal{P}^{\dagger}] = \sum_{a,b}|F_{ab}|^{2}Z_{a}Z_{b} + \sum_{b}Z_{b}X_{b}^{\dagger}X_{b} + \sum_{a}Z_{a}Y_{a}^{\dagger}Y_{a}. \label{eq:sharpness-spin-commutator-exact}
	\end{align}
	For any spin vector $\ket{\psi}$, define $\ket{v_{a}}:=\sigma_{a}^{-}\ket{\psi}$.
	Since $X_{b}\ket{\psi}=\sum_{a}\overline{F_{ab}}\ket{v_{a}}$, applying the matrix-norm bound for $F^{\dagger}$ to the Hilbert-space-valued vector $(\ket{v_{a}})_{a}$ gives
	\begin{align}
		\sum_{b}\norm{X_{b}\ket{\psi}}^{2}
		\le f^{2}\sum_{a}\norm{\sigma_{a}^{-}\ket{\psi}}^{2}
		=f^{2}\bra{\psi}\mathcal{N}_{A}\ket{\psi}. \label{eq:sharpness-X-bound}
	\end{align}
	Since Eq.~\eqref{eq:sharpness-X-bound} holds for every $\ket{\psi}$, we have $\sum_{b}X_{b}^{\dagger}X_{b}\le f^{2}\mathcal{N}_{A}$, and the analogous argument for $Y_{a}$ gives $\sum_{a}Y_{a}^{\dagger}Y_{a}\le f^{2}\mathcal{N}_{B}$.
	Moreover, $Z_{b}$ commutes with $X_{b}^{\dagger}X_{b}$ and $Z_{a}$ with $Y_{a}^{\dagger}Y_{a}$, because their supports lie on opposite sides.
	Since $-\mathds{1}\le Z_{x}\le\mathds{1}$, it follows that
	\begin{align}
		-f^{2}\mathcal{N}
		\le \sum_{b}Z_{b}X_{b}^{\dagger}X_{b}
		+\sum_{a}Z_{a}Y_{a}^{\dagger}Y_{a}
		\le f^{2}\mathcal{N}. \label{eq:sharpness-XY-bound}
	\end{align}
	For the first term in Eq.~\eqref{eq:sharpness-spin-commutator-exact}, $Z_{a}Z_{b}\le\mathds{1}$ gives the upper bound $W\mathds{1}$, while
	\begin{align}
		\sum_{a,b}|F_{ab}|^{2}Z_{a}Z_{b}
		=W\mathds{1}-\sum_{a,b}|F_{ab}|^{2}(\mathds{1}-Z_{a}Z_{b}). \label{eq:sharpness-ZZ-decomposition}
	\end{align}
	Using $0\le\mathds{1}-Z_{a}Z_{b}\le2(n_{a}+n_{b})$, together with $\sum_{b}|F_{ab}|^{2}=(FF^{\dagger})_{aa}\le f^{2}$ and $\sum_{a}|F_{ab}|^{2}=(F^{\dagger}F)_{bb}\le f^{2}$, Eq.~\eqref{eq:sharpness-ZZ-decomposition} gives
	\begin{align}
		W\mathds{1}-2f^{2}\mathcal{N}
		\le \sum_{a,b}|F_{ab}|^{2}Z_{a}Z_{b}
		\le W\mathds{1}. \label{eq:sharpness-ZZ-bound}
	\end{align}
	Combining Eqs.~\eqref{eq:sharpness-XY-bound} and \eqref{eq:sharpness-ZZ-bound} with Eq.~\eqref{eq:sharpness-spin-commutator-exact}, we obtain
	\begin{align}
		W\mathds{1}-3f^{2}\mathcal{N}
		\le[\mathcal{P},\mathcal{P}^{\dagger}]
		\le W\mathds{1}+f^{2}\mathcal{N}. \label{eq:sharpness-spin-commutator}
	\end{align}

	Let $\ket{\mathrm{vac}}:=\bigotimes_{x\in\Lambda}\ket{0}_{x}$ be the spin vacuum and set $\ket{\psi_{t}}:=e^{-t\mathcal{G}}\ket{\mathrm{vac}}$ and $m(t):=\bra{\psi_{t}}\mathcal{N}\ket{\psi_{t}}$.
	Since $\mathcal{G}^{\dagger}=-\mathcal{G}$, the vector $\ket{\psi_{t}}$ remains normalized.
	A direct differentiation gives
	\begin{align}
		m(0)&=m'(0)=0, \label{eq:sharpness-number-initial-data}\\
		m'(t)&=-2\bra{\psi_{t}}(\mathcal{P}^{\dagger}+\mathcal{P})\ket{\psi_{t}}, \label{eq:sharpness-number-first-derivative}\\
		m''(t)&=4\bra{\psi_{t}}[\mathcal{P},\mathcal{P}^{\dagger}]\ket{\psi_{t}}. \label{eq:sharpness-number-second-derivative}
	\end{align}
	The upper bound in Eq.~\eqref{eq:sharpness-spin-commutator}, together with Eq.~\eqref{eq:sharpness-number-second-derivative}, implies $m''(t)\le4W+4f^{2}m(t)$.
	To obtain an explicit upper bound on $m(t)$, compare it with the solution of the corresponding equality $m_{+}''(t)-4f^{2}m_{+}(t)=4W$ having the same initial data as in Eq.~\eqref{eq:sharpness-number-initial-data}.
	This solution is $m_{+}(t):=W[\cosh(2ft)-1]/f^{2}$.
	Since $m(t)-m_{+}(t)$ and its first derivative both vanish at $t=0$, while $(m-m_{+})''-4f^{2}(m-m_{+})=m''-4f^{2}m-4W$, solving this inhomogeneous equation with zero initial data gives the integral representation below, where $\sinh(2fs)/(2f)$ is the solution of the associated homogeneous equation that vanishes at $s=0$ and has first derivative one there.
	Since $\sinh(2f(t-u))/(2f)\ge0$ for $0\le u\le t$, while $m''(u)-4f^{2}m(u)-4W\le0$ by the differential inequality above, we obtain
	\begin{align}
		m(t)-m_{+}(t) = \int_{0}^{t}\frac{\sinh(2f(t-u))}{2f}\left[m''(u)-4f^{2}m(u)-4W\right]du \le 0.
	\end{align}
	Since $f \le 1/2$, one has $0\le 2ft \le1/2$ for $0\le t\le1/2$, so $\cosh(2ft)-1\le4f^{2}t^{2}$ and
	\begin{align}
		m(t)\le\frac{W}{f^{2}}\left[\cosh(2ft)-1\right]\le4Wt^{2}\qquad \br{0\le t\le\frac{1}{2}}. \label{eq:sharpness-number-growth}
	\end{align}
	At $t=1/2$, using $H_{F}=\mathcal{N}+\mathcal{P}^{\dagger}+\mathcal{P}$ together with Eq.~\eqref{eq:sharpness-number-first-derivative}, we obtain
	\begin{align}
		\bra{\psi_{\frac{1}{2}}}H_{F}\ket{\psi_{\frac{1}{2}}} = m \br{\frac{1}{2}} + \bra{\psi_{\frac{1}{2}}}(\mathcal{P}^{\dagger}+\mathcal{P})\ket{\psi_{\frac{1}{2}}} = m \br{\frac{1}{2}} - \frac{1}{2} m'\br{\frac{1}{2}} =-\int_{0}^{\frac{1}{2}}t\,m''(t)\,dt, \label{eq:sharpness-trial-energy-integration}
	\end{align}
	where the last equality follows by integration by parts using $m(0)=0$ from Eq.~\eqref{eq:sharpness-number-initial-data}.
	Using Eq.~\eqref{eq:sharpness-number-second-derivative} together with the lower bound in Eq.~\eqref{eq:sharpness-spin-commutator}, we have $m''(t)\ge 4W-12f^{2}m(t)$.
	Inserting this into Eq.~\eqref{eq:sharpness-trial-energy-integration} and then using Eq.~\eqref{eq:sharpness-number-growth}, we obtain
	\begin{align}
		\bra{\psi_{\frac{1}{2}}}H_{F}\ket{\psi_{\frac{1}{2}}}  \le -\int_{0}^{\frac{1}{2}}t\left(4W-12f^{2}m(t)\right)dt \le-\frac{W}{2}+48f^{2}W\int_{0}^{\frac{1}{2}}t^{3}\,dt = -\frac{W}{2}+\frac{3}{4}f^{2}W. \label{eq:sharpness-spin-trial-energy}
	\end{align}
	
	For $0\le x\le1/4$, one has $0\le1-\sqrt{1-x}-x/2=x^{2}/[2(1+\sqrt{1-x})^{2}]\le x^{2}/2$.
	Since $s_{k}\le f\le1/2$, and since $\sum_{k}s_{k}^{2}=W$ and $\sum_{k}s_{k}^{4}\le f^{2}W$, this gives
	\begin{align}
		E_{\mathrm{b}} = \sum_{k}\left(\sqrt{1-s_{k}^{2}}-1\right) \ge-\frac{W}{2}-\frac{f^{2}W}{2}. \label{eq:sharpness-bosonic-energy-lower}
	\end{align}
	The variational principle applied to Eq.~\eqref{eq:sharpness-spin-trial-energy}, together with Eq.~\eqref{eq:sharpness-bosonic-energy-lower}, therefore gives
	\begin{align}
		E_{\mathrm{s}}-E_{\mathrm{b}} \le \left(-\frac{W}{2}+\frac{3}{4}f^{2}W\right) - \left(-\frac{W}{2}-\frac{1}{2}f^{2}W\right) = \frac{5}{4}f^{2}W
		\le2f^{2}W. \label{eq:sharpness-defect-upper}
	\end{align}
	Together with the lower bound established at the beginning of the proof, Eq.~\eqref{eq:sharpness-defect-upper} proves Eq.~\eqref{eq:sharpness-defect}.

	Finally, let $\ket{\Omega_{F}}$ be any normalized spin ground state.
	Using the compression identity $\mathsf{J}^{\dagger}H_{\mathrm{b}}\mathsf{J}=H_{F}$ and the isometry property $\mathsf{J}^{\dagger}\mathsf{J}=\mathds{1}$, we have
	\begin{align}
		\bra{\Omega_{F}}\mathsf{J}^{\dagger}\left(H_{\mathrm{b}}-E_{\mathrm{b}}\mathds{1}\right)\mathsf{J}\ket{\Omega_{F}}
		=E_{\mathrm{s}}-E_{\mathrm{b}}.
	\end{align}
	Equation~\eqref{eq:sharpness-reference-gap} therefore gives
	\begin{align}
		\varepsilon_{*}\sum_{k}\bra{\Omega_{F}}\mathsf{J}^{\dagger}Q_{k}^{\mathrm{b}}\mathsf{J}\ket{\Omega_{F}}
		\le E_{\mathrm{s}}-E_{\mathrm{b}}
		\le2f^{2}W.
	\end{align}
	Since $f\le1/2$ implies $\varepsilon_{*}=\sqrt{1-f^{2}}\ge\sqrt{3}/2>2/3$, we obtain
	\begin{align}
		\sum_{k}\bra{\Omega_{F}}\mathsf{J}^{\dagger}Q_{k}^{\mathrm{b}}\mathsf{J}\ket{\Omega_{F}}
		\le3f^{2}W,
	\end{align}
	which proves Eq.~\eqref{eq:sharpness-defect-projectors}.
\end{proof}

\subsection{Entanglement extraction from weakly squeezed mode blocks}

We next convert the small energy defect into a lower bound on the spin entanglement.
The useful quantity is squashed entanglement~\cite{christandl2004squashed},
\begin{align}
	E_{\mathrm{sq}}(X:Y)_{\rho}
	:=\frac{1}{2}\inf_{\rho_{XYE}}I(X:Y\mid E)_{\rho},\qquad
	\operatorname{Tr}_{E}\rho_{XYE}=\rho_{XY}.
\end{align}
Here $I(X:Y\mid E)_{\rho}:=S(\rho_{XE})+S(\rho_{YE})-S(\rho_{XYE})-S(\rho_{E})$ is the conditional mutual information.
The infimum is over finite-dimensional extensions with no fixed upper bound on the dimension of $E$.

We use three standard properties of squashed entanglement~\cite{christandl2004squashed}, and give the required derivations below.
First, for a pure state $\ket{\psi}_{XY}$, with $\rho_{X}^{\psi}:=\operatorname{Tr}_{Y}\ket{\psi}\bra{\psi}$, squashed entanglement coincides with the entanglement entropy:
\begin{align}
	E_{\mathrm{sq}}(X:Y)_{\psi}=S(\rho_{X}^{\psi}). \label{eq:sharpness-Esq-pure}
\end{align}
This identity will allow us to convert the final squashed-entanglement lower bound into a lower bound on the entanglement entropy of the pure spin ground state.
Second, squashed entanglement is strongly superadditive.
Thus, entanglement lower bounds obtained on disjoint mode blocks can be summed:
\begin{align}
	E_{\mathrm{sq}}(X_{1}\cdots X_{m}:Y_{1}\cdots Y_{m})_{\rho}
	\ge\sum_{j=1}^{m}E_{\mathrm{sq}}(X_{j}:Y_{j})_{\rho_{j}}, \label{eq:sharpness-Esq-superadditivity}
\end{align}
where $\rho_{j}$ is the reduced state on $X_{j}Y_{j}$.
Third, squashed entanglement is monotone on average under LOCC.
If a finite-round LOCC protocol produces normalized states $\rho_{\ell}$ with probabilities $p_{\ell}$, where $\ell$ denotes the complete classical outcome, then
\begin{align}
	E_{\mathrm{sq}}(X:Y)_{\rho}
	\ge\sum_{\ell}p_{\ell}E_{\mathrm{sq}}(X:Y)_{\rho_{\ell}}. \label{eq:sharpness-Esq-LOCC}
\end{align}
This property will be used for the two local occupation-number projections introduced below, allowing the entanglement found in the successful postselected branch to give a lower bound on the entanglement before the projections.

We now record the ingredients needed for these properties in the form used below.
Strong subadditivity of the von Neumann entropy gives $I(X:Y\mid E)_{\rho} \ge 0$, and hence $E_{\mathrm{sq}}(X:Y)_{\rho} \ge 0$.
If $\rho_{XY}=\ket{\psi}\bra{\psi}$ is pure, every extension factors as $\ket{\psi}\bra{\psi}\otimes\rho_{E}$.
Therefore
\begin{align}
	I(X:Y\mid E)_{\rho}
	=I(X:Y)_{\psi}
	=2S(\rho_{X}^{\psi}),
\end{align}
independently of the extension.
Taking the infimum proves Eq.~\eqref{eq:sharpness-Esq-pure}.

To prove Eq.~\eqref{eq:sharpness-Esq-superadditivity}, it suffices to treat two paired subsystems and then iterate.
For any extension on $X_{1}X_{2}Y_{1}Y_{2}E$, the chain rule for conditional mutual information, together with strong subadditivity of the von Neumann entropy, gives
\begin{align}
	I(X_{1}X_{2}:Y_{1}Y_{2}\mid E)
	&=I(X_{1}:Y_{1}Y_{2}\mid E)+I(X_{2}:Y_{1}Y_{2}\mid EX_{1})\notag\\
	&\ge I(X_{1}:Y_{1}\mid E)+I(X_{2}:Y_{2}\mid EX_{1})\notag\\
	&\ge2E_{\mathrm{sq}}(X_{1}:Y_{1})_{\rho_{1}}+2E_{\mathrm{sq}}(X_{2}:Y_{2})_{\rho_{2}}.
\end{align}
Here $\rho_{j}:=\operatorname{Tr}_{\{X_{k}Y_{k}:k\ne j\}}\rho$ is the reduced state on $X_{j}Y_{j}$.
The first equality is the chain rule for conditional mutual information, which follows directly from its definition.
For the second line, the chain rule gives $I(X_{1}:Y_{1}Y_{2}\mid E)=I(X_{1}:Y_{1}\mid E)+I(X_{1}:Y_{2}\mid EY_{1})$ and $I(X_{2}:Y_{1}Y_{2}\mid EX_{1})=I(X_{2}:Y_{2}\mid EX_{1})+I(X_{2}:Y_{1}\mid EX_{1}Y_{2})$.
The second line then follows because $I(X_{1}:Y_{2}\mid EY_{1})\ge 0$ and $I(X_{2}:Y_{1}\mid EX_{1}Y_{2})\ge 0$ by strong subadditivity of the von Neumann entropy.
For the last line, the marginal on $X_{1}Y_{1}E$ is an extension of $\rho_{1}$, while the marginal on $X_{2}Y_{2}EX_{1}$ is an extension of $\rho_{2}$, so the definition of squashed entanglement gives the stated lower bound.
Taking the infimum over the original extension proves the two-pair inequality, and iteration gives Eq.~\eqref{eq:sharpness-Esq-superadditivity}.

For Eq.~\eqref{eq:sharpness-Esq-LOCC}, first consider a local instrument acting on $X$.
Let $X'$ denote the quantum subsystem after the local operation, and let $C$ be a classical register storing its outcome.
Applied to an arbitrary extension $\rho_{XYE}$, the instrument produces $\widetilde{\rho}_{X'YEC}=\sum_{\ell}p_{\ell}\rho_{X'YE}^{(\ell)}\otimes\ket{\ell}\bra{\ell}_{C}$.
A Kraus representation realizes this channel as a local isometry followed by a partial trace.
Conditional mutual information is invariant under the local isometry, while the subsequent partial trace cannot increase it: by the chain rule for conditional mutual information,
$I(X'CR:Y\mid E)=I(X'C:Y\mid E)+I(R:Y\mid EX'C)$, and the second term is nonnegative by strong subadditivity of the von Neumann entropy.
Therefore,
\begin{align}
	I(X:Y\mid E)_{\rho}
	&\ge I(X'C:Y\mid E)_{\widetilde{\rho}}\notag\\
	&=I(C:Y\mid E)_{\widetilde{\rho}}+I(X':Y\mid EC)_{\widetilde{\rho}}\notag\\
	&\ge\sum_{\ell}p_{\ell}I(X':Y\mid E)_{\rho^{(\ell)}}\notag\\
	&\ge2\sum_{\ell}p_{\ell}E_{\mathrm{sq}}(X':Y)_{\rho_{\ell}}.
\end{align}
Here $\rho_{\ell}:=\operatorname{Tr}_{E}\rho_{X'YE}^{(\ell)}$.
The second equality is the chain rule for conditional mutual information, and the third line uses $I(C:Y\mid E)_{\widetilde{\rho}}\ge0$ together with the classical block structure of $C$, for which $I(X':Y\mid EC)_{\widetilde{\rho}}=\sum_{\ell}p_{\ell}I(X':Y\mid E)_{\rho^{(\ell)}}$.
For the last line, $\rho_{X'YE}^{(\ell)}$ is an extension of $\rho_{\ell}$ for each $\ell$, so the definition of squashed entanglement applies separately to every branch.
The branch probabilities and the states $\rho_{\ell}$ depend only on the input marginal, not on its extension.
Taking the infimum over the original extension proves average monotonicity for one local instrument, and the same argument applies on $Y$.
Applying the same average-monotonicity bound at each round, and then summing over all classical outcome branches, proves Eq.~\eqref{eq:sharpness-Esq-LOCC} for any finite-round LOCC protocol.

We also use the following lower bound in terms of coherent information:
\begin{align}
	E_{\mathrm{sq}}(X:Y)_{\rho} \ge \max\left\{0,S(\rho_{Y})-S(\rho_{XY})\right\}. \label{eq:sharpness-coherent-info}
\end{align}
To see this, let $\rho_{XYE}$ be an arbitrary extension and let $\rho_{XYER}$ be a purification.
Writing $S(X\mid Z):=S(\rho_{XZ})-S(\rho_{Z})$ for the conditional entropy, purity of $\rho_{XYER}$ gives
\begin{align}
	I(X:Y\mid E)_{\rho}=S(X\mid E)+S(X\mid R).
\end{align}
Weak monotonicity of the von Neumann entropy, applied to the triples $XEY$ and $XRY$, gives $S(X\mid E)\ge-S(X\mid Y)$ and $S(X\mid R)\ge - S(X\mid Y)$.
Therefore
\begin{align}
	I(X:Y\mid E)_{\rho}
	\ge-2S(X\mid Y)
	=2\left[S(\rho_{Y})-S(\rho_{XY})\right].
\end{align}
Taking the infimum over all extensions and using $E_{\mathrm{sq}}(X:Y)_{\rho}\ge0$ proves Eq.~\eqref{eq:sharpness-coherent-info}.

We also need an elementary entropy-continuity estimate.
If $t:=\norm{\rho-\sigma}_{1}/2$ for states on a $d$-dimensional space, then
\begin{align}
	|S(\rho)-S(\sigma)|
	\le t\log d+\omega(t),\qquad
	\omega(t):=(1+t)\log(1+t)-t\log t. \label{eq:sharpness-continuity}
\end{align}
We set $\omega(0):=0$, and the case $t=0$ is immediate.
For $t>0$, let $(\rho-\sigma)_{+}$ and $(\rho-\sigma)_{-}$ be the positive and negative parts of $\rho-\sigma$.
Since $\operatorname{Tr}(\rho-\sigma)=0$ and $t=\norm{\rho-\sigma}_{1}/2$, both have trace $t$.
Thus, defining $\tau_{\pm}:=(\rho-\sigma)_{\pm}/t$, one has $\rho-\sigma=t(\tau_{+}-\tau_{-})$, where $\tau_{+}$ and $\tau_{-}$ are density matrices.
It follows that $\rho+t\tau_{-}=\sigma+t\tau_{+}$, and we define their normalized common value by $\xi:=(\rho+t\tau_{-})/(1+t)=(\sigma+t\tau_{+})/(1+t)$.
Using $\xi=(\rho+t\tau_{-})/(1+t)$, concavity of the von Neumann entropy gives $S(\xi)\ge[S(\rho)+tS(\tau_{-})]/(1+t)$.
On the other hand, using $\xi=(\sigma+t\tau_{+})/(1+t)$, the entropy upper bound for a binary mixture gives $S(\xi)\le[S(\sigma)+tS(\tau_{+})]/(1+t)+\hbin(t/(1+t))$.
Multiplying these two inequalities by $1+t$ gives
\begin{align}
	S(\rho)+tS(\tau_{-})
	\le(1+t)S(\xi)
	\le S(\sigma)+tS(\tau_{+})+\omega(t).
\end{align}
Here $(1+t)\hbin(t/(1+t))=\omega(t)$, with $\hbin$ denoting the binary entropy defined above.
Rearranging the preceding bounds and using $0\le S(\tau_{\pm})\le\log d$ gives $S(\rho)-S(\sigma)\le t[S(\tau_{+})-S(\tau_{-})]+\omega(t)\le t\log d+\omega(t)$.
Exchanging $\rho$ and $\sigma$ gives the opposite inequality and proves Eq.~\eqref{eq:sharpness-continuity}.

Every state to which $E_{\mathrm{sq}}$ is applied below has finite local support.
After the exact hard-core embedding, there are at most $N_{\mathrm{h}}$ bosons on each side, and the number-preserving mode rotations preserve these total-number bounds.
Thus, each individual mode may be restricted to occupations $0,\ldots,N_{\mathrm{h}}$ for the entanglement argument.
The infinite-dimensional vectors in Eq.~\eqref{eq:sharpness-bosonic-ground-state} enter only as auxiliary reference states, while all states to which the finite-dimensional entropic bounds are applied have finite local support.

\begin{lemma}
	\label{lem:sharpness-entropy-from-defect}
	For fixed $D \ge 1$, $0 \le \alpha < D$, and all sufficiently large even $\ell_{\Lambda}$, every normalized pure ground state of Eq.~\eqref{eq:sharpness-spin-model} satisfies
	\begin{align}
		S_{A}(\ket{\Omega_{F}}) \ge \frac{D-\alpha}{256}\,W\log \ell_{\Lambda}. \label{eq:sharpness-entropy-W}
	\end{align}
\end{lemma}

\begin{proof}
	Recall that $F=U_{F}\operatorname{diag}(s_{k})V_{F}^{\dagger}$, so that $s_{k}$ are the singular values of $F$, and that $\varepsilon_{k}:=\sqrt{1-s_{k}^{2}}$, $r_{k}:=s_{k}/(1+\varepsilon_{k})$, and $\zeta_{k}:=r_{k}^{2}$.
	The corresponding reference vectors $\ket{\phi_{k}}$ are given in Eq.~\eqref{eq:sharpness-bosonic-ground-state}.
	The proof identifies many disjoint blocks of rotated mode pairs on which the embedded ground state $\mathsf{J}\ket{\Omega_{F}}$ remains close to the corresponding bosonic reference state.
	On each such block, we project locally onto total occupation one on each side; the blocks will be chosen so that this event has a uniformly positive probability for the reference state.
	The resulting entanglement lower bounds can then be summed using strong superadditivity of squashed entanglement.
	
	Since $s_{k}/2\le r_{k}\le s_{k}$ for $s_{k}\le1/2$, and since $W=\sum_{k}s_{k}^{2}$, set
	\begin{align}
		T:=\sum_{k}\zeta_{k} = \sum_{k} r_{k}^{2}.
	\end{align}
	Then $T\ge W/4$.
	Moreover, Lemma~\ref{lem:sharpness-matrix} gives $W\ge1/(4D^{2\alpha})$ and $f=o(1)$, while $\zeta_{k}=r_{k}^{2}\le s_{k}^{2}\le f^{2}$.
	Therefore,
	\begin{align}
		T\ge\frac{W}{4}\ge\frac{1}{16D^{2\alpha}},\qquad
		\max_{k}\zeta_{k}\le f^{2}=o(1). \label{eq:sharpness-total-squeezing-weight}
	\end{align}
	Set $u_{*}:=1/(64D^{2\alpha})$.
	Then $T\ge4u_{*}$, while Eq.~\eqref{eq:sharpness-total-squeezing-weight} implies $\max_{k}\zeta_{k}<u_{*}$ for all sufficiently large $\ell_{\Lambda}$.
	Starting from any nonzero $\zeta_{k}$ and adding further nonzero terms until their sum first reaches $u_{*}$ therefore produces an index set $I_{j}$ whose total weight lies between $u_{*}$ and $2u_{*}$.
	Repeating this procedure with the remaining indices gives disjoint index sets $I_{j}$ satisfying
	\begin{align}
		u_{*}\le u_{j}:=\sum_{k\in I_{j}}\zeta_{k}\le2u_{*}, \label{eq:sharpness-mode-block-weight}
	\end{align}
	with a remainder of total weight less than $u_{*}$.
	If $m$ is the number of index sets $I_{j}$ constructed above, their total weight is at most $2mu_{*}$, so $T<2mu_{*}+u_{*}$.
	Since $T\ge4u_{*}$, this gives
	\begin{align}
		m\ge\frac{T}{4u_{*}}\ge\frac{W}{16u_{*}}. \label{eq:sharpness-number-mode-blocks}
	\end{align}

	After embedding the spin ground state into the bosonic Fock space by $\mathsf{J}$, let $\rho_{j}$ denote the reduced state on the rotated mode pairs $(\widehat{a}_{k},\widehat{b}_{k})$ with $k\in I_{j}$, obtained by tracing out all rotated modes with $k\notin I_{j}$.
	Define
	\begin{align}
		\mathfrak{d}_{j}:=\sum_{k\in I_{j}}\operatorname{Tr}(\rho_{j}Q_{k}^{\mathrm{b}}),\qquad
		Q_{k}^{\mathrm{b}}:=\mathds{1}-\ket{\phi_{k}}\bra{\phi_{k}}.
	\end{align}
	Each term $\operatorname{Tr}(\rho_{j}Q_{k}^{\mathrm{b}})$ measures the weight outside the reference state $\ket{\phi_{k}}$ for the $k$th mode pair, so $\mathfrak{d}_{j}$ quantifies the total deviation from the corresponding bosonic reference states over the indices in $I_{j}$.
	Since the index sets $I_{j}$ are disjoint and $\rho_{j}$ is the corresponding reduced state of $\mathsf{J}\ket{\Omega_{F}}$, summing over $j$ gives a sub-sum of the total defect in Eq.~\eqref{eq:sharpness-defect-projectors}.
	Hence,
	\begin{align}
		\sum_{j}\mathfrak{d}_{j}\le3f^{2}W. \label{eq:sharpness-total-index-set-defect}
	\end{align}
	We choose $\vartheta:=(D-\alpha)/(24D)$ and $\mathfrak{d}_{*}:=(u_{*}\vartheta/8)^{2}$, and call an index set $I_{j}$ good when $\mathfrak{d}_{j}\le\mathfrak{d}_{*}$.
	Since every index set that is not good contributes more than $\mathfrak{d}_{*}$ to Eq.~\eqref{eq:sharpness-total-index-set-defect}, the number of such index sets is at most $3f^{2}W/\mathfrak{d}_{*}$.
	Moreover, $u_{*}$ and $\mathfrak{d}_{*}$ are fixed positive constants, whereas $f^{2}=o(1)$.
	Thus, for all sufficiently large $\ell_{\Lambda}$,
	\begin{align}
		\#\{j:\mathfrak{d}_{j}>\mathfrak{d}_{*}\}
		\le\frac{3f^{2}W}{\mathfrak{d}_{*}}
		\le\frac{W}{32u_{*}}.
	\end{align}
	Combining this with Eq.~\eqref{eq:sharpness-number-mode-blocks}, there are therefore at least $W/(32u_{*})$ good index sets.

	Fix one good index set $I_{j}$ and write $M_{j}:=|I_{j}|$ and $\ket{\Phi_{j}}:=\bigotimes_{k\in I_{j}}\ket{\phi_{k}}$.
	Since $Q_{k}^{\mathrm{b}}=\mathds{1}-\ket{\phi_{k}}\bra{\phi_{k}}$ acts on the $k$th mode pair, one has $\ket{\Phi_{j}}\bra{\Phi_{j}}=\prod_{k\in I_{j}}(\mathds{1}-Q_{k}^{\mathrm{b}})$.
	The projectors commute, so $\mathds{1}-\ket{\Phi_{j}}\bra{\Phi_{j}}\le\sum_{k\in I_{j}}Q_{k}^{\mathrm{b}}$, which gives $1-\bra{\Phi_{j}}\rho_{j}\ket{\Phi_{j}}\le\mathfrak{d}_{j}$.
	Using the trace-distance bound against a pure state, we therefore obtain
	\begin{align}
		\frac{1}{2}\norm{\rho_{j}-\ket{\Phi_{j}}\bra{\Phi_{j}}}_{1}
		\le\sqrt{1-\bra{\Phi_{j}}\rho_{j}\ket{\Phi_{j}}}
		\le\sqrt{\mathfrak{d}_{j}}. \label{eq:sharpness-good-block-distance}
	\end{align}

	Let $\ket{k}_{A}$ denote the Fock state with one boson in the rotated mode $\widehat{a}_{k}$ and zero bosons in all other modes $\widehat{a}_{\ell}$ with $\ell\in I_{j}$, and define $\ket{k}_{B}$ analogously for the modes $\widehat{b}_{k}$.
	Let
	\begin{align}
		\Pi_{A,j}^{(1)}:=\sum_{k\in I_{j}}\ket{k}_{A}\bra{k}_{A},\qquad
		\Pi_{B,j}^{(1)}:=\sum_{k\in I_{j}}\ket{k}_{B}\bra{k}_{B},\qquad
		\Pi_{j}^{\mathrm{succ}}:=\Pi_{A,j}^{(1)}\otimes\Pi_{B,j}^{(1)}.
	\end{align}
	Thus $\Pi_{j}^{\mathrm{succ}}$ selects the sector in which the modes indexed by $I_{j}$ contain exactly one boson on each side.
	For the reference state $\ket{\Phi_{j}}=\bigotimes_{k\in I_{j}}\ket{\phi_{k}}$, Eq.~\eqref{eq:sharpness-bosonic-ground-state} shows that the projection $\Pi_{j}^{\mathrm{succ}}$ retains exactly those terms for which one index $k\in I_{j}$ has occupation number one on both sides, while every other index $\ell\in I_{j}\setminus\{k\}$ has occupation number zero on both sides.
	Thus, $\Pi_{j}^{\mathrm{succ}}\ket{\Phi_{j}} = -[\prod_{k\in I_{j}}(1-\zeta_{k})]^{1/2} \sum_{k\in I_{j}}r_{k}\ket{k}_{A}\ket{k}_{B}$.
	Its squared norm therefore gives the reference success probability
	\begin{align}
		p_{j}^{\mathrm{ref}}
		:=\bra{\Phi_{j}}\Pi_{j}^{\mathrm{succ}}\ket{\Phi_{j}}
		=\left(\prod_{k\in I_{j}}(1-\zeta_{k})\right) \sum_{k\in I_{j}}r_{k}^{2}
		=\left(\prod_{k\in I_{j}}(1-\zeta_{k})\right)u_{j}
		\ge (1-2u_{*})u_{*} \ge \frac{u_{*}}{2}. \label{eq:sharpness-reference-success-probability}
	\end{align}
	Here we used $\prod_{k\in I_{j}}(1-\zeta_{k})\ge1-u_{j}$ and $u_{*}\le u_{j}\le2u_{*}$.
	The last inequality follows from $u_{*}=1/(64D^{2\alpha}) \le 1/64$.
	After normalization, the successful reference state is, up to an irrelevant overall phase,
	\begin{align}
		\ket{\psi_{j}^{\mathrm{pair}}}
		:=\frac{1}{\sqrt{u_{j}}}\sum_{k\in I_{j}}r_{k}\ket{k}_{A}\ket{k}_{B}. \label{eq:sharpness-one-pair-output}
	\end{align}

	Let $p_{j}^{\mathrm{succ}}:=\operatorname{Tr}(\rho_{j}\Pi_{j}^{\mathrm{succ}})$ be the success probability for the actual state $\rho_{j}$, and set $t_{j}:=\norm{\rho_{j}-\ket{\Phi_{j}}\bra{\Phi_{j}}}_{1}/2$.
	Equation~\eqref{eq:sharpness-good-block-distance} gives $t_{j}\le\sqrt{\mathfrak{d}_{j}}$.
	Since $\Pi_{j}^{\mathrm{succ}}$ is a measurement projector, the difference between the actual and reference success probabilities satisfies $|p_{j}^{\mathrm{succ}}-p_{j}^{\mathrm{ref}}|\le t_{j}$.
	Moreover, multiplying the two states in Eq.~\eqref{eq:sharpness-good-block-distance} by $\Pi_{j}^{\mathrm{succ}}$ on both sides cannot increase their trace distance, since $\norm{\Pi_{j}^{\mathrm{succ}}X\Pi_{j}^{\mathrm{succ}}}_{1}\le\norm{X}_{1}$ for every operator $X$.
	Using also
	$\Pi_{j}^{\mathrm{succ}}\ket{\Phi_{j}}\bra{\Phi_{j}}\Pi_{j}^{\mathrm{succ}}
	=p_{j}^{\mathrm{ref}}\ket{\psi_{j}^{\mathrm{pair}}}\bra{\psi_{j}^{\mathrm{pair}}}$,
	we obtain
	\begin{align}
		\norm{\Pi_{j}^{\mathrm{succ}}\rho_{j}\Pi_{j}^{\mathrm{succ}} 		- p_{j}^{\mathrm{ref}}\ket{\psi_{j}^{\mathrm{pair}}}\bra{\psi_{j}^{\mathrm{pair}}}}_{1} = \norm{\Pi_{j}^{\mathrm{succ}} \br{\rho_{j} - \ket{\Phi_{j}} \bra{\Phi_{j}}} \Pi_{j}^{\mathrm{succ}}}_{1}
		\le \norm{\rho_{j} - \ket{\Phi_{j}} \bra{\Phi_{j}}}_{1} = 2t_{j}.
	\end{align}
	Since $I_{j}$ is good, Eq.~\eqref{eq:sharpness-good-block-distance} gives
	$t_{j}\le\sqrt{\mathfrak{d}_{j}}\le\sqrt{\mathfrak{d}_{*}}=u_{*}\vartheta/8$.
	Combining this with $|p_{j}^{\mathrm{succ}}-p_{j}^{\mathrm{ref}}|\le t_{j}$ and Eq.~\eqref{eq:sharpness-reference-success-probability}, we obtain
	\begin{align}
		p_{j}^{\mathrm{succ}}
		\ge p_{j}^{\mathrm{ref}}-t_{j}
		\ge\frac{u_{*}}{2}-\frac{u_{*}\vartheta}{8}
		\ge\frac{u_{*}}{4}>0, \label{eq:sharpness-actual-success-probability}
	\end{align}
	where the last inequality uses $0 < \vartheta \le 1/24$.
	We may therefore define the normalized successful output by $\rho_{j}^{\mathrm{succ}}:=\Pi_{j}^{\mathrm{succ}}\rho_{j}\Pi_{j}^{\mathrm{succ}}/p_{j}^{\mathrm{succ}}$.
	Using the preceding trace-distance bound, $|p_{j}^{\mathrm{succ}}-p_{j}^{\mathrm{ref}}|\le t_{j}$, and Eq.~\eqref{eq:sharpness-actual-success-probability}, the triangle inequality gives
	\begin{align}
		\frac{1}{2}\norm{\rho_{j}^{\mathrm{succ}}-\ket{\psi_{j}^{\mathrm{pair}}}\bra{\psi_{j}^{\mathrm{pair}}}}_{1}
		\le \frac{2t_{j}+|p_{j}^{\mathrm{succ}}-p_{j}^{\mathrm{ref}}|}{2p_{j}^{\mathrm{succ}}}
		\le \frac{6t_{j}}{u_{*}}
		\le\frac{8\sqrt{\mathfrak{d}_{j}}}{u_{*}}
		\le\vartheta. \label{eq:sharpness-output-distance}
	\end{align}

	Since the one-particle states $\{\ket{k}_{A}\}_{k\in I_{j}}$ and $\{\ket{k}_{B}\}_{k\in I_{j}}$ are orthonormal, Eq.~\eqref{eq:sharpness-one-pair-output} is already a Schmidt decomposition.
	Its Schmidt coefficients are $r_{k}/\sqrt{u_{j}}$, and hence the eigenvalues of the reduced state on $B$ are
	\begin{align}
		\frac{r_{k}^{2}}{u_{j}}=\frac{\zeta_{k}}{u_{j}},\qquad k\in I_{j}.
	\end{align}
	By Eqs.~\eqref{eq:sharpness-total-squeezing-weight} and \eqref{eq:sharpness-mode-block-weight},
	$\zeta_{k}\le f^{2}$ and $u_{j}\ge u_{*}$, so
	$\max_{k\in I_{j}}\zeta_{k}/u_{j}\le f^{2}/u_{*}$.
	Since the Shannon entropy of a probability distribution is bounded below by minus the logarithm of its largest probability, the entanglement entropy of the pure state $\ket{\psi_{j}^{\mathrm{pair}}}$ satisfies
	\begin{align}
		S((\psi_{j}^{\mathrm{pair}})_{B})
		\ge -\log \left(\max_{k\in I_{j}}\frac{\zeta_{k}}{u_{j}}\right)
		\ge\log\frac{u_{*}}{f^{2}}. \label{eq:sharpness-reference-output-entropy}
	\end{align}
	Equation~\eqref{eq:sharpness-output-distance} shows that $\rho_{j}^{\mathrm{succ}}$ is within trace distance $\vartheta$ of $\ket{\psi_{j}^{\mathrm{pair}}}\bra{\psi_{j}^{\mathrm{pair}}}$.
	By contractivity of the trace distance under the partial trace, their $B$-marginals are also within trace distance $\vartheta$.
	Both $\rho_{j}^{\mathrm{succ}}$ and $\ket{\psi_{j}^{\mathrm{pair}}}\bra{\psi_{j}^{\mathrm{pair}}}$ are supported on the tensor product of the one-particle sectors of the $M_{j}$ rotated modes on the two sides.
	Thus, the local space on $B$ has dimension $M_{j}$, while the joint space has dimension $M_{j}^{2}$.
	Equation~\eqref{eq:sharpness-output-distance}, together with contractivity of the trace distance under the partial trace, allows us to apply Eq.~\eqref{eq:sharpness-continuity} both to the $B$-marginals and to the joint states.
	Since $\ket{\psi_{j}^{\mathrm{pair}}}$ is pure, this gives
	\begin{align}
		S((\rho_{j}^{\mathrm{succ}})_{B})
		&\ge S((\psi_{j}^{\mathrm{pair}})_{B})-\vartheta\log M_{j}-\omega(\vartheta),\notag\\
		S(\rho_{j}^{\mathrm{succ}})
		&\le 2\vartheta\log M_{j}+\omega(\vartheta).
	\end{align}
	Subtracting the second estimate from the first and using Eq.~\eqref{eq:sharpness-reference-output-entropy}, we obtain
	\begin{align}
		S((\rho_{j}^{\mathrm{succ}})_{B})-S(\rho_{j}^{\mathrm{succ}})
		\ge \log\frac{u_{*}}{f^{2}}
		-3\vartheta\log M_{j}
		-2\omega(\vartheta). \label{eq:sharpness-block-coherent-information-pre}
	\end{align}
	Lemma~\ref{lem:sharpness-matrix} gives $f^{2}\lesssim\ell_{\Lambda}^{-(D-\alpha)}$, and hence
	$\log(u_{*}/f^{2})\ge(D-\alpha)\log\ell_{\Lambda}-\mathcal{O}_{D,\alpha}(1)$.
	Moreover, $M_{j}\le N_{\mathrm{h}}\le\ell_{\Lambda}^{D}$ and $3\vartheta D=(D-\alpha)/8$.
	Therefore, Eq.~\eqref{eq:sharpness-block-coherent-information-pre} gives
	\begin{align}
		S((\rho_{j}^{\mathrm{succ}})_{B})-S(\rho_{j}^{\mathrm{succ}})
		\ge\left(D-\alpha-3\vartheta D\right)\log\ell_{\Lambda}-\mathcal{O}_{D,\alpha}(1) 
		=\frac{7(D-\alpha)}{8}\log\ell_{\Lambda}-\mathcal{O}_{D,\alpha}(1) 
		\ge\frac{D-\alpha}{2}\log\ell_{\Lambda} \label{eq:sharpness-block-coherent-information}
	\end{align}
	for all sufficiently large $\ell_{\Lambda}$.

	Let $A_{I_{j}}$ and $B_{I_{j}}$ denote the rotated modes indexed by $I_{j}$ on the two sides.
	Average LOCC monotonicity in Eq.~\eqref{eq:sharpness-Esq-LOCC}, together with the nonnegativity of squashed entanglement, gives
	\begin{align}
		E_{\mathrm{sq}}(A_{I_{j}}:B_{I_{j}})_{\rho_{j}} \ge p_{j}^{\mathrm{succ}} E_{\mathrm{sq}}(A_{I_{j}}:B_{I_{j}})_{\rho_{j}^{\mathrm{succ}}}.
	\end{align}
	Using the coherent-information lower bound in Eq.~\eqref{eq:sharpness-coherent-info}, together with Eqs.~\eqref{eq:sharpness-actual-success-probability} and \eqref{eq:sharpness-block-coherent-information}, we obtain
	\begin{align}
		E_{\mathrm{sq}}(A_{I_{j}}:B_{I_{j}})_{\rho_{j}} \ge \frac{u_{*}(D-\alpha)}{8} \log \ell_{\Lambda}. \label{eq:sharpness-good-block-Esq}
	\end{align}
	There are at least $W/(32u_{*})$ good index sets.
	Since $\mathsf{J}$ and the mode rotations are local with respect to $A|B$, they preserve the entanglement across this bipartition.
	Using the pure-state identity in Eq.~\eqref{eq:sharpness-Esq-pure} and the strong superadditivity of squashed entanglement in Eq.~\eqref{eq:sharpness-Esq-superadditivity}, we therefore obtain
	\begin{align}
		S_{A}(\ket{\Omega_{F}}) \ge \sum_{j:\,I_{j}\,\mathrm{good}}
		E_{\mathrm{sq}}(A_{I_{j}}:B_{I_{j}})_{\rho_{j}} \ge \frac{W}{32u_{*}}\frac{u_{*}(D-\alpha)}{8}\log \ell_{\Lambda}
		= \frac{D-\alpha}{256}W\log \ell_{\Lambda}.
	\end{align}
	This proves Eq.~\eqref{eq:sharpness-entropy-W}.
\end{proof}

\subsection{Uniqueness from a hard-core projection estimate}

The entropy lower bound above applies to every pure ground state and does not use uniqueness.
To prove that the obstruction persists within the unique uniformly gapped class, we now compare the first two spin eigenvalues with those of the auxiliary bosonic Hamiltonian.

\begin{lemma}
	\label{lem:sharpness-unique}
	Assume $f\le1/2$ and $\mathfrak{R}<1$.
	Then
	\begin{align}
		\Delta_{F}
		\ge\sqrt{1-f^{2}}-\frac{4\mathfrak{R}}{1-\mathfrak{R}}. \label{eq:sharpness-finite-gap}
	\end{align}
	In particular, positivity of the right-hand side proves uniqueness of the spin ground state.
\end{lemma}

\begin{proof}
	The proof compares the two lowest spin energies with the corresponding auxiliary bosonic energies.
	For the spin ground energy, we project the bosonic ground state $\ket{\Phi}$ onto the original-mode hard-core subspace and control the resulting energy cost by $\mathfrak{R}$.
	For the first excited energy, the bosonic gap and the min--max principle give a lower bound directly through the compression to the spin Hilbert space.

	We first control both the weight of the bosonic ground state outside the hard-core subspace and the energy cost of projecting it onto that subspace.
	For this purpose, define the mean occupations in the original site modes and their total collision weight by
	\begin{align}
		\bar{n}_{A,a}:=\bra{\Phi}a_{a}^{\dagger}a_{a}\ket{\Phi},\qquad
		\bar{n}_{B,b}:=\bra{\Phi}b_{b}^{\dagger}b_{b}\ket{\Phi},\qquad
		\mathfrak{c}:=\sum_{a}\bar{n}_{A,a}^{2}+\sum_{b}\bar{n}_{B,b}^{2}. \label{eq:sharpness-site-occupations}
	\end{align}
	From Eq.~\eqref{eq:sharpness-bosonic-ground-state}, the mean occupation of either rotated mode in the $k$th paired factor is
	$\bar{n}_{k}:=\bra{\phi_{k}}\widehat{a}_{k}^{\dagger}\widehat{a}_{k}\ket{\phi_{k}}
	=\bra{\phi_{k}}\widehat{b}_{k}^{\dagger}\widehat{b}_{k}\ket{\phi_{k}}
	=r_{k}^{2}/(1-r_{k}^{2})$.
	Using $r_{k}=s_{k}/(1+\sqrt{1-s_{k}^{2}})$, this becomes
	$\bar{n}_{k}=[1/\sqrt{1-s_{k}^{2}}-1]/2\le s_{k}^{2}$ for $s_{k}\le1/2$.
	Transforming back to the original site modes using Eq.~\eqref{eq:sharpness-mode-rotations}, and using $\bar{n}_{k}\le s_{k}^{2}$, gives
	\begin{align}
		\bar{n}_{A,a}
		&=\sum_{k}|(U_{F})_{ak}|^{2}\bar{n}_{k}
		\le\sum_{k}|(U_{F})_{ak}|^{2}s_{k}^{2}
		=(FF^{\dagger})_{aa}, \\
		\bar{n}_{B,b}
		&=\sum_{k}|(V_{F})_{bk}|^{2}\bar{n}_{k}
		\le\sum_{k}|(V_{F})_{bk}|^{2}s_{k}^{2}
		=(F^{\dagger}F)_{bb}.
	\end{align}
	Squaring the preceding bounds and summing over the original site modes gives
	\begin{align}
		\mathfrak{c} = \sum_{a}\bar{n}_{A,a}^{2}+\sum_{b}\bar{n}_{B,b}^{2}  \le \sum_{a}(FF^{\dagger})_{aa}^{2}+\sum_{b}(F^{\dagger}F)_{bb}^{2} = \sum_{a}\left(\sum_{b}|F_{ab}|^{2}\right)^{2}
		+\sum_{b}\left(\sum_{a}|F_{ab}|^{2}\right)^{2} = \mathfrak{R}.
	\end{align}

	To control both violations of the hard-core constraint and the commutators with the hard-core projection, we need the second factorial moment of the bosonic occupation at each original site.
	Write $c_{x_{A}(a)}:=a_{a}$ and $c_{x_{B}(b)}:=b_{b}$, and define the original-mode bosonic number operator $n_{x}^{\mathrm{b}}:=c_{x}^{\dagger}c_{x}$.
	This operator is distinct from the spin projection $n_{x}$ in Eq.~\eqref{eq:sharpness-spin-model}.
	Let $\bar{n}_{x}$ denote $\bar{n}_{A,a}$ or $\bar{n}_{B,b}$ according to the side containing $x$.
	Since $\ket{\Phi}$ is a product of two-mode squeezed Gaussian states, the rotated $A$ modes satisfy
	$\bra{\Phi}\widehat{a}_{k}^{\dagger}\widehat{a}_{\ell}\ket{\Phi} =\bar{n}_{k}\delta_{k\ell}$, while
	$\bra{\Phi}\widehat{a}_{k}\widehat{a}_{\ell}\ket{\Phi}  = \bra{\Phi}\widehat{a}_{k}^{\dagger}\widehat{a}_{\ell}^{\dagger}\ket{\Phi}=0$.
	Hence, Wick's theorem gives
	\begin{align}
		\bra{\Phi}\widehat{a}_{k}^{\dagger}\widehat{a}_{\ell}^{\dagger}\widehat{a}_{m}\widehat{a}_{p}\ket{\Phi}
		=\bar{n}_{k}\bar{n}_{\ell}\left(\delta_{km}\delta_{\ell p}+\delta_{kp}\delta_{\ell m}\right).
	\end{align}
	Substituting $a_{a}=\sum_{k}(U_{F})_{ak}\widehat{a}_{k}$ from Eq.~\eqref{eq:sharpness-mode-rotations} gives
	\begin{align}
		\bra{\Phi}(a_{a}^{\dagger})^{2}a_{a}^{2}\ket{\Phi}
		=2\left(\sum_{k}|(U_{F})_{ak}|^{2}\bar{n}_{k}\right)^{2}
		=2\bar{n}_{A,a}^{2}.
	\end{align}
	The same calculation with $b_{b}=\sum_{k}\overline{(V_{F})_{bk}} \, \widehat{b}_{k}$ gives the corresponding identity on $B$.
	Since $(c_{x}^{\dagger})^{2}c_{x}^{2}=n_{x}^{\mathrm{b}}(n_{x}^{\mathrm{b}}-\mathds{1})$, we obtain
	$\bra{\Phi}n_{x}^{\mathrm{b}}(n_{x}^{\mathrm{b}}-\mathds{1})\ket{\Phi}=2\bar{n}_{x}^{2}$.
	The spectral inequalities
	$\mathds{1}_{[2,\infty)}(n_{x}^{\mathrm{b}})\le n_{x}^{\mathrm{b}}(n_{x}^{\mathrm{b}}-\mathds{1})/2$
	and
	$\mathds{1}_{\{1\}}(n_{x}^{\mathrm{b}})\le n_{x}^{\mathrm{b}}$
	therefore give
	\begin{align}
		\bra{\Phi}\mathds{1}_{[2,\infty)}(n_{x}^{\mathrm{b}})\ket{\Phi}
		\le \bar{n}_{x}^{2},\qquad
		\bra{\Phi}\mathds{1}_{\{2\}}(n_{x}^{\mathrm{b}})\ket{\Phi}
		\le \bar{n}_{x}^{2},\qquad
		\bra{\Phi}\mathds{1}_{\{1\}}(n_{x}^{\mathrm{b}})\ket{\Phi}
		\le \bar{n}_{x}.
		\label{eq:sharpness-occupation-probabilities}
	\end{align}

	Let $\Pi_{x}^{\mathrm{hc}}:=\ket{0}_{x}\bra{0}_{x}+\ket{1}_{x}\bra{1}_{x}$ and $\Pi_{\mathrm{hc}} := \prod_{x}\Pi_{x}^{\mathrm{hc}}$ be the projection onto the original-mode hard-core subspace.
	A violation of the hard-core constraint requires $n_{x}^{\mathrm{b}}\ge2$ for at least one site $x$, and hence
	\begin{align}
		\mathds{1}-\Pi_{\mathrm{hc}} \le \sum_{x}\mathds{1}_{[2,\infty)}(n_{x}^{\mathrm{b}}).
	\end{align}
	Using Eq.~\eqref{eq:sharpness-occupation-probabilities} and the definition of $\mathfrak{c}$ in Eq.~\eqref{eq:sharpness-site-occupations}, we obtain
	\begin{align}
		1-\norm{\Pi_{\mathrm{hc}}\ket{\Phi}}^{2}  =\bra{\Phi}(\mathds{1}-\Pi_{\mathrm{hc}})\ket{\Phi}  \le\sum_{x}\bra{\Phi}\mathds{1}_{[2,\infty)}(n_{x}^{\mathrm{b}})\ket{\Phi}  \le\sum_{a}\bar{n}_{A,a}^{2}+\sum_{b}\bar{n}_{B,b}^{2}
		=\mathfrak{c}.
	\end{align}
	Therefore,
	\begin{align}
		\norm{\Pi_{\mathrm{hc}} \ket{\Phi}}^{2}\ge1-\mathfrak{c}. \label{eq:sharpness-hard-core-weight}
	\end{align}

	We next control the energy increase caused by the hard-core projection.
	To express the Bogoliubov annihilation operators in the original site basis, define
	$R_{F}:=U_{F}\operatorname{diag}(r_{k})V_{F}^{\dagger}$, which is obtained from the singular-value decomposition
	$F=U_{F}\operatorname{diag}(s_{k})V_{F}^{\dagger}$ by replacing each $s_{k}$ with $r_{k}$.
	Then define
	\begin{align}
		X_{a}&:=a_{a}+\sum_{b}(R_{F})_{ab}b_{b}^{\dagger},\qquad
		Y_{b}:=b_{b}+\sum_{a}(R_{F})_{ab}a_{a}^{\dagger}. \label{eq:sharpness-unnormalized-annihilators}
	\end{align}
	Using Eq.~\eqref{eq:sharpness-mode-rotations} and the definition of $R_{F}$, these operators can be written as
	$X_{a}=\sum_{k}(U_{F})_{ak}(\widehat{a}_{k}+r_{k}\widehat{b}_{k}^{\dagger})$
	and
	$Y_{b}=\sum_{k}\overline{(V_{F})_{bk}}(\widehat{b}_{k}+r_{k}\widehat{a}_{k}^{\dagger})$.
	Since the operators in parentheses annihilate the corresponding factors $\ket{\phi_{k}}$ in Eq.~\eqref{eq:sharpness-bosonic-ground-state}, we have
	$X_{a}\ket{\Phi}=Y_{b}\ket{\Phi}=0$.
	Moreover, Eq.~\eqref{eq:sharpness-bosonic-diagonal-block} gives
	\begin{align}
		H_{\mathrm{b}}-E_{\mathrm{b}}\mathds{1}
		=\sum_{k}\frac{\varepsilon_{k}}{1-r_{k}^{2}}
		\left[
		(\widehat{a}_{k}+r_{k}\widehat{b}_{k}^{\dagger})^{\dagger}
		(\widehat{a}_{k}+r_{k}\widehat{b}_{k}^{\dagger})
		+
		(\widehat{b}_{k}+r_{k}\widehat{a}_{k}^{\dagger})^{\dagger}
		(\widehat{b}_{k}+r_{k}\widehat{a}_{k}^{\dagger})
		\right].
	\end{align}
	Since
	$\varepsilon_{k}/(1-r_{k}^{2})=(1+\varepsilon_{k})/2\le1$
	and $U_{F}$ and $V_{F}$ are unitary,
	\begin{align}
		\sum_{k}
		(\widehat{a}_{k}+r_{k}\widehat{b}_{k}^{\dagger})^{\dagger}
		(\widehat{a}_{k}+r_{k}\widehat{b}_{k}^{\dagger})
		&=\sum_{a}X_{a}^{\dagger}X_{a}, \\
		\sum_{k}
		(\widehat{b}_{k}+r_{k}\widehat{a}_{k}^{\dagger})^{\dagger}
		(\widehat{b}_{k}+r_{k}\widehat{a}_{k}^{\dagger})
		&=\sum_{b}Y_{b}^{\dagger}Y_{b}.
	\end{align}
	Therefore,
	\begin{align}
		0\le H_{\mathrm{b}}-E_{\mathrm{b}}\mathds{1}
		\le\sum_{a}X_{a}^{\dagger}X_{a}+\sum_{b}Y_{b}^{\dagger}Y_{b}. \label{eq:sharpness-bosonic-upper-factorization}
	\end{align}

	Because $X_{a}$ and $Y_{b}$ annihilate $\ket{\Phi}$, the energy of the projected state is controlled by their commutators with $\Pi_{\mathrm{hc}}$.
	For any original bosonic annihilation operator $c_{x}$,
	\begin{align}
		[c_{x},\Pi_{\mathrm{hc}}]&=-\sqrt{2}\ket{1}_{x}\bra{2}_{x}\prod_{y\ne x}\Pi_{y}^{\mathrm{hc}}, \qquad 		[c_{x}^{\dagger},\Pi_{\mathrm{hc}}] =\sqrt{2}\ket{2}_{x}\bra{1}_{x}\prod_{y\ne x}\Pi_{y}^{\mathrm{hc}}. \label{eq:sharpness-hard-core-commutators}
	\end{align}
	The first commutator has range inside the hard-core subspace, whereas the second creates exactly one hard-core violation at $x$.
	Thus, the two types of ranges are orthogonal, and the ranges of $[c_{x}^{\dagger},\Pi_{\mathrm{hc}}]$ corresponding to distinct sites $x$ are mutually orthogonal.
	Taking the squared norms in Eq.~\eqref{eq:sharpness-hard-core-commutators} and using that the hard-core projections on the other sites are bounded by the identity gives
	$\norm{[c_{x},\Pi_{\mathrm{hc}}]\ket{\Phi}}^{2}
	\le2\bra{\Phi}\mathds{1}_{\{2\}}(n_{x}^{\mathrm{b}})\ket{\Phi}$
	and
	$\norm{[c_{x}^{\dagger},\Pi_{\mathrm{hc}}]\ket{\Phi}}^{2}
	\le2\bra{\Phi}\mathds{1}_{\{1\}}(n_{x}^{\mathrm{b}})\ket{\Phi}$.
	Hence, Eq.~\eqref{eq:sharpness-occupation-probabilities} gives
	$\norm{[c_{x},\Pi_{\mathrm{hc}}]\ket{\Phi}}^{2}\le2\bar{n}_{x}^{2}$
	and
	$\norm{[c_{x}^{\dagger},\Pi_{\mathrm{hc}}]\ket{\Phi}}^{2}\le2\bar{n}_{x}$.
	Moreover,
	$(R_{F}^{\dagger}R_{F})_{bb}
	=\sum_{k}|(V_{F})_{bk}|^{2}r_{k}^{2}
	\le\sum_{k}|(V_{F})_{bk}|^{2}r_{k}^{2}/(1-r_{k}^{2})
	=\bar{n}_{B,b}$,
	and similarly
	$(R_{F}R_{F}^{\dagger})_{aa}\le\bar{n}_{A,a}$.
	Using Eq.~\eqref{eq:sharpness-unnormalized-annihilators} and the orthogonality above, we therefore obtain
	\begin{align}
		\sum_{a}\norm{[X_{a},\Pi_{\mathrm{hc}}]\ket{\Phi}}^{2}
		&\le2\sum_{a}\bar{n}_{A,a}^{2}
		+2\sum_{a,b}|(R_{F})_{ab}|^{2}\bar{n}_{B,b}
		=2\sum_{a}\bar{n}_{A,a}^{2}
		+2\sum_{b}(R_{F}^{\dagger}R_{F})_{bb}\bar{n}_{B,b}
		\le2\mathfrak{c}, \label{eq:sharpness-X-commutator-bound}\\
		\sum_{b}\norm{[Y_{b},\Pi_{\mathrm{hc}}]\ket{\Phi}}^{2}
		&\le2\sum_{b}\bar{n}_{B,b}^{2}
		+2\sum_{a,b}|(R_{F})_{ab}|^{2}\bar{n}_{A,a}
		=2\sum_{b}\bar{n}_{B,b}^{2}
		+2\sum_{a}(R_{F}R_{F}^{\dagger})_{aa}\bar{n}_{A,a}
		\le2\mathfrak{c}. \label{eq:sharpness-Y-commutator-bound}
	\end{align}

	Since $\operatorname{Ran}\Pi_{\mathrm{hc}}=\operatorname{Ran}\mathsf{J}$, the normalized vector
	$\Pi_{\mathrm{hc}}\ket{\Phi}/\norm{\Pi_{\mathrm{hc}}\ket{\Phi}}$
	corresponds to a valid normalized spin trial state.
	Before applying the variational principle, Eq.~\eqref{eq:sharpness-bosonic-upper-factorization} and
	$X_{a}\ket{\Phi}=Y_{b}\ket{\Phi}=0$ give
	\begin{align}
		\bra{\Phi}\Pi_{\mathrm{hc}}(H_{\mathrm{b}}-E_{\mathrm{b}}\mathds{1})\Pi_{\mathrm{hc}}\ket{\Phi}
		&\le\sum_{a}\norm{X_{a}\Pi_{\mathrm{hc}}\ket{\Phi}}^{2}
		+\sum_{b}\norm{Y_{b}\Pi_{\mathrm{hc}}\ket{\Phi}}^{2}\notag\\
		&=\sum_{a}\norm{[X_{a},\Pi_{\mathrm{hc}}]\ket{\Phi}}^{2}
		+\sum_{b}\norm{[Y_{b},\Pi_{\mathrm{hc}}]\ket{\Phi}}^{2}\notag\\
		&\le4\mathfrak{c},
	\end{align}
	where the last inequality follows from Eqs.~\eqref{eq:sharpness-X-commutator-bound} and \eqref{eq:sharpness-Y-commutator-bound}.
	Applying the variational principle to the normalized projected state and using Eq.~\eqref{eq:sharpness-hard-core-weight}, we therefore obtain
	\begin{align}
		E_{\mathrm{s}}-E_{\mathrm{b}}  \le \frac{\bra{\Phi}\Pi_{\mathrm{hc}}(H_{\mathrm{b}}-E_{\mathrm{b}}\mathds{1})\Pi_{\mathrm{hc}}\ket{\Phi}}
		{\norm{\Pi_{\mathrm{hc}}\ket{\Phi}}^{2}} \le \frac{4\mathfrak{c}}{1-\mathfrak{c}}
		\le\frac{4\mathfrak{R}}{1-\mathfrak{R}}.
		\label{eq:sharpness-ground-energy-projection-bound}
	\end{align}

	We now derive a lower bound on the first excited spin energy.
	By Eq.~\eqref{eq:sharpness-bosonic-diagonal-block}, the unique bosonic ground state is $\ket{\Phi}$ and the first excitation energy above it is $\sqrt{1-f^{2}}$.
	Since $\mathsf{J}$ is an isometry, $H_{F}=\mathsf{J}^{\dagger}H_{\mathrm{b}}\mathsf{J}$ is the compression of $H_{\mathrm{b}}$ to the subspace $\operatorname{Ran}\mathsf{J}$, and the min--max principle gives
	\begin{align}
		E_{1}(H_{F}) \ge E_{1}(H_{\mathrm{b}}) = E_{\mathrm{b}}+\sqrt{1-f^{2}}. \label{eq:sharpness-first-excited-minmax}
	\end{align}
	Subtracting Eq.~\eqref{eq:sharpness-ground-energy-projection-bound} from Eq.~\eqref{eq:sharpness-first-excited-minmax} yields
	\begin{align}
		\Delta_{F}
		=E_{1}(H_{F})-E_{\mathrm{s}}
		\ge\sqrt{1-f^{2}}-\frac{4\mathfrak{R}}{1-\mathfrak{R}},
	\end{align}
	which proves Eq.~\eqref{eq:sharpness-finite-gap}.
\end{proof}

\subsection{Proof of the regular-cut obstruction and optimality corollary}

\begin{proof}[Proof of Proposition~\ref{prop:regular-cut-obstruction}]
	Let $\ket{\Omega_{F}}$ be any normalized pure ground state of Eq.~\eqref{eq:sharpness-spin-model}.
	By Lemma~\ref{lem:sharpness-entropy-from-defect},
	\begin{align}
		S_{A}(\ket{\Omega_{F}})
		\ge\frac{D-\alpha}{256}\,W\log \ell_{\Lambda}. \label{eq:sharpness-final-entropy-combination}
	\end{align}
	Combining Eq.~\eqref{eq:sharpness-final-entropy-combination} with the three regimes for $W$ in Eq.~\eqref{eq:sharpness-HS-bound} gives
	\begin{align}
		S_{A}(\ket{\Omega_{F}})
		\gtrsim
		\begin{cases}
			\log \ell_{\Lambda},&0\le\alpha<(D+1)/2,\\
			(\log \ell_{\Lambda})^{2},&\alpha=(D+1)/2,\\
			\ell_{\Lambda}^{2\alpha-D-1}\log \ell_{\Lambda},&(D+1)/2<\alpha<D.
		\end{cases}
	\end{align}
	Since $n=\ell_{\Lambda}^{D}$, one has $\log \ell_{\Lambda}=D^{-1}\log n$ and $\ell_{\Lambda}^{2\alpha-D-1}=n^{(2\alpha-D-1)/D}$.
	Therefore
	\begin{align}
		S_{A}(\ket{\Omega_{F}})
		\gtrsim
		\begin{cases}
			\log n,&0\le\alpha<(D+1)/2,\\
			(\log n)^{2},&\alpha=(D+1)/2,\\
			n^{(2\alpha-D-1)/D}\log n,&(D+1)/2<\alpha<D,
		\end{cases}
	\end{align}
	which proves Eq.~\eqref{eq:sharpness-main-lower}.
	No uniqueness assumption has been used, so the bound holds for every pure ground state.
	
	Now assume $0\le\alpha<(3D+1)/4$.
	Lemma~\ref{lem:sharpness-matrix} gives $f=o(1)$ and $\mathfrak{R}=o(1)$, so $f\le1/2$ and $\mathfrak{R}<1$ for all sufficiently large even $\ell_{\Lambda}$.
	Lemma~\ref{lem:sharpness-unique} therefore applies, and Eq.~\eqref{eq:sharpness-finite-gap} gives
	\begin{align}
		\Delta_{F} \ge \sqrt{1-f^{2}}-\frac{4\mathfrak{R}}{1-\mathfrak{R}} = 1-o(1). \label{eq:sharpness-final-gap}
	\end{align}
	In particular, the right-hand side is positive for all sufficiently large even $\ell_{\Lambda}$, so the same lemma implies that the ground state is unique.
	The same estimate proves Eq.~\eqref{eq:sharpness-asymptotic-gap}.
\end{proof}

\begin{proof}[Proof of Corollary~\ref{cor:regular-cut-threshold-sharp}]
	First suppose that $D\ge1$ and $0\le\alpha<D/2$.
	Since $\alpha<D/2<(3D+1)/4$, Proposition~\ref{prop:regular-cut-obstruction} gives a unique ground state for all sufficiently large even $\ell_{\Lambda}$, with
	\begin{align}
		\Delta_{F}\ge1-o(1),\qquad
		S_{A}(\ket{\Omega_{F}})\ge c_{D,\alpha}\log n.
	\end{align}
	In particular, the gap is bounded below by a positive constant for all sufficiently large $\ell_{\Lambda}$.
	On the other hand, Theorem~\ref{thm:single-scale-arbitrary-cut} applies to every bipartition and gives
	\begin{align}
		S_{A}(\ket{\Omega_{F}})\le C_{D,\alpha}\log n.
	\end{align}
	Thus, for the regular codimension-one cut considered here,
	\begin{align}
		c_{D,\alpha}\log n
		\le S_{A}(\ket{\Omega_{F}})
		\le C_{D,\alpha}\log n,
	\end{align}
	so the logarithmic scaling is optimal throughout $0\le\alpha<D/2$.

	Assume now that $D\ge2$ and $\alpha=(D+1)/2$.
	Since $(D+1)/2<(3D+1)/4$ for $D>1$, Proposition~\ref{prop:regular-cut-obstruction} again gives a unique ground state for all sufficiently large even $\ell_{\Lambda}$ and $\Delta_{F} \ge 1 - o(1)$.
	After increasing the starting size if necessary, one therefore has $\Delta_{F} \ge 1/2$.
	The critical case of Eq.~\eqref{eq:sharpness-main-lower} gives
	\begin{align}
		S_{A}(\ket{\Omega_{F}})
		\ge c_{D}\,(\log n)^{2},
	\end{align}
	which is Eq.~\eqref{eq:sharpness-critical-lower}.
	
	Next suppose that $(D+1)/2<\alpha<(3D+1)/4$.
	Proposition~\ref{prop:regular-cut-obstruction} gives a unique ground state with $\Delta_{F}\ge1-o(1)$ for all sufficiently large even $\ell_{\Lambda}$, and the corresponding case of Eq.~\eqref{eq:sharpness-main-lower} gives
	\begin{align}
		S_{A}(\ket{\Omega_{F}})
		\ge c_{D,\alpha}\,n^{(2\alpha-D-1)/D}\log n. \label{eq:sharpness-corollary-polynomial}
	\end{align}
	Here $(2\alpha-D-1)/D>0$, so the right-hand side grows polynomially in $n$, up to the additional logarithmic factor.
	This proves Eq.~\eqref{eq:sharpness-polynomial-lower}.
	
	For every $\alpha>(D+1)/2$ sufficiently close to $(D+1)/2$, Proposition~\ref{prop:regular-cut-obstruction} gives a unique uniformly gapped example with polynomially growing regular-cut entanglement.
	In contrast, Theorem~\ref{thm:regular-cut-polylog-entropy} gives a universal polylogarithmic entropy bound for regular codimension-one cuts throughout $D/2<\alpha<(D+1)/2$.
	Therefore, $\alpha=(D+1)/2$ is the optimal threshold for such a universal polylogarithmic bound under the present microscopic assumptions.
\end{proof}

The physical Hamiltonian in Eq.~\eqref{eq:sharpness-spin-model} remains the qubit Hamiltonian throughout the construction.
The bosonic system serves only as an auxiliary local isometric representation and variational reference.
The entropy lower bound is obtained from the spin-space energy defect together with the local entanglement-extraction argument above.

For $0\le\alpha<D/2$, the logarithmic lower bound from this construction matches the universal logarithmic upper bound of Theorem~\ref{thm:single-scale-arbitrary-cut}, so the logarithmic scaling is optimal in this range.
Moreover, Theorem~\ref{thm:regular-cut-polylog-entropy} gives a universal polylogarithmic bound for regular codimension-one cuts throughout $D/2<\alpha<(D+1)/2$, while Proposition~\ref{prop:regular-cut-obstruction} provides unique uniformly gapped examples with polynomially growing regular-cut entanglement for exponents arbitrarily close to $(D+1)/2$ from above.
Thus, $\alpha=(D+1)/2$ is the optimal threshold for a universal polylogarithmic entropy bound across regular codimension-one cuts.
The precise optimal entropy scaling at and above this threshold is not determined by the present bounds.

\section{Numerical illustration with a quadratic fermionic pairing model}
\label{sec:numerical-fermionic-illustration}

The results above are formulated for finite-dimensional spin or qudit systems with Kac-normalized two-body interactions.
Motivated by the Fourier-paired spin family in Section~\ref{sec:regular-cut-sharpness}, we now consider a quadratic fermionic analogue obtained by replacing the spin raising and lowering operators in the pairing interaction by fermionic creation and annihilation operators, while keeping the same Kac normalization, power-law geometry, Fourier phases, and codimension-one bipartition.
The fermionic model is not identical to the spin Hamiltonians considered in the rigorous analysis, and the discussion below should therefore not be interpreted as a numerical proof of the spin-system statements.
Its purpose is instead to provide a simple solvable setting in which the same geometric structure can be followed to substantially larger system sizes.
In this sense, the numerics provide a useful indication of the finite-size behavior suggested by the threshold $\alpha=(D+1)/2$.

\subsection{Quadratic model}

We consider an open $D$-dimensional hypercubic lattice $\Lambda=\{0,\ldots,\ell_{\Lambda}-1\}^{D}$ with even side length $\ell_{\Lambda}$ and $n=\ell_{\Lambda}^{D}$, where $x=(x_{1},\ldots,x_{D})\in\Lambda$ denotes the lattice coordinate and $x_{1}$ is its first coordinate.
The lattice is divided into two equal geometric halves,
\begin{align}
	A :=\{x\in\Lambda:x_{1}<\ell_{\Lambda}/2\}, \qquad B :=\{x\in\Lambda:x_{1}\ge \ell_{\Lambda}/2\}, \qquad |A|=|B|=:N_{\mathrm{h}}=\frac{n}{2} .
\end{align}
Thus, the entanglement cut is a regular codimension-one cut of the same type considered in the geometric entropy analysis above.
We enumerate the sites in each half in the standard lexicographic order of their coordinates and write the corresponding sites as $x_{A}(a)\in A$ and $x_{B}(b)\in B$, with $a,b=0,\ldots,N_{\mathrm{h}}-1$.
This fixed ordering is used both in the Fourier phase and in all numerical data below.
We study the pairing Hamiltonian
\begin{align}
	H_{\mathrm{fer}} &= \mu\sum_{x\in\Lambda}c_{x}^{\dagger} c_{x} + \gamma_{\mathrm{sc}} \sum_{a,b=0}^{N_{\mathrm{h}}-1} \left( F_{ab}c_{x_{A}(a)}^{\dagger} c_{x_{B}(b)}^{\dagger} + \overline{F_{ab}}c_{x_{B}(b)} c_{x_{A}(a)} \right), \qquad F_{ab} := \frac{J}{n^{1-\alpha/D}} \frac{\exp(2\pi iab/N_{\mathrm{h}})}{r_{ab}^{\alpha}} . \label{eq:numerical-fermion}
\end{align}
For $J=\gamma_{\mathrm{sc}}=\mu=1$, this is the direct fermionic counterpart of Eq.~\eqref{eq:sharpness-spin-model}, with the same coefficient matrix $F$.
Here $r_{ab}$ is the Manhattan distance between the two sites $x_{A}(a)$ and $x_{B}(b)$.
The factor $n^{-(1-\alpha/D)}$ is exactly the Kac normalization used in the spin problem.
We use the unit-modulus phase factor $e^{2\pi iab/N_{\mathrm{h}}}$ rather than the normalized Fourier-matrix element $N_{\mathrm{h}}^{-1/2}e^{2\pi iab/N_{\mathrm{h}}}$.
Thus, the Fourier phase does not introduce any additional system-size normalization beyond the Kac factor $n^{-(1-\alpha/D)}$.
We use $J=\gamma_{\mathrm{sc}}=\mu=1$ throughout the data shown below.

\subsection{Analytic gap and entropy upper bounds}

The quadratic structure allows both the half-cut entropy and the spectral gap to be expressed exactly in terms of the singular values of the pairing matrix.
Let $F=U\,\operatorname{diag}(s_{1},\ldots,s_{N_{\mathrm{h}}})V^{\dagger}$ be a singular-value decomposition.
The corresponding mode rotations act separately within $A$ and within $B$ and therefore preserve the entanglement across the geometric half cut.
The Hamiltonian then reduces to independent two-mode pairing blocks.
For each singular value $s_{k}$, define
\begin{align}
	E_{k} := \sqrt{\mu^{2}+\gamma_{\mathrm{sc}}^{2}s_{k}^{2}}, \qquad p_{k} := \frac{1}{2}\left(1-\frac{\mu}{E_{k}}\right).
\end{align}
The half-cut ground-state entanglement entropy is therefore~\cite{peschel2003calculation}
\begin{align}
	S_{A} &= \sum_{k=1}^{N_{\mathrm{h}}} \left[ -p_{k}\log p_{k} -(1-p_{k})\log(1-p_{k}) \right], \qquad \Delta_{\mathrm{fer}} = \min_{k} E_{k}. \label{eq:numerical-fermion-gap}
\end{align}
Here $\Delta_{\mathrm{fer}}$ denotes the spectral gap on the full fermionic Fock space; no fixed-parity-sector restriction is imposed.
For a fixed singular value $s_{k}$, the corresponding two-mode block has even-parity eigenvalues $\mu\pm E_{k}$ and two odd-parity eigenvalues equal to $\mu$.
Hence its excitation gap is exactly $E_{k}$, and therefore
\begin{align}
	\Delta_{\mathrm{fer}} = \min_{k}\sqrt{\mu^{2}+\gamma_{\mathrm{sc}}^{2}s_{k}^{2}} \ge \mu. \label{eq:numerical-analytic-gap-bound}
\end{align}
In particular, for $\mu>0$ every finite-size pairing block has a unique ground state, and hence the full quadratic Hamiltonian has a unique ground state.
For the parameters used below, $\mu=1$, so $\Delta_{\mathrm{fer}}\ge1$ holds exactly at every finite size.

For the entropy, the total singular-value weight is the Frobenius norm
\begin{align}
	\Sigma_{n} := \sum_{k=1}^{N_{\mathrm{h}}}s_{k}^{2} = \norm{F}_{\mathrm{F}}^{2} = \frac{J^{2}}{n^{2(1-\alpha/D)}} \sum_{a,b=0}^{N_{\mathrm{h}}-1}r_{ab}^{-2\alpha}. \label{eq:numerical-Frobenius-weight}
\end{align}
For the geometric half cut, a site in $A$ that belongs to a cross-cut pair at distance $r$ must lie within distance $r$ of the interface, giving $\mathcal{O}(\ell_{\Lambda}^{D-1}r)$ choices.
For each such site, there are at most $\mathcal{O}(r^{D-1})$ sites at distance $r$.
Hence the number of cross-cut pairs at distance $r$ is at most $C_{D}\ell_{\Lambda}^{D-1}r^{D}$, and therefore
\begin{align}
	\sum_{a,b=0}^{N_{\mathrm{h}}-1}r_{ab}^{-2\alpha} &\le C_{D}\ell_{\Lambda}^{D-1}\sum_{r=1}^{D\ell_{\Lambda}}r^{D-2\alpha}  \le C_{D,\alpha}\begin{cases}
		\ell_{\Lambda}^{2D-2\alpha},&0\le\alpha<(D+1)/2,\quad D\ge1,\\
		\ell_{\Lambda}^{D-1}\log(e\ell_{\Lambda}),&\alpha=(D+1)/2,\quad D\ge2,\\
		\ell_{\Lambda}^{D-1},&(D+1)/2<\alpha<D,\quad D\ge2.
	\end{cases} \label{eq:numerical-cross-cut-square-sum}
\end{align}
Using $n=\ell_{\Lambda}^{D}$ in Eq.~\eqref{eq:numerical-Frobenius-weight} gives
\begin{align}
	\Sigma_{n} &\le C_{D,\alpha,J}\begin{cases}
		1,&0\le\alpha<(D+1)/2,\quad D\ge1,\\
		\log(en),&\alpha=(D+1)/2,\quad D\ge2,\\
		n^{(2\alpha-D-1)/D},&(D+1)/2<\alpha<D,\quad D\ge2.
	\end{cases} \label{eq:numerical-Frobenius-scaling}
\end{align}

We next convert this square-sum estimate into an entropy bound.
Set $x_{k}:=\gamma_{\mathrm{sc}}^{2}s_{k}^{2}/\mu^{2}$.
Since $(1+x)^{-1/2}\ge1-x/2$ for $x\ge0$,
\begin{align}
	p_{k} = \frac{1}{2}\left(1-\frac{1}{\sqrt{1+x_{k}}}\right)
	\le \frac{x_{k}}{4}
	= \frac{\gamma_{\mathrm{sc}}^{2}}{4\mu^{2}}s_{k}^{2}.
\end{align}
Writing $Q_{n}:=\sum_{k=1}^{N_{\mathrm{h}}}p_{k}$, we therefore have $Q_{n}\le \frac{\gamma_{\mathrm{sc}}^{2}}{4\mu^{2}}\Sigma_{n}$.
Since $S_{A}=\sum_{k=1}^{N_{\mathrm{h}}}\hbin(p_{k})$ and the binary entropy is concave, Jensen's inequality gives
\begin{align}
	S_{A} \le N_{\mathrm{h}}\hbin\left(\frac{Q_{n}}{N_{\mathrm{h}}}\right).
\end{align}
Moreover, Eq.~\eqref{eq:numerical-Frobenius-scaling} gives $\Sigma_{n}=o(N_{\mathrm{h}})$ for every fixed $\alpha<D$, and hence $Q_{n}/N_{\mathrm{h}}\le1/2$ for all sufficiently large $n$.
Using $\hbin(q)\le q\log(e/q)$ then gives
\begin{align}
	S_{A}\le Q_{n}\log\frac{eN_{\mathrm{h}}}{Q_{n}}.
\end{align}
The function $q\mapsto q\log(eN_{\mathrm{h}}/q)$ is increasing for $0<q<N_{\mathrm{h}}$.
Thus, with $C:=\gamma_{\mathrm{sc}}^{2}/(4\mu^{2})$ and $Q_{n}\le C\Sigma_{n}$,
\begin{align}
	S_{A}
	\le Q_{n}\log\frac{eN_{\mathrm{h}}}{Q_{n}}
	\le C\Sigma_{n}\log\frac{eN_{\mathrm{h}}}{C\Sigma_{n}}.
\end{align}
The half cut contains $\ell_{\Lambda}^{D-1}$ nearest-neighbor cross-cut pairs, so Eq.~\eqref{eq:numerical-Frobenius-weight} gives $\Sigma_{n}\ge J^{2}n^{(2\alpha-D-1)/D}$.
Since $N_{\mathrm{h}}=n/2$, it follows that $\log(eN_{\mathrm{h}}/(C\Sigma_{n}))=\mathcal{O}(\log(en))$.
Therefore,
\begin{align}
	S_{A}=\mathcal{O}\left(\Sigma_{n}\log(en)\right).
\end{align}
Combining this with Eq.~\eqref{eq:numerical-Frobenius-scaling}, we obtain the analytic upper bounds
\begin{align}
	S_{A} &= \begin{cases}
		\mathcal{O}(\log(en)),&0\le\alpha<(D+1)/2,\quad D\ge1,\\
		\mathcal{O}((\log(en))^{2}),&\alpha=(D+1)/2,\quad D\ge2,\\
		\mathcal{O}\left(n^{(2\alpha-D-1)/D}\log(en)\right),&(D+1)/2<\alpha<D,\quad D\ge2.
	\end{cases} \label{eq:numerical-analytic-entropy-upper}
\end{align}
These bounds determine the same geometric threshold $\alpha=(D+1)/2$ that appears in the square-sum estimates of the spin-system analysis.
They are upper bounds only; the singular-value argument above does not by itself show that the logarithmic factor is saturated.
The numerical data below therefore test how closely the actual entropy follows these analytic scaling forms.

Finally, Eq.~\eqref{eq:numerical-Frobenius-scaling} also implies that the gap approaches its lower bound.
Indeed,
\begin{align}
	\min_{k}s_{k}^{2} \le \frac{1}{N_{\mathrm{h}}}\sum_{k=1}^{N_{\mathrm{h}}}s_{k}^{2} = \frac{\Sigma_{n}}{N_{\mathrm{h}}} = o(1),
\end{align}
and hence
\begin{align}
	\mu \le \Delta_{\mathrm{fer}} \le \sqrt{\mu^{2}+\gamma_{\mathrm{sc}}^{2}\Sigma_{n}/N_{\mathrm{h}}} = \mu+o(1). \label{eq:numerical-gap-convergence}
\end{align}
Thus, for $\mu=1$, the finite-size gap is exactly bounded below by one and converges to one as $n\to\infty$.

\subsection{Numerical parameters}

For $D=1$ we use $n=16,\ldots,16384$ along powers of two.
For $D=2$, we use $\ell_{\Lambda}=10,20,\ldots,150$, corresponding to $n=100,400,\ldots,22500$.
For $D=3$ we use $\ell_{\Lambda}=4,6,\ldots,28$, corresponding to $n=64,\ldots,21952$.
For $D=2,3$, the values of $\alpha$ are chosen in both ranges $0\le\alpha<(D+1)/2$ and $(D+1)/2<\alpha<D$, while excluding $\alpha=(D+1)/2$.
For $D=1$, all sampled values satisfy $\alpha<(D+1)/2=1$.

\subsection{Overall entropy growth and the spectral gap}

Supplementary Figure~\ref{fig:numerics-entropy-all} shows the half-vs-half ground-state entanglement entropy for all values of $\alpha$ included in the scan.
The entropy grows monotonically with the system size, and the growth becomes visibly stronger as $\alpha$ approaches $D$.
The log--log representation is useful for displaying all regimes in a common figure: for small and intermediate $\alpha$, the curves grow very slowly, whereas for $(D+1)/2<\alpha<D$ in $D=2,3$ a clear algebraic component develops.

\begin{figure}[h]
	\centering
	\includegraphics[width=\textwidth]{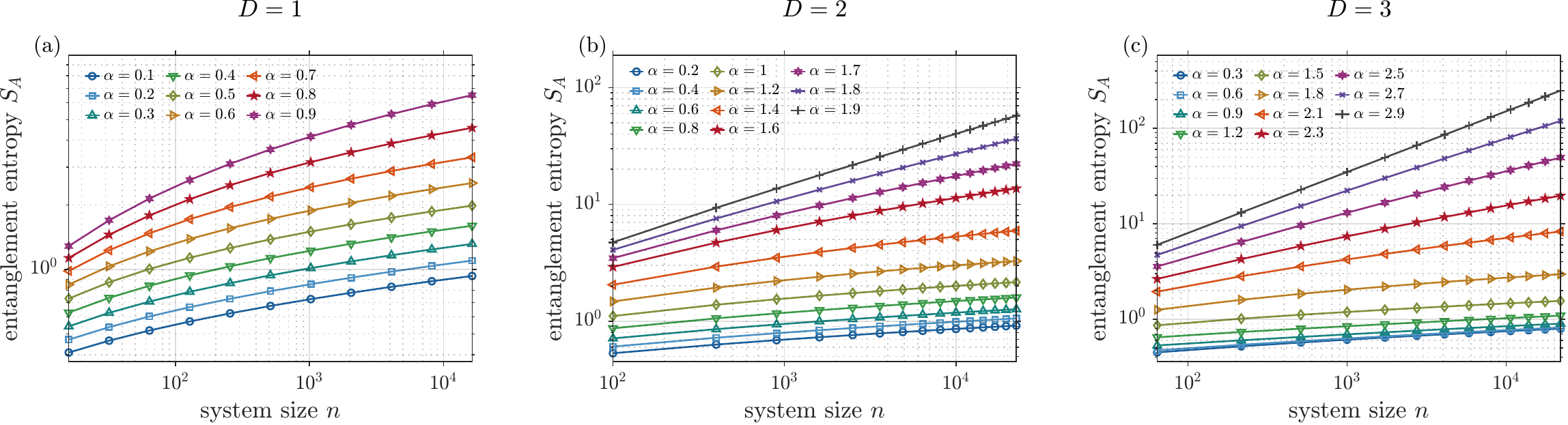}
	\caption{
		Ground-state entanglement entropy across the geometric half cut for $D=1,2,3$.
		All values of $\alpha$ used in the numerical scan are shown.
		The axes are logarithmic in both $n$ and $S_{A}$.
	}
	\label{fig:numerics-entropy-all}
\end{figure}

Supplementary Figure~\ref{fig:numerics-gap} shows the corresponding spectral gap.
In all three dimensions, the gap remains above $\mu=1$ and numerically approaches this lower bound as the system size increases.
This behavior is consistent with Eqs.~\eqref{eq:numerical-analytic-gap-bound} and \eqref{eq:numerical-gap-convergence}, which give $\Delta_{\mathrm{fer}}\ge1$ at every finite size and $\Delta_{\mathrm{fer}}\to1$ as $n\to\infty$.
The numerical data therefore probe entanglement growth while the ground state remains unique and the gap stays uniformly open.

\begin{figure}[h]
	\centering
	\includegraphics[width=\textwidth]{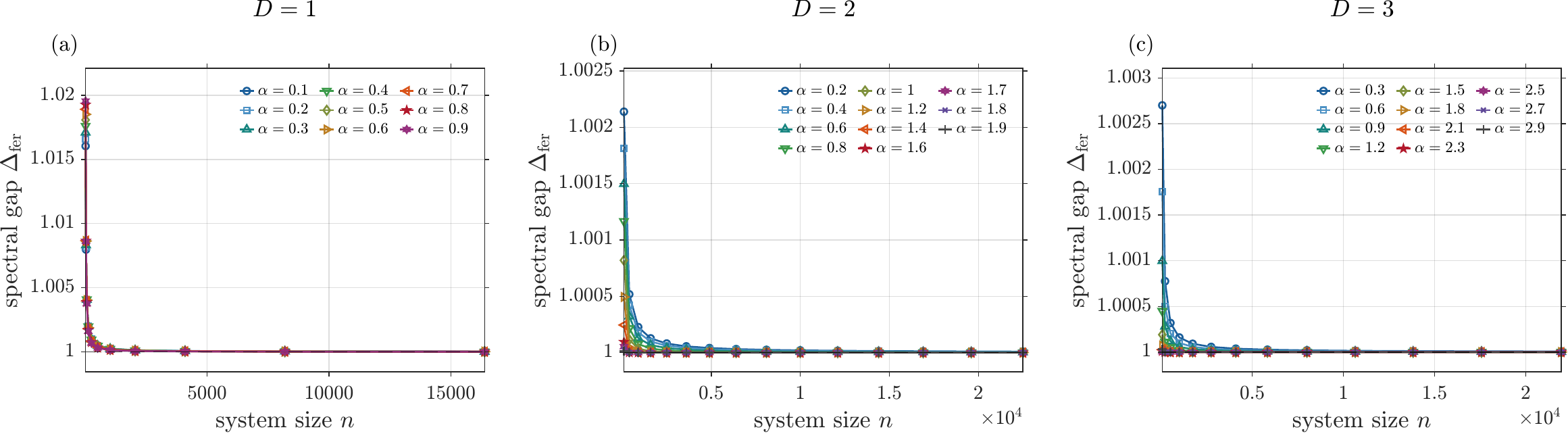}
	\caption{
		Spectral gap of the quadratic fermionic model for the same parameter sets as in Supplementary Fig.~\ref{fig:numerics-entropy-all}.
		The horizontal dashed line indicates $\Delta_{\mathrm{fer}}=\mu=1$.
		The finite-size gaps numerically approach this value from above, in agreement with the analytic bounds $\Delta_{\mathrm{fer}}\ge1$ and $\Delta_{\mathrm{fer}}\to1$.
	}
	\label{fig:numerics-gap}
\end{figure}

\subsection{Logarithmic behavior for \texorpdfstring{$0\le\alpha<(D+1)/2$}{0 <= alpha < (D+1)/2}}

To resolve the slowly growing regime more directly, Supplementary Fig.~\ref{fig:numerics-subthreshold} plots $S_{A}$ against $\log n$ and retains only the range $0\le\alpha<(D+1)/2$.
In this entire range, Eq.~\eqref{eq:numerical-analytic-entropy-upper} gives the analytic upper bound $S_{A}=\mathcal{O}(\log n)$.
Over the accessible sizes, the curves are close to linear in $\log n$.
Thus, the data are consistent with
\begin{align}
	S_{A}(n) \simeq a_{D,\alpha}\log n+b_{D,\alpha}
\end{align}
throughout this range.

\begin{figure}[h]
	\centering
	\includegraphics[width=\textwidth]{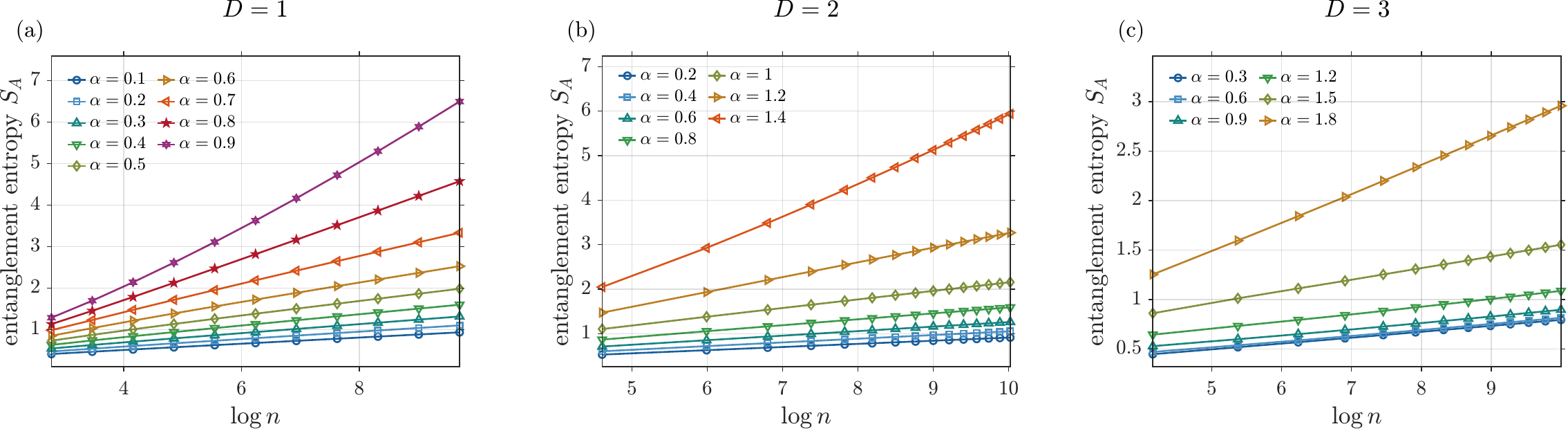}
	\caption{
		Half-cut entanglement entropy plotted against $\log n$, restricted to $0\le\alpha<(D+1)/2$.
		The analytic singular-value estimate gives $S_{A}=\mathcal{O}(\log n)$ throughout this regime, while the near-linear finite-size behavior is consistent with logarithmic entropy growth over the sizes shown.
	}
	\label{fig:numerics-subthreshold}
\end{figure}

An important point is that this near-linearity is not restricted to $\alpha<D/2$.
It persists throughout the nontrivial interval $D/2<\alpha<(D+1)/2$, as seen for $0.6\le\alpha\le0.9$ in $D=1$, $\alpha=1.2,1.4$ in $D=2$, and $\alpha=1.8$ in $D=3$.
The rigorous spin-system result in this interval is a polylogarithmic bound for geometrically regular cuts, whereas the present solvable fermionic model admits the sharper analytic upper bound $\mathcal{O}(\log n)$.
The numerical data suggest that the entropy growth in this structured model is close to saturating this logarithmic upper behavior, while remaining fully consistent with the spin-system bound.

\subsection{Behavior for \texorpdfstring{$(D+1)/2<\alpha<D$}{(D+1)/2 < alpha < D}}

For $D\ge2$, the interval $(D+1)/2<\alpha<D$ is nonempty.
In this regime Eq.~\eqref{eq:numerical-analytic-entropy-upper} gives the analytic upper scaling
\begin{align}
	S_{A}(n) &= \mathcal{O}\left(n^{p(\alpha,D)}\log n\right), \qquad p(\alpha,D) := \frac{2\alpha-D-1}{D}. \label{eq:numerical-guide-exponent}
\end{align}
The same exponent $p(\alpha,D)$ appears in the entropy lower bound for the Fourier-paired spin family in Proposition~\ref{prop:regular-cut-obstruction}, where the scaling is bounded below by a constant times $n^{p(\alpha,D)}\log n$.
It also coincides with the power appearing in the boundary-sensitive dimer estimate Eq.~\eqref{eq:dimer-regular-cut-entropy} for a regular codimension-one cut.
We therefore use this analytically determined scaling as the guide for the numerical comparison below.

Supplementary Figure~\ref{fig:numerics-superthreshold} shows the comparison.
In $D=2$ we use $(\alpha,p)=(1.6,0.1),(1.7,0.2),(1.8,0.3),(1.9,0.4)$.
In $D=3$ we use $(\alpha,p)=(2.1,1/15),(2.3,1/5),(2.5,1/3),(2.7,7/15),(2.9,3/5)$.
The dotted curves have the form $c_{D,\alpha}n^{p(\alpha,D)}\log n$.

The exponent $p(\alpha,D)$ is fixed analytically by Eq.~\eqref{eq:numerical-guide-exponent}; it is not obtained by fitting the slope of the entropy data.
Only the overall multiplicative normalization $c_{D,\alpha}$ is chosen from the large-$n$ data in order to compare the shapes of the curves.

\begin{figure}[h]
	\centering
	\includegraphics[width=0.9\textwidth]{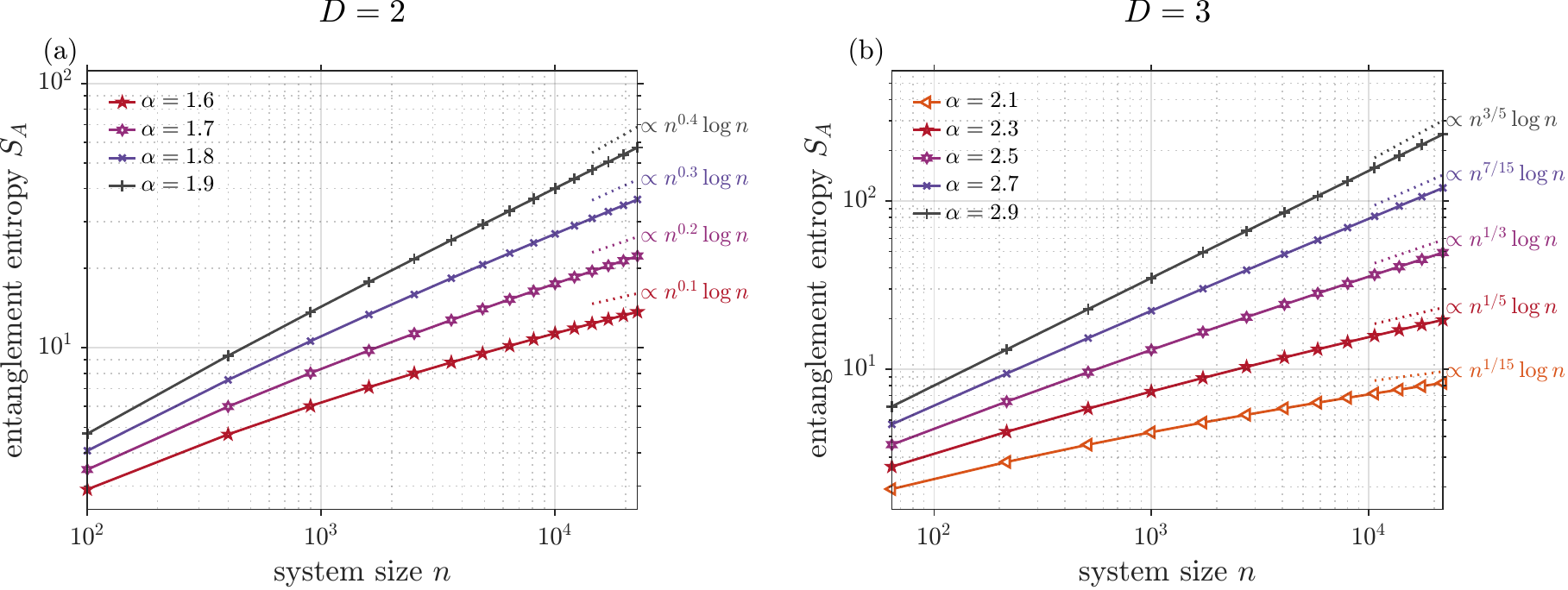}
	\caption{
		Entanglement entropy for $(D+1)/2<\alpha<D$ in $D=2$ and $D=3$.
		The dotted curves are fixed-exponent guides proportional to $n^{p(\alpha,D)}\log n$, with $p(\alpha,D)=(2\alpha-D-1)/D$ as determined by the analytic singular-value bound.
		Only an overall multiplicative normalization is set from the largest-size data; the exponent itself is not fitted.
	}
	\label{fig:numerics-superthreshold}
\end{figure}

For the largest system sizes shown, the data follow these guides closely.
The comparison therefore suggests that the analytic upper scaling is close to saturation in this quadratic model: below the threshold, the half-cut entropy is numerically compatible with logarithmic growth, including throughout $D/2<\alpha<(D+1)/2$, whereas for $(D+1)/2<\alpha<D$ the data are well described by $n^{p(\alpha,D)}\log n$ with the analytically determined positive exponent $p(\alpha,D)$.
This behavior is qualitatively consistent with the role of $(D+1)/2$ in the geometric square-sum estimates and with the loss of the closed polylogarithmic induction beyond that value.

\end{document}